\documentclass[11pt]{article} %%%%%%%%%%%%%%%%%%%%%%%%%%%%%%%%%%%%%%%%%%%%%%%%%

\usepackage{textcomp, amssymb}  % additional symbols (there are more packages)
\usepackage{anysize}            % margin package sets tighter margins
\usepackage{amsfonts}
\usepackage{amsmath}
\usepackage{dsfont}
\usepackage{MnSymbol}
\usepackage{mathtools}
\usepackage{graphicx}
\usepackage{epsfig}
\usepackage{epstopdf}
\usepackage{lmerge}
\usepackage{tikz}
\usepackage{amsthm}
\usepackage{ifpdf}
\usetikzlibrary{calc,decorations.pathmorphing,shapes,graphs}

\newcounter{sarrow}
\newcommand\xrsquigarrow[1]{%
\stepcounter{sarrow}%
\mathrel{\begin{tikzpicture}[baseline= {( $ (current bounding box.south) + (0,-0.5ex) $ )}]
\node[inner sep=.5ex] (\thesarrow) {$\scriptstyle #1$};
\path[draw,<-,decorate,
  decoration={zigzag,amplitude=0.7pt,segment length=1.2mm,pre=lineto,pre length=4pt}]
    (\thesarrow.south east) -- (\thesarrow.south west);
    \end{tikzpicture}}%
}

\newcounter{sarrow1}
\newcommand\xnrsquigarrow[1]{%
\stepcounter{sarrow1}%
\mathrel{\begin{tikzpicture}[baseline= {( $ (current bounding box.south) + (0,-0.5ex) $ )}]
\node[inner sep=.5ex] (\thesarrow) {$\scriptstyle #1$};
\path[draw,<-,decorate,
  decoration={zigzag,amplitude=0.7pt,segment length=1.2mm,pre=lineto,pre length=4pt}]
    (\thesarrow1.south east) -- (\thesarrow1.south west);
    $\slashedarrowfill@\relbar\relbar/$
    \end{tikzpicture}}%
}

\makeatletter
\def\slashedarrowfill@#1#2#3#4#5{%
  $\m@th\thickmuskip0mu\medmuskip\thickmuskip\thinmuskip\thickmuskip
   \relax#5#1\mkern-7mu%
   \cleaders\hbox{$#5\mkern-2mu#2\mkern-2mu$}\hfill
   \mathclap{#3}\mathclap{#2}%
   \cleaders\hbox{$#5\mkern-2mu#2\mkern-2mu$}\hfill
   \mkern-7mu#4$%
}
\def\rightslashedarrowfillb@{%
  \slashedarrowfill@\relbar\relbar/\rightarrow}
\newcommand\xnrightarrow[2][]{%
  \ext@arrow 0055{\rightslashedarrowfillb@}{#1}{#2}}

\def\rightslashedarrowfille@{%
  \slashedarrowfill@\relbar\relbar/\twoheadrightarrow}
\newcommand\xntworightarrow[2][]{%
  \ext@arrow 0055{\rightslashedarrowfille@}{#1}{#2}}

\def\rightslashedarrowfillg@{%
  \slashedarrowfill@\relbar\relbar{\raisebox{.12em}{}}\twoheadrightarrow}
\newcommand\xtworightarrow[2][]{%
  \ext@arrow 0055{\rightslashedarrowfillg@}{#1}{#2}}

\def\rightslashedarrowfillx@{%
  \slashedarrowfill@\Relbar\Relbar/\rightrightarrows}
\newcommand\xnTworightarrow[2][]{%
  \ext@arrow 0055{\rightslashedarrowfillx@}{#1}{#2}}

\def\rightslashedarrowfilly@{%
  \slashedarrowfill@\Relbar\Relbar{\raisebox{.12em}{}}\rightrightarrows}
\newcommand\xTworightarrow[2][]{%
  \ext@arrow 0055{\rightslashedarrowfilly@}{#1}{#2}}

\pgfdeclareshape{slash underlined}
{
  \inheritsavedanchors[from=rectangle] % this is nearly a circle
  \inheritanchorborder[from=rectangle]
  \inheritanchor[from=rectangle]{north}
  \inheritanchor[from=rectangle]{north west}
  \inheritanchor[from=rectangle]{north east}
  \inheritanchor[from=rectangle]{center}
  \inheritanchor[from=rectangle]{west}
  \inheritanchor[from=rectangle]{east}
  \inheritanchor[from=rectangle]{mid}
  \inheritanchor[from=rectangle]{mid west}
  \inheritanchor[from=rectangle]{mid east}
  \inheritanchor[from=rectangle]{base}
  \inheritanchor[from=rectangle]{base west}
  \inheritanchor[from=rectangle]{base east}
  \inheritanchor[from=rectangle]{south}
  \inheritanchor[from=rectangle]{south west}
  \inheritanchor[from=rectangle]{south east}
  \inheritanchorborder[from=rectangle]
  \foregroundpath{
    \southwest \pgf@xa=\pgf@x \pgf@ya=\pgf@y
    \northeast \pgf@xb=\pgf@x \pgf@yb=\pgf@y
    \pgf@xc=\pgf@xa
    \advance\pgf@xc by .5\pgf@xb
    \pgf@yc=\pgf@ya
    \advance\pgf@xc by -1.3pt
    \advance\pgf@yc by -1.8pt
    \pgfpathmoveto{\pgfqpoint{\pgf@xc}{\pgf@yc}}
    \advance\pgf@xc by  2.6pt
    \advance\pgf@yc by  3.6pt
    \pgfpathlineto{\pgfqpoint{\pgf@xc}{\pgf@yc}}
    \pgfpathmoveto{\pgfqpoint{\pgf@xa}{\pgf@ya}}
    \pgfpathlineto{\pgfqpoint{\pgf@xb}{\pgf@ya}}
 }
}
\tikzset{nomorepostaction/.code=\let\tikz@postactions\pgfutil@empty}

\newtheorem{theorem}{Theorem}[section]
\newtheorem{definition}[theorem]{Definition}
\newtheorem{proposition}[theorem]{Proposition}
\newtheorem{lemma}[theorem]{Lemma}

\begin{document}
%%% tile page %%%%%%%%%%%%%%%%%%%%%%%%%%%%%%%%%%%%%%%%%%%%%%%%%%%%%%%%%%%%%%%%%

\begin{titlepage}
\thispagestyle{empty}

\hrule
\begin{center}
{\bf\LARGE Probabilistic Process Algebra for True Concurrency}

\vspace{0.7cm}
--- Yong Wang ---

\vspace{2cm}
\begin{figure}[!htbp]
 \centering
 \includegraphics[width=1.0\textwidth]{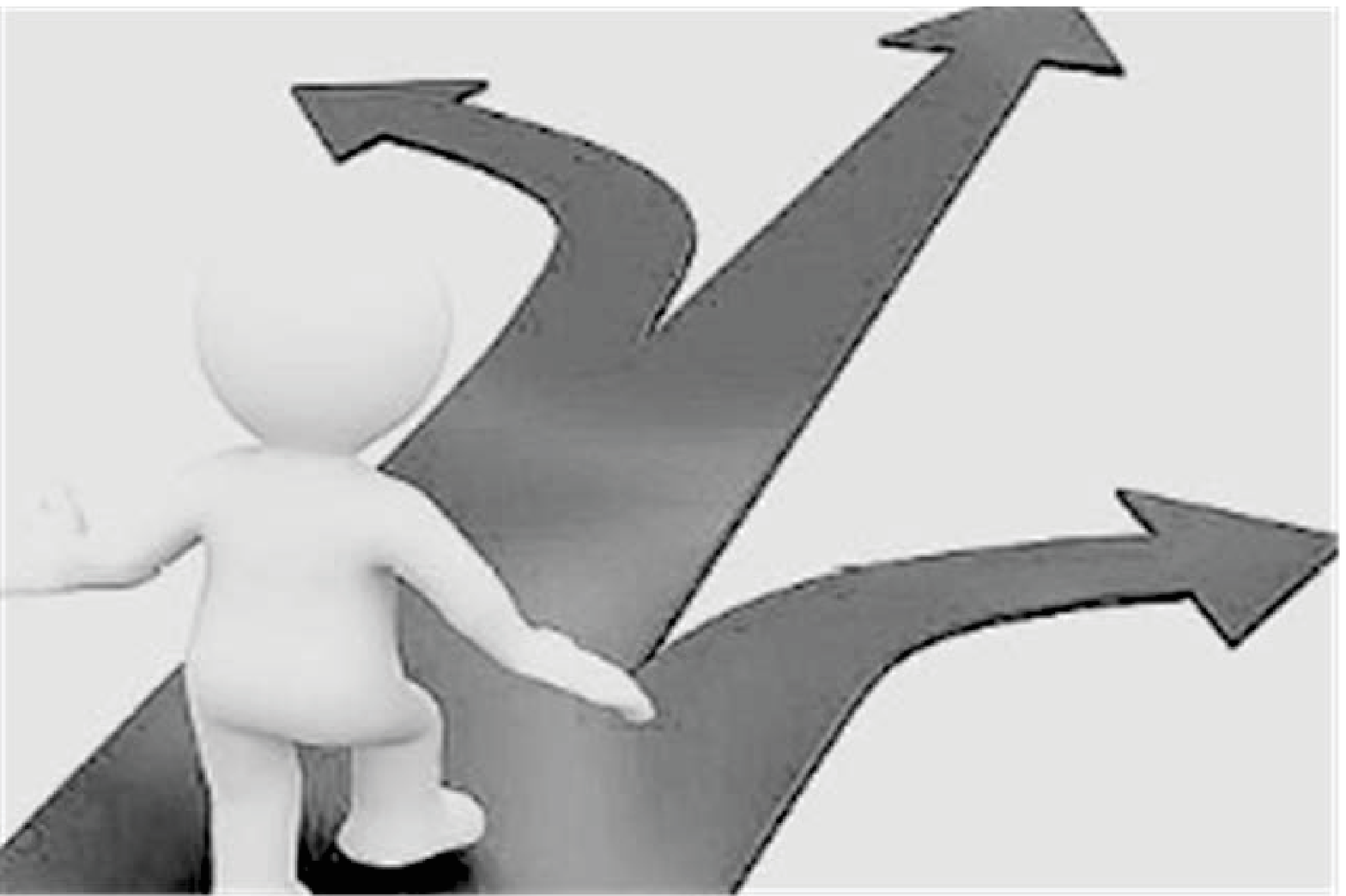}
\end{figure}

\end{center}
\end{titlepage}

\newpage %%%%%%%%%%%%%%%%%%%%%%%%%%%%%%%%%%%%%%%%%%%%%%%%%%%%%%%%%%%%%%%%%%%%%%

\setcounter{page}{1}\pagenumbering{roman}

\tableofcontents

\newpage

\setcounter{page}{1}\pagenumbering{arabic}

        \section{Introduction}

The well-known process algebras, such as CCS \cite{CC} \cite{CCS}, ACP \cite{ACP} and $\pi$-calculus \cite{PI1} \cite{PI2}, capture the interleaving concurrency based on bisimilarity semantics.
We did some work on truly concurrent process algebras, such as CTC \cite{CTC}, APTC \cite{ATC} and $\pi_{tc}$ \cite{PITC}, capture the true concurrency based on truly concurrent bisimilarities, such as
pomset bisimilarity, step bisimilarity, history-preserving (hp-) bisimilarity and hereditary history-preserving (hhp-) bisimilarity. Truly concurrent process algebras are generalizations
of the corresponding traditional process algebras.

In this book, we introduce probabilism into truly concurrent process algebras, based on the work on probabilistic process algebra \cite{PPA} \cite{PPA2} \cite{PPA3}. We introduce the
preliminaries in chapter \ref{bg}. A Calculus for Probabilistic True Concurrency abbreviated $CPTC$ in chapter \ref{cptc}, Algebra of Probabilistic Processes for True Concurrency abbreviated
$APPTC$ in chapter \ref{apptc4}, a calculus for mobile processes called $\pi_{ptc}$ in chapter \ref{piptc} and guards in chapter \ref{pguards}.

\newpage\section{Backgrounds}\label{bg}

To make this book self-satisfied, we introduce some preliminaries in this chapter, including some introductions on operational semantics, proof techniques, truly concurrent process algebra \cite{ATC} \cite{CTC} \cite{PITC}, which is based on truly
concurrent operational semantics, and also probabilistic operational semantics.

\subsection{Operational Semantics}\label{OS}

The semantics of $ACP$ is based on bisimulation/rooted branching bisimulation equivalences, and the modularity of $ACP$ relies on the concept of conservative extension, for the
conveniences, we introduce some concepts and conclusions on them.

\begin{definition}[Bisimulation]
A bisimulation relation $R$ is a binary relation on processes such that: (1) if $p R q$ and $p\xrightarrow{a}p'$ then $q\xrightarrow{a}q'$ with $p' R q'$; (2) if $p R q$ and
$q\xrightarrow{a}q'$ then $p\xrightarrow{a}p'$ with $p' R q'$; (3) if $p R q$ and $pP$, then $qP$; (4) if $p R q$ and $qP$, then $pP$. Two processes $p$ and $q$ are bisimilar,
denoted by $p\sim_{HM} q$, if there is a bisimulation relation $R$ such that $p R q$.
\end{definition}

\begin{definition}[Congruence]
Let $\Sigma$ be a signature. An equivalence relation $R$ on $\mathcal{T}(\Sigma)$ is a congruence if for each $f\in\Sigma$, if $s_i R t_i$ for $i\in\{1,\cdots,ar(f)\}$, then
$f(s_1,\cdots,s_{ar(f)}) R f(t_1,\cdots,t_{ar(f)})$.
\end{definition}

\begin{definition}[Branching bisimulation]
A branching bisimulation relation $R$ is a binary relation on the collection of processes such that: (1) if $p R q$ and $p\xrightarrow{a}p'$ then either $a\equiv \tau$ and $p' R q$ or there is a sequence of (zero or more) $\tau$-transitions $q\xrightarrow{\tau}\cdots\xrightarrow{\tau}q_0$ such that $p R q_0$ and $q_0\xrightarrow{a}q'$ with $p' R q'$; (2) if $p R q$ and $q\xrightarrow{a}q'$ then either $a\equiv \tau$ and $p R q'$ or there is a sequence of (zero or more) $\tau$-transitions $p\xrightarrow{\tau}\cdots\xrightarrow{\tau}p_0$ such that $p_0 R q$ and $p_0\xrightarrow{a}p'$ with $p' R q'$; (3) if $p R q$ and $pP$, then there is a sequence of (zero or more) $\tau$-transitions $q\xrightarrow{\tau}\cdots\xrightarrow{\tau}q_0$ such that $p R q_0$ and $q_0P$; (4) if $p R q$ and $qP$, then there is a sequence of (zero or more) $\tau$-transitions $p\xrightarrow{\tau}\cdots\xrightarrow{\tau}p_0$ such that $p_0 R q$ and $p_0P$. Two processes $p$ and $q$ are branching bisimilar, denoted by $p\approx_{bHM} q$, if there is a branching bisimulation relation $R$ such that $p R q$.
\end{definition}

\begin{definition}[Rooted branching bisimulation]
A rooted branching bisimulation relation $R$ is a binary relation on processes such that: (1) if $p R q$ and $p\xrightarrow{a}p'$ then $q\xrightarrow{a}q'$ with $p'\approx_{bHM} q'$;
(2) if $p R q$ and $q\xrightarrow{a}q'$ then $p\xrightarrow{a}p'$ with $p'\approx_{bHM} q'$; (3) if $p R q$ and $pP$, then $qP$; (4) if $p R q$ and $qP$, then $pP$. Two processes $p$ and $q$ are rooted branching bisimilar, denoted by $p\approx_{rbHM} q$, if there is a rooted branching bisimulation relation $R$ such that $p R q$.
\end{definition}

\begin{definition}[Conservative extension]
Let $T_0$ and $T_1$ be TSSs (transition system specifications) over signatures $\Sigma_0$ and $\Sigma_1$, respectively. The TSS $T_0\oplus T_1$ is a conservative extension of $T_0$ if
the LTSs (labeled transition systems) generated by $T_0$ and $T_0\oplus T_1$ contain exactly the same transitions $t\xrightarrow{a}t'$ and $tP$ with $t\in \mathcal{T}(\Sigma_0)$.
\end{definition}

\begin{definition}[Source-dependency]
The source-dependent variables in a transition rule of $\rho$ are defined inductively as follows: (1) all variables in the source of $\rho$ are source-dependent; (2) if
$t\xrightarrow{a}t'$ is a premise of $\rho$ and all variables in $t$ are source-dependent, then all variables in $t'$ are source-dependent. A transition rule is source-dependent if
all its variables are. A TSS is source-dependent if all its rules are.
\end{definition}

\begin{definition}[Freshness]
Let $T_0$ and $T_1$ be TSSs over signatures $\Sigma_0$ and $\Sigma_1$, respectively. A term in $\mathbb{T}(T_0\oplus T_1)$ is said to be fresh if it contains a function symbol from
$\Sigma_1\setminus\Sigma_0$. Similarly, a transition label or predicate symbol in $T_1$ is fresh if it does not occur in $T_0$.
\end{definition}

\begin{theorem}[Conservative extension]\label{TCE}
Let $T_0$ and $T_1$ be TSSs over signatures $\Sigma_0$ and $\Sigma_1$, respectively, where $T_0$ and $T_0\oplus T_1$ are positive after reduction. Under the following conditions,
$T_0\oplus T_1$ is a conservative extension of $T_0$. (1) $T_0$ is source-dependent. (2) For each $\rho\in T_1$, either the source of $\rho$ is fresh, or $\rho$ has a premise of the
form $t\xrightarrow{a}t'$ or $tP$, where $t\in \mathbb{T}(\Sigma_0)$, all variables in $t$ occur in the source of $\rho$ and $t'$, $a$ or $P$ is fresh.
\end{theorem}

\subsection{Proof Techniques}\label{PT}

In this subsection, we introduce the concepts and conclusions about elimination, which is very important in the proof of completeness theorem.

\begin{definition}[Elimination property]
Let a process algebra with a defined set of basic terms as a subset of the set of closed terms over the process algebra. Then the process algebra has the elimination to basic terms
property if for every closed term $s$ of the algebra, there exists a basic term $t$ of the algebra such that the algebra$\vdash s=t$.
\end{definition}

\begin{definition}[Strongly normalizing]
A term $s_0$ is called strongly normalizing if does not an infinite series of reductions beginning in $s_0$.
\end{definition}

\begin{definition}
We write $s>_{lpo} t$ if $s\rightarrow^+ t$ where $\rightarrow^+$ is the transitive closure of the reduction relation defined by the transition rules of an algebra.
\end{definition}

\begin{theorem}[Strong normalization]\label{SN}
Let a term rewriting system (TRS) with finitely many rewriting rules and let $>$ be a well-founded ordering on the signature of the corresponding algebra. If $s>_{lpo} t$ for each
rewriting rule $s\rightarrow t$ in the TRS, then the term rewriting system is strongly normalizing.
\end{theorem}

\subsection{Truly Concurrent Process Algebra -- APTC}

APTC eliminates the differences of structures of transition system, event structure, etc, and discusses their behavioral equivalences. It considers that there are two kinds of causality
relations: the chronological order modeled by the sequential composition and the causal order between different parallel branches modeled by the communication merge. It also considers
that there exist two kinds of confliction relations: the structural confliction modeled by the alternative composition and the conflictions in different parallel branches which should
be eliminated. Based on conservative extension, there are four modules in APTC: BATC (Basic Algebra for True Concurrency), APTC (Algebra for Parallelism in True Concurrency), recursion
and abstraction.

\subsubsection{Basic Algebra for True Concurrency}

BATC has sequential composition $\cdot$ and alternative composition $+$ to capture the chronological ordered causality and the structural confliction. The constants are ranged over $A$,
the set of atomic actions. The algebraic laws on $\cdot$ and $+$ are sound and complete modulo truly concurrent bisimulation equivalences (including pomset bisimulation, step
bisimulation, hp-bisimulation and hhp-bisimulation).

\begin{definition}[Prime event structure with silent event]\label{PES}
Let $\Lambda$ be a fixed set of labels, ranged over $a,b,c,\cdots$ and $\tau$. A ($\Lambda$-labelled) prime event structure with silent event $\tau$ is a tuple
$\mathcal{E}=\langle \mathbb{E}, \leq, \sharp, \lambda\rangle$, where $\mathbb{E}$ is a denumerable set of events, including the silent event $\tau$. Let
$\hat{\mathbb{E}}=\mathbb{E}\backslash\{\tau\}$, exactly excluding $\tau$, it is obvious that $\hat{\tau^*}=\epsilon$, where $\epsilon$ is the empty event.
Let $\lambda:\mathbb{E}\rightarrow\Lambda$ be a labelling function and let $\lambda(\tau)=\tau$. And $\leq$, $\sharp$ are binary relations on $\mathbb{E}$,
called causality and conflict respectively, such that:

\begin{enumerate}
  \item $\leq$ is a partial order and $\lceil e \rceil = \{e'\in \mathbb{E}|e'\leq e\}$ is finite for all $e\in \mathbb{E}$. It is easy to see that
  $e\leq\tau^*\leq e'=e\leq\tau\leq\cdots\leq\tau\leq e'$, then $e\leq e'$.
  \item $\sharp$ is irreflexive, symmetric and hereditary with respect to $\leq$, that is, for all $e,e',e''\in \mathbb{E}$, if $e\sharp e'\leq e''$, then $e\sharp e''$.
\end{enumerate}

Then, the concepts of consistency and concurrency can be drawn from the above definition:

\begin{enumerate}
  \item $e,e'\in \mathbb{E}$ are consistent, denoted as $e\frown e'$, if $\neg(e\sharp e')$. A subset $X\subseteq \mathbb{E}$ is called consistent, if $e\frown e'$ for all
  $e,e'\in X$.
  \item $e,e'\in \mathbb{E}$ are concurrent, denoted as $e\parallel e'$, if $\neg(e\leq e')$, $\neg(e'\leq e)$, and $\neg(e\sharp e')$.
\end{enumerate}
\end{definition}

\begin{definition}[Configuration]
Let $\mathcal{E}$ be a PES. A (finite) configuration in $\mathcal{E}$ is a (finite) consistent subset of events $C\subseteq \mathcal{E}$, closed with respect to causality
(i.e. $\lceil C\rceil=C$). The set of finite configurations of $\mathcal{E}$ is denoted by $\mathcal{C}(\mathcal{E})$. We let $\hat{C}=C\backslash\{\tau\}$.
\end{definition}

A consistent subset of $X\subseteq \mathbb{E}$ of events can be seen as a pomset. Given $X, Y\subseteq \mathbb{E}$, $\hat{X}\sim \hat{Y}$ if $\hat{X}$ and $\hat{Y}$ are
isomorphic as pomsets. In the following of the paper, we say $C_1\sim C_2$, we mean $\hat{C_1}\sim\hat{C_2}$.

\begin{definition}[Pomset transitions and step]
Let $\mathcal{E}$ be a PES and let $C\in\mathcal{C}(\mathcal{E})$, and $\emptyset\neq X\subseteq \mathbb{E}$, if $C\cap X=\emptyset$ and $C'=C\cup X\in\mathcal{C}(\mathcal{E})$,
then $C\xrightarrow{X} C'$ is called a pomset transition from $C$ to $C'$. When the events in $X$ are pairwise concurrent, we say that $C\xrightarrow{X}C'$ is a step.
\end{definition}

\begin{definition}[Pomset, step bisimulation]\label{PSB}
Let $\mathcal{E}_1$, $\mathcal{E}_2$ be PESs. A pomset bisimulation is a relation $R\subseteq\mathcal{C}(\mathcal{E}_1)\times\mathcal{C}(\mathcal{E}_2)$, such that if
$(C_1,C_2)\in R$, and $C_1\xrightarrow{X_1}C_1'$ then $C_2\xrightarrow{X_2}C_2'$, with $X_1\subseteq \mathbb{E}_1$, $X_2\subseteq \mathbb{E}_2$, $X_1\sim X_2$ and $(C_1',C_2')\in R$,
and vice-versa. We say that $\mathcal{E}_1$, $\mathcal{E}_2$ are pomset bisimilar, written $\mathcal{E}_1\sim_p\mathcal{E}_2$, if there exists a pomset bisimulation $R$, such that
$(\emptyset,\emptyset)\in R$. By replacing pomset transitions with steps, we can get the definition of step bisimulation. When PESs $\mathcal{E}_1$ and $\mathcal{E}_2$ are step
bisimilar, we write $\mathcal{E}_1\sim_s\mathcal{E}_2$.
\end{definition}

\begin{definition}[Posetal product]
Given two PESs $\mathcal{E}_1$, $\mathcal{E}_2$, the posetal product of their configurations, denoted $\mathcal{C}(\mathcal{E}_1)\overline{\times}\mathcal{C}(\mathcal{E}_2)$,
is defined as

$$\{(C_1,f,C_2)|C_1\in\mathcal{C}(\mathcal{E}_1),C_2\in\mathcal{C}(\mathcal{E}_2),f:C_1\rightarrow C_2 \textrm{ isomorphism}\}.$$

A subset $R\subseteq\mathcal{C}(\mathcal{E}_1)\overline{\times}\mathcal{C}(\mathcal{E}_2)$ is called a posetal relation. We say that $R$ is downward closed when for any
$(C_1,f,C_2),(C_1',f',C_2')\in \mathcal{C}(\mathcal{E}_1)\overline{\times}\mathcal{C}(\mathcal{E}_2)$, if $(C_1,f,C_2)\subseteq (C_1',f',C_2')$ pointwise and $(C_1',f',C_2')\in R$,
then $(C_1,f,C_2)\in R$.

For $f:X_1\rightarrow X_2$, we define $f[x_1\mapsto x_2]:X_1\cup\{x_1\}\rightarrow X_2\cup\{x_2\}$, $z\in X_1\cup\{x_1\}$,(1)$f[x_1\mapsto x_2](z)=
x_2$,if $z=x_1$;(2)$f[x_1\mapsto x_2](z)=f(z)$, otherwise. Where $X_1\subseteq \mathbb{E}_1$, $X_2\subseteq \mathbb{E}_2$, $x_1\in \mathbb{E}_1$, $x_2\in \mathbb{E}_2$.
\end{definition}

\begin{definition}[(Hereditary) history-preserving bisimulation]\label{HHPB}
A history-preserving (hp-) bisimulation is a posetal relation $R\subseteq\mathcal{C}(\mathcal{E}_1)\overline{\times}\mathcal{C}(\mathcal{E}_2)$ such that if $(C_1,f,C_2)\in R$,
and $C_1\xrightarrow{e_1} C_1'$, then $C_2\xrightarrow{e_2} C_2'$, with $(C_1',f[e_1\mapsto e_2],C_2')\in R$, and vice-versa. $\mathcal{E}_1,\mathcal{E}_2$ are history-preserving
(hp-)bisimilar and are written $\mathcal{E}_1\sim_{hp}\mathcal{E}_2$ if there exists a hp-bisimulation $R$ such that $(\emptyset,\emptyset,\emptyset)\in R$.

A hereditary history-preserving (hhp-)bisimulation is a downward closed hp-bisimulation. $\mathcal{E}_1,\mathcal{E}_2$ are hereditary history-preserving (hhp-)bisimilar and are
written $\mathcal{E}_1\sim_{hhp}\mathcal{E}_2$.
\end{definition}

In the following, let $e_1, e_2, e_1', e_2'\in \mathbb{E}$, and let variables $x,y,z$ range over the set of terms for true concurrency, $p,q,s$ range over the set of closed terms.
The set of axioms of BATC consists of the laws given in Table \ref{AxiomsForBATC}.

\begin{center}
    \begin{table}
        \begin{tabular}{@{}ll@{}}
            \hline No. &Axiom\\
            $A1$ & $x+ y = y+ x$\\
            $A2$ & $(x+ y)+ z = x+ (y+ z)$\\
            $A3$ & $x+ x = x$\\
            $A4$ & $(x+ y)\cdot z = x\cdot z + y\cdot z$\\
            $A5$ & $(x\cdot y)\cdot z = x\cdot(y\cdot z)$\\
        \end{tabular}
        \caption{Axioms of BATC}
        \label{AxiomsForBATC}
    \end{table}
\end{center}

\begin{definition}[Basic terms of $BATC$]
The set of basic terms of $BATC$, $\mathcal{B}(BATC)$, is inductively defined as follows:
\begin{enumerate}
  \item $\mathbb{E}\subset\mathcal{B}(BATC)$;
  \item if $e\in \mathbb{E}, t\in\mathcal{B}(BATC)$ then $e\cdot t\in\mathcal{B}(BATC)$;
  \item if $t,s\in\mathcal{B}(BATC)$ then $t+ s\in\mathcal{B}(BATC)$.
\end{enumerate}
\end{definition}

\begin{theorem}[Elimination theorem of $BATC$]
Let $p$ be a closed $BATC$ term. Then there is a basic $BATC$ term $q$ such that $BATC\vdash p=q$.
\end{theorem}

We give the operational transition rules of operators $\cdot$ and $+$ as Table \ref{TRForBATC} shows. And the predicate $\xrightarrow{e}\surd$ represents successful termination after
execution of the event $e$.

\begin{center}
    \begin{table}
        $$\frac{}{e\xrightarrow{e}\surd}$$
        $$\frac{x\xrightarrow{e}\surd}{x+ y\xrightarrow{e}\surd} \quad\frac{x\xrightarrow{e}x'}{x+ y\xrightarrow{e}x'} \quad\frac{y\xrightarrow{e}\surd}{x+ y\xrightarrow{e}\surd}
        \quad\frac{y\xrightarrow{e}y'}{x+ y\xrightarrow{e}y'}$$
        $$\frac{x\xrightarrow{e}\surd}{x\cdot y\xrightarrow{e} y} \quad\frac{x\xrightarrow{e}x'}{x\cdot y\xrightarrow{e}x'\cdot y}$$
        \caption{Transition rules of BATC}
        \label{TRForBATC}
    \end{table}
\end{center}

\begin{theorem}[Congruence of $BATC$ with respect to truly concurrent bisimulation equivalences]
Truly concurrent bisimulation equivalences $\sim_{p}$, $\sim_s$, $\sim_{hp}$ and $\sim_{hhp}$ are all congruences with respect to $BATC$.
\end{theorem}

\begin{theorem}[Soundness of BATC modulo truly concurrent bisimulation equivalences]\label{SBATC}
The axiomatization of BATC is sound modulo truly concurrent bisimulation equivalences $\sim_{p}$, $\sim_{s}$, $\sim_{hp}$ and $\sim_{hhp}$. That is,

\begin{enumerate}
  \item let $x$ and $y$ be BATC terms. If BATC $\vdash x=y$, then $x\sim_{p} y$;
  \item let $x$ and $y$ be BATC terms. If BATC $\vdash x=y$, then $x\sim_{s} y$;
  \item let $x$ and $y$ be BATC terms. If BATC $\vdash x=y$, then $x\sim_{hp} y$;
  \item let $x$ and $y$ be BATC terms. If BATC $\vdash x=y$, then $x\sim_{hhp} y$.
\end{enumerate}

\end{theorem}

\begin{theorem}[Completeness of BATC modulo truly concurrent bisimulation equivalences]\label{CBATC}
The axiomatization of BATC is complete modulo truly concurrent bisimulation equivalences $\sim_{p}$, $\sim_{s}$, $\sim_{hp}$ and $\sim_{hhp}$. That is,

\begin{enumerate}
  \item let $p$ and $q$ be closed BATC terms, if $p\sim_{p} q$ then $p=q$;
  \item let $p$ and $q$ be closed BATC terms, if $p\sim_{s} q$ then $p=q$;
  \item let $p$ and $q$ be closed BATC terms, if $p\sim_{hp} q$ then $p=q$;
  \item let $p$ and $q$ be closed BATC terms, if $p\sim_{hhp} q$ then $p=q$.
\end{enumerate}

\end{theorem}

Since hhp-bisimilarity is a downward closed hp-bisimilarity and can be downward closed to single atomic event, which implies bisimilarity. As Moller \cite{ILM} proven, there is not a
finite sound and complete axiomatization for parallelism $\parallel$ modulo bisimulation equivalence, so there is not a finite sound and complete axiomatization for parallelism
$\parallel$ modulo hhp-bisimulation equivalence either. Inspired by the way of left merge to modeling the full merge for bisimilarity, we introduce a left parallel composition
$\leftmerge$ to model the full parallelism $\parallel$ for hhp-bisimilarity.

In the following subsection, we add left parallel composition $\leftmerge$ to the whole theory. Because the resulting theory is similar to the former, we only list the significant
differences, and all proofs of the conclusions are left to the reader.

\subsubsection{$APTC$ with Left Parallel Composition}

We give the transition rules of APTC in Table \ref{TRForAPTC}, it is suitable for all truly concurrent behavioral equivalence, including pomset bisimulation, step bisimulation,
hp-bisimulation and hhp-bisimulation.

\begin{center}
    \begin{table}
        $$\frac{x\xrightarrow{e_1}\surd\quad y\xrightarrow{e_2}\surd}{x\parallel y\xrightarrow{\{e_1,e_2\}}\surd} \quad\frac{x\xrightarrow{e_1}x'\quad y\xrightarrow{e_2}\surd}{x\parallel y\xrightarrow{\{e_1,e_2\}}x'}$$
        $$\frac{x\xrightarrow{e_1}\surd\quad y\xrightarrow{e_2}y'}{x\parallel y\xrightarrow{\{e_1,e_2\}}y'} \quad\frac{x\xrightarrow{e_1}x'\quad y\xrightarrow{e_2}y'}{x\parallel y\xrightarrow{\{e_1,e_2\}}x'\between y'}$$
        $$\frac{x\xrightarrow{e_1}\surd\quad y\xrightarrow{e_2}\surd}{x\mid y\xrightarrow{\gamma(e_1,e_2)}\surd} \quad\frac{x\xrightarrow{e_1}x'\quad y\xrightarrow{e_2}\surd}{x\mid y\xrightarrow{\gamma(e_1,e_2)}x'}$$
        $$\frac{x\xrightarrow{e_1}\surd\quad y\xrightarrow{e_2}y'}{x\mid y\xrightarrow{\gamma(e_1,e_2)}y'} \quad\frac{x\xrightarrow{e_1}x'\quad y\xrightarrow{e_2}y'}{x\mid y\xrightarrow{\gamma(e_1,e_2)}x'\between y'}$$
        $$\frac{x\xrightarrow{e_1}\surd\quad (\sharp(e_1,e_2))}{\Theta(x)\xrightarrow{e_1}\surd} \quad\frac{x\xrightarrow{e_2}\surd\quad (\sharp(e_1,e_2))}{\Theta(x)\xrightarrow{e_2}\surd}$$
        $$\frac{x\xrightarrow{e_1}x'\quad (\sharp(e_1,e_2))}{\Theta(x)\xrightarrow{e_1}\Theta(x')} \quad\frac{x\xrightarrow{e_2}x'\quad (\sharp(e_1,e_2))}{\Theta(x)\xrightarrow{e_2}\Theta(x')}$$
        $$\frac{x\xrightarrow{e_1}\surd \quad y\nrightarrow^{e_2}\quad (\sharp(e_1,e_2))}{x\triangleleft y\xrightarrow{\tau}\surd}
        \quad\frac{x\xrightarrow{e_1}x' \quad y\nrightarrow^{e_2}\quad (\sharp(e_1,e_2))}{x\triangleleft y\xrightarrow{\tau}x'}$$
        $$\frac{x\xrightarrow{e_1}\surd \quad y\nrightarrow^{e_3}\quad (\sharp(e_1,e_2),e_2\leq e_3)}{x\triangleleft y\xrightarrow{e_1}\surd}
        \quad\frac{x\xrightarrow{e_1}x' \quad y\nrightarrow^{e_3}\quad (\sharp(e_1,e_2),e_2\leq e_3)}{x\triangleleft y\xrightarrow{e_1}x'}$$
        $$\frac{x\xrightarrow{e_3}\surd \quad y\nrightarrow^{e_2}\quad (\sharp(e_1,e_2),e_1\leq e_3)}{x\triangleleft y\xrightarrow{\tau}\surd}
        \quad\frac{x\xrightarrow{e_3}x' \quad y\nrightarrow^{e_2}\quad (\sharp(e_1,e_2),e_1\leq e_3)}{x\triangleleft y\xrightarrow{\tau}x'}$$
        $$\frac{x\xrightarrow{e}\surd}{\partial_H(x)\xrightarrow{e}\surd}\quad (e\notin H)\quad\quad\frac{x\xrightarrow{e}x'}{\partial_H(x)\xrightarrow{e}\partial_H(x')}\quad(e\notin H)$$
        \caption{Transition rules of APTC}
        \label{TRForAPTC}
    \end{table}
\end{center}

The transition rules of left parallel composition $\leftmerge$ are shown in Table \ref{TRForLeftParallel}. With a little abuse, we extend the causal order relation $\leq$ on
$\mathbb{E}$ to include the original partial order (denoted by $<$) and concurrency (denoted by $=$).

\begin{center}
    \begin{table}
        $$\frac{x\xrightarrow{e_1}\surd\quad y\xrightarrow{e_2}\surd \quad(e_1\leq e_2)}{x\leftmerge y\xrightarrow{\{e_1,e_2\}}\surd} \quad\frac{x\xrightarrow{e_1}x'\quad y\xrightarrow{e_2}\surd \quad(e_1\leq e_2)}{x\leftmerge y\xrightarrow{\{e_1,e_2\}}x'}$$
        $$\frac{x\xrightarrow{e_1}\surd\quad y\xrightarrow{e_2}y' \quad(e_1\leq e_2)}{x\leftmerge y\xrightarrow{\{e_1,e_2\}}y'} \quad\frac{x\xrightarrow{e_1}x'\quad y\xrightarrow{e_2}y' \quad(e_1\leq e_2)}{x\leftmerge y\xrightarrow{\{e_1,e_2\}}x'\between y'}$$
        \caption{Transition rules of left parallel operator $\leftmerge$}
        \label{TRForLeftParallel}
    \end{table}
\end{center}

The new axioms for parallelism are listed in Table \ref{AxiomsForLeftParallelism}.

\begin{center}
    \begin{table}
        \begin{tabular}{@{}ll@{}}
            \hline No. &Axiom\\
            $A6$ & $x+ \delta = x$\\
            $A7$ & $\delta\cdot x =\delta$\\
            $P1$ & $x\between y = x\parallel y + x\mid y$\\
            $P2$ & $x\parallel y = y \parallel x$\\
            $P3$ & $(x\parallel y)\parallel z = x\parallel (y\parallel z)$\\
            $P4$ & $x\parallel y = x\leftmerge y + y\leftmerge x$\\
            $P5$ & $(e_1\leq e_2)\quad e_1\leftmerge (e_2\cdot y) = (e_1\leftmerge e_2)\cdot y$\\
            $P6$ & $(e_1\leq e_2)\quad (e_1\cdot x)\leftmerge e_2 = (e_1\leftmerge e_2)\cdot x$\\
            $P7$ & $(e_1\leq e_2)\quad (e_1\cdot x)\leftmerge (e_2\cdot y) = (e_1\leftmerge e_2)\cdot (x\between y)$\\
            $P8$ & $(x+ y)\leftmerge z = (x\leftmerge z)+ (y\leftmerge z)$\\
            $P9$ & $\delta\leftmerge x = \delta$\\
            $C10$ & $e_1\mid e_2 = \gamma(e_1,e_2)$\\
            $C11$ & $e_1\mid (e_2\cdot y) = \gamma(e_1,e_2)\cdot y$\\
            $C12$ & $(e_1\cdot x)\mid e_2 = \gamma(e_1,e_2)\cdot x$\\
            $C13$ & $(e_1\cdot x)\mid (e_2\cdot y) = \gamma(e_1,e_2)\cdot (x\between y)$\\
            $C14$ & $(x+ y)\mid z = (x\mid z) + (y\mid z)$\\
            $C15$ & $x\mid (y+ z) = (x\mid y)+ (x\mid z)$\\
            $C16$ & $\delta\mid x = \delta$\\
            $C17$ & $x\mid\delta = \delta$\\
            $CE18$ & $\Theta(e) = e$\\
            $CE19$ & $\Theta(\delta) = \delta$\\
            $CE20$ & $\Theta(x+ y) = \Theta(x)\triangleleft y + \Theta(y)\triangleleft x$\\
            $CE21$ & $\Theta(x\cdot y)=\Theta(x)\cdot\Theta(y)$\\
            $CE22$ & $\Theta(x\leftmerge y) = ((\Theta(x)\triangleleft y)\leftmerge y)+ ((\Theta(y)\triangleleft x)\leftmerge x)$\\
            $CE23$ & $\Theta(x\mid y) = ((\Theta(x)\triangleleft y)\mid y)+ ((\Theta(y)\triangleleft x)\mid x)$\\
            $U24$ & $(\sharp(e_1,e_2))\quad e_1\triangleleft e_2 = \tau$\\
            $U25$ & $(\sharp(e_1,e_2),e_2\leq e_3)\quad e_1\triangleleft e_3 = e_1$\\
            $U26$ & $(\sharp(e_1,e_2),e_2\leq e_3)\quad e_3\triangleleft e_1 = \tau$\\
            $U27$ & $e\triangleleft \delta = e$\\
            $U28$ & $\delta \triangleleft e = \delta$\\
            $U29$ & $(x+ y)\triangleleft z = (x\triangleleft z)+ (y\triangleleft z)$\\
            $U30$ & $(x\cdot y)\triangleleft z = (x\triangleleft z)\cdot (y\triangleleft z)$\\
            $U31$ & $(x\leftmerge y)\triangleleft z = (x\triangleleft z)\leftmerge (y\triangleleft z)$\\
            $U32$ & $(x\mid y)\triangleleft z = (x\triangleleft z)\mid (y\triangleleft z)$\\
            $U33$ & $x\triangleleft (y+ z) = (x\triangleleft y)\triangleleft z$\\
            $U34$ & $x\triangleleft (y\cdot z)=(x\triangleleft y)\triangleleft z$\\
            $U35$ & $x\triangleleft (y\leftmerge z) = (x\triangleleft y)\triangleleft z$\\
            $U36$ & $x\triangleleft (y\mid z) = (x\triangleleft y)\triangleleft z$\\
        \end{tabular}
        \caption{Axioms of parallelism with left parallel composition}
        \label{AxiomsForLeftParallelism}
    \end{table}
\end{center}

\begin{definition}[Basic terms of $APTC$ with left parallel composition]
The set of basic terms of $APTC$, $\mathcal{B}(APTC)$, is inductively defined as follows:
\begin{enumerate}
  \item $\mathbb{E}\subset\mathcal{B}(APTC)$;
  \item if $e\in \mathbb{E}, t\in\mathcal{B}(APTC)$ then $e\cdot t\in\mathcal{B}(APTC)$;
  \item if $t,s\in\mathcal{B}(APTC)$ then $t+ s\in\mathcal{B}(APTC)$;
  \item if $t,s\in\mathcal{B}(APTC)$ then $t\leftmerge s\in\mathcal{B}(APTC)$.
\end{enumerate}
\end{definition}

\begin{theorem}[Generalization of the algebra for left parallelism with respect to $BATC$]
The algebra for left parallelism is a generalization of $BATC$.
\end{theorem}

\begin{theorem}[Congruence theorem of $APTC$ with left parallel composition]
Truly concurrent bisimulation equivalences $\sim_{p}$, $\sim_s$, $\sim_{hp}$ and $\sim_{hhp}$ are all congruences with respect to $APTC$ with left parallel composition.
\end{theorem}

\begin{theorem}[Elimination theorem of parallelism with left parallel composition]
Let $p$ be a closed $APTC$ with left parallel composition term. Then there is a basic $APTC$ term $q$ such that $APTC\vdash p=q$.
\end{theorem}

\begin{theorem}[Soundness of parallelism  with left parallel composition modulo truly concurrent bisimulation equivalences]
Let $x$ and $y$ be $APTC$ with left parallel composition terms. If $APTC\vdash x=y$, then

\begin{enumerate}
  \item $x\sim_{s} y$;
  \item $x\sim_{p} y$;
  \item $x\sim_{hp} y$;
  \item $x\sim_{hhp} y$.
\end{enumerate}
\end{theorem}

\begin{theorem}[Completeness of parallelism with left parallel composition modulo truly concurrent bisimulation equivalences]
Let $x$ and $y$ be $APTC$ terms.

\begin{enumerate}
  \item If $x\sim_{s} y$, then $APTC\vdash x=y$;
  \item if $x\sim_{p} y$, then $APTC\vdash x=y$;
  \item if $x\sim_{hp} y$, then $APTC\vdash x=y$;
  \item if $x\sim_{hhp} y$, then $APTC\vdash x=y$.
\end{enumerate}
\end{theorem}

The axioms of encapsulation operator are shown in \ref{AxiomsForEncapsulationLeft}.

\begin{center}
    \begin{table}
        \begin{tabular}{@{}ll@{}}
            \hline No. &Axiom\\
            $D1$ & $e\notin H\quad\partial_H(e) = e$\\
            $D2$ & $e\in H\quad \partial_H(e) = \delta$\\
            $D3$ & $\partial_H(\delta) = \delta$\\
            $D4$ & $\partial_H(x+ y) = \partial_H(x)+\partial_H(y)$\\
            $D5$ & $\partial_H(x\cdot y) = \partial_H(x)\cdot\partial_H(y)$\\
            $D6$ & $\partial_H(x\leftmerge y) = \partial_H(x)\leftmerge\partial_H(y)$\\
        \end{tabular}
        \caption{Axioms of encapsulation operator with left parallel composition}
        \label{AxiomsForEncapsulationLeft}
    \end{table}
\end{center}

\begin{theorem}[Conservativity of $APTC$ with respect to the algebra for parallelism with left parallel composition]
$APTC$ is a conservative extension of the algebra for parallelism with left parallel composition.
\end{theorem}

\begin{theorem}[Congruence theorem of encapsulation operator $\partial_H$]
Truly concurrent bisimulation equivalences $\sim_{p}$, $\sim_s$, $\sim_{hp}$ and $\sim_{hhp}$ are all congruences with respect to encapsulation operator $\partial_H$.
\end{theorem}

\begin{theorem}[Elimination theorem of $APTC$]
Let $p$ be a closed $APTC$ term including the encapsulation operator $\partial_H$. Then there is a basic $APTC$ term $q$ such that $APTC\vdash p=q$.
\end{theorem}

\begin{theorem}[Soundness of $APTC$ modulo truly concurrent bisimulation equivalences]
Let $x$ and $y$ be $APTC$ terms including encapsulation operator $\partial_H$. If $APTC\vdash x=y$, then

\begin{enumerate}
  \item $x\sim_{s} y$;
  \item $x\sim_{p} y$;
  \item $x\sim_{hp} y$;
  \item $x\sim_{hhp} y$.
\end{enumerate}
\end{theorem}

\begin{theorem}[Completeness of $APTC$ modulo truly concurrent bisimulation equivalences]
Let $p$ and $q$ be closed $APTC$ terms including encapsulation operator $\partial_H$,

\begin{enumerate}
  \item if $p\sim_{s} q$ then $p=q$;
  \item if $p\sim_{p} q$ then $p=q$;
  \item if $p\sim_{hp} q$ then $p=q$;
  \item if $p\sim_{hhp} q$ then $p=q$.
\end{enumerate}
\end{theorem}

\subsubsection{Recursion}

\begin{definition}[Recursive specification]
A recursive specification is a finite set of recursive equations

$$X_1=t_1(X_1,\cdots,X_n)$$
$$\cdots$$
$$X_n=t_n(X_1,\cdots,X_n)$$

where the left-hand sides of $X_i$ are called recursion variables, and the right-hand sides $t_i(X_1,\cdots,X_n)$ are process terms in $APTC$ with possible occurrences of the recursion
variables $X_1,\cdots,X_n$.
\end{definition}

\begin{definition}[Solution]
Processes $p_1,\cdots,p_n$ are a solution for a recursive specification $\{X_i=t_i(X_1,\cdots,X_n)|i\in\{1,\cdots,n\}\}$ (with respect to truly concurrent bisimulation equivalences
$\sim_s$($\sim_p$, $\sim_{hp}$, $\sim_{hhp}$)) if $p_i\sim_s (\sim_p, \sim_{hp},\sim{hhp})t_i(p_1,\cdots,p_n)$ for $i\in\{1,\cdots,n\}$.
\end{definition}

\begin{definition}[Guarded recursive specification]
A recursive specification

$$X_1=t_1(X_1,\cdots,X_n)$$
$$...$$
$$X_n=t_n(X_1,\cdots,X_n)$$

is guarded if the right-hand sides of its recursive equations can be adapted to the form by applications of the axioms in $APTC$ and replacing recursion variables by the right-hand
sides of their recursive equations,

$$(a_{11}\leftmerge\cdots\leftmerge a_{1i_1})\cdot s_1(X_1,\cdots,X_n)+\cdots+(a_{k1}\leftmerge\cdots\leftmerge a_{ki_k})\cdot s_k(X_1,\cdots,X_n)+(b_{11}\leftmerge\cdots\leftmerge b_{1j_1})+\cdots+(b_{1j_1}\leftmerge\cdots\leftmerge b_{lj_l})$$

where $a_{11},\cdots,a_{1i_1},a_{k1},\cdots,a_{ki_k},b_{11},\cdots,b_{1j_1},b_{1j_1},\cdots,b_{lj_l}\in \mathbb{E}$, and the sum above is allowed to be empty, in which case it
represents the deadlock $\delta$.
\end{definition}

\begin{definition}[Linear recursive specification]
A recursive specification is linear if its recursive equations are of the form

$$(a_{11}\leftmerge\cdots\leftmerge a_{1i_1})X_1+\cdots+(a_{k1}\leftmerge\cdots\leftmerge a_{ki_k})X_k+(b_{11}\leftmerge\cdots\leftmerge b_{1j_1})+\cdots+(b_{1j_1}\leftmerge\cdots\leftmerge b_{lj_l})$$

where $a_{11},\cdots,a_{1i_1},a_{k1},\cdots,a_{ki_k},b_{11},\cdots,b_{1j_1},b_{1j_1},\cdots,b_{lj_l}\in \mathbb{E}$, and the sum above is allowed to be empty, in which case it
represents the deadlock $\delta$.
\end{definition}

The behavior of the solution $\langle X_i|E\rangle$ for the recursion variable $X_i$ in $E$, where $i\in\{1,\cdots,n\}$, is exactly the behavior of their right-hand sides
$t_i(X_1,\cdots,X_n)$, which is captured by the two transition rules in Table \ref{TRForGR}.

\begin{center}
    \begin{table}
        $$\frac{t_i(\langle X_1|E\rangle,\cdots,\langle X_n|E\rangle)\xrightarrow{\{e_1,\cdots,e_k\}}\surd}{\langle X_i|E\rangle\xrightarrow{\{e_1,\cdots,e_k\}}\surd}$$
        $$\frac{t_i(\langle X_1|E\rangle,\cdots,\langle X_n|E\rangle)\xrightarrow{\{e_1,\cdots,e_k\}} y}{\langle X_i|E\rangle\xrightarrow{\{e_1,\cdots,e_k\}} y}$$
        \caption{Transition rules of guarded recursion}
        \label{TRForGR}
    \end{table}
\end{center}

The $RDP$ (Recursive Definition Principle) and the $RSP$ (Recursive Specification Principle) are shown in Table \ref{RDPRSP}.

\begin{center}
\begin{table}
  \begin{tabular}{@{}ll@{}}
\hline No. &Axiom\\
  $RDP$ & $\langle X_i|E\rangle = t_i(\langle X_1|E,\cdots,X_n|E\rangle)\quad (i\in\{1,\cdots,n\})$\\
  $RSP$ & if $y_i=t_i(y_1,\cdots,y_n)$ for $i\in\{1,\cdots,n\}$, then $y_i=\langle X_i|E\rangle \quad(i\in\{1,\cdots,n\})$\\
\end{tabular}
\caption{Recursive definition and specification principle}
\label{RDPRSP}
\end{table}
\end{center}

\begin{theorem}[Conservitivity of $APTC$ with guarded recursion]
$APTC$ with guarded recursion is a conservative extension of $APTC$.
\end{theorem}

\begin{theorem}[Congruence theorem of $APTC$ with guarded recursion]
Truly concurrent bisimulation equivalences $\sim_{p}$, $\sim_s$, $\sim_{hp}$, $\sim_{hhp}$ are all congruences with respect to $APTC$ with guarded recursion.
\end{theorem}

\begin{theorem}[Elimination theorem of $APTC$ with linear recursion]
Each process term in $APTC$ with linear recursion is equal to a process term $\langle X_1|E\rangle$ with $E$ a linear recursive specification.
\end{theorem}

\begin{theorem}[Soundness of $APTC$ with guarded recursion]
Let $x$ and $y$ be $APTC$ with guarded recursion terms. If $APTC\textrm{ with guarded recursion}\vdash x=y$, then
\begin{enumerate}
  \item $x\sim_{s} y$;
  \item $x\sim_{p} y$;
  \item $x\sim_{hp} y$;
  \item $x\sim_{hhp} y$.
\end{enumerate}
\end{theorem}

\begin{theorem}[Completeness of $APTC$ with linear recursion]
Let $p$ and $q$ be closed $APTC$ with linear recursion terms, then,
\begin{enumerate}
  \item if $p\sim_{s} q$ then $p=q$;
  \item if $p\sim_{p} q$ then $p=q$;
  \item if $p\sim_{hp} q$ then $p=q$;
  \item if $p\sim_{hhp} q$ then $p=q$.
\end{enumerate}
\end{theorem}

\subsubsection{Abstraction}

\begin{definition}[Weak pomset transitions and weak step]
Let $\mathcal{E}$ be a PES and let $C\in\mathcal{C}(\mathcal{E})$, and $\emptyset\neq X\subseteq \hat{\mathbb{E}}$, if $C\cap X=\emptyset$ and
$\hat{C'}=\hat{C}\cup X\in\mathcal{C}(\mathcal{E})$, then $C\xRightarrow{X} C'$ is called a weak pomset transition from $C$ to $C'$, where we define
$\xRightarrow{e}\triangleq\xrightarrow{\tau^*}\xrightarrow{e}\xrightarrow{\tau^*}$. And $\xRightarrow{X}\triangleq\xrightarrow{\tau^*}\xrightarrow{e}\xrightarrow{\tau^*}$,
for every $e\in X$. When the events in $X$ are pairwise concurrent, we say that $C\xRightarrow{X}C'$ is a weak step.
\end{definition}

\begin{definition}[Branching pomset, step bisimulation]\label{BPSB}
Assume a special termination predicate $\downarrow$, and let $\surd$ represent a state with $\surd\downarrow$. Let $\mathcal{E}_1$, $\mathcal{E}_2$ be PESs. A branching pomset
bisimulation is a relation $R\subseteq\mathcal{C}(\mathcal{E}_1)\times\mathcal{C}(\mathcal{E}_2)$, such that:
 \begin{enumerate}
   \item if $(C_1,C_2)\in R$, and $C_1\xrightarrow{X}C_1'$ then
   \begin{itemize}
     \item either $X\equiv \tau^*$, and $(C_1',C_2)\in R$;
     \item or there is a sequence of (zero or more) $\tau$-transitions $C_2\xrightarrow{\tau^*} C_2^0$, such that $(C_1,C_2^0)\in R$ and $C_2^0\xRightarrow{X}C_2'$ with
     $(C_1',C_2')\in R$;
   \end{itemize}
   \item if $(C_1,C_2)\in R$, and $C_2\xrightarrow{X}C_2'$ then
   \begin{itemize}
     \item either $X\equiv \tau^*$, and $(C_1,C_2')\in R$;
     \item or there is a sequence of (zero or more) $\tau$-transitions $C_1\xrightarrow{\tau^*} C_1^0$, such that $(C_1^0,C_2)\in R$ and $C_1^0\xRightarrow{X}C_1'$ with
     $(C_1',C_2')\in R$;
   \end{itemize}
   \item if $(C_1,C_2)\in R$ and $C_1\downarrow$, then there is a sequence of (zero or more) $\tau$-transitions $C_2\xrightarrow{\tau^*}C_2^0$ such that $(C_1,C_2^0)\in R$
   and $C_2^0\downarrow$;
   \item if $(C_1,C_2)\in R$ and $C_2\downarrow$, then there is a sequence of (zero or more) $\tau$-transitions $C_1\xrightarrow{\tau^*}C_1^0$ such that $(C_1^0,C_2)\in R$
   and $C_1^0\downarrow$.
 \end{enumerate}

We say that $\mathcal{E}_1$, $\mathcal{E}_2$ are branching pomset bisimilar, written $\mathcal{E}_1\approx_{bp}\mathcal{E}_2$, if there exists a branching pomset bisimulation $R$,
such that $(\emptyset,\emptyset)\in R$.

By replacing pomset transitions with steps, we can get the definition of branching step bisimulation. When PESs $\mathcal{E}_1$ and $\mathcal{E}_2$ are branching step bisimilar,
we write $\mathcal{E}_1\approx_{bs}\mathcal{E}_2$.
\end{definition}

\begin{definition}[Rooted branching pomset, step bisimulation]\label{RBPSB}
Assume a special termination predicate $\downarrow$, and let $\surd$ represent a state with $\surd\downarrow$. Let $\mathcal{E}_1$, $\mathcal{E}_2$ be PESs. A branching pomset
bisimulation is a relation $R\subseteq\mathcal{C}(\mathcal{E}_1)\times\mathcal{C}(\mathcal{E}_2)$, such that:
 \begin{enumerate}
   \item if $(C_1,C_2)\in R$, and $C_1\xrightarrow{X}C_1'$ then $C_2\xrightarrow{X}C_2'$ with $C_1'\approx_{bp}C_2'$;
   \item if $(C_1,C_2)\in R$, and $C_2\xrightarrow{X}C_2'$ then $C_1\xrightarrow{X}C_1'$ with $C_1'\approx_{bp}C_2'$;
   \item if $(C_1,C_2)\in R$ and $C_1\downarrow$, then $C_2\downarrow$;
   \item if $(C_1,C_2)\in R$ and $C_2\downarrow$, then $C_1\downarrow$.
 \end{enumerate}

We say that $\mathcal{E}_1$, $\mathcal{E}_2$ are rooted branching pomset bisimilar, written $\mathcal{E}_1\approx_{rbp}\mathcal{E}_2$, if there exists a rooted branching pomset
bisimulation $R$, such that $(\emptyset,\emptyset)\in R$.

By replacing pomset transitions with steps, we can get the definition of rooted branching step bisimulation. When PESs $\mathcal{E}_1$ and $\mathcal{E}_2$ are rooted branching step
bisimilar, we write $\mathcal{E}_1\approx_{rbs}\mathcal{E}_2$.
\end{definition}

\begin{definition}[Branching (hereditary) history-preserving bisimulation]\label{BHHPB}
Assume a special termination predicate $\downarrow$, and let $\surd$ represent a state with $\surd\downarrow$. A branching history-preserving (hp-) bisimulation is a posetal
relation $R\subseteq\mathcal{C}(\mathcal{E}_1)\overline{\times}\mathcal{C}(\mathcal{E}_2)$ such that:

 \begin{enumerate}
   \item if $(C_1,f,C_2)\in R$, and $C_1\xrightarrow{e_1}C_1'$ then
   \begin{itemize}
     \item either $e_1\equiv \tau$, and $(C_1',f[e_1\mapsto \tau],C_2)\in R$;
     \item or there is a sequence of (zero or more) $\tau$-transitions $C_2\xrightarrow{\tau^*} C_2^0$, such that $(C_1,f,C_2^0)\in R$ and $C_2^0\xrightarrow{e_2}C_2'$ with
     $(C_1',f[e_1\mapsto e_2],C_2')\in R$;
   \end{itemize}
   \item if $(C_1,f,C_2)\in R$, and $C_2\xrightarrow{e_2}C_2'$ then
   \begin{itemize}
     \item either $X\equiv \tau$, and $(C_1,f[e_2\mapsto \tau],C_2')\in R$;
     \item or there is a sequence of (zero or more) $\tau$-transitions $C_1\xrightarrow{\tau^*} C_1^0$, such that $(C_1^0,f,C_2)\in R$ and $C_1^0\xrightarrow{e_1}C_1'$ with
     $(C_1',f[e_2\mapsto e_1],C_2')\in R$;
   \end{itemize}
   \item if $(C_1,f,C_2)\in R$ and $C_1\downarrow$, then there is a sequence of (zero or more) $\tau$-transitions $C_2\xrightarrow{\tau^*}C_2^0$ such that $(C_1,f,C_2^0)\in R$
   and $C_2^0\downarrow$;
   \item if $(C_1,f,C_2)\in R$ and $C_2\downarrow$, then there is a sequence of (zero or more) $\tau$-transitions $C_1\xrightarrow{\tau^*}C_1^0$ such that $(C_1^0,f,C_2)\in R$
   and $C_1^0\downarrow$.
 \end{enumerate}

$\mathcal{E}_1,\mathcal{E}_2$ are branching history-preserving (hp-)bisimilar and are written $\mathcal{E}_1\approx_{bhp}\mathcal{E}_2$ if there exists a branching hp-bisimulation
$R$ such that $(\emptyset,\emptyset,\emptyset)\in R$.

A branching hereditary history-preserving (hhp-)bisimulation is a downward closed branching hhp-bisimulation. $\mathcal{E}_1,\mathcal{E}_2$ are branching hereditary history-preserving
(hhp-)bisimilar and are written $\mathcal{E}_1\approx_{bhhp}\mathcal{E}_2$.
\end{definition}

\begin{definition}[Rooted branching (hereditary) history-preserving bisimulation]\label{RBHHPB}
Assume a special termination predicate $\downarrow$, and let $\surd$ represent a state with $\surd\downarrow$. A rooted branching history-preserving (hp-) bisimulation is a posetal relation $R\subseteq\mathcal{C}(\mathcal{E}_1)\overline{\times}\mathcal{C}(\mathcal{E}_2)$ such that:

 \begin{enumerate}
   \item if $(C_1,f,C_2)\in R$, and $C_1\xrightarrow{e_1}C_1'$, then $C_2\xrightarrow{e_2}C_2'$ with $C_1'\approx_{bhp}C_2'$;
   \item if $(C_1,f,C_2)\in R$, and $C_2\xrightarrow{e_2}C_1'$, then $C_1\xrightarrow{e_1}C_2'$ with $C_1'\approx_{bhp}C_2'$;
   \item if $(C_1,f,C_2)\in R$ and $C_1\downarrow$, then $C_2\downarrow$;
   \item if $(C_1,f,C_2)\in R$ and $C_2\downarrow$, then $C_1\downarrow$.
 \end{enumerate}

$\mathcal{E}_1,\mathcal{E}_2$ are rooted branching history-preserving (hp-)bisimilar and are written $\mathcal{E}_1\approx_{rbhp}\mathcal{E}_2$ if there exists rooted a branching
hp-bisimulation $R$ such that $(\emptyset,\emptyset,\emptyset)\in R$.

A rooted branching hereditary history-preserving (hhp-)bisimulation is a downward closed rooted branching hhp-bisimulation. $\mathcal{E}_1,\mathcal{E}_2$ are rooted branching
hereditary history-preserving (hhp-)bisimilar and are written $\mathcal{E}_1\approx_{rbhhp}\mathcal{E}_2$.
\end{definition}

\begin{definition}[Guarded linear recursive specification]
A recursive specification is linear if its recursive equations are of the form

$$(a_{11}\leftmerge\cdots\leftmerge a_{1i_1})X_1+\cdots+(a_{k1}\leftmerge\cdots\leftmerge a_{ki_k})X_k+(b_{11}\leftmerge\cdots\leftmerge b_{1j_1})+\cdots+(b_{1j_1}\leftmerge\cdots\leftmerge b_{lj_l})$$

where $a_{11},\cdots,a_{1i_1},a_{k1},\cdots,a_{ki_k},b_{11},\cdots,b_{1j_1},b_{1j_1},\cdots,b_{lj_l}\in \mathbb{E}\cup\{\tau\}$, and the sum above is allowed to be empty, in which case
it represents the deadlock $\delta$.

A linear recursive specification $E$ is guarded if there does not exist an infinite sequence of $\tau$-transitions
$\langle X|E\rangle\xrightarrow{\tau}\langle X'|E\rangle\xrightarrow{\tau}\langle X''|E\rangle\xrightarrow{\tau}\cdots$.
\end{definition}

The transition rules of $\tau$ are shown in Table \ref{TRForTau}, and axioms of $\tau$ are as Table \ref{AxiomsForTauLeft} shows.

\begin{center}
    \begin{table}
        $$\frac{}{\tau\xrightarrow{\tau}\surd}$$
        $$\frac{x\xrightarrow{e}\surd}{\tau_I(x)\xrightarrow{e}\surd}\quad e\notin I
        \quad\quad\frac{x\xrightarrow{e}x'}{\tau_I(x)\xrightarrow{e}\tau_I(x')}\quad e\notin I$$

        $$\frac{x\xrightarrow{e}\surd}{\tau_I(x)\xrightarrow{\tau}\surd}\quad e\in I
        \quad\quad\frac{x\xrightarrow{e}x'}{\tau_I(x)\xrightarrow{\tau}\tau_I(x')}\quad e\in I$$
        \caption{Transition rule of $\textrm{APTC}_{\tau}$}
        \label{TRForTau}
    \end{table}
\end{center}

\begin{theorem}[Conservitivity of $APTC$ with silent step and guarded linear recursion]
$APTC$ with silent step and guarded linear recursion is a conservative extension of $APTC$ with linear recursion.
\end{theorem}

\begin{theorem}[Congruence theorem of $APTC$ with silent step and guarded linear recursion]
Rooted branching truly concurrent bisimulation equivalences $\approx_{rbp}$, $\approx_{rbs}$, $\approx_{rbhp}$, and $\approx_{rbhhp}$ are all congruences with respect to $APTC$ with
silent step and guarded linear recursion.
\end{theorem}

\begin{center}
\begin{table}
  \begin{tabular}{@{}ll@{}}
\hline No. &Axiom\\
  $B1$ & $e\cdot\tau=e$\\
  $B2$ & $e\cdot(\tau\cdot(x+y)+x)=e\cdot(x+y)$\\
  $B3$ & $x\leftmerge\tau=x$\\
\end{tabular}
\caption{Axioms of silent step}
\label{AxiomsForTauLeft}
\end{table}
\end{center}

\begin{theorem}[Elimination theorem of $APTC$ with silent step and guarded linear recursion]
Each process term in $APTC$ with silent step and guarded linear recursion is equal to a process term $\langle X_1|E\rangle$ with $E$ a guarded linear recursive specification.
\end{theorem}

\begin{theorem}[Soundness of $APTC$ with silent step and guarded linear recursion]
Let $x$ and $y$ be $APTC$ with silent step and guarded linear recursion terms. If $APTC$ with silent step and guarded linear recursion $\vdash x=y$, then
\begin{enumerate}
  \item $x\approx_{rbs} y$;
  \item $x\approx_{rbp} y$;
  \item $x\approx_{rbhp} y$;
  \item $x\approx_{rbhhp} y$.
\end{enumerate}
\end{theorem}

\begin{theorem}[Completeness of $APTC$ with silent step and guarded linear recursion]
Let $p$ and $q$ be closed $APTC$ with silent step and guarded linear recursion terms, then,
\begin{enumerate}
  \item if $p\approx_{rbs} q$ then $p=q$;
  \item if $p\approx_{rbp} q$ then $p=q$;
  \item if $p\approx_{rbhp} q$ then $p=q$;
  \item if $p\approx_{rbhhp} q$ then $p=q$.
\end{enumerate}
\end{theorem}

The transition rules of $\tau_I$ are shown in Table \ref{TRForTau}, and the axioms are shown in Table \ref{AxiomsForAbstractionLeft}.

\begin{theorem}[Conservitivity of $APTC_{\tau}$ with guarded linear recursion]
$APTC_{\tau}$ with guarded linear recursion is a conservative extension of $APTC$ with silent step and guarded linear recursion.
\end{theorem}

\begin{theorem}[Congruence theorem of $APTC_{\tau}$ with guarded linear recursion]
Rooted branching truly concurrent bisimulation equivalences $\approx_{rbp}$, $\approx_{rbs}$, $\approx_{rbhp}$ and $\approx_{rbhhp}$ are all congruences with respect to $APTC_{\tau}$
with guarded linear recursion.
\end{theorem}

\begin{center}
\begin{table}
  \begin{tabular}{@{}ll@{}}
\hline No. &Axiom\\
  $TI1$ & $e\notin I\quad \tau_I(e)=e$\\
  $TI2$ & $e\in I\quad \tau_I(e)=\tau$\\
  $TI3$ & $\tau_I(\delta)=\delta$\\
  $TI4$ & $\tau_I(x+y)=\tau_I(x)+\tau_I(y)$\\
  $TI5$ & $\tau_I(x\cdot y)=\tau_I(x)\cdot\tau_I(y)$\\
  $TI6$ & $\tau_I(x\leftmerge y)=\tau_I(x)\leftmerge\tau_I(y)$\\
\end{tabular}
\caption{Axioms of abstraction operator}
\label{AxiomsForAbstractionLeft}
\end{table}
\end{center}

\begin{theorem}[Soundness of $APTC_{\tau}$ with guarded linear recursion]
Let $x$ and $y$ be $APTC_{\tau}$ with guarded linear recursion terms. If $APTC_{\tau}$ with guarded linear recursion $\vdash x=y$, then
\begin{enumerate}
  \item $x\approx_{rbs} y$;
  \item $x\approx_{rbp} y$;
  \item $x\approx_{rbhp} y$;
  \item $x\approx_{rbhhp} y$.
\end{enumerate}
\end{theorem}

\begin{definition}[Cluster]
Let $E$ be a guarded linear recursive specification, and $I\subseteq \mathbb{E}$. Two recursion variable $X$ and $Y$ in $E$ are in the same cluster for $I$ iff there exist sequences of
transitions $\langle X|E\rangle\xrightarrow{\{b_{11},\cdots, b_{1i}\}}\cdots\xrightarrow{\{b_{m1},\cdots, b_{mi}\}}\langle Y|E\rangle$ and $\langle Y|E\rangle\xrightarrow{\{c_{11},\cdots, c_{1j}\}}\cdots\xrightarrow{\{c_{n1},\cdots, c_{nj}\}}\langle X|E\rangle$, where $b_{11},\cdots,b_{mi},c_{11},\cdots,c_{nj}\in I\cup\{\tau\}$.

$a_1\leftmerge\cdots\leftmerge a_k$ or $(a_1\leftmerge\cdots\leftmerge a_k) X$ is an exit for the cluster $C$ iff: (1) $a_1\leftmerge\cdots\leftmerge a_k$ or
$(a_1\leftmerge\cdots\leftmerge a_k) X$ is a summand at the right-hand side of the recursive equation for a recursion variable in $C$, and (2) in the case of
$(a_1\leftmerge\cdots\leftmerge a_k) X$, either $a_l\notin I\cup\{\tau\}(l\in\{1,2,\cdots,k\})$ or $X\notin C$.
\end{definition}

\begin{center}
\begin{table}
  \begin{tabular}{@{}ll@{}}
\hline No. &Axiom\\
  $CFAR$ & If $X$ is in a cluster for $I$ with exits \\
           & $\{(a_{11}\leftmerge\cdots\leftmerge a_{1i})Y_1,\cdots,(a_{m1}\leftmerge\cdots\leftmerge a_{mi})Y_m, b_{11}\leftmerge\cdots\leftmerge b_{1j},\cdots,b_{n1}\leftmerge\cdots\leftmerge b_{nj}\}$, \\
           & then $\tau\cdot\tau_I(\langle X|E\rangle)=$\\
           & $\tau\cdot\tau_I((a_{11}\leftmerge\cdots\leftmerge a_{1i})\langle Y_1|E\rangle+\cdots+(a_{m1}\leftmerge\cdots\leftmerge a_{mi})\langle Y_m|E\rangle+b_{11}\leftmerge\cdots\leftmerge b_{1j}+\cdots+b_{n1}\leftmerge\cdots\leftmerge b_{nj})$\\
\end{tabular}
\caption{Cluster fair abstraction rule}
\label{CFARLeft}
\end{table}
\end{center}

\begin{theorem}[Soundness of $CFAR$]
$CFAR$ is sound modulo rooted branching truly concurrent bisimulation equivalences $\approx_{rbs}$, $\approx_{rbp}$, $\approx_{rbhp}$ and $\approx_{rbhhp}$.
\end{theorem}

\begin{theorem}[Completeness of $APTC_{\tau}$ with guarded linear recursion and $CFAR$]
Let $p$ and $q$ be closed $APTC_{\tau}$ with guarded linear recursion and $CFAR$ terms, then,
\begin{enumerate}
  \item if $p\approx_{rbs} q$ then $p=q$;
  \item if $p\approx_{rbp} q$ then $p=q$;
  \item if $p\approx_{rbhp} q$ then $p=q$;
  \item if $p\approx_{rbhhp} q$ then $p=q$.
\end{enumerate}
\end{theorem}

\subsection{Probabilistic Operational Semantics for True Concurrency}

In the following, the variables $x,x',y,y',z,z'$ range over the collection of process terms, $s,s',t,t',u,u'$ are closed terms, $\tau$ is the special constant silent step, $\delta$
is the special constant deadlock, $A$ is the collection of atomic actions, atomic actions $a,b\in A$, $A_{\delta}=A\cup\{\delta\}$, $A_{\tau}=A\cup\{\tau\}$.
$\rightsquigarrow$ denotes probabilistic transition, and action transition labelled by an atomic action $a\in A$, $\xrightarrow{a}$ and $\xrightarrow{a}\surd$.
$x\xrightarrow{a}p$ means that by performing action $a$ process $x$ evolves into $p$; while $x\xrightarrow{a}\surd$ means that $x$ performs an $a$ action and then terminates.
$p\rightsquigarrow x$ denotes that process $p$ chooses to behave like process $x$ with a non-zero probability $\pi >0$.

\begin{definition}[Probabilistic prime event structure with silent event]\label{PPES}
Let $\Lambda$ be a fixed set of labels, ranged over $a,b,c,\cdots$ and $\tau$. A ($\Lambda$-labelled) prime event structure with silent event $\tau$ is a quintuple
$\mathcal{E}=\langle \mathbb{E}, \leq, \sharp, \sharp_{\pi}, \lambda\rangle$, where $\mathbb{E}$ is a denumerable set of events, including the silent event $\tau$. Let
$\hat{\mathbb{E}}=\mathbb{E}\backslash\{\tau\}$, exactly excluding $\tau$, it is obvious that $\hat{\tau^*}=\epsilon$, where $\epsilon$ is the empty event.
Let $\lambda:\mathbb{E}\rightarrow\Lambda$ be a labelling function and let $\lambda(\tau)=\tau$. And $\leq$, $\sharp$, $\sharp_{\pi}$ are binary relations on $\mathbb{E}$,
called causality, conflict and probabilistic conflict respectively, such that:

\begin{enumerate}
  \item $\leq$ is a partial order and $\lceil e \rceil = \{e'\in \mathbb{E}|e'\leq e\}$ is finite for all $e\in \mathbb{E}$. It is easy to see that
  $e\leq\tau^*\leq e'=e\leq\tau\leq\cdots\leq\tau\leq e'$, then $e\leq e'$.
  \item $\sharp$ is irreflexive, symmetric and hereditary with respect to $\leq$, that is, for all $e,e',e''\in \mathbb{E}$, if $e\sharp e'\leq e''$, then $e\sharp e''$;
  \item $\sharp_{\pi}$ is irreflexive, symmetric and hereditary with respect to $\leq$, that is, for all $e,e',e''\in \mathbb{E}$, if $e\sharp_{\pi} e'\leq e''$, then $e\sharp_{\pi} e''$.
\end{enumerate}

Then, the concepts of consistency and concurrency can be drawn from the above definition:

\begin{enumerate}
  \item $e,e'\in \mathbb{E}$ are consistent, denoted as $e\frown e'$, if $\neg(e\sharp e')$ and $\neg(e\sharp_{\pi} e')$. A subset $X\subseteq \mathbb{E}$ is called consistent, if $e\frown e'$ for all
  $e,e'\in X$.
  \item $e,e'\in \mathbb{E}$ are concurrent, denoted as $e\parallel e'$, if $\neg(e\leq e')$, $\neg(e'\leq e)$, and $\neg(e\sharp e')$ and $\neg(e\sharp_{\pi} e')$.
\end{enumerate}
\end{definition}

\begin{definition}[Configuration]
Let $\mathcal{E}$ be a PES. A (finite) configuration in $\mathcal{E}$ is a (finite) consistent subset of events $C\subseteq \mathcal{E}$, closed with respect to causality
(i.e. $\lceil C\rceil=C$). The set of finite configurations of $\mathcal{E}$ is denoted by $\mathcal{C}(\mathcal{E})$. We let $\hat{C}=C\backslash\{\tau\}$.
\end{definition}

A consistent subset of $X\subseteq \mathbb{E}$ of events can be seen as a pomset. Given $X, Y\subseteq \mathbb{E}$, $\hat{X}\sim \hat{Y}$ if $\hat{X}$ and $\hat{Y}$ are
isomorphic as pomsets. In the following of the paper, we say $C_1\sim C_2$, we mean $\hat{C_1}\sim\hat{C_2}$.

\begin{definition}[Pomset transitions and step]
Let $\mathcal{E}$ be a PES and let $C\in\mathcal{C}(\mathcal{E})$, and $\emptyset\neq X\subseteq \mathbb{E}$, if $C\cap X=\emptyset$ and $C'=C\cup X\in\mathcal{C}(\mathcal{E})$,
then $C\xrightarrow{X} C'$ is called a pomset transition from $C$ to $C'$. When the events in $X$ are pairwise concurrent, we say that $C\xrightarrow{X}C'$ is a step.
\end{definition}

\begin{definition}[Probabilistic transitions]
Let $\mathcal{E}$ be a PES and let $C\in\mathcal{C}(\mathcal{E})$, the transition $C\xrsquigarrow{\pi} C^{\pi}$ is called a probabilistic transition from $C$ to $C^{\pi}$.
\end{definition}

\begin{definition}[Weak pomset transitions and weak step]
Let $\mathcal{E}$ be a PES and let $C\in\mathcal{C}(\mathcal{E})$, and $\emptyset\neq X\subseteq \hat{\mathbb{E}}$, if $C\cap X=\emptyset$ and
$\hat{C'}=\hat{C}\cup X\in\mathcal{C}(\mathcal{E})$, then $C\xRightarrow{X} C'$ is called a weak pomset transition from $C$ to $C'$, where we define
$\xRightarrow{e}\triangleq\xrightarrow{\tau^*}\xrightarrow{e}\xrightarrow{\tau^*}$. And $\xRightarrow{X}\triangleq\xrightarrow{\tau^*}\xrightarrow{e}\xrightarrow{\tau^*}$,
for every $e\in X$. When the events in $X$ are pairwise concurrent, we say that $C\xRightarrow{X}C'$ is a weak step.
\end{definition}

We will also suppose that all the PESs in this book are image finite, that is, for any PES $\mathcal{E}$ and $C\in \mathcal{C}(\mathcal{E})$ and $a\in \Lambda$,
$\{\langle C,s\rangle\xrsquigarrow{\pi} \langle C^{\pi},s\rangle\}$,
$\{e\in \mathbb{E}|\langle C,s\rangle\xrightarrow{e} \langle C',s'\rangle\wedge \lambda(e)=a\}$ and
$\{e\in\hat{\mathbb{E}}|\langle C,s\rangle\xRightarrow{e} \langle C',s'\rangle\wedge \lambda(e)=a\}$ is finite.

A probability distribution function (PDF) $\mu$ is a map $\mu:\mathcal{C}\times\mathcal{C}\rightarrow[0,1]$ and $\mu^*$ is the cumulative probability distribution function (cPDF).

\begin{definition}[Probabilistic pomset, step bisimulation]\label{PPSB}
Let $\mathcal{E}_1$, $\mathcal{E}_2$ be PESs. A probabilistic pomset bisimulation is a relation $R\subseteq\mathcal{C}(\mathcal{E}_1)\times\mathcal{C}(\mathcal{E}_2)$, such that (1) if
$(C_1,C_2)\in R$, and $C_1\xrightarrow{X_1}C_1'$ then $C_2\xrightarrow{X_2}C_2'$, with $X_1\subseteq \mathbb{E}_1$, $X_2\subseteq \mathbb{E}_2$, $X_1\sim X_2$ and $(C_1',C_2')\in R$,
and vice-versa; (2) if $(C_1,C_2)\in R$, and $C_1\xrsquigarrow{\pi}C_1^{\pi}$ then $C_2\xrsquigarrow{\pi}C_2^{\pi}$ and $(C_1^{\pi},C_2^{\pi})\in R$, and vice-versa; (3) if $(C_1,C_2)\in R$,
then $\mu(C_1,C)=\mu(C_2,C)$ for each $C\in\mathcal{C}(\mathcal{E})/R$; (4) $[\surd]_R=\{\surd\}$. We say that $\mathcal{E}_1$, $\mathcal{E}_2$ are probabilistic pomset bisimilar, written $\mathcal{E}_1\sim_{pp}\mathcal{E}_2$,
if there exists a probabilistic pomset bisimulation $R$, such that
$(\emptyset,\emptyset)\in R$. By replacing probabilistic pomset transitions with steps, we can get the definition of probabilistic step bisimulation. When PESs $\mathcal{E}_1$ and $\mathcal{E}_2$ are probabilistic step
bisimilar, we write $\mathcal{E}_1\sim_{ps}\mathcal{E}_2$.
\end{definition}

\begin{definition}[Posetal product]
Given two PESs $\mathcal{E}_1$, $\mathcal{E}_2$, the posetal product of their configurations, denoted $\mathcal{C}(\mathcal{E}_1)\overline{\times}\mathcal{C}(\mathcal{E}_2)$,
is defined as

$$\{(C_1,f,C_2)|C_1\in\mathcal{C}(\mathcal{E}_1),C_2\in\mathcal{C}(\mathcal{E}_2),f:C_1\rightarrow C_2 \textrm{ isomorphism}\}.$$

A subset $R\subseteq\mathcal{C}(\mathcal{E}_1)\overline{\times}\mathcal{C}(\mathcal{E}_2)$ is called a posetal relation. We say that $R$ is downward closed when for any
$(C_1,f,C_2),(C_1',f',C_2')\in \mathcal{C}(\mathcal{E}_1)\overline{\times}\mathcal{C}(\mathcal{E}_2)$, if $(C_1,f,C_2)\subseteq (C_1',f',C_2')$ pointwise and $(C_1',f',C_2')\in R$,
then $(C_1,f,C_2)\in R$.

For $f:X_1\rightarrow X_2$, we define $f[x_1\mapsto x_2]:X_1\cup\{x_1\}\rightarrow X_2\cup\{x_2\}$, $z\in X_1\cup\{x_1\}$,(1)$f[x_1\mapsto x_2](z)=
x_2$,if $z=x_1$;(2)$f[x_1\mapsto x_2](z)=f(z)$, otherwise. Where $X_1\subseteq \mathbb{E}_1$, $X_2\subseteq \mathbb{E}_2$, $x_1\in \mathbb{E}_1$, $x_2\in \mathbb{E}_2$.
\end{definition}

\begin{definition}[Probabilistic (hereditary) history-preserving bisimulation]\label{PHHPB}
A probabilistic history-preserving (hp-) bisimulation is a posetal relation $R\subseteq\mathcal{C}(\mathcal{E}_1)\overline{\times}\mathcal{C}(\mathcal{E}_2)$ such that (1) if $(C_1,f,C_2)\in R$,
and $C_1\xrightarrow{e_1} C_1'$, then $C_2\xrightarrow{e_2} C_2'$, with $(C_1',f[e_1\mapsto e_2],C_2')\in R$, and vice-versa; (2) if $(C_1,f,C_2)\in R$, and
$C_1\xrsquigarrow{\pi}C_1^{\pi}$ then $C_2\xrsquigarrow{\pi}C_2^{\pi}$ and $(C_1^{\pi},f,C_2^{\pi})\in R$, and vice-versa; (3) if $(C_1,f,C_2)\in R$,
then $\mu(C_1,C)=\mu(C_2,C)$ for each $C\in\mathcal{C}(\mathcal{E})/R$; (4) $[\surd]_R=\{\surd\}$. $\mathcal{E}_1,\mathcal{E}_2$ are probabilistic history-preserving
(hp-)bisimilar and are written $\mathcal{E}_1\sim_{php}\mathcal{E}_2$ if there exists a probabilistic hp-bisimulation $R$ such that $(\emptyset,\emptyset,\emptyset)\in R$.

A probabilistic hereditary history-preserving (hhp-)bisimulation is a downward closed probabilistic hp-bisimulation. $\mathcal{E}_1,\mathcal{E}_2$ are probabilistic hereditary history-preserving (hhp-)bisimilar and are
written $\mathcal{E}_1\sim_{phhp}\mathcal{E}_2$.
\end{definition}

\begin{definition}[Weakly probabilistic pomset, step bisimulation]\label{WPSB}
Let $\mathcal{E}_1$, $\mathcal{E}_2$ be PESs. A weakly probabilistic pomset bisimulation is a relation $R\subseteq\mathcal{C}(\mathcal{E}_1)\times\mathcal{C}(\mathcal{E}_2)$, such that (1) if
$(C_1,C_2)\in R$, and $C_1\xRightarrow{X_1}C_1'$ then $C_2\xRightarrow{X_2}C_2'$, with $X_1\subseteq \hat{\mathbb{E}_1}$, $X_2\subseteq \hat{\mathbb{E}_2}$, $X_1\sim X_2$ and
$(C_1',C_2')\in R$, and vice-versa; (2) if $(C_1,C_2)\in R$, and $C_1\xrsquigarrow{\pi}C_1^{\pi}$ then $C_2\xrsquigarrow{\pi}C_2^{\pi}$ and $(C_1^{\pi},C_2^{\pi})\in R$, and vice-versa; (3) if $(C_1,C_2)\in R$,
then $\mu(C_1,C)=\mu(C_2,C)$ for each $C\in\mathcal{C}(\mathcal{E})/R$; (4) $[\surd]_R=\{\surd\}$. We say that $\mathcal{E}_1$, $\mathcal{E}_2$ are weakly pomset bisimilar,
written $\mathcal{E}_1\approx_{pp}\mathcal{E}_2$, if there exists a weak probabilistic pomset bisimulation $R$, such that $(\emptyset,\emptyset)\in R$. By replacing weakly probabilistic pomset transitions with
weak steps, we can get the definition of weakly probabilistic step bisimulation. When PESs $\mathcal{E}_1$ and $\mathcal{E}_2$ are weakly probabilistic step bisimilar, we write
$\mathcal{E}_1\approx_{ps}\mathcal{E}_2$.
\end{definition}

\begin{definition}[Weakly posetal product]
Given two PESs $\mathcal{E}_1$, $\mathcal{E}_2$, the weakly posetal product of their configurations, denoted $\mathcal{C}(\mathcal{E}_1)\overline{\times}\mathcal{C}(\mathcal{E}_2)$,
is defined as

$$\{(C_1,f,C_2)|C_1\in\mathcal{C}(\mathcal{E}_1),C_2\in\mathcal{C}(\mathcal{E}_2),f:\hat{C_1}\rightarrow \hat{C_2} \textrm{ isomorphism}\}.$$

A subset $R\subseteq\mathcal{C}(\mathcal{E}_1)\overline{\times}\mathcal{C}(\mathcal{E}_2)$ is called a weakly posetal relation. We say that $R$ is downward closed when for any
$(C_1,f,C_2),(C_1',f,C_2')\in \mathcal{C}(\mathcal{E}_1)\overline{\times}\mathcal{C}(\mathcal{E}_2)$, if $(C_1,f,C_2)\subseteq (C_1',f',C_2')$ pointwise and $(C_1',f',C_2')\in R$,
then $(C_1,f,C_2)\in R$.

For $f:X_1\rightarrow X_2$, we define $f[x_1\mapsto x_2]:X_1\cup\{x_1\}\rightarrow X_2\cup\{x_2\}$, $z\in X_1\cup\{x_1\}$,(1)$f[x_1\mapsto x_2](z)=
x_2$,if $z=x_1$;(2)$f[x_1\mapsto x_2](z)=f(z)$, otherwise. Where $X_1\subseteq \hat{\mathbb{E}_1}$, $X_2\subseteq \hat{\mathbb{E}_2}$, $x_1\in \hat{\mathbb{E}}_1$,
$x_2\in \hat{\mathbb{E}}_2$. Also, we define $f(\tau^*)=f(\tau^*)$.
\end{definition}

\begin{definition}[Weakly probabilistic (hereditary) history-preserving bisimulation]\label{WHHPB}
A weakly probabilistic history-preserving (hp-) bisimulation is a weakly posetal relation $R\subseteq\mathcal{C}(\mathcal{E}_1)\overline{\times}\mathcal{C}(\mathcal{E}_2)$ such that
(1) if $(C_1,f,C_2)\in R$, and $C_1\xRightarrow{e_1} C_1'$, then $C_2\xRightarrow{e_2} C_2'$, with $(C_1',f[e_1\mapsto e_2],C_2')\in R$, and vice-versa; (2) if $(C_1,f, C_2)\in R$, and
$C_1\xrsquigarrow{\pi}C_1^{\pi}$ then $C_2\xrsquigarrow{\pi}C_2^{\pi}$ and $(C_1^{\pi},f,C_2^{\pi})\in R$, and vice-versa; (3) if $(C_1,f,C_2)\in R$,
then $\mu(C_1,C)=\mu(C_2,C)$ for each $C\in\mathcal{C}(\mathcal{E})/R$; (4) $[\surd]_R=\{\surd\}$. $\mathcal{E}_1,\mathcal{E}_2$
are weakly probabilistic history-preserving (hp-)bisimilar and are written $\mathcal{E}_1\approx_{php}\mathcal{E}_2$ if there exists a weakly probabilistic hp-bisimulation $R$ such that
$(\emptyset,\emptyset,\emptyset)\in R$.

A weakly probabilistic hereditary history-preserving (hhp-)bisimulation is a downward closed weakly probabilistic hp-bisimulation. $\mathcal{E}_1,\mathcal{E}_2$ are weakly probabilistic
hereditary history-preserving (hhp-)bisimilar and are written $\mathcal{E}_1\approx_{phhp}\mathcal{E}_2$.
\end{definition}

\begin{definition}[Probabilistic branching pomset, step bisimulation]\label{PBPSB}
Assume a special termination predicate $\downarrow$, and let $\surd$ represent a state with $\surd\downarrow$. Let $\mathcal{E}_1$, $\mathcal{E}_2$ be PESs. A probabilistic branching pomset
bisimulation is a relation $R\subseteq\mathcal{C}(\mathcal{E}_1)\times\mathcal{C}(\mathcal{E}_2)$, such that:

 \begin{enumerate}
   \item if $(C_1,C_2)\in R$, and $C_1\xrightarrow{X}C_1'$ then
   \begin{itemize}
     \item either $X\equiv \tau^*$, and $(C_1',C_2)\in R$;
     \item or there is a sequence of (zero or more) probabilistic transitions and $\tau$-transitions $C_2\rightsquigarrow^*\xrightarrow{\tau^*} C_2^0$, such that $(C_1,C_2^0)\in R$ and $C_2^0\xRightarrow{X}C_2'$ with
     $(C_1',C_2')\in R$;
   \end{itemize}
   \item if $(C_1,C_2)\in R$, and $C_2\xrightarrow{X}C_2'$ then
   \begin{itemize}
     \item either $X\equiv \tau^*$, and $(C_1,C_2')\in R$;
     \item or there is a sequence of (zero or more) probabilistic transitions and $\tau$-transitions $C_1\rightsquigarrow^*\xrightarrow{\tau^*} C_1^0$, such that $(C_1^0,C_2)\in R$ and $C_1^0\xRightarrow{X}C_1'$ with
     $(C_1',C_2')\in R$;
   \end{itemize}
   \item if $(C_1,C_2)\in R$ and $C_1\downarrow$, then there is a sequence of (zero or more) probabilistic transitions and $\tau$-transitions $C_2\rightsquigarrow^*\xrightarrow{\tau^*} C_2^0$ such that $(C_1,C_2^0)\in R$
   and $C_2^0\downarrow$;
   \item if $(C_1,C_2)\in R$ and $C_2\downarrow$, then there is a sequence of (zero or more) probabilistic transitions and $\tau$-transitions $C_1\rightsquigarrow^*\xrightarrow{\tau^*} C_1^0$ such that $(C_1^0,C_2)\in R$
   and $C_1^0\downarrow$.
   \item if $(C_1,C_2)\in R$,then $\mu(C_1,C)=\mu(C_2,C)$ for each $C\in\mathcal{C}(\mathcal{E})/R$;
   \item $[\surd]_R=\{\surd\}$.
 \end{enumerate}

We say that $\mathcal{E}_1$, $\mathcal{E}_2$ are probabilistic branching pomset bisimilar, written $\mathcal{E}_1\approx_{pbp}\mathcal{E}_2$, if there exists a probabilistic branching pomset bisimulation $R$,
such that $(\emptyset,\emptyset)\in R$.

By replacing probabilistic branching pomset transitions with steps, we can get the definition of probabilistic branching step bisimulation. When PESs $\mathcal{E}_1$ and $\mathcal{E}_2$ are probabilistic branching step bisimilar,
we write $\mathcal{E}_1\approx_{pbs}\mathcal{E}_2$.
\end{definition}

\begin{definition}[Probabilistic rooted branching pomset, step bisimulation]\label{PRBPSB}
Assume a special termination predicate $\downarrow$, and let $\surd$ represent a state with $\surd\downarrow$. Let $\mathcal{E}_1$, $\mathcal{E}_2$ be PESs. A branching pomset
bisimulation is a relation $R\subseteq\mathcal{C}(\mathcal{E}_1)\times\mathcal{C}(\mathcal{E}_2)$, such that:

 \begin{enumerate}
   \item if $(C_1,C_2)\in R$, and $C_1\rightsquigarrow C_1^{\pi}\xrightarrow{X}C_1'$ then $C_2\rightsquigarrow C_2^{\pi}\xrightarrow{X}C_2'$ with $C_1'\approx_{pbp}C_2'$;
   \item if $(C_1,C_2)\in R$, and $C_2\rightsquigarrow C_2^{\pi}\xrightarrow{X}C_2'$ then $C_1\rightsquigarrow C_1^{\pi}\xrightarrow{X}C_1'$ with $C_1'\approx_{pbp}C_2'$;
   \item if $(C_1,C_2)\in R$ and $C_1\downarrow$, then $C_2\downarrow$;
   \item if $(C_1,C_2)\in R$ and $C_2\downarrow$, then $C_1\downarrow$.
 \end{enumerate}

We say that $\mathcal{E}_1$, $\mathcal{E}_2$ are probabilistic rooted branching pomset bisimilar, written $\mathcal{E}_1\approx_{prbp}\mathcal{E}_2$, if there exists a probabilistic rooted branching pomset
bisimulation $R$, such that $(\emptyset,\emptyset)\in R$.

By replacing probabilistic pomset transitions with steps, we can get the definition of probabilistic rooted branching step bisimulation. When PESs $\mathcal{E}_1$ and $\mathcal{E}_2$ are probabilistic rooted branching step
bisimilar, we write $\mathcal{E}_1\approx_{prbs}\mathcal{E}_2$.
\end{definition}

\begin{definition}[Probabilistic branching (hereditary) history-preserving bisimulation]\label{PBHHPB}
Assume a special termination predicate $\downarrow$, and let $\surd$ represent a state with $\surd\downarrow$. A probabilistic branching history-preserving (hp-) bisimulation is a weakly
posetal relation $R\subseteq\mathcal{C}(\mathcal{E}_1)\overline{\times}\mathcal{C}(\mathcal{E}_2)$ such that:

 \begin{enumerate}
   \item if $(C_1,f,C_2)\in R$, and $C_1\xrightarrow{e_1}C_1'$ then
   \begin{itemize}
     \item either $e_1\equiv \tau$, and $(C_1',f[e_1\mapsto \tau],C_2)\in R$;
     \item or there is a sequence of (zero or more) probabilistic transitions and $\tau$-transitions $C_2\rightsquigarrow^*\xrightarrow{\tau^*} C_2^0$, such that $(C_1,f,C_2^0)\in R$ and $C_2^0\xrightarrow{e_2}C_2'$ with
     $(C_1',f[e_1\mapsto e_2],C_2')\in R$;
   \end{itemize}
   \item if $(C_1,f,C_2)\in R$, and $C_2\xrightarrow{e_2}C_2'$ then
   \begin{itemize}
     \item either $X\equiv \tau$, and $(C_1,f[e_2\mapsto \tau],C_2')\in R$;
     \item or there is a sequence of (zero or more) probabilistic transitions and $\tau$-transitions $C_1\rightsquigarrow^*\xrightarrow{\tau^*} C_1^0$, such that $(C_1^0,f,C_2)\in R$ and $C_1^0\xrightarrow{e_1}C_1'$ with
     $(C_1',f[e_2\mapsto e_1],C_2')\in R$;
   \end{itemize}
   \item if $(C_1,f,C_2)\in R$ and $C_1\downarrow$, then there is a sequence of (zero or more) probabilistic transitions and $\tau$-transitions $C_2\rightsquigarrow^*\xrightarrow{\tau^*} C_2^0$ such that $(C_1,f,C_2^0)\in R$
   and $C_2^0\downarrow$;
   \item if $(C_1,f,C_2)\in R$ and $C_2\downarrow$, then there is a sequence of (zero or more) probabilistic transitions and $\tau$-transitions $C_1\rightsquigarrow^*\xrightarrow{\tau^*} C_1^0$ such that $(C_1^0,f,C_2)\in R$
   and $C_1^0\downarrow$;
   \item if $(C_1,C_2)\in R$,then $\mu(C_1,C)=\mu(C_2,C)$ for each $C\in\mathcal{C}(\mathcal{E})/R$;
   \item $[\surd]_R=\{\surd\}$.
 \end{enumerate}

$\mathcal{E}_1,\mathcal{E}_2$ are probabilistic branching history-preserving (hp-)bisimilar and are written $\mathcal{E}_1\approx_{pbhp}\mathcal{E}_2$ if there exists a probabilistic branching hp-bisimulation
$R$ such that $(\emptyset,\emptyset,\emptyset)\in R$.

A probabilistic branching hereditary history-preserving (hhp-)bisimulation is a downward closed probabilistic branching hhp-bisimulation. $\mathcal{E}_1,\mathcal{E}_2$ are probabilistic branching hereditary history-preserving
(hhp-)bisimilar and are written $\mathcal{E}_1\approx_{pbhhp}\mathcal{E}_2$.
\end{definition}

\begin{definition}[Probabilistic rooted branching (hereditary) history-preserving bisimulation]\label{PRBHHPB}
Assume a special termination predicate $\downarrow$, and let $\surd$ represent a state with $\surd\downarrow$. A probabilistic rooted branching history-preserving (hp-) bisimulation is a
posetal relation $R\subseteq\mathcal{C}(\mathcal{E}_1)\overline{\times}\mathcal{C}(\mathcal{E}_2)$ such that:

 \begin{enumerate}
   \item if $(C_1,f,C_2)\in R$, and $C_1\rightsquigarrow C_1^{\pi}\xrightarrow{e_1}C_1'$, then $C_2\rightsquigarrow C_2^{\pi}\xrightarrow{e_2}C_2'$ with $C_1'\approx_{pbhp}C_2'$;
   \item if $(C_1,f,C_2)\in R$, and $C_2\rightsquigarrow C_2^{\pi}\xrightarrow{e_2}C_1'$, then $C_1\rightsquigarrow C_1^{\pi}\xrightarrow{e_1}C_2'$ with $C_1'\approx_{pbhp}C_2'$;
   \item if $(C_1,f,C_2)\in R$ and $C_1\downarrow$, then $C_2\downarrow$;
   \item if $(C_1,f,C_2)\in R$ and $C_2\downarrow$, then $C_1\downarrow$.
 \end{enumerate}

$\mathcal{E}_1,\mathcal{E}_2$ are probabilistic rooted branching history-preserving (hp-)bisimilar and are written $\mathcal{E}_1\approx_{prbhp}\mathcal{E}_2$ if there exists a probabilistic rooted branching
hp-bisimulation $R$ such that $(\emptyset,\emptyset,\emptyset)\in R$.

A probabilistic rooted branching hereditary history-preserving (hhp-)bisimulation is a downward closed probabilistic rooted branching hhp-bisimulation. $\mathcal{E}_1,\mathcal{E}_2$ are probabilistic rooted branching
hereditary history-preserving (hhp-)bisimilar and are written $\mathcal{E}_1\approx_{prbhhp}\mathcal{E}_2$.
\end{definition}

\newpage\section{A Calculus for Probabilistic True Concurrency}\label{cptc}

In this chapter, we design a calculus for probabilistic true concurrency (CPTC). This chapter is organized as follows. We introduce strongly probabilistic
truly concurrent bisimulations in section \ref{stcb3}, its properties for weakly probabilistic truly concurrent bisimulations in section \ref{wtcb3}.

\subsection{Syntax and Operational Semantics}\label{sos3}

We assume an infinite set $\mathcal{N}$ of (action or event) names, and use $a,b,c,\cdots$ to range over $\mathcal{N}$. We denote by $\overline{\mathcal{N}}$ the set of co-names and
let $\overline{a},\overline{b},\overline{c},\cdots$ range over $\overline{\mathcal{N}}$. Then we set $\mathcal{L}=\mathcal{N}\cup\overline{\mathcal{N}}$ as the set of labels, and use
$l,\overline{l}$ to range over $\mathcal{L}$. We extend complementation to $\mathcal{L}$ such that $\overline{\overline{a}}=a$. Let $\tau$ denote the silent step (internal action or
event) and define $Act=\mathcal{L}\cup\{\tau\}$ to be the set of actions, $\alpha,\beta$ range over $Act$. And $K,L$ are used to stand for subsets of $\mathcal{L}$ and $\overline{L}$
is used for the set of complements of labels in $L$. A relabelling function $f$ is a function from $\mathcal{L}$ to $\mathcal{L}$ such that $f(\overline{l})=\overline{f(l)}$. By
defining $f(\tau)=\tau$, we extend $f$ to $Act$.

Further, we introduce a set $\mathcal{X}$ of process variables, and a set $\mathcal{K}$ of process constants, and let $X,Y,\cdots$ range over $\mathcal{X}$, and $A,B,\cdots$ range over
$\mathcal{K}$, $\widetilde{X}$ is a tuple of distinct process variables, and also $E,F,\cdots$ range over the recursive expressions. We write $\mathcal{P}$ for the set of processes.
Sometimes, we use $I,J$ to stand for an indexing set, and we write $E_i:i\in I$ for a family of expressions indexed by $I$. $Id_D$ is the identity function or relation over set $D$.

For each process constant schema $A$, a defining equation of the form

$$A\overset{\text{def}}{=}P$$

is assumed, where $P$ is a process.

\subsubsection{Syntax}

We use the Prefix $.$ to model the causality relation $\leq$ in true concurrency, the Summation $+$ to model the conflict relation $\sharp$, and the Box-Summation $\boxplus_{\pi}$ to
model the probabilistic conflict relation $\sharp_{\pi}$ in true concurrency, and the Composition $\parallel$ to explicitly model concurrent relation in true concurrency. And we
follow the conventions of process algebra.

\begin{definition}[Syntax]\label{syntax3}
Truly concurrent processes are defined inductively by the following formation rules:

\begin{enumerate}
  \item $A\in\mathcal{P}$;
  \item $\textbf{nil}\in\mathcal{P}$;
  \item if $P\in\mathcal{P}$, then the Prefix $\alpha.P\in\mathcal{P}$, for $\alpha\in Act$;
  \item if $P,Q\in\mathcal{P}$, then the Summation $P+Q\in\mathcal{P}$;
  \item if $P,Q\in\mathcal{P}$, then the Box-Summation $P\boxplus_{\pi}Q\in\mathcal{P}$;
  \item if $P,Q\in\mathcal{P}$, then the Composition $P\parallel Q\in\mathcal{P}$;
  \item if $P\in\mathcal{P}$, then the Prefix $(\alpha_1\parallel\cdots\parallel\alpha_n).P\in\mathcal{P}\quad(n\in I)$, for $\alpha_,\cdots,\alpha_n\in Act$;
  \item if $P\in\mathcal{P}$, then the Restriction $P\setminus L\in\mathcal{P}$ with $L\in\mathcal{L}$;
  \item if $P\in\mathcal{P}$, then the Relabelling $P[f]\in\mathcal{P}$.
\end{enumerate}

The standard BNF grammar of syntax of CPTC can be summarized as follows:

$$P::=A\quad|\quad\textbf{nil}\quad|\quad\alpha.P\quad|\quad P+P\quad |\quad P\boxplus_{\pi} P\quad |\quad P\parallel P\quad |\quad (\alpha_1\parallel\cdots\parallel\alpha_n).P \quad|\quad P\setminus L\quad |\quad P[f].$$
\end{definition}

\subsubsection{Operational Semantics}

The operational semantics is defined by LTSs (labelled transition systems), and it is detailed by the following definition.

\begin{definition}[Semantics]\label{semantics3}
The operational semantics of CPTC corresponding to the syntax in Definition \ref{syntax3} is defined by a series of transition rules, named $\textbf{PAct}$, $\textbf{PSum}$, $\textbf{PBox-Sum}$,
$\textbf{Com}$, $\textbf{Res}$, $\textbf{Rel}$ and $\textbf{Con}$ and named $\textbf{Act}$, $\textbf{Sum}$,
$\textbf{Com}$, $\textbf{Res}$, $\textbf{Rel}$ and $\textbf{Con}$ indicate that the rules are associated respectively with Prefix, Summation, Composition, Restriction, Relabelling and
Constants in Definition \ref{syntax3}. They are shown in Table \ref{PTRForCTC3} and \ref{TRForCTC3}.

\begin{center}
    \begin{table}
        \[\textbf{PAct}_1\quad \frac{}{\alpha.P\rightsquigarrow\breve{\alpha}.P}\]

        \[\textbf{PSum}\quad \frac{P\rightsquigarrow P'\quad Q\rightsquigarrow Q'}{P+Q\rightsquigarrow P'+Q'}\]

        \[\textbf{PBox-Sum}\quad \frac{P\rightsquigarrow P'}{P\boxplus_{\pi}Q\rightsquigarrow P'}\quad \frac{Q\rightsquigarrow Q'}{P\boxplus_{\pi}Q\rightsquigarrow Q'}\]

        \[\textbf{PCom}\quad \frac{P\rightsquigarrow P'\quad Q\rightsquigarrow Q'}{P\parallel Q\rightsquigarrow P'+Q'}\]

        \[\textbf{PAct}_2\quad \frac{}{(\alpha_1\parallel\cdots\parallel\alpha_n).P\rightsquigarrow(\breve{\alpha_1}\parallel\cdots\parallel\breve{\alpha_n}).P}\]

        \[\textbf{PRes}\quad \frac{P\rightsquigarrow P'}{P\setminus L\rightsquigarrow P'\setminus L}\]

        \[\textbf{PRel}\quad \frac{P\rightsquigarrow P'}{P[f]\rightsquigarrow P'[f]}\]

        \[\textbf{PCon}\quad\frac{P\rightsquigarrow P'}{A\rightsquigarrow P'}\quad (A\overset{\text{def}}{=}P)\]

        \caption{Probabilistic transition rules of CPTC}
        \label{PTRForCTC3}
    \end{table}
\end{center}

\begin{center}
    \begin{table}
        \[\textbf{Act}_1\quad \frac{}{\breve{\alpha}.P\xrightarrow{\alpha}P}\]

        \[\textbf{Sum}_1\quad \frac{P\xrightarrow{\alpha}P'}{P+Q\xrightarrow{\alpha}P'}\]

        \[\textbf{Com}_1\quad \frac{P\xrightarrow{\alpha}P'\quad Q\nrightarrow}{P\parallel Q\xrightarrow{\alpha}P'\parallel Q}\]

        \[\textbf{Com}_2\quad \frac{Q\xrightarrow{\alpha}Q'\quad P\nrightarrow}{P\parallel Q\xrightarrow{\alpha}P\parallel Q'}\]

        \[\textbf{Com}_3\quad \frac{P\xrightarrow{\alpha}P'\quad Q\xrightarrow{\beta}Q'}{P\parallel Q\xrightarrow{\{\alpha,\beta\}}P'\parallel Q'}\quad (\beta\neq\overline{\alpha})\]

        \[\textbf{Com}_4\quad \frac{P\xrightarrow{l}P'\quad Q\xrightarrow{\overline{l}}Q'}{P\parallel Q\xrightarrow{\tau}P'\parallel Q'}\]

        \[\textbf{Act}_2\quad \frac{}{(\breve{\alpha_1}\parallel\cdots\parallel\breve{\alpha_n}).P\xrightarrow{\{\alpha_1,\cdots,\alpha_n\}}P}\quad (\alpha_i\neq\overline{\alpha_j}\quad i,j\in\{1,\cdots,n\})\]

        \[\textbf{Sum}_2\quad \frac{P\xrightarrow{\{\alpha_1,\cdots,\alpha_n\}}P'}{P+Q\xrightarrow{\{\alpha_1,\cdots,\alpha_n\}}P'}\]

        \[\textbf{Res}_1\quad \frac{P\xrightarrow{\alpha}P'}{P\setminus L\xrightarrow{\alpha}P'\setminus L}\quad (\alpha,\overline{\alpha}\notin L)\]

        \[\textbf{Res}_2\quad \frac{P\xrightarrow{\{\alpha_1,\cdots,\alpha_n\}}P'}{P\setminus L\xrightarrow{\{\alpha_1,\cdots,\alpha_n\}}P'\setminus L}\quad (\alpha_1,\overline{\alpha_1},\cdots,\alpha_n,\overline{\alpha_n}\notin L)\]

        \[\textbf{Rel}_1\quad \frac{P\xrightarrow{\alpha}P'}{P[f]\xrightarrow{f(\alpha)}P'[f]}\]

        \[\textbf{Rel}_2\quad \frac{P\xrightarrow{\{\alpha_1,\cdots,\alpha_n\}}P'}{P[f]\xrightarrow{\{f(\alpha_1),\cdots,f(\alpha_n)\}}P'[f]}\]

        \[\textbf{Con}_1\quad\frac{P\xrightarrow{\alpha}P'}{A\xrightarrow{\alpha}P'}\quad (A\overset{\text{def}}{=}P)\]

        \[\textbf{Con}_2\quad\frac{P\xrightarrow{\{\alpha_1,\cdots,\alpha_n\}}P'}{A\xrightarrow{\{\alpha_1,\cdots,\alpha_n\}}P'}\quad (A\overset{\text{def}}{=}P)\]

        \caption{Action transition rules of CPTC}
        \label{TRForCTC3}
    \end{table}
\end{center}
\end{definition}

\subsubsection{Properties of Transitions}

\begin{definition}[Sorts]\label{sorts3}
Given the sorts $\mathcal{L}(A)$ and $\mathcal{L}(X)$ of constants and variables, we define $\mathcal{L}(P)$ inductively as follows.

\begin{enumerate}
  \item $\mathcal{L}(l.P)=\{l\}\cup\mathcal{L}(P)$;
  \item $\mathcal{L}((l_1\parallel \cdots\parallel l_n).P)=\{l_1,\cdots,l_n\}\cup\mathcal{L}(P)$;
  \item $\mathcal{L}(\tau.P)=\mathcal{L}(P)$;
  \item $\mathcal{L}(P+Q)=\mathcal{L}(P)\cup\mathcal{L}(Q)$;
  \item $\mathcal{L}(P\boxplus_{\pi}Q)=\mathcal{L}(P)\cup\mathcal{L}(Q)$;
  \item $\mathcal{L}(P\parallel Q)=\mathcal{L}(P)\cup\mathcal{L}(Q)$;
  \item $\mathcal{L}(P\setminus L)=\mathcal{L}(P)-(L\cup\overline{L})$;
  \item $\mathcal{L}(P[f])=\{f(l):l\in\mathcal{L}(P)\}$;
  \item for $A\overset{\text{def}}{=}P$, $\mathcal{L}(P)\subseteq\mathcal{L}(A)$.
\end{enumerate}
\end{definition}

Now, we present some properties of the transition rules defined in Table \ref{TRForCTC3}.

\begin{proposition}
If $P\xrightarrow{\alpha}P'$, then
\begin{enumerate}
  \item $\alpha\in\mathcal{L}(P)\cup\{\tau\}$;
  \item $\mathcal{L}(P')\subseteq\mathcal{L}(P)$.
\end{enumerate}

If $P\xrightarrow{\{\alpha_1,\cdots,\alpha_n\}}P'$, then
\begin{enumerate}
  \item $\alpha_1,\cdots,\alpha_n\in\mathcal{L}(P)\cup\{\tau\}$;
  \item $\mathcal{L}(P')\subseteq\mathcal{L}(P)$.
\end{enumerate}
\end{proposition}

\begin{proof}
By induction on the inference of $P\xrightarrow{\alpha}P'$ and $P\xrightarrow{\{\alpha_1,\cdots,\alpha_n\}}P'$, there are fourteen cases corresponding to the transition rules named
$\textbf{Act}_{1,2}$, $\textbf{Sum}_{1,2}$, $\textbf{Com}_{1,2,3,4}$, $\textbf{Res}_{1,2}$, $\textbf{Rel}_{1,2}$ and $\textbf{Con}_{1,2}$ in Table \ref{TRForCTC3}, we just prove the
one case $\textbf{Act}_1$ and $\textbf{Act}_2$, and omit the others.

Case $\textbf{Act}_1$: by $\textbf{Act}_1$, with $P\equiv\alpha.P'$. Then by Definition \ref{sorts3}, we have (1) $\mathcal{L}(P)=\{\alpha\}\cup\mathcal{L}(P')$ if $\alpha\neq\tau$;
(2) $\mathcal{L}(P)=\mathcal{L}(P')$ if $\alpha=\tau$. So, $\alpha\in\mathcal{L}(P)\cup\{\tau\}$, and $\mathcal{L}(P')\subseteq\mathcal{L}(P)$, as desired.

Case $\textbf{Act}_2$: by $\textbf{Act}_2$, with $P\equiv(\alpha_1\parallel\cdots\parallel\alpha_n).P'$. Then by Definition \ref{sorts3}, we have (1)
$\mathcal{L}(P)=\{\alpha_1,\cdots,\alpha_n\}\cup\mathcal{L}(P')$ if $\alpha_i\neq\tau$ for $i\leq n$;
(2) $\mathcal{L}(P)=\mathcal{L}(P')$ if $\alpha_1,\cdots,\alpha_n=\tau$. So, $\alpha_1,\cdots,\alpha_n\in\mathcal{L}(P)\cup\{\tau\}$, and $\mathcal{L}(P')\subseteq\mathcal{L}(P)$,
as desired.
\end{proof}

\subsection{Strongly Probabilistic Truly Concurrent Bisimulations}\label{stcb3}

\subsubsection{Laws and Congruence}

Based on the concepts of strongly probabilistic truly concurrent bisimulation equivalences, we get the following laws.

\begin{proposition}[Monoid laws for strongly probabilistic pomset bisimulation] The monoid laws for strongly probabilistic pomset bisimulation are as follows.
\begin{enumerate}
  \item $P+Q\sim_{pp} Q+P$;
  \item $P+(Q+R)\sim_{pp} (P+Q)+R$;
  \item $P+P\sim_{pp} P$;
  \item $P+\textbf{nil}\sim_{pp} P$.
\end{enumerate}
\end{proposition}

\begin{proof}
According to the definition of strongly probabilistic pomset bisimulation, we can easily prove the above equations, and we omit the proof.
\end{proof}

\begin{proposition}[Monoid laws for strongly probabilistic step bisimulation] The monoid laws for strongly probabilistic step bisimulation are as follows.
\begin{enumerate}
  \item $P+Q\sim_{ps} Q+P$;
  \item $P+(Q+R)\sim_{ps} (P+Q)+R$;
  \item $P+P\sim_{ps} P$;
  \item $P+\textbf{nil}\sim_{ps} P$.
\end{enumerate}
\end{proposition}

\begin{proof}
According to the definition of strongly probabilistic step bisimulation, we can easily prove the above equations, and we omit the proof.
\end{proof}

\begin{proposition}[Monoid laws for strongly probabilistic hp-bisimulation] The monoid laws for strongly probabilistic hp-bisimulation are as follows.
\begin{enumerate}
  \item $P+Q\sim_{php} Q+P$;
  \item $P+(Q+R)\sim_{php} (P+Q)+R$;
  \item $P+P\sim_{php} P$;
  \item $P+\textbf{nil}\sim_{php} P$.
\end{enumerate}
\end{proposition}

\begin{proof}
According to the definition of strongly probabilistic hp-bisimulation, we can easily prove the above equations, and we omit the proof.
\end{proof}

\begin{proposition}[Monoid laws for strongly probabilistic hhp-bisimulation] The monoid laws for strongly probabilistic hhp-bisimulation are as follows.
\begin{enumerate}
  \item $P+Q\sim_{phhp} Q+P$;
  \item $P+(Q+R)\sim_{phhp} (P+Q)+R$;
  \item $P+P\sim_{phhp} P$;
  \item $P+\textbf{nil}\sim_{phhp} P$.
\end{enumerate}
\end{proposition}

\begin{proof}
According to the definition of strongly probabilistic hhp-bisimulation, we can easily prove the above equations, and we omit the proof.
\end{proof}

\begin{proposition}[Monoid laws 2 for strongly probabilistic pomset bisimulation] 
The monoid laws 2 for strongly probabilistic pomset bisimulation are as follows.

\begin{enumerate}
  \item $P\boxplus_{\pi} Q\sim_{pp} Q\boxplus_{1-\pi} P$\\
  \item $P\boxplus_{\pi}(Q\boxplus_{\rho} R)\sim_{pp} (P\boxplus_{\frac{\pi}{\pi+\rho-\pi\rho}}Q)\boxplus_{\pi+\rho-\pi\rho} R$\\
  \item $P\boxplus_{\pi}P\sim_{pp} P$\\
  \item $P\boxplus_{\pi}\textbf{nil}\sim_{pp} P$.
\end{enumerate}
\end{proposition}

\begin{proof}
According to the definition of strongly probabilistic pomset bisimulation, we can easily prove the above equations, and we omit the proof.
\end{proof}

\begin{proposition}[Monoid laws 2 for strongly probabilistic step bisimulation] 
The monoid laws 2 for strongly probabilistic step bisimulation are as follows.

\begin{enumerate}
  \item $P\boxplus_{\pi} Q\sim_{ps} Q\boxplus_{1-\pi} P$\\
  \item $P\boxplus_{\pi}(Q\boxplus_{\rho} R)\sim_{ps} (P\boxplus_{\frac{\pi}{\pi+\rho-\pi\rho}}Q)\boxplus_{\pi+\rho-\pi\rho} R$\\
  \item $P\boxplus_{\pi}P\sim_{ps} P$\\
  \item $P\boxplus_{\pi}\textbf{nil}\sim_{ps} P$.
\end{enumerate}
\end{proposition}

\begin{proof}
According to the definition of strongly probabilistic step bisimulation, we can easily prove the above equations, and we omit the proof.
\end{proof}

\begin{proposition}[Monoid laws 2 for strongly probabilistic hp-bisimulation] 
The monoid laws 2 for strongly probabilistic hp-bisimulation are as follows.

\begin{enumerate}
  \item $P\boxplus_{\pi} Q\sim_{php} Q\boxplus_{1-\pi} P$\\
  \item $P\boxplus_{\pi}(Q\boxplus_{\rho} R)\sim_{php} (P\boxplus_{\frac{\pi}{\pi+\rho-\pi\rho}}Q)\boxplus_{\pi+\rho-\pi\rho} R$\\
  \item $P\boxplus_{\pi}P\sim_{php} P$\\
  \item $P\boxplus_{\pi}\textbf{nil}\sim_{php} P$.
\end{enumerate}
\end{proposition}

\begin{proof}
According to the definition of strongly probabilistic hp-bisimulation, we can easily prove the above equations, and we omit the proof.
\end{proof}

\begin{proposition}[Monoid laws 2 for strongly probabilistic hhp-bisimulation] 
The monoid laws 2 for strongly probabilistic hhp-bisimulation are as follows.

\begin{enumerate}
  \item $P\boxplus_{\pi} Q\sim_{phhp} Q\boxplus_{1-\pi} P$\\
  \item $P\boxplus_{\pi}(Q\boxplus_{\rho} R)\sim_{phhp} (P\boxplus_{\frac{\pi}{\pi+\rho-\pi\rho}}Q)\boxplus_{\pi+\rho-\pi\rho} R$\\
  \item $P\boxplus_{\pi}P\sim_{phhp} P$\\
  \item $P\boxplus_{\pi}\textbf{nil}\sim_{phhp} P$.
\end{enumerate}
\end{proposition}

\begin{proof}
According to the definition of strongly probabilistic hhp-bisimulation, we can easily prove the above equations, and we omit the proof.
\end{proof}

\begin{proposition}[Static laws for strongly probabilistic pomset bisimulation] \label{SLSPB3}
The static laws for strongly probabilistic pomset bisimulation are as follows.

\begin{enumerate}
  \item $P\parallel Q\sim_{pp} Q\parallel P$;
  \item $P\parallel(Q\parallel R)\sim_{pp} (P\parallel Q)\parallel R$;
  \item $P\parallel \textbf{nil}\sim_{pp} P$;
  \item $P\setminus L\sim_{pp} P$, if $\mathcal{L}(P)\cap(L\cup\overline{L})=\emptyset$;
  \item $P\setminus K\setminus L\sim_{pp} P\setminus(K\cup L)$;
  \item $P[f]\setminus L\sim_{pp} P\setminus f^{-1}(L)[f]$;
  \item $(P\parallel Q)\setminus L\sim_{pp} P\setminus L\parallel Q\setminus L$, if $\mathcal{L}(P)\cap\overline{\mathcal{L}(Q)}\cap(L\cup\overline{L})=\emptyset$;
  \item $P[Id]\sim_{pp} P$;
  \item $P[f]\sim_{pp} P[f']$, if $f\upharpoonright\mathcal{L}(P)=f'\upharpoonright\mathcal{L}(P)$;
  \item $P[f][f']\sim_{pp} P[f'\circ f]$;
  \item $(P\parallel Q)[f]\sim_{pp} P[f]\parallel Q[f]$, if $f\upharpoonright(L\cup\overline{L})$ is one-to-one, where $L=\mathcal{L}(P)\cup\mathcal{L}(Q)$.
\end{enumerate}
\end{proposition}

\begin{proof}
According to the definition of strongly probabilistic pomset bisimulation, we can easily prove the above equations, and we omit the proof.
\end{proof}

\begin{proposition}[Static laws for strongly probabilistic step bisimulation] \label{SLSSB3}
The static laws for strongly probabilistic step bisimulation are as follows.

\begin{enumerate}
  \item $P\parallel Q\sim_{ps} Q\parallel P$;
  \item $P\parallel(Q\parallel R)\sim_{ps} (P\parallel Q)\parallel R$;
  \item $P\parallel \textbf{nil}\sim_{ps} P$;
  \item $P\setminus L\sim_{ps} P$, if $\mathcal{L}(P)\cap(L\cup\overline{L})=\emptyset$;
  \item $P\setminus K\setminus L\sim_{ps} P\setminus(K\cup L)$;
  \item $P[f]\setminus L\sim_{ps} P\setminus f^{-1}(L)[f]$;
  \item $(P\parallel Q)\setminus L\sim_{ps} P\setminus L\parallel Q\setminus L$, if $\mathcal{L}(P)\cap\overline{\mathcal{L}(Q)}\cap(L\cup\overline{L})=\emptyset$;
  \item $P[Id]\sim_{ps} P$;
  \item $P[f]\sim_{ps} P[f']$, if $f\upharpoonright\mathcal{L}(P)=f'\upharpoonright\mathcal{L}(P)$;
  \item $P[f][f']\sim_{ps} P[f'\circ f]$;
  \item $(P\parallel Q)[f]\sim_{ps} P[f]\parallel Q[f]$, if $f\upharpoonright(L\cup\overline{L})$ is one-to-one, where $L=\mathcal{L}(P)\cup\mathcal{L}(Q)$.
\end{enumerate}
\end{proposition}

\begin{proof}
According to the definition of strongly probabilistic step bisimulation, we can easily prove the above equations, and we omit the proof.
\end{proof}

\begin{proposition}[Static laws for strongly probabilistic hp-bisimulation] \label{SLSHPB3}
The static laws for strongly probabilistic hp-bisimulation are as follows.

\begin{enumerate}
  \item $P\parallel Q\sim_{php} Q\parallel P$;
  \item $P\parallel(Q\parallel R)\sim_{php} (P\parallel Q)\parallel R$;
  \item $P\parallel \textbf{nil}\sim_{php} P$;
  \item $P\setminus L\sim_{php} P$, if $\mathcal{L}(P)\cap(L\cup\overline{L})=\emptyset$;
  \item $P\setminus K\setminus L\sim_{php} P\setminus(K\cup L)$;
  \item $P[f]\setminus L\sim_{php} P\setminus f^{-1}(L)[f]$;
  \item $(P\parallel Q)\setminus L\sim_{php} P\setminus L\parallel Q\setminus L$, if $\mathcal{L}(P)\cap\overline{\mathcal{L}(Q)}\cap(L\cup\overline{L})=\emptyset$;
  \item $P[Id]\sim_{php} P$;
  \item $P[f]\sim_{php} P[f']$, if $f\upharpoonright\mathcal{L}(P)=f'\upharpoonright\mathcal{L}(P)$;
  \item $P[f][f']\sim_{php} P[f'\circ f]$;
  \item $(P\parallel Q)[f]\sim_{php} P[f]\parallel Q[f]$, if $f\upharpoonright(L\cup\overline{L})$ is one-to-one, where $L=\mathcal{L}(P)\cup\mathcal{L}(Q)$.
\end{enumerate}
\end{proposition}

\begin{proof}
According to the definition of strongly probabilistic hp-bisimulation, we can easily prove the above equations, and we omit the proof.
\end{proof}

\begin{proposition}[Static laws for strongly probabilistic hhp-bisimulation] \label{SLSHHPB3}
The static laws for strongly probabilistic hhp-bisimulation are as follows.

\begin{enumerate}
  \item $P\parallel Q\sim_{phhp} Q\parallel P$;
  \item $P\parallel(Q\parallel R)\sim_{phhp} (P\parallel Q)\parallel R$;
  \item $P\parallel \textbf{nil}\sim_{phhp} P$;
  \item $P\setminus L\sim_{phhp} P$, if $\mathcal{L}(P)\cap(L\cup\overline{L})=\emptyset$;
  \item $P\setminus K\setminus L\sim_{phhp} P\setminus(K\cup L)$;
  \item $P[f]\setminus L\sim_{phhp} P\setminus f^{-1}(L)[f]$;
  \item $(P\parallel Q)\setminus L\sim_{phhp} P\setminus L\parallel Q\setminus L$, if $\mathcal{L}(P)\cap\overline{\mathcal{L}(Q)}\cap(L\cup\overline{L})=\emptyset$;
  \item $P[Id]\sim_{phhp} P$;
  \item $P[f]\sim_{phhp} P[f']$, if $f\upharpoonright\mathcal{L}(P)=f'\upharpoonright\mathcal{L}(P)$;
  \item $P[f][f']\sim_{phhp} P[f'\circ f]$;
  \item $(P\parallel Q)[f]\sim_{phhp} P[f]\parallel Q[f]$, if $f\upharpoonright(L\cup\overline{L})$ is one-to-one, where $L=\mathcal{L}(P)\cup\mathcal{L}(Q)$.
\end{enumerate}
\end{proposition}

\begin{proof}
According to the definition of strongly probabilistic hhp-bisimulation, we can easily prove the above equations, and we omit the proof.
\end{proof}

\begin{proposition}[Expansion law for strongly probabilistic step bisimulation]\label{NELSSB3}
Let $P\equiv (P_1[f_1]\parallel\cdots\parallel P_n[f_n])\setminus L$, with $n\geq 1$. Then

\begin{eqnarray}
P\sim_{ps} \{(f_1(\alpha_1)\parallel\cdots\parallel f_n(\alpha_n)).(P_1'[f_1]\parallel\cdots\parallel P_n'[f_n])\setminus L: \nonumber\\
P_i\rightsquigarrow\xrightarrow{\alpha_i}P_i',i\in\{1,\cdots,n\},f_i(\alpha_i)\notin L\cup\overline{L}\} \nonumber\\
+\sum\{\tau.(P_1[f_1]\parallel\cdots\parallel P_i'[f_i]\parallel\cdots\parallel P_j'[f_j]\parallel\cdots\parallel P_n[f_n])\setminus L: \nonumber\\
P_i\rightsquigarrow\xrightarrow{l_1}P_i',P_j\rightsquigarrow\xrightarrow{l_2}P_j',f_i(l_1)=\overline{f_j(l_2)},i<j\} \nonumber
\end{eqnarray}
\end{proposition}

\begin{proof}
Though transition rules in Table \ref{TRForCTC3} are defined in the flavor of single event, they can be modified into a step (a set of events within which each event is pairwise
concurrent), we omit them. If we treat a single event as a step containing just one event, the proof of the new expansion law has not any problem, so we use this way and still use the
transition rules in Table \ref{TRForCTC3}.

Firstly, we consider the case without Restriction and Relabeling. That is, we suffice to prove the following case by induction on the size $n$.

For $P\equiv P_1\parallel\cdots\parallel P_n$, with $n\geq 1$, we need to prove

\begin{eqnarray}
P\sim_{ps} \{(\alpha_1\parallel\cdots\parallel \alpha_n).(P_1'\parallel\cdots\parallel P_n'): P_i\rightsquigarrow\xrightarrow{\alpha_i}P_i',i\in\{1,\cdots,n\}\nonumber\\
+\sum\{\tau.(P_1\parallel\cdots\parallel P_i'\parallel\cdots\parallel P_j'\parallel\cdots\parallel P_n): P_i\rightsquigarrow\xrightarrow{l}P_i',P_j\rightsquigarrow\xrightarrow{\overline{l}}P_j',i<j\} \nonumber
\end{eqnarray}

For $n=1$, $P_1\sim_{ps} \alpha_1.P_1':P_1\rightsquigarrow\xrightarrow{\alpha_1}P_1'$ is obvious. Then with a hypothesis $n$, we consider $R\equiv P\parallel P_{n+1}$. By the
transition rules $\textbf{Com}_{1,2,3,4}$, we can get

\begin{eqnarray}
R\sim_{ps} \{(p\parallel \alpha_{n+1}).(P'\parallel P_{n+1}'): P\rightsquigarrow\xrightarrow{p}P',P_{n+1}\rightsquigarrow\xrightarrow{\alpha_{n+1}}P_{n+1}',p\subseteq P\}\nonumber\\
+\sum\{\tau.(P'\parallel P_{n+1}'): P\rightsquigarrow\xrightarrow{l}P',P_{n+1}\rightsquigarrow\xrightarrow{\overline{l}}P_{n+1}'\} \nonumber
\end{eqnarray}

Now with the induction assumption $P\equiv P_1\parallel\cdots\parallel P_n$, the right-hand side can be reformulated as follows.

\begin{eqnarray}
\{(\alpha_1\parallel\cdots\parallel \alpha_n\parallel \alpha_{n+1}).(P_1'\parallel\cdots\parallel P_n'\parallel P_{n+1}'): \nonumber\\
P_i\rightsquigarrow\xrightarrow{\alpha_i}P_i',i\in\{1,\cdots,n+1\}\nonumber\\
+\sum\{\tau.(P_1\parallel\cdots\parallel P_i'\parallel\cdots\parallel P_j'\parallel\cdots\parallel P_n\parallel P_{n+1}): \nonumber\\
P_i\rightsquigarrow\xrightarrow{l}P_i',P_j\rightsquigarrow\xrightarrow{\overline{l}}P_j',i<j\} \nonumber\\
+\sum\{\tau.(P_1\parallel\cdots\parallel P_i'\parallel\cdots\parallel P_j\parallel\cdots\parallel P_n\parallel P_{n+1}'): \nonumber\\
P_i\xrightarrow{l}P_i',P_{n+1}\rightsquigarrow\xrightarrow{\overline{l}}P_{n+1}',i\in\{1,\cdots, n\}\} \nonumber
\end{eqnarray}

So,

\begin{eqnarray}
R\sim_{ps} \{(\alpha_1\parallel\cdots\parallel \alpha_n\parallel \alpha_{n+1}).(P_1'\parallel\cdots\parallel P_n'\parallel P_{n+1}'): \nonumber\\
P_i\rightsquigarrow\xrightarrow{\alpha_i}P_i',i\in\{1,\cdots,n+1\}\nonumber\\
+\sum\{\tau.(P_1\parallel\cdots\parallel P_i'\parallel\cdots\parallel P_j'\parallel\cdots\parallel P_n): \nonumber\\
P_i\rightsquigarrow\xrightarrow{l}P_i',P_j\rightsquigarrow\xrightarrow{\overline{l}}P_j',1 \leq i<j\geq n+1\} \nonumber
\end{eqnarray}

Then, we can easily add the full conditions with Restriction and Relabeling.
\end{proof}

\begin{proposition}[Expansion law for strongly probabilistic pomset bisimulation]\label{NELSPB3}
Let $P\equiv (P_1[f_1]\parallel\cdots\parallel P_n[f_n])\setminus L$, with $n\geq 1$. Then

\begin{eqnarray}
P\sim_{pp} \{(f_1(\alpha_1)\parallel\cdots\parallel f_n(\alpha_n)).(P_1'[f_1]\parallel\cdots\parallel P_n'[f_n])\setminus L: \nonumber\\
P_i\rightsquigarrow\xrightarrow{\alpha_i}P_i',i\in\{1,\cdots,n\},f_i(\alpha_i)\notin L\cup\overline{L}\} \nonumber\\
+\sum\{\tau.(P_1[f_1]\parallel\cdots\parallel P_i'[f_i]\parallel\cdots\parallel P_j'[f_j]\parallel\cdots\parallel P_n[f_n])\setminus L: \nonumber\\
P_i\rightsquigarrow\xrightarrow{l_1}P_i',P_j\rightsquigarrow\xrightarrow{l_2}P_j',f_i(l_1)=\overline{f_j(l_2)},i<j\} \nonumber
\end{eqnarray}
\end{proposition}

\begin{proof}
Similarly to the proof of expansion law for strongly probabilistic step bisimulation (see Proposition \ref{NELSSB3}), we can prove that the new expansion law holds for strongly
probabilistic pomset bisimulation, we omit it.
\end{proof}

\begin{proposition}[Expansion law for strongly probabilistic hp-bisimulation]\label{NELSHPB3}
Let $P\equiv (P_1[f_1]\parallel\cdots\parallel P_n[f_n])\setminus L$, with $n\geq 1$. Then

\begin{eqnarray}
P\sim_{php} \{(f_1(\alpha_1)\parallel\cdots\parallel f_n(\alpha_n)).(P_1'[f_1]\parallel\cdots\parallel P_n'[f_n])\setminus L: \nonumber\\
P_i\rightsquigarrow\xrightarrow{\alpha_i}P_i',i\in\{1,\cdots,n\},f_i(\alpha_i)\notin L\cup\overline{L}\} \nonumber\\
+\sum\{\tau.(P_1[f_1]\parallel\cdots\parallel P_i'[f_i]\parallel\cdots\parallel P_j'[f_j]\parallel\cdots\parallel P_n[f_n])\setminus L: \nonumber\\
P_i\rightsquigarrow\xrightarrow{l_1}P_i',P_j\rightsquigarrow\xrightarrow{l_2}P_j',f_i(l_1)=\overline{f_j(l_2)},i<j\} \nonumber
\end{eqnarray}
\end{proposition}

\begin{proof}
Similarly to the proof of expansion law for strongly probabilistic pomset bisimulation (see Proposition \ref{NELSPB3}), we can prove that the expansion law holds for strongly probabilistic
hp-bisimulation, we just need additionally to check the above conditions on hp-bisimulation, we omit it.
\end{proof}

\begin{proposition}[Expansion law for strongly probabilistic hhp-bisimulation]\label{NELSHHPB3}
Let $P\equiv (P_1[f_1]\parallel\cdots\parallel P_n[f_n])\setminus L$, with $n\geq 1$. Then

\begin{eqnarray}
P\sim_{phhp} \{(f_1(\alpha_1)\parallel\cdots\parallel f_n(\alpha_n)).(P_1'[f_1]\parallel\cdots\parallel P_n'[f_n])\setminus L: \nonumber\\
P_i\rightsquigarrow\xrightarrow{\alpha_i}P_i',i\in\{1,\cdots,n\},f_i(\alpha_i)\notin L\cup\overline{L}\} \nonumber\\
+\sum\{\tau.(P_1[f_1]\parallel\cdots\parallel P_i'[f_i]\parallel\cdots\parallel P_j'[f_j]\parallel\cdots\parallel P_n[f_n])\setminus L: \nonumber\\
P_i\rightsquigarrow\xrightarrow{l_1}P_i',P_j\rightsquigarrow\xrightarrow{l_2}P_j',f_i(l_1)=\overline{f_j(l_2)},i<j\} \nonumber
\end{eqnarray}
\end{proposition}

\begin{proof}
From the definition of strongly probabilistic hhp-bisimulation (see Definition \ref{HHPB}), we know that strongly hhp-bisimulation is downward closed for strongly probabilistic hp-bisimulation.

Similarly to the proof of the expansion law for strongly probabilistic hp-bisimulation (see Proposition \ref{NELSHPB3}), we can prove that the expansion law holds for strongly probabilistic
hhp-bisimulation, that is, they are downward closed for strongly probabilistic hp-bisimulation, we omit it.
\end{proof}

\begin{theorem}[Congruence for strongly probabilistic pomset bisimulation] \label{CSPB3}
We can enjoy the full congruence for strongly probabilistic pomset bisimulation as follows.

\begin{enumerate}
  \item If $A\overset{\text{def}}{=}P$, then $A\sim_{pp} P$;
  \item Let $P_1\sim_{pp} P_2$. Then
        \begin{enumerate}
           \item $\alpha.P_1\sim_{pp} \alpha.P_2$;
           \item $(\alpha_1\parallel\cdots\parallel\alpha_n).P_1\sim_{pp} (\alpha_1\parallel\cdots\parallel\alpha_n).P_2$;
           \item $P_1+Q\sim_{pp} P_2 +Q$;
           \item $P_1\boxplus_{\pi}Q\sim_{pp} P_2 \boxplus_{\pi}Q$;
           \item $P_1\parallel Q\sim_{pp} P_2\parallel Q$;
           \item $P_1\setminus L\sim_{pp} P_2\setminus L$;
           \item $P_1[f]\sim_{pp} P_2[f]$.
         \end{enumerate}
\end{enumerate}
\end{theorem}

\begin{proof}
According to the definition of strongly probabilistic pomset bisimulation, we can easily prove the above equations, and we omit the proof.
\end{proof}

\begin{theorem}[Congruence for strongly probabilistic step bisimulation] \label{CSSB3}
We can enjoy the full congruence for strongly probabilistic step bisimulation as follows.

\begin{enumerate}
  \item If $A\overset{\text{def}}{=}P$, then $A\sim_{ps} P$;
  \item Let $P_1\sim_{ps} P_2$. Then
        \begin{enumerate}
           \item $\alpha.P_1\sim_{ps} \alpha.P_2$;
           \item $(\alpha_1\parallel\cdots\parallel\alpha_n).P_1\sim_{ps} (\alpha_1\parallel\cdots\parallel\alpha_n).P_2$;
           \item $P_1+Q\sim_{ps} P_2 +Q$;
           \item $P_1\boxplus_{\pi}Q\sim_{ps} P_2 \boxplus_{\pi}Q$;
           \item $P_1\parallel Q\sim_{ps} P_2\parallel Q$;
           \item $P_1\setminus L\sim_{ps} P_2\setminus L$;
           \item $P_1[f]\sim_{ps} P_2[f]$.
         \end{enumerate}
\end{enumerate}
\end{theorem}

\begin{proof}
According to the definition of strongly probabilistic step bisimulation, we can easily prove the above equations, and we omit the proof.
\end{proof}

\begin{theorem}[Congruence for strongly probabilistic hp-bisimulation] \label{CSHPB3}
We can enjoy the full congruence for strongly probabilistic hp-bisimulation as follows.

\begin{enumerate}
  \item If $A\overset{\text{def}}{=}P$, then $A\sim_{php} P$;
  \item Let $P_1\sim_{php} P_2$. Then
        \begin{enumerate}
           \item $\alpha.P_1\sim_{php} \alpha.P_2$;
           \item $(\alpha_1\parallel\cdots\parallel\alpha_n).P_1\sim_{php} (\alpha_1\parallel\cdots\parallel\alpha_n).P_2$;
           \item $P_1+Q\sim_{php} P_2 +Q$;
           \item $P_1\boxplus_{\pi}Q\sim_{php} P_2 \boxplus_{\pi}Q$;
           \item $P_1\parallel Q\sim_{php} P_2\parallel Q$;
           \item $P_1\setminus L\sim_{php} P_2\setminus L$;
           \item $P_1[f]\sim_{php} P_2[f]$.
         \end{enumerate}
\end{enumerate}
\end{theorem}

\begin{proof}
According to the definition of strongly probabilistic hp-bisimulation, we can easily prove the above equations, and we omit the proof.
\end{proof}

\begin{theorem}[Congruence for strongly probabilistic hhp-bisimulation] \label{CSHHPB3}
We can enjoy the full congruence for strongly probabilistic hhp-bisimulation as follows.

\begin{enumerate}
  \item If $A\overset{\text{def}}{=}P$, then $A\sim_{phhp} P$;
  \item Let $P_1\sim_{phhp} P_2$. Then
        \begin{enumerate}
           \item $\alpha.P_1\sim_{phhp} \alpha.P_2$;
           \item $(\alpha_1\parallel\cdots\parallel\alpha_n).P_1\sim_{phhp} (\alpha_1\parallel\cdots\parallel\alpha_n).P_2$;
           \item $P_1+Q\sim_{phhp} P_2 +Q$;
           \item $P_1\boxplus_{\pi}Q\sim_{phhp} P_2 \boxplus_{\pi}Q$;
           \item $P_1\parallel Q\sim_{phhp} P_2\parallel Q$;
           \item $P_1\setminus L\sim_{phhp} P_2\setminus L$;
           \item $P_1[f]\sim_{phhp} P_2[f]$.
         \end{enumerate}
\end{enumerate}
\end{theorem}

\begin{proof}
According to the definition of strongly probabilistic hhp-bisimulation, we can easily prove the above equations, and we omit the proof.
\end{proof}

\subsubsection{Recursion}

\begin{definition}[Weakly guarded recursive expression]
$X$ is weakly guarded in $E$ if each occurrence of $X$ is with some subexpression $\alpha.F$ or $(\alpha_1\parallel\cdots\parallel\alpha_n).F$ of $E$.
\end{definition}

\begin{lemma}\label{LUS3}
If the variables $\widetilde{X}$ are weakly guarded in $E$, and $E\{\widetilde{P}/\widetilde{X}\}\rightsquigarrow\xrightarrow{\{\alpha_1,\cdots,\alpha_n\}}P'$, then $P'$ takes the
form $E'\{\widetilde{P}/\widetilde{X}\}$ for some expression $E'$, and moreover, for any $\widetilde{Q}$, $E\{\widetilde{Q}/\widetilde{X}\}\rightsquigarrow\xrightarrow{\{\alpha_1,\cdots,\alpha_n\}}E'\{\widetilde{Q}/\widetilde{X}\}$.
\end{lemma}

\begin{proof}
It needs to induct on the depth of the inference of $E\{\widetilde{P}/\widetilde{X}\}\rightsquigarrow\xrightarrow{\{\alpha_1,\cdots,\alpha_n\}}P'$.

\begin{enumerate}
  \item Case $E\equiv Y$, a variable. Then $Y\notin \widetilde{X}$. Since $\widetilde{X}$ are weakly guarded, $Y\{\widetilde{P}/\widetilde{X}\equiv Y\}\nrightarrow$, this case is
  impossible.
  \item Case $E\equiv\beta.F$. Then we must have $\alpha=\beta$, and $P'\equiv F\{\widetilde{P}/\widetilde{X}\}$, and $E\{\widetilde{Q}/\widetilde{X}\}\equiv \beta.F\{\widetilde{Q}/\widetilde{X}\}
  \rightsquigarrow\xrightarrow{\beta}F\{\widetilde{Q}/\widetilde{X}\}$, then, let $E'$ be $F$, as desired.
  \item Case $E\equiv(\beta_1\parallel\cdots\parallel\beta_n).F$. Then we must have $\alpha_i=\beta_i$ for $1\leq i\leq n$, and $P'\equiv F\{\widetilde{P}/\widetilde{X}\}$, and
  $E\{\widetilde{Q}/\widetilde{X}\}\equiv (\beta_1\parallel\cdots\parallel\beta_n).F\{\widetilde{Q}/\widetilde{X}\} \rightsquigarrow\xrightarrow{\{\beta_1,\cdots,\beta_n\}}F\{\widetilde{Q}/\widetilde{X}\}$,
  then, let $E'$ be $F$, as desired.
  \item Case $E\equiv E_1+E_2$. Then either $E_1\{\widetilde{P}/\widetilde{X}\} \rightsquigarrow\xrightarrow{\{\alpha_1,\cdots,\alpha_n\}}P'$ or $E_2\{\widetilde{P}/\widetilde{X}\}
  \rightsquigarrow\xrightarrow{\{\alpha_1,\cdots,\alpha_n\}}P'$, then, we can apply this lemma in either case, as desired.
  \item Case $E\equiv E_1\boxplus_{\pi}E_2$. Then either $E_1\{\widetilde{P}/\widetilde{X}\} \rightsquigarrow P'$ or $E_2\{\widetilde{P}/\widetilde{X}\}
  \rightsquigarrow P'$, then, we can apply this lemma in either case, as desired.
  \item Case $E\equiv E_1\parallel E_2$. There are four possibilities.
  \begin{enumerate}
    \item We may have $E_1\{\widetilde{P}/\widetilde{X}\} \rightsquigarrow\xrightarrow{\alpha}P_1'$ and $E_2\{\widetilde{P}/\widetilde{X}\}\nrightarrow$ with
    $P'\equiv P_1'\parallel (E_2\{\widetilde{P}/\widetilde{X}\})$, then by applying this lemma, $P_1'$ is of the form $E_1'\{\widetilde{P}/\widetilde{X}\}$, and for any $Q$,
    $E_1\{\widetilde{Q}/\widetilde{X}\}\rightsquigarrow\xrightarrow{\alpha} E_1'\{\widetilde{Q}/\widetilde{X}\}$. So, $P'$ is of the form
    $E_1'\parallel E_2\{\widetilde{P}/\widetilde{X}\}$, and for any $Q$, $E\{\widetilde{Q}/\widetilde{X}\}\equiv E_1\{\widetilde{Q}/\widetilde{X}\}\parallel E_2\{\widetilde{Q}/
    \widetilde{X}\}\rightsquigarrow\xrightarrow{\alpha} (E_1'\parallel E_2)\{\widetilde{Q}/\widetilde{X}\}$, then, let $E'$ be $E_1'\parallel E_2$, as desired.
    \item We may have $E_2\{\widetilde{P}/\widetilde{X}\} \rightsquigarrow\xrightarrow{\alpha}P_2'$ and $E_1\{\widetilde{P}/\widetilde{X}\}\nrightarrow$ with
    $P'\equiv P_2'\parallel (E_1\{\widetilde{P}/\widetilde{X}\})$, this case can be prove similarly to the above subcase, as desired.
    \item We may have $E_1\{\widetilde{P}/\widetilde{X}\} \rightsquigarrow\xrightarrow{\alpha}P_1'$ and $E_2\{\widetilde{P}/\widetilde{X}\}\rightsquigarrow\xrightarrow{\beta}P_2'$
    with $\alpha\neq\overline{\beta}$ and $P'\equiv P_1'\parallel P_2'$, then by applying this lemma, $P_1'$ is of the form $E_1'\{\widetilde{P}/\widetilde{X}\}$, and for any $Q$,
    $E_1\{\widetilde{Q}/\widetilde{X}\}\rightsquigarrow\xrightarrow{\alpha} E_1'\{\widetilde{Q}/\widetilde{X}\}$; $P_2'$ is of the form $E_2'\{\widetilde{P}/\widetilde{X}\}$, and for
    any $Q$, $E_2\{\widetilde{Q}/\widetilde{X}\}\rightsquigarrow\xrightarrow{\alpha} E_2'\{\widetilde{Q}/\widetilde{X}\}$. So, $P'$ is of the form
    $E_1'\parallel E_2'\{\widetilde{P}/\widetilde{X}\}$, and for any $Q$, $E\{\widetilde{Q}/\widetilde{X}\}\equiv E_1\{\widetilde{Q}/\widetilde{X}\}\parallel E_2\{\widetilde{Q}/
    \widetilde{X}\}\rightsquigarrow\xrightarrow{\{\alpha,\beta\}} (E_1'\parallel E_2')\{\widetilde{Q}/\widetilde{X}\}$, then, let $E'$ be $E_1'\parallel E_2'$, as desired.
    \item We may have $E_1\{\widetilde{P}/\widetilde{X}\} \rightsquigarrow\xrightarrow{l}P_1'$ and $E_2\{\widetilde{P}/\widetilde{X}\}\rightsquigarrow\xrightarrow{\overline{l}}P_2'$
    with $P'\equiv P_1'\parallel P_2'$, then by applying this lemma, $P_1'$ is of the form $E_1'\{\widetilde{P}/\widetilde{X}\}$, and for any $Q$, $E_1\{\widetilde{Q}/\widetilde{X}\}
    \rightsquigarrow\xrightarrow{l} E_1'\{\widetilde{Q}/\widetilde{X}\}$; $P_2'$ is of the form $E_2'\{\widetilde{P}/\widetilde{X}\}$, and for any $Q$,
    $E_2\{\widetilde{Q}/\widetilde{X}\}\rightsquigarrow\xrightarrow{\overline{l}} E_2'\{\widetilde{Q}/\widetilde{X}\}$. So, $P'$ is of the form
    $E_1'\parallel E_2'\{\widetilde{P}/\widetilde{X}\}$, and for any $Q$, $E\{\widetilde{Q}/\widetilde{X}\}\equiv E_1\{\widetilde{Q}/\widetilde{X}\}\parallel E_2\{\widetilde{Q}/
    \widetilde{X}\}\rightsquigarrow\xrightarrow{\tau} (E_1'\parallel E_2')\{\widetilde{Q}/\widetilde{X}\}$, then, let $E'$ be $E_1'\parallel E_2'$, as desired.
  \end{enumerate}
  \item Case $E\equiv F[R]$ and $E\equiv F\setminus L$. These cases can be prove similarly to the above case.
  \item Case $E\equiv C$, an agent constant defined by $C\overset{\text{def}}{=}R$. Then there is no $X\in\widetilde{X}$ occurring in $E$, so
  $C\rightsquigarrow\xrightarrow{\{\alpha_1,\cdots,\alpha_n\}}P'$, let $E'$ be $P'$, as desired.
\end{enumerate}
\end{proof}

\begin{theorem}[Unique solution of equations for strongly probabilistic step bisimulation]\label{USSSB3}
Let the recursive expressions $E_i(i\in I)$ contain at most the variables $X_i(i\in I)$, and let each $X_j(j\in I)$ be weakly guarded in each $E_i$. Then,

If $\widetilde{P}\sim_{ps} \widetilde{E}\{\widetilde{P}/\widetilde{X}\}$ and $\widetilde{Q}\sim_{ps} \widetilde{E}\{\widetilde{Q}/\widetilde{X}\}$,
then $\widetilde{P}\sim_{ps} \widetilde{Q}$.
\end{theorem}

\begin{proof}
It is sufficient to induct on the depth of the inference of $E\{\widetilde{P}/\widetilde{X}\}\rightsquigarrow\xrightarrow{\{\alpha_1,\cdots,\alpha_n\}}P'$.

\begin{enumerate}
  \item Case $E\equiv X_i$. Then we have $E\{\widetilde{P}/\widetilde{X}\}\equiv P_i\rightsquigarrow\xrightarrow{\{\alpha_1,\cdots,\alpha_n\}}P'$, since
  $P_i\sim_{ps} E_i\{\widetilde{P}/\widetilde{X}\}$, we have $E_i\{\widetilde{P}/\widetilde{X}\}\rightsquigarrow\xrightarrow{\{\alpha_1,\cdots,\alpha_n\}}P''\sim_{ps} P'$. Since
  $\widetilde{X}$ are weakly guarded in $E_i$, by Lemma \ref{LUS3}, $P''\equiv E'\{\widetilde{P}/\widetilde{X}\}$ and
  $E_i\{\widetilde{P}/\widetilde{X}\}\rightsquigarrow\xrightarrow{\{\alpha_1,\cdots,\alpha_n\}} E'\{\widetilde{P}/\widetilde{X}\}$. Since
  $E\{\widetilde{Q}/\widetilde{X}\}\equiv X_i\{\widetilde{Q}/\widetilde{X}\} \equiv Q_i\sim_{ps} E_i\{\widetilde{Q}/\widetilde{X}\}$,
  $E\{\widetilde{Q}/\widetilde{X}\}\rightsquigarrow\xrightarrow{\{\alpha_1,\cdots,\alpha_n\}}Q'\sim_{ps} E'\{\widetilde{Q}/\widetilde{X}\}$. So, $P'\sim_{ps} Q'$, as desired.
  \item Case $E\equiv\alpha.F$. This case can be proven similarly.
  \item Case $E\equiv(\alpha_1\parallel\cdots\parallel\alpha_n).F$. This case can be proven similarly.
  \item Case $E\equiv E_1+E_2$. We have $E_i\{\widetilde{P}/\widetilde{X}\} \rightsquigarrow\xrightarrow{\{\alpha_1,\cdots,\alpha_n\}}P'$,
  $E_i\{\widetilde{Q}/\widetilde{X}\} \rightsquigarrow\xrightarrow{\{\alpha_1,\cdots,\alpha_n\}}Q'$, then, $P'\sim_{ps} Q'$, as desired.
  \item Case $E\equiv E_1\boxplus_{\pi}E_2$. We have $E_i\{\widetilde{P}/\widetilde{X}\} \rightsquigarrow P'$,
  $E_i\{\widetilde{Q}/\widetilde{X}\} \rightsquigarrow Q'$, then, $P'\sim_{ps} Q'$, as desired.
  \item Case $E\equiv E_1\parallel E_2$, $E\equiv F[R]$ and $E\equiv F\setminus L$, $E\equiv C$. These cases can be prove similarly to the above case.
\end{enumerate}
\end{proof}

\begin{theorem}[Unique solution of equations for strongly probabilistic pomset bisimulation]\label{USSPB3}
Let the recursive expressions $E_i(i\in I)$ contain at most the variables $X_i(i\in I)$, and let each $X_j(j\in I)$ be weakly guarded in each $E_i$. Then,

If $\widetilde{P}\sim_{pp} \widetilde{E}\{\widetilde{P}/\widetilde{X}\}$ and $\widetilde{Q}\sim_{pp} \widetilde{E}\{\widetilde{Q}/\widetilde{X}\}$, then
$\widetilde{P}\sim_{pp} \widetilde{Q}$.
\end{theorem}

\begin{proof}
Similarly to the proof of unique solution of equations for strongly probabilistic step bisimulation (see Theorem \ref{USSSB3}), we can prove that the unique solution of equations holds
for strongly probabilistic pomset bisimulation, we omit it.
\end{proof}

\begin{theorem}[Unique solution of equations for strongly probabilistic hp-bisimulation]\label{USSHPB3}
Let the recursive expressions $E_i(i\in I)$ contain at most the variables $X_i(i\in I)$, and let each $X_j(j\in I)$ be weakly guarded in each $E_i$. Then,

If $\widetilde{P}\sim_{php} \widetilde{E}\{\widetilde{P}/\widetilde{X}\}$ and $\widetilde{Q}\sim_{php} \widetilde{E}\{\widetilde{Q}/\widetilde{X}\}$, then
$\widetilde{P}\sim_{php} \widetilde{Q}$.
\end{theorem}

\begin{proof}
Similarly to the proof of unique solution of equations for strongly probabilistic pomset bisimulation (see Theorem \ref{USSPB3}), we can prove that the unique solution of equations
holds for strongly probabilistic hp-bisimulation, we just need additionally to check the above conditions on hp-bisimulation, we omit it.
\end{proof}

\begin{theorem}[Unique solution of equations for strongly probabilistic hhp-bisimulation]\label{USSHHPB3}
Let the recursive expressions $E_i(i\in I)$ contain at most the variables $X_i(i\in I)$, and let each $X_j(j\in I)$ be weakly guarded in each $E_i$. Then,

If $\widetilde{P}\sim_{phhp} \widetilde{E}\{\widetilde{P}/\widetilde{X}\}$ and $\widetilde{Q}\sim_{phhp} \widetilde{E}\{\widetilde{Q}/\widetilde{X}\}$, then
$\widetilde{P}\sim_{phhp} \widetilde{Q}$.
\end{theorem}

\begin{proof}
Similarly to the proof of unique solution of equations for strongly probabilistic hp-bisimulation (see Theorem \ref{USSHPB3}), we can prove that the unique solution of equations holds
for strongly probabilistic hhp-bisimulation, we omit it.
\end{proof}

\subsection{Weakly Probabilistic Truly Concurrent Bisimulations}\label{wtcb3}

The weak probabilistic transition rules of CPTC are the same as the strong one in Table \ref{PTRForCTC3}. And the weak action transition rules of CPTC are listed in Table
\ref{WTRForCTC3}.

\begin{center}
    \begin{table}
        \[\textbf{WAct}_1\quad \frac{}{\alpha.P\xRightarrow{\alpha}P}\]

        \[\textbf{WSum}_1\quad \frac{P\xRightarrow{\alpha}P'}{P+Q\xRightarrow{\alpha}P'}\]

        \[\textbf{WCom}_1\quad \frac{P\xRightarrow{\alpha}P'\quad Q\nrightarrow}{P\parallel Q\xRightarrow{\alpha}P'\parallel Q}\]

        \[\textbf{WCom}_2\quad \frac{Q\xRightarrow{\alpha}Q'\quad P\nrightarrow}{P\parallel Q\xRightarrow{\alpha}P\parallel Q'}\]

        \[\textbf{WCom}_3\quad \frac{P\xRightarrow{\alpha}P'\quad Q\xRightarrow{\beta}Q'}{P\parallel Q\xRightarrow{\{\alpha,\beta\}}P'\parallel Q'}\quad (\beta\neq\overline{\alpha})\]

        \[\textbf{WCom}_4\quad \frac{P\xRightarrow{l}P'\quad Q\xRightarrow{\overline{l}}Q'}{P\parallel Q\xRightarrow{\tau}P'\parallel Q'}\]

        \[\textbf{WAct}_2\quad \frac{}{(\alpha_1\parallel\cdots\parallel\alpha_n).P\xRightarrow{\{\alpha_1,\cdots,\alpha_n\}}P}\quad (\alpha_i\neq\overline{\alpha_j}\quad i,j\in\{1,\cdots,n\})\]

        \[\textbf{WSum}_2\quad \frac{P\xRightarrow{\{\alpha_1,\cdots,\alpha_n\}}P'}{P+Q\xRightarrow{\{\alpha_1,\cdots,\alpha_n\}}P'}\]

        \[\textbf{WRes}_1\quad \frac{P\xRightarrow{\alpha}P'}{P\setminus L\xRightarrow{\alpha}P'\setminus L}\quad (\alpha,\overline{\alpha}\notin L)\]

        \[\textbf{WRes}_2\quad \frac{P\xRightarrow{\{\alpha_1,\cdots,\alpha_n\}}P'}{P\setminus L\xRightarrow{\{\alpha_1,\cdots,\alpha_n\}}P'\setminus L}\quad (\alpha_1,\overline{\alpha_1},\cdots,\alpha_n,\overline{\alpha_n}\notin L)\]

        \[\textbf{WRel}_1\quad \frac{P\xRightarrow{\alpha}P'}{P[f]\xRightarrow{f(\alpha)}P'[f]}\]

        \[\textbf{WRel}_2\quad \frac{P\xRightarrow{\{\alpha_1,\cdots,\alpha_n\}}P'}{P[f]\xRightarrow{\{f(\alpha_1),\cdots,f(\alpha_n)\}}P'[f]}\]

        \[\textbf{WCon}_1\quad\frac{P\xRightarrow{\alpha}P'}{A\xRightarrow{\alpha}P'}\quad (A\overset{\text{def}}{=}P)\]

        \[\textbf{WCon}_2\quad\frac{P\xRightarrow{\{\alpha_1,\cdots,\alpha_n\}}P'}{A\xRightarrow{\{\alpha_1,\cdots,\alpha_n\}}P'}\quad (A\overset{\text{def}}{=}P)\]
        \caption{Weak action transition rules of CPTC}
        \label{WTRForCTC3}
    \end{table}
\end{center}

\subsubsection{Laws and Congruence}

Remembering that $\tau$ can neither be restricted nor relabeled, we know that the monoid laws, the monoid laws 2, the static laws and the expansion law in
section \ref{stcb3} still hold with respect to the corresponding weakly probabilistic truly concurrent bisimulations. And also, we can enjoy the full congruence of Prefix, Summation, Composition,
Restriction, Relabelling and Constants with respect to corresponding weakly probabilistic truly concurrent bisimulations. We will not retype these laws, and just give the
$\tau$-specific laws.

\begin{proposition}[$\tau$ laws for weakly probabilistic pomset bisimulation]\label{TAUWPB3}
The $\tau$ laws for weakly probabilistic pomset bisimulation is as follows.

\begin{enumerate}
  \item $P\approx_{pp} \tau.P$;
  \item $\alpha.\tau.P\approx_{pp} \alpha.P$;
  \item $(\alpha_1\parallel\cdots\parallel\alpha_n).\tau.P\approx_{pp} (\alpha_1\parallel\cdots\parallel\alpha_n).P$;
  \item $P+\tau.P\approx_{pp} \tau.P$;
  \item $P\cdot((Q+\tau\cdot(Q+R))\boxplus_{\pi}S)\approx_{pp}P\cdot((Q+R)\boxplus_{\pi}S)$
  \item $P\approx_{pp} \tau\parallel P$.
\end{enumerate}
\end{proposition}

\begin{proof}
According to the definition of weakly probabilistic pomset bisimulation, we can easily prove the above equations, and we omit the proof.
\end{proof}

\begin{proposition}[$\tau$ laws for weakly probabilistic step bisimulation]\label{TAUWSB3}
The $\tau$ laws for weakly probabilistic step bisimulation is as follows.

\begin{enumerate}
  \item $P\approx_{ps} \tau.P$;
  \item $\alpha.\tau.P\approx_{ps} \alpha.P$;
  \item $(\alpha_1\parallel\cdots\parallel\alpha_n).\tau.P\approx_{ps} (\alpha_1\parallel\cdots\parallel\alpha_n).P$;
  \item $P+\tau.P\approx_{ps} \tau.P$;
  \item $P\cdot((Q+\tau\cdot(Q+R))\boxplus_{\pi}S)\approx_{ps}P\cdot((Q+R)\boxplus_{\pi}S)$
  \item $P\approx_{ps} \tau\parallel P$.
\end{enumerate}
\end{proposition}

\begin{proof}
According to the definition of weakly probabilistic step bisimulation, we can easily prove the above equations, and we omit the proof.
\end{proof}

\begin{proposition}[$\tau$ laws for weakly probabilistic hp-bisimulation]\label{TAUWHB3}
The $\tau$ laws for weakly probabilistic hp-bisimulation is as follows.

\begin{enumerate}
  \item $P\approx_{php} \tau.P$;
  \item $\alpha.\tau.P\approx_{php} \alpha.P$;
  \item $(\alpha_1\parallel\cdots\parallel\alpha_n).\tau.P\approx_{php} (\alpha_1\parallel\cdots\parallel\alpha_n).P$;
  \item $P+\tau.P\approx_{php} \tau.P$;
  \item $P\cdot((Q+\tau\cdot(Q+R))\boxplus_{\pi}S)\approx_{php}P\cdot((Q+R)\boxplus_{\pi}S)$
  \item $P\approx_{php} \tau\parallel P$.
\end{enumerate}
\end{proposition}

\begin{proof}
According to the definition of weakly probabilistic hp-bisimulation, we can easily prove the above equations, and we omit the proof.
\end{proof}

\begin{proposition}[$\tau$ laws for weakly probabilistic hhp-bisimulation]\label{TAUWHHPB3}
The $\tau$ laws for weakly probabilistic hhp-bisimulation is as follows.

\begin{enumerate}
  \item $P\approx_{phhp} \tau.P$;
  \item $\alpha.\tau.P\approx_{phhp} \alpha.P$;
  \item $(\alpha_1\parallel\cdots\parallel\alpha_n).\tau.P\approx_{phhp} (\alpha_1\parallel\cdots\parallel\alpha_n).P$;
  \item $P+\tau.P\approx_{phhp} \tau.P$;
  \item $P\cdot((Q+\tau\cdot(Q+R))\boxplus_{\pi}S)\approx_{phhp}P\cdot((Q+R)\boxplus_{\pi}S)$
  \item $P\approx_{phhp} \tau\parallel P$.
\end{enumerate}
\end{proposition}

\begin{proof}
According to the definition of weakly probabilistic hhp-bisimulation, we can easily prove the above equations, and we omit the proof.
\end{proof}

\subsubsection{Recursion}

\begin{definition}[Sequential]
$X$ is sequential in $E$ if every subexpression of $E$ which contains $X$, apart from $X$ itself, is of the form $\alpha.F$, or $(\alpha_1\parallel\cdots\parallel\alpha_n).F$, or
$\sum\widetilde{F}$.
\end{definition}

\begin{definition}[Guarded recursive expression]
$X$ is guarded in $E$ if each occurrence of $X$ is with some subexpression $l.F$ or $(l_1\parallel\cdots\parallel l_n).F$ of $E$.
\end{definition}

\begin{lemma}\label{LUSWW3}
Let $G$ be guarded and sequential, $Vars(G)\subseteq\widetilde{X}$, and let $G\{\widetilde{P}/\widetilde{X}\}\rightsquigarrow\xrightarrow{\{\alpha_1,\cdots,\alpha_n\}}P'$. Then there is an expression
$H$ such that $G\rightsquigarrow\xrightarrow{\{\alpha_1,\cdots,\alpha_n\}}H$, $P'\equiv H\{\widetilde{P}/\widetilde{X}\}$, and for any $\widetilde{Q}$,
$G\{\widetilde{Q}/\widetilde{X}\}\rightsquigarrow\xrightarrow{\{\alpha_1,\cdots,\alpha_n\}} H\{\widetilde{Q}/\widetilde{X}\}$. Moreover $H$ is sequential, $Vars(H)\subseteq\widetilde{X}$, and if
$\alpha_1=\cdots=\alpha_n=\tau$, then $H$ is also guarded.
\end{lemma}

\begin{proof}
We need to induct on the structure of $G$.

If $G$ is a Constant, a Composition, a Restriction or a Relabeling then it contains no variables, since $G$ is sequential and guarded, then
$G\rightsquigarrow\xrightarrow{\{\alpha_1,\cdots,\alpha_n\}}P'$, then let $H\equiv P'$, as desired.

$G$ cannot be a variable, since it is guarded.

If $G\equiv G_1+G_2$. Then either $G_1\{\widetilde{P}/\widetilde{X}\} \rightsquigarrow\xrightarrow{\{\alpha_1,\cdots,\alpha_n\}}P'$ or $G_2\{\widetilde{P}/\widetilde{X}\} \rightsquigarrow
\xrightarrow{\{\alpha_1,\cdots,\alpha_n\}}P'$, then, we can apply this lemma in either case, as desired.

If $G\equiv G_1\boxplus_{\pi}G_2$. Then either $G_1\{\widetilde{P}/\widetilde{X}\} \rightsquigarrow P'$ or $G_2\{\widetilde{P}/\widetilde{X}\} \rightsquigarrow
P'$, then, we can apply this lemma in either case, as desired.

If $G\equiv\beta.H$. Then we must have $\alpha=\beta$, and $P'\equiv H\{\widetilde{P}/\widetilde{X}\}$, and
$G\{\widetilde{Q}/\widetilde{X}\}\equiv \beta.H\{\widetilde{Q}/\widetilde{X}\} \rightsquigarrow\xrightarrow{\beta}H\{\widetilde{Q}/\widetilde{X}\}$, then, let $G'$ be $H$, as desired.

If $G\equiv(\beta_1\parallel\cdots\parallel\beta_n).H$. Then we must have $\alpha_i=\beta_i$ for $1\leq i\leq n$, and $P'\equiv H\{\widetilde{P}/\widetilde{X}\}$, and
$G\{\widetilde{Q}/\widetilde{X}\}\equiv (\beta_1\parallel\cdots\parallel\beta_n).H\{\widetilde{Q}/\widetilde{X}\} \rightsquigarrow\xrightarrow{\{\beta_1,\cdots,\beta_n\}}H\{\widetilde{Q}/\widetilde{X}\}$,
then, let $G'$ be $H$, as desired.

If $G\equiv\tau.H$. Then we must have $\tau=\tau$, and $P'\equiv H\{\widetilde{P}/\widetilde{X}\}$, and $G\{\widetilde{Q}/\widetilde{X}\}\equiv \tau.H\{\widetilde{Q}/\widetilde{X}\}
\rightsquigarrow\xrightarrow{\tau}H\{\widetilde{Q}/\widetilde{X}\}$, then, let $G'$ be $H$, as desired.
\end{proof}

\begin{theorem}[Unique solution of equations for weakly probabilistic step bisimulation]\label{USWSB3}
Let the guarded and sequential expressions $\widetilde{E}$ contain free variables $\subseteq \widetilde{X}$, then,

If $\widetilde{P}\approx_{ps} \widetilde{E}\{\widetilde{P}/\widetilde{X}\}$ and $\widetilde{Q}\approx_{ps} \widetilde{E}\{\widetilde{Q}/\widetilde{X}\}$, then
$\widetilde{P}\approx_{ps} \widetilde{Q}$.
\end{theorem}

\begin{proof}
Like the corresponding theorem in CCS, without loss of generality, we only consider a single equation $X=E$. So we assume $P\approx_{ps} E(P)$, $Q\approx_{ps} E(Q)$, then $P\approx_{ps} Q$.

We will prove $\{(H(P),H(Q)): H\}$ sequential, if $H(P)\rightsquigarrow\xrightarrow{\{\alpha_1,\cdots,\alpha_n\}}P'$, then, for some $Q'$,
$H(Q)\rightsquigarrow\xRightarrow{\{\alpha_1.\cdots,\alpha_n\}}Q'$ and $P'\approx_{ps} Q'$.

Let $H(P)\rightsquigarrow\xrightarrow{\{\alpha_1,\cdot,\alpha_n\}}P'$, then $H(E(P))\rightsquigarrow\xRightarrow{\{\alpha_1,\cdots,\alpha_n\}}P''$ and $P'\approx_{ps} P''$.

By Lemma \ref{LUSWW3}, we know there is a sequential $H'$ such that $H(E(P))\rightsquigarrow\xRightarrow{\{\alpha_1,\cdots,\alpha_n\}}H'(P)\Rightarrow P''\approx_{ps} P'$.

And, $H(E(Q))\rightsquigarrow\xRightarrow{\{\alpha_1,\cdots,\alpha_n\}}H'(Q)\Rightarrow Q''$ and $P''\approx_{ps} Q''$. And $H(Q)\rightsquigarrow\xrightarrow{\{\alpha_1,\cdots,\alpha_n\}}Q'\approx_{ps} Q''$.
Hence, $P'\approx_{ps} Q'$, as desired.
\end{proof}

\begin{theorem}[Unique solution of equations for weakly probabilistic pomset bisimulation]\label{USWPB3}
Let the guarded and sequential expressions $\widetilde{E}$ contain free variables $\subseteq \widetilde{X}$, then,

If $\widetilde{P}\approx_{pp} \widetilde{E}\{\widetilde{P}/\widetilde{X}\}$ and $\widetilde{Q}\approx_{pp} \widetilde{E}\{\widetilde{Q}/\widetilde{X}\}$, then
$\widetilde{P}\approx_{pp} \widetilde{Q}$.
\end{theorem}

\begin{proof}
Similarly to the proof of unique solution of equations for weakly probabilistic step bisimulation $\approx_{ps}$ (Theorem \ref{USWSB3}), we can prove that unique solution of equations
holds for weakly probabilistic pomset bisimulation $\approx_{pp}$, we omit it.
\end{proof}

\begin{theorem}[Unique solution of equations for weakly probabilistic hp-bisimulation]\label{USWHPB3}
Let the guarded and sequential expressions $\widetilde{E}$ contain free variables $\subseteq \widetilde{X}$, then,

If $\widetilde{P}\approx_{php} \widetilde{E}\{\widetilde{P}/\widetilde{X}\}$ and $\widetilde{Q}\approx_{php} \widetilde{E}\{\widetilde{Q}/\widetilde{X}\}$, then $\widetilde{P}\approx_{php} \widetilde{Q}$.
\end{theorem}

\begin{proof}
Similarly to the proof of unique solution of equations for weakly probabilistic pomset bisimulation (Theorem \ref{USWPB3}), we can prove that unique solution of equations holds for weakly
probabilistic hp-bisimulation, we just need additionally to check the above conditions on weakly probabilistic hp-bisimulation, we omit it.
\end{proof}

\begin{theorem}[Unique solution of equations for weakly probabilistic hhp-bisimulation]\label{USWHHPB3}
Let the guarded and sequential expressions $\widetilde{E}$ contain free variables $\subseteq \widetilde{X}$, then,

If $\widetilde{P}\approx_{phhp} \widetilde{E}\{\widetilde{P}/\widetilde{X}\}$ and $\widetilde{Q}\approx_{phhp} \widetilde{E}\{\widetilde{Q}/\widetilde{X}\}$, then $\widetilde{P}\approx_{phhp} \widetilde{Q}$.
\end{theorem}

\begin{proof}
Similarly to the proof of unique solution of equations for weakly probabilistic hp-bisimulation (see Theorem \ref{USWHPB3}), we can prove that the unique solution of equations holds
for weakly probabilistic hhp-bisimulation, we omit it.
\end{proof}

\newpage\section{Algebraic Laws for Probabilistic True Concurrency}\label{apptc4}

The theory $APPTC$ (Algebra of Probabilistic Processes for True Concurrency) has four modules: $BAPTC$ (Basic Algebra for Probabilistic True Concurrency), $APPTC$ (Algebra for Parallelism
in Probabilistic True Concurrency), recursion and abstraction.

This chapter is organized as follows. We introduce $BAPTC$ in section \ref{baptc}, $APPTC$ in section \ref{apptc}, recursion in section \ref{rec}, and abstraction in section \ref{abs}.

\subsection{Basic Algebra for Probabilistic True Concurrency}{\label{baptc}}

In this section, we will discuss the algebraic laws for prime event structure $\mathcal{E}$, exactly for causality $\leq$, conflict $\sharp$ and probabilistic conflict $\sharp_{\pi}$.
We will follow the conventions of process algebra, using $\cdot$ instead of $\leq$, $+$ instead of $\sharp$ and $\boxplus_{\pi}$ instead of $\sharp_{\pi}$. The resulted algebra is called
Basic Algebra for Probabilistic True Concurrency, abbreviated $BAPTC$.

\subsubsection{Axiom System of $BAPTC$}

In the following, let $e_1, e_2, e_1', e_2'\in \mathbb{E}$, and let variables $x,y,z$ range over the set of terms for true concurrency, $p,q,s$ range over the set of closed terms. The
set of axioms of $BAPTC$ consists of the laws given in Table \ref{AxiomsForBAPTC}.

\begin{center}
    \begin{table}
        \begin{tabular}{@{}ll@{}}
            \hline No. &Axiom\\
            $A1$ & $x+ y = y+ x$\\
            $A2$ & $(x+ y)+ z = x+ (y+ z)$\\
            $A3$ & $x+ x = x$\\
            $A4$ & $(x+ y)\cdot z = x\cdot z + y\cdot z$\\
            $A5$ & $(x\cdot y)\cdot z = x\cdot(y\cdot z)$\\
            $PA1$ & $x\boxplus_{\pi} y=y\boxplus_{1-\pi} x$\\
            $PA2$ & $x\boxplus_{\pi}(y\boxplus_{\rho} z)=(x\boxplus_{\frac{\pi}{\pi+\rho-\pi\rho}}y)\boxplus_{\pi+\rho-\pi\rho} z$\\
            $PA3$ & $x\boxplus_{\pi}x=x$\\
            $PA4$ & $(x\boxplus_{\pi}y)\cdot z=x\cdot z\boxplus_{\pi}y\cdot z$\\
            $PA5$ & $(x\boxplus_{\pi}y)+z=(x+z)\boxplus_{\pi}(y+z)$\\
        \end{tabular}
        \caption{Axioms of $BAPTC$}
        \label{AxiomsForBAPTC}
    \end{table}
\end{center}

Intuitively, the axiom $A1$ says that the binary operator $+$ satisfies commutative law. The axiom $A2$ says that $+$ satisfies associativity. $A3$ says that $+$ satisfies idempotency.
The axiom $A4$ is the right distributivity of the binary operator $\cdot$ to $+$. The axiom $A5$ is the associativity of $\cdot$. The axiom $PA1$ is the commutativity of $\boxplus_{\pi}$.
The axiom $PA2$ is the associativity of $\boxplus_{\pi}$. The axiom $PA3$ says that $\boxplus_{\pi}$ satisfies idempotency. The axiom $PA4$ is the right distributivity of $\cdot$ to $\boxplus_{\pi}$.
And the axiom $PA5$ is the right distributivity of $+$ to $\boxplus_{\pi}$.

\subsubsection{Properties of $BAPTC$}

\begin{definition}[Basic terms of $BAPTC$]\label{BTBAPTC}
The set of basic terms of $BAPTC$, $\mathcal{B}(BAPTC)$, is inductively defined as follows:

\begin{enumerate}
  \item $\mathbb{E}\subset\mathcal{B}(BAPTC)$;
  \item if $e\in \mathbb{E}, t\in\mathcal{B}(BAPTC)$ then $e\cdot t\in\mathcal{B}(BAPTC)$;
  \item if $t,s\in\mathcal{B}(BAPTC)$ then $t+ s\in\mathcal{B}(BAPTC)$;
  \item if $t,s\in\mathcal{B}(BAPTC)$ then $t\boxplus_{\pi} s\in\mathcal{B}(BAPTC)$.
\end{enumerate}
\end{definition}

\begin{theorem}[Elimination theorem of $BAPTC$]\label{ETBAPTC}
Let $p$ be a closed $BAPTC$ term. Then there is a basic $BAPTC$ term $q$ such that $BAPTC\vdash p=q$.
\end{theorem}

\begin{proof}
(1) Firstly, suppose that the following ordering on the signature of $BAPTC$ is defined: $\cdot > +>\boxplus_{\pi}$ and the symbol $\cdot$ is given the lexicographical status for the
first argument, then for each rewrite rule $p\rightarrow q$ in Table \ref{TRSForBAPTC} relation $p>_{lpo} q$ can easily be proved. We obtain that the term rewrite system shown in
Table \ref{TRSForBAPTC} is strongly normalizing, for it has finitely many rewriting rules, and $>$ is a well-founded ordering on the signature of $BAPTC$, and if $s>_{lpo} t$,
for each rewriting rule $s\rightarrow t$ is in Table \ref{TRSForBAPTC} (see Theorem \ref{SN}).

\begin{center}
    \begin{table}
        \begin{tabular}{@{}ll@{}}
            \hline No. &Rewriting Rule\\
            $RA3$ & $x+ x \rightarrow x$\\
            $RA4$ & $(x+ y)\cdot z \rightarrow x\cdot z + y\cdot z$\\
            $RA5$ & $(x\cdot y)\cdot z \rightarrow x\cdot(y\cdot z)$\\
            $RPA1$ & $x\boxplus_{\pi} y\rightarrow y\boxplus_{1-\pi} x$\\
            $RPA2$ & $x\boxplus_{\pi}(y\boxplus_{\rho} z)\rightarrow (x\boxplus_{\frac{\pi}{\pi+\rho-\pi\rho}}y)\boxplus_{\pi+\rho-\pi\rho} z$\\
            $RPA3$ & $x\boxplus_{\pi}x\rightarrow x$\\
            $RPA4$ & $(x\boxplus_{\pi}y)\cdot z\rightarrow x\cdot z\boxplus_{\pi}y\cdot z$\\
            $RPA5$ & $(x\boxplus_{\pi}y)+z\rightarrow (x+z)\boxplus_{\pi}(y+z)$\\
        \end{tabular}
        \caption{Term rewrite system of $BAPTC$}
        \label{TRSForBAPTC}
    \end{table}
\end{center}

(2) Then we prove that the normal forms of closed $BAPTC$ terms are basic $BAPTC$ terms.

Suppose that $p$ is a normal form of some closed $BAPTC$ term and suppose that $p$ is not a basic term. Let $p'$ denote the smallest sub-term of $p$ which is not a basic term. It
implies that each sub-term of $p'$ is a basic term. Then we prove that $p$ is not a term in normal form. It is sufficient to induct on the structure of $p'$:

\begin{itemize}
  \item Case $p'\equiv e, e\in \mathbb{E}$. $p'$ is a basic term, which contradicts the assumption that $p'$ is not a basic term, so this case should not occur.
  \item Case $p'\equiv p_1\cdot p_2$. By induction on the structure of the basic term $p_1$:
      \begin{itemize}
        \item Subcase $p_1\in \mathbb{E}$. $p'$ would be a basic term, which contradicts the assumption that $p'$ is not a basic term;
        \item Subcase $p_1\equiv e\cdot p_1'$. $RA5$ rewriting rule can be applied. So $p$ is not a normal form;
        \item Subcase $p_1\equiv p_1'+ p_1''$. $RA4$ rewriting rule can be applied. So $p$ is not a normal form.
      \end{itemize}
  \item Case $p'\equiv p_1+ p_2$. By induction on the structure of the basic terms both $p_1$ and $p_2$, all subcases will lead to that $p'$ would be a basic term, which contradicts
  the assumption that $p'$ is not a basic term;
  \item Case $p'\equiv p_1\boxplus_{\pi} p_2$. By induction on the structure of the basic terms both $p_1$ and $p_2$, all subcases will lead to that $p'$ would be a basic term, which contradicts
  the assumption that $p'$ is not a basic term.
\end{itemize}
\end{proof}

\subsubsection{Structured Operational Semantics of $BAPTC$}

In this subsection, we will define a term-deduction system which gives the operational semantics of $BAPTC$. Like the way in \cite{PPA}, we also introduce the counterpart $\breve{e}$
of the event $e$, and also the set $\breve{\mathbb{E}}=\{\breve{e}|e\in\mathbb{E}\}$.

Firstly, we give the definition of PDFs in Table \ref{PDFBAPTC}.

\begin{center}
    \begin{table}
        $$\mu(e,\breve{e})=1$$
        $$\mu(x\cdot y, x'\cdot y)=\mu(x,x')$$
        $$\mu(x+y,x'+y')=\mu(x,x')\cdot \mu(y,y')$$
        $$\mu(x\boxplus_{\pi}y,z)=\pi\mu(x,z)+(1-\pi)\mu(y,z)$$
        $$\mu(x,y)=0,\textrm{otherwise}$$
        \caption{PDF definitions of $BAPTC$}
        \label{PDFBAPTC}
    \end{table}
\end{center}

We give the operational transition rules for operators $\cdot$, $+$ and
$\boxplus_{\pi}$ as Table \ref{SETRForBAPTC} shows. And the predicate $\xrightarrow{e}\surd$ represents successful termination after execution of the event $e$.

\begin{center}
    \begin{table}
        $$\frac{}{e\rightsquigarrow\breve{e}}$$
        $$\frac{x\rightsquigarrow x'}{x\cdot y\rightsquigarrow x'\cdot y}$$
        $$\frac{x\rightsquigarrow x'\quad y\rightsquigarrow y'}{x+y\rightsquigarrow x'+y'}$$
        $$\frac{x\rightsquigarrow x'}{x\boxplus_{\pi}y\rightsquigarrow x'}\quad \frac{y\rightsquigarrow y'}{x\boxplus_{\pi}y\rightsquigarrow y'}$$
        $$\frac{}{\breve{e}\xrightarrow{e}\surd}$$
        $$\frac{x\xrightarrow{e}\surd}{x+ y\xrightarrow{e}\surd} \quad\frac{x\xrightarrow{e}x'}{x+ y\xrightarrow{e}x'} \quad\frac{y\xrightarrow{e}\surd}{x+ y\xrightarrow{e}\surd} \quad\frac{y\xrightarrow{e}y'}{x+ y\xrightarrow{e}y'}$$
        $$\frac{x\xrightarrow{e}\surd}{x\cdot y\xrightarrow{e} y} \quad\frac{x\xrightarrow{e}x'}{x\cdot y\xrightarrow{e}x'\cdot y}$$
        \caption{Single event transition rules of $BAPTC$}
        \label{SETRForBAPTC}
    \end{table}
\end{center}

The pomset transition rules are shown in Table \ref{PTRForBAPTC}, different to single event transition rules in Table \ref{SETRForBAPTC}, the pomset transition rules are labeled by
pomsets, which are defined by causality $\cdot$, conflict $+$ and $\boxplus_{\pi}$.

\begin{center}
    \begin{table}
        $$\frac{}{X\rightsquigarrow\breve{X}}$$
        $$\frac{x\rightsquigarrow x'}{x\cdot y\rightsquigarrow x'\cdot y}$$
        $$\frac{x\rightsquigarrow x'\quad y\rightsquigarrow y'}{x+y\rightsquigarrow x'+y'}$$
        $$\frac{x\rightsquigarrow x'}{x\boxplus_{\pi}y\rightsquigarrow x'}\quad \frac{y\rightsquigarrow y'}{x\boxplus_{\pi}y\rightsquigarrow y'}$$
        $$\frac{}{\breve{X}\xrightarrow{X}\surd}$$
        $$\frac{x\xrightarrow{X}\surd}{x+ y\xrightarrow{X}\surd} (X\subseteq x)\quad\frac{x\xrightarrow{X}x'}{x+ y\xrightarrow{X}x'} (X\subseteq x) \quad\frac{y\xrightarrow{Y}\surd}{x+ y\xrightarrow{Y}\surd} (Y\subseteq y)\quad\frac{y\xrightarrow{Y}y'}{x+ y\xrightarrow{Y}y'}(Y\subseteq y)$$
        $$\frac{x\xrightarrow{X}\surd}{x\cdot y\xrightarrow{X} y} (X\subseteq x)\quad\frac{x\xrightarrow{X}x'}{x\cdot y\xrightarrow{X}x'\cdot y} (X\subseteq x)$$
        \caption{Pomset transition rules of $BAPTC$}
        \label{PTRForBAPTC}
    \end{table}
\end{center}

\begin{theorem}[Congruence of $BAPTC$ with respect to probabilistic pomset bisimulation equivalence]
Probabilistic pomset bisimulation equivalence $\sim_{pp}$ is a congruence with respect to $BAPTC$.
\end{theorem}

\begin{proof}
It is easy to see that probabilistic pomset bisimulation is an equivalent relation on $BAPTC$ terms, we only need to prove that $\sim_{pp}$ is preserved by the operators $\cdot$, $+$ and $\boxplus_{\pi}$.
That is, if $x\sim_{pp} x'$ and $y\sim_{pp}y'$, we need to prove that $x\cdot y\sim_{pp}x'\cdot y'$, $x+ y\sim_{pp}x'+ y'$ and $x\boxplus_{\pi} y\sim_{pp}x'\boxplus_{\pi} y'$. The proof
is quite trivial and we omit it.
\end{proof}

\begin{theorem}[Soundness of $BAPTC$ modulo probabilistic pomset bisimulation equivalence]\label{SBAPTCPBE}
Let $x$ and $y$ be $BAPTC$ terms. If $BAPTC\vdash x=y$, then $x\sim_{pp} y$.
\end{theorem}

\begin{proof}
Since probabilistic pomset bisimulation $\sim_{pp}$ is both an equivalent and a congruent relation, we only need to check if each axiom in Table \ref{AxiomsForBAPTC} is sound modulo
probabilistic pomset bisimulation equivalence. The proof is quite trivial and we omit it.
\end{proof}

\begin{theorem}[Completeness of $BAPTC$ modulo probabilistic pomset bisimulation equivalence]\label{CBAPTCPBE}
Let $p$ and $q$ be closed $BAPTC$ terms, if $p\sim_{pp} q$ then $p=q$.
\end{theorem}

\begin{proof}
Firstly, by the elimination theorem of $BAPTC$, we know that for each closed $BAPTC$ term $p$, there exists a closed basic $BAPTC$ term $p'$, such that $BAPTC\vdash p=p'$, so, we only
need to consider closed basic $BAPTC$ terms.

The basic terms (see Definition \ref{BTBAPTC}) modulo associativity and commutativity (AC) of conflict $+$ (defined by axioms $A1$ and $A2$ in Table \ref{AxiomsForBAPTC}), and this
equivalence is denoted by $=_{AC}$. Then, each equivalence class $s$ modulo AC of $+$ has the following normal form

$$s_1\boxplus_{\pi_1}\cdots\boxplus_{\pi_{k-1}} s_k$$

with each $s_i$ has the following form

$$t_1+\cdots+ t_l$$

with each $t_j$ either an atomic event or of the form $u_1\cdot u_2$, and each $t_j$ is called the summand of $s$.

Now, we prove that for normal forms $n$ and $n'$, if $n\sim_{pp} n'$ then $n=_{AC}n'$. It is sufficient to induct on the sizes of $n$ and $n'$.

\begin{itemize}
  \item Consider a summand $e$ of $n$. Then $n\rightsquigarrow\breve{e}\xrightarrow{e}\surd$, so $n\sim_{pp} n'$ implies $n'\rightsquigarrow\breve{e}\xrightarrow{e}\surd$, meaning that
  $n'$ also contains the summand $e$.
  \item Consider a summand $t_1\cdot t_2$ of $n$. Then $n\rightsquigarrow\breve{t_1}\xrightarrow{t_1}t_2$, so $n\sim_{pp} n'$ implies $n'\rightsquigarrow\breve{t_1}\xrightarrow{t_1}t_2'$
  with $t_2\sim_{pp} t_2'$, meaning that $n'$ contains a summand $t_1\cdot t_2'$. Since $t_2$ and $t_2'$ are normal forms and have sizes smaller than $n$ and $n'$, by
  the induction hypotheses $t_2\sim_{pp} t_2'$ implies $t_2=_{AC} t_2'$.
\end{itemize}

So, we get $n=_{AC} n'$.

Finally, let $s$ and $t$ be basic terms, and $s\sim_{pp} t$, there are normal forms $n$ and $n'$, such that $s=n$ and $t=n'$. The soundness theorem of $BAPTC$ modulo probabilistic
pomset bisimulation equivalence (see Theorem \ref{SBAPTCPBE}) yields $s\sim_{pp} n$ and $t\sim_{pp} n'$, so $n\sim_{pp} s\sim_{pp} t\sim_{pp} n'$. Since if $n\sim_{pp} n'$ then
$n=_{AC}n'$, $s=n=_{AC}n'=t$, as desired.
\end{proof}

The step transition rules are similar in Table \ref{PTRForBAPTC}, different to pomset transition rules, the step transition rules are labeled by steps, in which every event is pairwise
concurrent.

\begin{theorem}[Congruence of $BAPTC$ with respect to probabilistic step bisimulation equivalence]
Probabilistic probabilistic step bisimulation equivalence $\sim_{ps}$ is a congruence with respect to $BAPTC$.
\end{theorem}

\begin{proof}
It is easy to see that probabilistic step bisimulation is an equivalent relation on $BAPTC$ terms, we only need to prove that $\sim_{ps}$ is preserved by the operators $\cdot$, $+$ and
$\boxplus_{\pi}$.
That is, if $x\sim_{ps} x'$ and $y\sim_{ps}y'$, we need to prove that $x\cdot y\sim_{ps}x'\cdot y'$, $x+ y\sim_{ps}x'+ y'$ and $x\boxplus_{\pi} y\sim_{ps}x'\boxplus_{\pi} y'$. The proof
is quite trivial and we omit it
\end{proof}

\begin{theorem}[Soundness of $BAPTC$ modulo probabilistic step bisimulation equivalence]\label{SBAPTCSBE}
Let $x$ and $y$ be $BAPTC$ terms. If $BAPTC\vdash x=y$, then $x\sim_{ps} y$.
\end{theorem}

\begin{proof}
Since probabilistic step bisimulation $\sim_{ps}$ is both an equivalent and a congruent relation, we only need to check if each axiom in Table \ref{AxiomsForBAPTC} is sound modulo
probabilistic step bisimulation equivalence. The proof is quite trivial and we omit it.
\end{proof}

\begin{theorem}[Completeness of $BAPTC$ modulo probabilistic step bisimulation equivalence]\label{CBAPTCSBE}
Let $p$ and $q$ be closed $BAPTC$ terms, if $p\sim_{ps} q$ then $p=q$.
\end{theorem}

\begin{proof}
Firstly, by the elimination theorem of $BAPTC$, we know that for each closed $BAPTC$ term $p$, there exists a closed basic $BAPTC$ term $p'$, such that $BAPTC\vdash p=p'$, so, we only
need to consider closed basic $BAPTC$ terms.

The basic terms (see Definition \ref{BTBAPTC}) modulo associativity and commutativity (AC) of conflict $+$ (defined by axioms $A1$ and $A2$ in Table \ref{AxiomsForBAPTC}), and this
equivalence is denoted by $=_{AC}$. Then, each equivalence class $s$ modulo AC of $+$ has the following normal form

$$s_1\boxplus_{\pi_1}\cdots\boxplus_{\pi_{k-1}} s_k$$

with each $s_i$ has the following form

$$t_1+\cdots+ t_l$$

with each $t_j$ either an atomic event or of the form $u_1\cdot u_2$, and each $t_j$ is called the summand of $s$.

Now, we prove that for normal forms $n$ and $n'$, if $n\sim_{ps} n'$ then $n=_{AC}n'$. It is sufficient to induct on the sizes of $n$ and $n'$.

\begin{itemize}
  \item Consider a summand $e$ of $n$. Then $n\rightsquigarrow\breve{e}\xrightarrow{e}\surd$, so $n\sim_{ps} n'$ implies $n'\rightsquigarrow\breve{e}\xrightarrow{e}\surd$, meaning that
  $n'$ also contains the summand $e$.
  \item Consider a summand $t_1\cdot t_2$ of $n$. Then $n\rightsquigarrow\breve{t_1}\xrightarrow{t_1}t_2$, so $n\sim_{ps} n'$ implies $n'\rightsquigarrow\breve{t_1}\xrightarrow{t_1}t_2'$
  with $t_2\sim_{ps} t_2'$, meaning that $n'$ contains a summand $t_1\cdot t_2'$. Since $t_2$ and $t_2'$ are normal forms and have sizes smaller than $n$ and $n'$, by
  the induction hypotheses $t_2\sim_{ps} t_2'$ implies $t_2=_{AC} t_2'$.
\end{itemize}

So, we get $n=_{AC} n'$.

Finally, let $s$ and $t$ be basic terms, and $s\sim_{ps} t$, there are normal forms $n$ and $n'$, such that $s=n$ and $t=n'$. The soundness theorem of $BAPTC$ modulo probabilistic
pomset bisimulation equivalence (see Theorem \ref{SBAPTCSBE}) yields $s\sim_{ps} n$ and $t\sim_{ps} n'$, so $n\sim_{ps} s\sim_{ps} t\sim_{ps} n'$. Since if $n\sim_{ps} n'$ then
$n=_{AC}n'$, $s=n=_{AC}n'=t$, as desired.
\end{proof}

The transition rules for (hereditary) hp-bisimulation of $BAPTC$ are same as single event transition rules in Table \ref{SETRForBAPTC}.

\begin{theorem}[Congruence of $BAPTC$ with respect to probabilistic hp-bisimulation equivalence]
Probabilistic hp-bisimulation equivalence $\sim_{php}$ is a congruence with respect to $BAPTC$.
\end{theorem}

\begin{proof}
It is easy to see that probabilistic hp-bisimulation is an equivalent relation on $BAPTC$ terms, we only need to prove that $\sim_{php}$ is preserved by the operators $\cdot$, $+$
and $\boxplus_{\pi}$.
That is, if $x\sim_{php} x'$ and $y\sim_{php}y'$, we need to prove that $x\cdot y\sim_{php}x'\cdot y'$, $x+ y\sim_{php}x'+ y'$ and $x\boxplus_{\pi} y\sim_{php}x'\boxplus_{\pi} y'$. The proof
is quite trivial and we omit it.
\end{proof}

\begin{theorem}[Soundness of $BAPTC$ modulo probabilistic hp-bisimulation equivalence]\label{SBAPTCHPBE}
Let $x$ and $y$ be $BAPTC$ terms. If $BAPTC\vdash x=y$, then $x\sim_{php} y$.
\end{theorem}

\begin{proof}
Since probabilistic hp-bisimulation $\sim_{php}$ is both an equivalent and a congruent relation, we only need to check if each axiom in Table \ref{AxiomsForBAPTC} is sound modulo
probabilistic hp-bisimulation equivalence. The proof is quite trivial and we omit it.
\end{proof}

\begin{theorem}[Completeness of $BAPTC$ modulo probabilistic hp-bisimulation equivalence]\label{CBAPTCHPBE}
Let $p$ and $q$ be closed $BAPTC$ terms, if $p\sim_{php} q$ then $p=q$.
\end{theorem}

\begin{proof}
Firstly, by the elimination theorem of $BAPTC$, we know that for each closed $BAPTC$ term $p$, there exists a closed basic $BAPTC$ term $p'$, such that $BAPTC\vdash p=p'$, so, we only
need to consider closed basic $BAPTC$ terms.

The basic terms (see Definition \ref{BTBAPTC}) modulo associativity and commutativity (AC) of conflict $+$ (defined by axioms $A1$ and $A2$ in Table \ref{AxiomsForBAPTC}), and this
equivalence is denoted by $=_{AC}$. Then, each equivalence class $s$ modulo AC of $+$ has the following normal form

$$s_1\boxplus_{\pi_1}\cdots\boxplus_{\pi_{k-1}} s_k$$

with each $s_i$ has the following form

$$t_1+\cdots+ t_l$$

with each $t_j$ either an atomic event or of the form $u_1\cdot u_2$, and each $t_j$ is called the summand of $s$.

Now, we prove that for normal forms $n$ and $n'$, if $n\sim_{php} n'$ then $n=_{AC}n'$. It is sufficient to induct on the sizes of $n$ and $n'$.

\begin{itemize}
  \item Consider a summand $e$ of $n$. Then $n\rightsquigarrow\breve{e}\xrightarrow{e}\surd$, so $n\sim_{php} n'$ implies $n'\rightsquigarrow\breve{e}\xrightarrow{e}\surd$, meaning that
  $n'$ also contains the summand $e$.
  \item Consider a summand $t_1\cdot t_2$ of $n$. Then $n\rightsquigarrow\breve{t_1}\xrightarrow{t_1}t_2$, so $n\sim_{php} n'$ implies $n'\rightsquigarrow\breve{t_1}\xrightarrow{t_1}t_2'$
  with $t_2\sim_{php} t_2'$, meaning that $n'$ contains a summand $t_1\cdot t_2'$. Since $t_2$ and $t_2'$ are normal forms and have sizes smaller than $n$ and $n'$, by
  the induction hypotheses $t_2\sim_{php} t_2'$ implies $t_2=_{AC} t_2'$.
\end{itemize}

So, we get $n=_{AC} n'$.

Finally, let $s$ and $t$ be basic terms, and $s\sim_{php} t$, there are normal forms $n$ and $n'$, such that $s=n$ and $t=n'$. The soundness theorem of $BAPTC$ modulo probabilistic
pomset bisimulation equivalence (see Theorem \ref{SBAPTCHPBE}) yields $s\sim_{php} n$ and $t\sim_{php} n'$, so $n\sim_{php} s\sim_{php} t\sim_{php} n'$. Since if $n\sim_{php} n'$ then
$n=_{AC}n'$, $s=n=_{AC}n'=t$, as desired.
\end{proof}

\begin{theorem}[Congruence of $BAPTC$ with respect to probabilistic hhp-bisimulation equivalence]
Probabilistic hhp-bisimulation equivalence $\sim_{phhp}$ is a congruence with respect to $BAPTC$.
\end{theorem}

\begin{proof}
It is easy to see that probabilistic hhp-bisimulation is an equivalent relation on $BAPTC$ terms, we only need to prove that $\sim_{phhp}$ is preserved by the operators $\cdot$, $+$ and
$\boxplus_{\pi}$.
That is, if $x\sim_{phhp} x'$ and $y\sim_{phhp}y'$, we need to prove that $x\cdot y\sim_{phhp}x'\cdot y'$, $x+ y\sim_{phhp}x'+ y'$ and $x\boxplus_{\pi} y\sim_{phhp}x'\boxplus_{\pi} y'$. The proof
is quite trivial and we omit it.
\end{proof}

\begin{theorem}[Soundness of $BAPTC$ modulo probabilistic hhp-bisimulation equivalence]\label{SBAPTCHHPBE}
Let $x$ and $y$ be $BAPTC$ terms. If $BAPTC\vdash x=y$, then $x\sim_{phhp} y$.
\end{theorem}

\begin{proof}
Since probabilistic hhp-bisimulation $\sim_{phhp}$ is both an equivalent and a congruent relation, we only need to check if each axiom in Table \ref{AxiomsForBAPTC} is sound modulo
probabilistic hhp-bisimulation equivalence. It is quite trivial and we omit it.
\end{proof}

\begin{theorem}[Completeness of $BAPTC$ modulo probabilistic hhp-bisimulation equivalence]\label{CBAPTCHHPBE}
Let $p$ and $q$ be closed $BAPTC$ terms, if $p\sim_{phhp} q$ then $p=q$.
\end{theorem}

\begin{proof}
Firstly, by the elimination theorem of $BAPTC$, we know that for each closed $BAPTC$ term $p$, there exists a closed basic $BAPTC$ term $p'$, such that $BAPTC\vdash p=p'$, so, we only
need to consider closed basic $BAPTC$ terms.

The basic terms (see Definition \ref{BTBAPTC}) modulo associativity and commutativity (AC) of conflict $+$ (defined by axioms $A1$ and $A2$ in Table \ref{AxiomsForBAPTC}), and this
equivalence is denoted by $=_{AC}$. Then, each equivalence class $s$ modulo AC of $+$ has the following normal form

$$s_1\boxplus_{\pi_1}\cdots\boxplus_{\pi_{k-1}} s_k$$

with each $s_i$ has the following form

$$t_1+\cdots+ t_l$$

with each $t_j$ either an atomic event or of the form $u_1\cdot u_2$, and each $t_j$ is called the summand of $s$.

Now, we prove that for normal forms $n$ and $n'$, if $n\sim_{phhp} n'$ then $n=_{AC}n'$. It is sufficient to induct on the sizes of $n$ and $n'$.

\begin{itemize}
  \item Consider a summand $e$ of $n$. Then $n\rightsquigarrow\breve{e}\xrightarrow{e}\surd$, so $n\sim_{phhp} n'$ implies $n'\rightsquigarrow\breve{e}\xrightarrow{e}\surd$, meaning that
  $n'$ also contains the summand $e$.
  \item Consider a summand $t_1\cdot t_2$ of $n$. Then $n\rightsquigarrow\breve{t_1}\xrightarrow{t_1}t_2$, so $n\sim_{phhp} n'$ implies $n'\rightsquigarrow\breve{t_1}\xrightarrow{t_1}t_2'$
  with $t_2\sim_{phhp} t_2'$, meaning that $n'$ contains a summand $t_1\cdot t_2'$. Since $t_2$ and $t_2'$ are normal forms and have sizes smaller than $n$ and $n'$, by
  the induction hypotheses $t_2\sim_{phhp} t_2'$ implies $t_2=_{AC} t_2'$.
\end{itemize}

So, we get $n=_{AC} n'$.

Finally, let $s$ and $t$ be basic terms, and $s\sim_{phhp} t$, there are normal forms $n$ and $n'$, such that $s=n$ and $t=n'$. The soundness theorem of $BAPTC$ modulo probabilistic
pomset bisimulation equivalence (see Theorem \ref{SBAPTCHHPBE}) yields $s\sim_{phhp} n$ and $t\sim_{phhp} n'$, so $n\sim_{phhp} s\sim_{phhp} t\sim_{phhp} n'$. Since if $n\sim_{phhp} n'$ then
$n=_{AC}n'$, $s=n=_{AC}n'=t$, as desired.
\end{proof}

\subsection{Algebra for Parallelism in Probabilistic True Concurrency}\label{apptc}

In this section, we will discuss parallelism in probabilistic true concurrency. The resulted algebra is called Algebra for Parallelism in Probabilistic True Concurrency, abbreviated $APPTC$.

\subsubsection{Axiom System of Parallelism}

We design the axioms of parallelism in Table \ref{AxiomsForPParallelism}, including algebraic laws for parallel operator $\parallel$, communication operator $\mid$, conflict elimination
operator $\Theta$ and unless operator $\triangleleft$, and also the whole parallel operator $\between$. Since the communication between two communicating events in different parallel
branches may cause deadlock (a state of inactivity), which is caused by mismatch of two communicating events or the imperfectness of the communication channel. We introduce a new
constant $\delta$ to denote the deadlock, and let the atomic event $e\in \mathbb{E}\cup\{\delta\}$.

\begin{center}
    \begin{table}
        \begin{tabular}{@{}ll@{}}
            \hline No. &Axiom\\
            $A3$ & $e+e=e$\\
            $A6$ & $x+ \delta = x$\\
            $A7$ & $\delta\cdot x =\delta$\\
            $P1$ & $(x+x=x,y+y=y)\quad x\between y = x\parallel y + x\mid y$\\
            $P2$ & $x\parallel y = y \parallel x$\\
            $P3$ & $(x\parallel y)\parallel z = x\parallel (y\parallel z)$\\
            $P4$ & $(x+x=x,y+y=y)\quad x\parallel y = x\leftmerge y + y\leftmerge x$\\
            $P5$ & $(e_1\leq e_2)\quad e_1\leftmerge (e_2\cdot y) = (e_1\leftmerge e_2)\cdot y$\\
            $P6$ & $(e_1\leq e_2)\quad (e_1\cdot x)\leftmerge e_2 = (e_1\leftmerge e_2)\cdot x$\\
            $P7$ & $(e_1\leq e_2)\quad (e_1\cdot x)\leftmerge (e_2\cdot y) = (e_1\leftmerge e_2)\cdot (x\between y)$\\
            $P8$ & $(x+ y)\leftmerge z = (x\leftmerge z)+ (y\leftmerge z)$\\
            $P9$ & $\delta\leftmerge x = \delta$\\
            $C10$ & $e_1\mid e_2 = \gamma(e_1,e_2)$\\
            $C11$ & $e_1\mid (e_2\cdot y) = \gamma(e_1,e_2)\cdot y$\\
            $C12$ & $(e_1\cdot x)\mid e_2 = \gamma(e_1,e_2)\cdot x$\\
            $C13$ & $(e_1\cdot x)\mid (e_2\cdot y) = \gamma(e_1,e_2)\cdot (x\between y)$\\
            $C14$ & $(x+ y)\mid z = (x\mid z) + (y\mid z)$\\
            $C15$ & $x\mid (y+ z) = (x\mid y)+ (x\mid z)$\\
            $C16$ & $\delta\mid x = \delta$\\
            $C17$ & $x\mid\delta = \delta$\\
            $PM1$ & $x\parallel (y\boxplus_{\pi} z)=(x\parallel y)\boxplus_{\pi}(x\parallel z)$\\
            $PM2$ & $(x\boxplus_{\pi} y)\parallel z=(x\parallel z)\boxplus_{\pi}(y\parallel z)$\\
            $PM3$ & $x\mid (y\boxplus_{\pi} z)=(x\mid y)\boxplus_{\pi}(x\mid z)$\\
            $PM4$ & $(x\boxplus_{\pi} y)\mid z=(x\mid z)\boxplus_{\pi}(y\mid z)$\\
            $CE18$ & $\Theta(e) = e$\\
            $CE19$ & $\Theta(\delta) = \delta$\\
            $CE20$ & $\Theta(x+ y) = \Theta(x)\triangleleft y + \Theta(y)\triangleleft x$\\
            $PCE1$ & $\Theta(x\boxplus_{\pi} y) = \Theta(x)\triangleleft y \boxplus_{\pi} \Theta(y)\triangleleft x$\\
            $CE21$ & $\Theta(x\cdot y)=\Theta(x)\cdot\Theta(y)$\\
            $CE22$ & $\Theta(x\leftmerge y) = ((\Theta(x)\triangleleft y)\leftmerge y)+ ((\Theta(y)\triangleleft x)\leftmerge x)$\\
            $CE23$ & $\Theta(x\mid y) = ((\Theta(x)\triangleleft y)\mid y)+ ((\Theta(y)\triangleleft x)\mid x)$\\
            $U24$ & $(\sharp(e_1,e_2))\quad e_1\triangleleft e_2 = \tau$\\
            $U25$ & $(\sharp(e_1,e_2),e_2\leq e_3)\quad e_1\triangleleft e_3 = e_1$\\
            $U26$ & $(\sharp(e_1,e_2),e_2\leq e_3)\quad e3\triangleleft e_1 = \tau$\\
            $PU1$ & $(\sharp_{\pi}(e_1,e_2))\quad e_1\triangleleft e_2 = \tau$\\
            $PU2$ & $(\sharp_{\pi}(e_1,e_2),e_2\leq e_3)\quad e_1\triangleleft e_3 = e_1$\\
            $PU3$ & $(\sharp_{\pi}(e_1,e_2),e_2\leq e_3)\quad e_3\triangleleft e_1 = \tau$\\
            $U27$ & $e\triangleleft \delta = e$\\
            $U28$ & $\delta \triangleleft e = \delta$\\
            $U29$ & $(x+ y)\triangleleft z = (x\triangleleft z)+ (y\triangleleft z)$\\
            $PU4$ & $(x\boxplus_{\pi} y)\triangleleft z = (x\triangleleft z)\boxplus_{\pi} (y\triangleleft z)$\\
            $U30$ & $(x\cdot y)\triangleleft z = (x\triangleleft z)\cdot (y\triangleleft z)$\\
            $U31$ & $(x\leftmerge y)\triangleleft z = (x\triangleleft z)\leftmerge (y\triangleleft z)$\\
            $U32$ & $(x\mid y)\triangleleft z = (x\triangleleft z)\mid (y\triangleleft z)$\\
            $U33$ & $x\triangleleft (y+ z) = (x\triangleleft y)\triangleleft z$\\
            $PU5$ & $x\triangleleft (y\boxplus_{\pi} z) = (x\triangleleft y)\triangleleft z$\\
            $U34$ & $x\triangleleft (y\cdot z)=(x\triangleleft y)\triangleleft z$\\
            $U35$ & $x\triangleleft (y\leftmerge z) = (x\triangleleft y)\triangleleft z$\\
            $U36$ & $x\triangleleft (y\mid z) = (x\triangleleft y)\triangleleft z$\\
        \end{tabular}
        \caption{Axioms of parallelism}
        \label{AxiomsForPParallelism}
    \end{table}
\end{center}

We explain the intuitions of the axioms of parallelism in Table \ref{AxiomsForPParallelism} in the following. The axiom $A6$ says that the deadlock $\delta$ is redundant in the process
term $t+ \delta$. $A7$ says that the deadlock blocks all behaviors of the process term $\delta\cdot t$.

The axiom $P1$ is the definition of the whole parallelism $\between$, which says that $s\between t$ either is the form of $s\parallel t$ or $s\mid t$. $P2$ says that $\parallel$
satisfies commutative law, while $P3$ says that $\parallel$ satisfies associativity. $P4$, $P5$ and $P6$ are the defining axioms of $\parallel$, say the $s\parallel t$ executes $s$ and
$t$ concurrently. $P8$ is the right of $\leftmerge$ to $+$. $P9$ says that $\delta\leftmerge t$ blocks any event.

$C10$, $C11$, $C12$ and $C13$ are the defining axioms of the communication operator $\mid$ which say that $s\mid t$ makes a communication between $s$ and $t$. $C14$ and $C15$ are the
right and left distributivity of $\mid$ to $+$. $C16$ and $C17$ say that both $\delta\mid t$ and $t\mid\delta$ all block any event.

$CE18$ and $CE19$ say that the conflict elimination operator $\Theta$ leaves atomic events and the deadlock unchanged. $CE20-CE23$ are the functions of $\Theta$ acting on the operators
$+$, $\cdot$, $\parallel$ and $\mid$. $U24$, $U25$ and $U26$ are the defining laws of the unless operator $\triangleleft$, in $U24$ and $U26$, there is a new constant $\tau$, the
silent step, we will discuss $\tau$ in details in section \ref{abs}, in these two axioms, we just need to remember that $\tau$ really keeps silent. $U27$ says that the deadlock
$\delta$ cannot block any event in the process term $e\triangleleft\delta$, while $U28$ says that $\delta\triangleleft e$ does not exhibit any behavior. $U29-U36$ are the disguised
right and left distributivity of $\triangleleft$ to the operators $+$, $\cdot$, $\parallel$ and $\mid$.

The axiom $A3$ in the above section is replaced by the new one, because of the introduction of $\delta$. The axioms $PM1$, $PM2$, $PM3$ and $PM4$ are the distributivity of $\parallel$ and
$\mid$ to $\boxplus_{\pi}$. $PCE1$ is the function of $\Theta$ acting on the operator $\boxplus_{\pi}$. $PU1$, $PU2$ and $PU3$ are the defining laws of the unless operator $\triangleleft$
for $\sharp_{\pi}$. $PU4$ and $PU5$ is the disguised right and left distributivity of $\triangleleft$ to the operators $\boxplus_{\pi}$.

\begin{definition}[Basic terms of $APPTC$]\label{BTAPPTC}
The set of basic terms of $APPTC$, $\mathcal{B}(APPTC)$, is inductively defined as follows:
\begin{enumerate}
  \item $\mathbb{E}\subset\mathcal{B}(APPTC)$;
  \item if $e\in \mathbb{E}, t\in\mathcal{B}(APPTC)$ then $e\cdot t\in\mathcal{B}(APPTC)$;
  \item if $t,s\in\mathcal{B}(APPTC)$ then $t+ s\in\mathcal{B}(APPTC)$;
  \item if $t,s\in\mathcal{B}(APPTC)$ then $t\boxplus_{\pi} s\in\mathcal{B}(APPTC)$;
  \item if $t,s\in\mathcal{B}(APPTC)$ then $t\leftmerge s\in\mathcal{B}(APPTC)$.
\end{enumerate}
\end{definition}

Based on the definition of basic terms for $APPTC$ (see Definition \ref{BTAPPTC}) and axioms of parallelism (see Table \ref{AxiomsForPParallelism}), we can prove the elimination
theorem of parallelism.

\begin{theorem}[Elimination theorem of parallelism]\label{ETPParallelism}
Let $p$ be a closed $APPTC$ term. Then there is a basic $APPTC$ term $q$ such that $APPTC\vdash p=q$.
\end{theorem}

\begin{proof}
(1) Firstly, suppose that the following ordering on the signature of $APPTC$ is defined: $\leftmerge > \cdot > +>\boxplus_{\pi}$ and the symbol $\parallel$ is given the lexicographical
status for the first argument, then for each rewrite rule $p\rightarrow q$ in Table \ref{TRSForAPPTC} relation $p>_{lpo} q$ can easily be proved. We obtain that the term rewrite system
shown in Table \ref{TRSForAPPTC} is strongly normalizing, for it has finitely many rewriting rules, and $>$ is a well-founded ordering on the signature of $APPTC$, and if $s>_{lpo} t$,
for each rewriting rule $s\rightarrow t$ is in Table \ref{TRSForAPPTC} (see Theorem \ref{SN}).

\begin{center}
    \begin{table}
        \begin{tabular}{@{}ll@{}}
            \hline No. &Rewriting Rule\\
            $RA3$ & $e+e\rightarrow e$\\
            $RA6$ & $x+ \delta \rightarrow x$\\
            $RA7$ & $\delta\cdot x \rightarrow\delta$\\
            $RP1$ & $(x+x=x,y+y=y)\quad x\between y \rightarrow x\parallel y + x\mid y$\\
            $RP2$ & $x\parallel y \rightarrow y \parallel x$\\
            $RP3$ & $(x\parallel y)\parallel z \rightarrow x\parallel (y\parallel z)$\\
            $RP4$ & $(x+x=x,y+y=y)\quad x\parallel y \rightarrow x\leftmerge y + y\leftmerge x$\\
            $RP5$ & $(e_1\leq e_2)\quad e_1\leftmerge (e_2\cdot y) \rightarrow (e_1\leftmerge e_2)\cdot y$\\
            $RP6$ & $(e_1\leq e_2)\quad (e_1\cdot x)\leftmerge e_2 \rightarrow (e_1\leftmerge e_2)\cdot x$\\
            $RP7$ & $(e_1\leq e_2)\quad (e_1\cdot x)\leftmerge (e_2\cdot y) \rightarrow (e_1\leftmerge e_2)\cdot (x\between y)$\\
            $RP8$ & $(x+ y)\leftmerge z \rightarrow (x\leftmerge z)+ (y\leftmerge z)$\\
            $RP9$ & $\delta\leftmerge x \rightarrow \delta$\\
            $RC10$ & $e_1\mid e_2 \rightarrow \gamma(e_1,e_2)$\\
            $RC11$ & $e_1\mid (e_2\cdot y) \rightarrow \gamma(e_1,e_2)\cdot y$\\
            $RC12$ & $(e_1\cdot x)\mid e_2 \rightarrow \gamma(e_1,e_2)\cdot x$\\
            $RC13$ & $(e_1\cdot x)\mid (e_2\cdot y) \rightarrow \gamma(e_1,e_2)\cdot (x\between y)$\\
            $RC14$ & $(x+ y)\mid z \rightarrow (x\mid z) + (y\mid z)$\\
            $RC15$ & $x\mid (y+ z) \rightarrow (x\mid y)+ (x\mid z)$\\
            $RC16$ & $\delta\mid x \rightarrow \delta$\\
            $RC17$ & $x\mid\delta \rightarrow \delta$\\
            $RPM1$ & $x\parallel (y\boxplus_{\pi} z)\rightarrow(x\parallel y)\boxplus_{\pi}(x\parallel z)$\\
            $RPM2$ & $(x\boxplus_{\pi} y)\parallel z\rightarrow(x\parallel z)\boxplus_{\pi}(y\parallel z)$\\
            $RPM3$ & $x\mid (y\boxplus_{\pi} z)\rightarrow(x\mid y)\boxplus_{\pi}(x\mid z)$\\
            $RPM4$ & $(x\boxplus_{\pi} y)\mid z\rightarrow(x\mid z)\boxplus_{\pi}(y\mid z)$\\
            $RCE18$ & $\Theta(e) \rightarrow e$\\
            $RCE19$ & $\Theta(\delta) \rightarrow \delta$\\
            $RCE20$ & $\Theta(x+ y) \rightarrow \Theta(x)\triangleleft y + \Theta(y)\triangleleft x$\\
            $RPCE1$ & $\Theta(x\boxplus_{\pi} y) \rightarrow \Theta(x)\triangleleft y \boxplus_{\pi} \Theta(y)\triangleleft x$\\
            $RCE21$ & $\Theta(x\cdot y)\rightarrow\Theta(x)\cdot\Theta(y)$\\
            $RCE22$ & $\Theta(x\leftmerge y) \rightarrow ((\Theta(x)\triangleleft y)\leftmerge y)+ ((\Theta(y)\triangleleft x)\leftmerge x)$\\
            $RCE23$ & $\Theta(x\mid y) \rightarrow ((\Theta(x)\triangleleft y)\mid y)+ ((\Theta(y)\triangleleft x)\mid x)$\\
            $RU24$ & $(\sharp(e_1,e_2))\quad e_1\triangleleft e_2 \rightarrow \tau$\\
            $RU25$ & $(\sharp(e_1,e_2),e_2\leq e_3)\quad e_1\triangleleft e_3 \rightarrow e_1$\\
            $RU26$ & $(\sharp(e_1,e_2),e_2\leq e_3)\quad e3\triangleleft e_1 \rightarrow \tau$\\
            $RPU1$ & $(\sharp_{\pi}(e_1,e_2))\quad e_1\triangleleft e_2 \rightarrow \tau$\\
            $RPU2$ & $(\sharp_{\pi}(e_1,e_2),e_2\leq e_3)\quad e_1\triangleleft e_3 \rightarrow e_1$\\
            $RPU3$ & $(\sharp_{\pi}(e_1,e_2),e_2\leq e_3)\quad e_3\triangleleft e_1 \rightarrow \tau$\\
            $RU27$ & $e\triangleleft \delta \rightarrow e$\\
            $RU28$ & $\delta \triangleleft e \rightarrow \delta$\\
            $RU29$ & $(x+ y)\triangleleft z \rightarrow (x\triangleleft z)+ (y\triangleleft z)$\\
            $RPU4$ & $(x\boxplus_{\pi} y)\triangleleft z \rightarrow (x\triangleleft z)\boxplus_{\pi} (y\triangleleft z)$\\
            $RU30$ & $(x\cdot y)\triangleleft z \rightarrow (x\triangleleft z)\cdot (y\triangleleft z)$\\
            $RU31$ & $(x\leftmerge y)\triangleleft z \rightarrow (x\triangleleft z)\leftmerge (y\triangleleft z)$\\
            $RU32$ & $(x\mid y)\triangleleft z \rightarrow (x\triangleleft z)\mid (y\triangleleft z)$\\
            $RU33$ & $x\triangleleft (y+ z) \rightarrow (x\triangleleft y)\triangleleft z$\\
            $RPU5$ & $x\triangleleft (y\boxplus_{\pi} z) \rightarrow (x\triangleleft y)\triangleleft z$\\
            $RU34$ & $x\triangleleft (y\cdot z)\rightarrow(x\triangleleft y)\triangleleft z$\\
            $RU35$ & $x\triangleleft (y\leftmerge z) \rightarrow (x\triangleleft y)\triangleleft z$\\
            $RU36$ & $x\triangleleft (y\mid z) \rightarrow (x\triangleleft y)\triangleleft z$\\
        \end{tabular}
        \caption{Term rewrite system of $APPTC$}
        \label{TRSForAPPTC}
    \end{table}
\end{center}

(2) Then we prove that the normal forms of closed $APPTC$ terms are basic $APPTC$ terms.

Suppose that $p$ is a normal form of some closed $APPTC$ term and suppose that $p$ is not a basic $APPTC$ term. Let $p'$ denote the smallest sub-term of $p$ which is not a basic $APPTC$ term. It implies that each sub-term of $p'$ is a basic $APPTC$ term. Then we prove that $p$ is not a term in normal form. It is sufficient to induct on the structure of $p'$:

\begin{itemize}
  \item Case $p'\equiv e, e\in \mathbb{E}$. $p'$ is a basic $APPTC$ term, which contradicts the assumption that $p'$ is not a basic $APPTC$ term, so this case should not occur.
  \item Case $p'\equiv p_1\cdot p_2$. By induction on the structure of the basic $APPTC$ term $p_1$:
      \begin{itemize}
        \item Subcase $p_1\in \mathbb{E}$. $p'$ would be a basic $APPTC$ term, which contradicts the assumption that $p'$ is not a basic $APPTC$ term;
        \item Subcase $p_1\equiv e\cdot p_1'$. $RA5$ rewriting rule in Table \ref{TRSForBAPTC} can be applied. So $p$ is not a normal form;
        \item Subcase $p_1\equiv p_1'+ p_1''$. $RA4$ rewriting rule in Table \ref{TRSForBAPTC} can be applied. So $p$ is not a normal form;
        \item Subcase $p_1\equiv p_1'\boxplus_{\pi} p_1''$. $RPA4$ and $RPA5$ rewriting rules in Table \ref{TRSForBAPTC} can be applied. So $p$ is not a normal form;
        \item Subcase $p_1\equiv p_1'\leftmerge p_1''$. $p'$ would be a basic $APPTC$ term, which contradicts the assumption that $p'$ is not a basic $APPTC$ term;
        \item Subcase $p_1\equiv p_1'\mid p_1''$. $RC11$ rewrite rule in Table \ref{TRSForAPPTC} can be applied. So $p$ is not a normal form;
        \item Subcase $p_1\equiv \Theta(p_1')$. $RCE19$ and $RCE20$ rewrite rules in Table \ref{TRSForAPPTC} can be applied. So $p$ is not a normal form.
      \end{itemize}
  \item Case $p'\equiv p_1+ p_2$. By induction on the structure of the basic $APPTC$ terms both $p_1$ and $p_2$, all subcases will lead to that $p'$ would be a basic $APPTC$ term,
  which contradicts the assumption that $p'$ is not a basic $APPTC$ term.
  \item Case $p'\equiv p_1\boxplus_{\pi} p_2$. By induction on the structure of the basic $APPTC$ terms both $p_1$ and $p_2$, all subcases will lead to that $p'$ would be a basic $APPTC$ term,
  which contradicts the assumption that $p'$ is not a basic $APPTC$ term.
  \item Case $p'\equiv p_1\leftmerge p_2$. By induction on the structure of the basic $APPTC$ terms both $p_1$ and $p_2$, all subcases will lead to that $p'$ would be a basic $APPTC$
  term, which contradicts the assumption that $p'$ is not a basic $APPTC$ term.
  \item Case $p'\equiv p_1\mid p_2$. By induction on the structure of the basic $APPTC$ terms both $p_1$ and $p_2$, all subcases will lead to that $p'$ would be a basic $APPTC$ term,
  which contradicts the assumption that $p'$ is not a basic $APPTC$ term.
  \item Case $p'\equiv \Theta(p_1)$. By induction on the structure of the basic $APPTC$ term $p_1$, $RCE19-RCE24$ rewrite rules in Table \ref{TRSForAPPTC} can be applied. So $p$ is not
  a normal form.
  \item Case $p'\equiv p_1\triangleleft p_2$. By induction on the structure of the basic $APPTC$ terms both $p_1$ and $p_2$, all subcases will lead to that $p'$ would be a basic
  $APPTC$ term, which contradicts the assumption that $p'$ is not a basic $APPTC$ term.
\end{itemize}
\end{proof}

\subsubsection{Structured Operational Semantics of Parallelism}

Firstly, we give the definition of PDFs in Table \ref{PDFAPPTC}.

\begin{center}
    \begin{table}
        $$\mu(\delta,\breve{\delta})=1$$
        $$\mu(x\between y,x'\parallel y'+x'\mid y')=\mu(x,x')\cdot\mu(y,y')$$
        $$\mu(x\parallel y,x'\leftmerge y+y'\leftmerge x)=\mu(x,x')\cdot \mu(y,y')$$
        $$\mu(x\leftmerge y, x'\leftmerge y)=\mu(x,x')$$
        $$\mu(x\mid y,x'\mid y')=\mu(x,x')\cdot \mu(y,y')$$
        $$\mu(\Theta(x),\Theta(x'))=\mu(x,x')$$
        $$\mu(x\triangleleft y, x'\triangleleft y)=\mu(x,x')$$
        $$\mu(x,y)=0,\textrm{otherwise}$$
        \caption{PDF definitions of $APPTC$}
        \label{PDFAPPTC}
    \end{table}
\end{center}

We give the transition rules of APTC in Table \ref{TRForAPPTC1}, \ref{TRForAPPTC}, it is suitable for all truly concurrent behavioral equivalence, including probabilistic pomset bisimulation,
probabilistic step bisimulation, probabilistic hp-bisimulation and probabilistic hhp-bisimulation.

\begin{center}
    \begin{table}
        $$\frac{x\rightsquigarrow x'\quad y\rightsquigarrow y'}{x\between y\rightsquigarrow x'\parallel y'+x'\mid y'}$$
        $$\frac{x\rightsquigarrow x'\quad y\rightsquigarrow y'}{x\parallel y\rightsquigarrow x'\leftmerge y+y'\leftmerge x}$$
        $$\frac{x\rightsquigarrow x'}{x\leftmerge y\rightsquigarrow x'\leftmerge y}$$
        $$\frac{x\rightsquigarrow x'\quad y\rightsquigarrow y'}{x\mid y\rightsquigarrow x'\mid y'}$$
        $$\frac{x\rightsquigarrow x'}{\Theta(x)\rightsquigarrow \Theta(x')}$$
        $$\frac{x\rightsquigarrow x'}{x\triangleleft y\rightsquigarrow x'\triangleleft y}$$
        \caption{Probabilistic transition rules of APPTC}
        \label{TRForAPPTC1}
    \end{table}
\end{center}

\begin{center}
    \begin{table}
        $$\frac{x\xrightarrow{e_1}\surd\quad y\xrightarrow{e_2}\surd}{x\parallel y\xrightarrow{\{e_1,e_2\}}\surd} \quad\frac{x\xrightarrow{e_1}x'\quad y\xrightarrow{e_2}\surd}{x\parallel y\xrightarrow{\{e_1,e_2\}}x'}$$
        $$\frac{x\xrightarrow{e_1}\surd\quad y\xrightarrow{e_2}y'}{x\parallel y\xrightarrow{\{e_1,e_2\}}y'} \quad\frac{x\xrightarrow{e_1}x'\quad y\xrightarrow{e_2}y'}{x\parallel y\xrightarrow{\{e_1,e_2\}}x'\between y'}$$
        $$\frac{x\xrightarrow{e_1}\surd\quad y\xrightarrow{e_2}\surd \quad(e_1\leq e_2)}{x\leftmerge y\xrightarrow{\{e_1,e_2\}}\surd} \quad\frac{x\xrightarrow{e_1}x'\quad y\xrightarrow{e_2}\surd \quad(e_1\leq e_2)}{x\leftmerge y\xrightarrow{\{e_1,e_2\}}x'}$$
        $$\frac{x\xrightarrow{e_1}\surd\quad y\xrightarrow{e_2}y' \quad(e_1\leq e_2)}{x\leftmerge y\xrightarrow{\{e_1,e_2\}}y'} \quad\frac{x\xrightarrow{e_1}x'\quad y\xrightarrow{e_2}y' \quad(e_1\leq e_2)}{x\leftmerge y\xrightarrow{\{e_1,e_2\}}x'\between y'}$$
        $$\frac{x\xrightarrow{e_1}\surd\quad y\xrightarrow{e_2}\surd}{x\mid y\xrightarrow{\gamma(e_1,e_2)}\surd} \quad\frac{x\xrightarrow{e_1}x'\quad y\xrightarrow{e_2}\surd}{x\mid y\xrightarrow{\gamma(e_1,e_2)}x'}$$
        $$\frac{x\xrightarrow{e_1}\surd\quad y\xrightarrow{e_2}y'}{x\mid y\xrightarrow{\gamma(e_1,e_2)}y'} \quad\frac{x\xrightarrow{e_1}x'\quad y\xrightarrow{e_2}y'}{x\mid y\xrightarrow{\gamma(e_1,e_2)}x'\between y'}$$
        $$\frac{x\xrightarrow{e_1}\surd\quad (\sharp(e_1,e_2))}{\Theta(x)\xrightarrow{e_1}\surd} \quad\frac{x\xrightarrow{e_2}\surd\quad (\sharp(e_1,e_2))}{\Theta(x)\xrightarrow{e_2}\surd}$$
        $$\frac{x\xrightarrow{e_1}x'\quad (\sharp(e_1,e_2))}{\Theta(x)\xrightarrow{e_1}\Theta(x')} \quad\frac{x\xrightarrow{e_2}x'\quad (\sharp(e_1,e_2))}{\Theta(x)\xrightarrow{e_2}\Theta(x')}$$
        $$\frac{x\xrightarrow{e_1}\surd\quad (\sharp_{\pi}(e_1,e_2))}{\Theta(x)\xrightarrow{e_1}\surd} \quad\frac{x\xrightarrow{e_2}\surd\quad (\sharp_{\pi}(e_1,e_2))}{\Theta(x)\xrightarrow{e_2}\surd}$$
        $$\frac{x\xrightarrow{e_1}x'\quad (\sharp_{\pi}(e_1,e_2))}{\Theta(x)\xrightarrow{e_1}\Theta(x')} \quad\frac{x\xrightarrow{e_2}x'\quad (\sharp_{\pi}(e_1,e_2))}{\Theta(x)\xrightarrow{e_2}\Theta(x')}$$
        $$\frac{x\xrightarrow{e_1}\surd \quad y\nrightarrow^{e_2}\quad (\sharp(e_1,e_2))}{x\triangleleft y\xrightarrow{\tau}\surd}
        \quad\frac{x\xrightarrow{e_1}x' \quad y\nrightarrow^{e_2}\quad (\sharp(e_1,e_2))}{x\triangleleft y\xrightarrow{\tau}x'}$$
        $$\frac{x\xrightarrow{e_1}\surd \quad y\nrightarrow^{e_3}\quad (\sharp(e_1,e_2),e_2\leq e_3)}{x\triangleleft y\xrightarrow{e_1}\surd}
        \quad\frac{x\xrightarrow{e_1}x' \quad y\nrightarrow^{e_3}\quad (\sharp(e_1,e_2),e_2\leq e_3)}{x\triangleleft y\xrightarrow{e_1}x'}$$
        $$\frac{x\xrightarrow{e_3}\surd \quad y\nrightarrow^{e_2}\quad (\sharp(e_1,e_2),e_1\leq e_3)}{x\triangleleft y\xrightarrow{\tau}\surd}
        \quad\frac{x\xrightarrow{e_3}x' \quad y\nrightarrow^{e_2}\quad (\sharp(e_1,e_2),e_1\leq e_3)}{x\triangleleft y\xrightarrow{\tau}x'}$$
        $$\frac{x\xrightarrow{e_1}\surd \quad y\nrightarrow^{e_2}\quad (\sharp_{\pi}(e_1,e_2))}{x\triangleleft y\xrightarrow{\tau}\surd}
        \quad\frac{x\xrightarrow{e_1}x' \quad y\nrightarrow^{e_2}\quad (\sharp_{\pi}(e_1,e_2))}{x\triangleleft y\xrightarrow{\tau}x'}$$
        $$\frac{x\xrightarrow{e_1}\surd \quad y\nrightarrow^{e_3}\quad (\sharp_{\pi}(e_1,e_2),e_2\leq e_3)}{x\triangleleft y\xrightarrow{e_1}\surd}
        \quad\frac{x\xrightarrow{e_1}x' \quad y\nrightarrow^{e_3}\quad (\sharp_{\pi}(e_1,e_2),e_2\leq e_3)}{x\triangleleft y\xrightarrow{e_1}x'}$$
        $$\frac{x\xrightarrow{e_3}\surd \quad y\nrightarrow^{e_2}\quad (\sharp_{\pi}(e_1,e_2),e_1\leq e_3)}{x\triangleleft y\xrightarrow{\tau}\surd}
        \quad\frac{x\xrightarrow{e_3}x' \quad y\nrightarrow^{e_2}\quad (\sharp_{\pi}(e_1,e_2),e_1\leq e_3)}{x\triangleleft y\xrightarrow{\tau}x'}$$
%        $$\frac{x\xrightarrow{e}\surd}{\partial_H(x)\xrightarrow{e}\surd}\quad (e\notin H)\quad\quad\frac{x\xrightarrow{e}x'}{\partial_H(x)\xrightarrow{e}\partial_H(x')}\quad(e\notin H)$$
        \caption{Action transition rules of APPTC}
        \label{TRForAPPTC}
    \end{table}
\end{center}

\begin{theorem}[Generalization of the algebra for parallelism with respect to $BAPTC$]
The algebra for parallelism is a generalization of $BAPTC$.
\end{theorem}

\begin{proof}
It follows from the following three facts.

\begin{enumerate}
  \item The transition rules of $BAPTC$ in section \ref{baptc} are all source-dependent;
  \item The sources of the transition rules for the algebra for parallelism contain an occurrence of $\between$, or $\parallel$, or $\leftmerge$, or $\mid$, or $\Theta$, or $\triangleleft$;
  \item The transition rules of $APPTC$ are all source-dependent.
\end{enumerate}

So, the algebra for parallelism is a generalization of $BAPTC$, that is, $BAPTC$ is an embedding of the algebra for parallelism, as desired.
\end{proof}

\begin{theorem}[Congruence of $APPTC$ with respect to probabilistic pomset bisimulation equivalence]
Probabilistic pomset bisimulation equivalence $\sim_{pp}$ is a congruence with respect to $APPTC$.
\end{theorem}

\begin{proof}
It is easy to see that probabilistic pomset bisimulation is an equivalent relation on $APPTC$ terms, we only need to prove that $\sim_{pp}$ is preserved by the operators
$\between$, $\parallel$, $\leftmerge$, $\mid$, $\Theta$ and $\triangleleft$.
That is, if $x\sim_{pp} x'$ and $y\sim_{pp}y'$, we need to prove that $x\between y\sim_{pp}x'\between y'$, $x\parallel y\sim_{pp}x'\parallel y'$, $x\leftmerge y\sim_{pp}x'\leftmerge y'$,
$x\mid y\sim_{pp}x'\mid y'$, $x\triangleleft y\sim_{pp}x'\triangleleft y'$, and $\Theta(x)\sim_{pp}\Theta(x')$. The proof is quite trivial and we omit it.
\end{proof}

\begin{theorem}[Soundness of parallelism modulo probabilistic pomset bisimulation equivalence]\label{SPPPBE}
Let $x$ and $y$ be $APPTC$ terms. If $APPTC\vdash x=y$, then $x\sim_{pp} y$.
\end{theorem}

\begin{proof}
Since probabilistic pomset bisimulation $\sim_{pp}$ is both an equivalent and a congruent relation with respect to the operators $\between$, $\parallel$, $\leftmerge$, $\mid$, $\Theta$ and
$\triangleleft$, we only need to check if each axiom in Table \ref{AxiomsForPParallelism} is sound modulo probabilistic pomset bisimulation equivalence. The proof is quite trivial, and we omit it.
\end{proof}

\begin{theorem}[Completeness of parallelism modulo probabilistic pomset bisimulation equivalence]\label{CPPPBE}
Let $p$ and $q$ be closed $APPTC$ terms, if $p\sim_{pp} q$ then $p=q$.
\end{theorem}

\begin{proof}
Firstly, by the elimination theorem of $APPTC$ (see Theorem \ref{ETPParallelism}), we know that for each closed $APPTC$ term $p$, there exists a closed basic $APPTC$ term $p'$,
such that $APPTC\vdash p=p'$, so, we only need to consider closed basic $APPTC$ terms.

The basic terms (see Definition \ref{BTAPPTC}) modulo associativity and commutativity (AC) of conflict $+$ (defined by axioms $A1$ and $A2$ in Table \ref{AxiomsForBAPTC} and these
equivalences is denoted by $=_{AC}$. Then, each equivalence class $s$ modulo AC of $+$ has the following normal form

$$s_1\boxplus_{\pi_1}\cdots\boxplus_{\pi_{k-1}} s_k$$

with each $s_i$ has the following form

$$t_1+\cdots+ t_l$$

with each $t_j$ either an atomic event or of the form

$$u_1\cdot\cdots\cdot u_m$$

with each $u_l$ either an atomic event or of the form

$$v_1\leftmerge\cdots\leftmerge v_m$$

with each $v_m$ an atomic event, and each $t_j$ is called the summand of $s$.

Now, we prove that for normal forms $n$ and $n'$, if $n\sim_{pp} n'$ then $n=_{AC}n'$. It is sufficient to induct on the sizes of $n$ and $n'$.

\begin{itemize}
  \item Consider a summand $e$ of $n$. Then $n\rightsquigarrow\breve{e}\xrightarrow{e}\surd$, so $n\sim_{pp} n'$ implies $n'\rightsquigarrow\breve{e}\xrightarrow{e}\surd$, meaning that
  $n'$ also contains the summand $e$.
  \item Consider a summand $t_1\cdot t_2$ of $n$,
  \begin{itemize}
    \item if $t_1\equiv e'$, then $n\rightsquigarrow\breve{e'}\xrightarrow{e'}t_2$, so $n\sim_{pp} n'$ implies $n'\rightsquigarrow\breve{e'}\xrightarrow{e'}t_2'$ with $t_2\sim_{pp} t_2'$,
    meaning that $n'$ contains a summand $e'\cdot t_2'$. Since $t_2$ and $t_2'$ are normal forms and have sizes smaller than $n$ and $n'$, by the induction hypotheses if
    $t_2\sim_{pp} t_2'$ then $t_2=_{AC} t_2'$;
    \item if $t_1\equiv e_1\leftmerge\cdots\leftmerge e_m$, then $n\rightsquigarrow\breve{e_1}\leftmerge\cdots\leftmerge\breve{e_m}\xrightarrow{\{e_1,\cdots,e_m\}}t_2$,
    so $n\sim_{pp} n'$ implies $n'\rightsquigarrow\breve{e_1}\leftmerge\cdots\leftmerge\breve{e_m}\xrightarrow{\{e_1,\cdots,e_m\}}t_2'$ with $t_2\sim_{pp} t_2'$, meaning that $n'$
    contains a summand $(e_1\leftmerge\cdots\leftmerge e_m)\cdot t_2'$. Since $t_2$ and $t_2'$ are normal forms and have sizes smaller than $n$ and $n'$, by the induction hypotheses
    if $t_2\sim_{pp} t_2'$ then $t_2=_{AC} t_2'$.
  \end{itemize}
\end{itemize}

So, we get $n=_{AC} n'$.

Finally, let $s$ and $t$ be basic $APPTC$ terms, and $s\sim_{pp} t$, there are normal forms $n$ and $n'$, such that $s=n$ and $t=n'$. The soundness theorem of parallelism modulo
probabilistic pomset bisimulation equivalence (see Theorem \ref{SPPPBE}) yields $s\sim_{pp} n$ and $t\sim_{pp} n'$, so $n\sim_{pp} s\sim_{pp} t\sim_{pp} n'$. Since if $n\sim_{pp} n'$
then $n=_{AC}n'$, $s=n=_{AC}n'=t$, as desired.
\end{proof}

\begin{theorem}[Congruence of $APPTC$ with respect to probabilistic step bisimulation equivalence]
Probabilistic step bisimulation equivalence $\sim_{ps}$ is a congruence with respect to $APPTC$.
\end{theorem}

\begin{proof}
It is easy to see that probabilistic step bisimulation is an equivalent relation on $APPTC$ terms, we only need to prove that $\sim_{ps}$ is preserved by the operators
$\between$, $\parallel$, $\leftmerge$, $\mid$, $\Theta$ and $\triangleleft$.
That is, if $x\sim_{ps} x'$ and $y\sim_{ps}y'$, we need to prove that $x\between y\sim_{ps}x'\between y'$, $x\parallel y\sim_{ps}x'\parallel y'$, $x\leftmerge y\sim_{ps}x'\leftmerge y'$,
$x\mid y\sim_{ps}x'\mid y'$, $x\triangleleft y\sim_{ps}x'\triangleleft y'$, and $\Theta(x)\sim_{ps}\Theta(x')$. The proof is quite trivial and we omit it.
\end{proof}

\begin{theorem}[Soundness of parallelism modulo probabilistic step bisimulation equivalence]\label{SPPSBE}
Let $x$ and $y$ be $APPTC$ terms. If $APPTC\vdash x=y$, then $x\sim_{ps} y$.
\end{theorem}

\begin{proof}
Since probabilistic step bisimulation $\sim_{ps}$ is both an equivalent and a congruent relation with respect to the operators $\between$, $\parallel$, $\leftmerge$, $\mid$, $\Theta$ and
$\triangleleft$, we only need to check if each axiom in Table \ref{AxiomsForPParallelism} is sound modulo probabilistic step bisimulation equivalence. The proof is quite trivial, and we omit it.
\end{proof}

\begin{theorem}[Completeness of parallelism modulo probabilistic step bisimulation equivalence]\label{CPPSBE}
Let $p$ and $q$ be closed $APPTC$ terms, if $p\sim_{ps} q$ then $p=q$.
\end{theorem}

\begin{proof}
Firstly, by the elimination theorem of $APPTC$ (see Theorem \ref{ETPParallelism}), we know that for each closed $APPTC$ term $p$, there exists a closed basic $APPTC$ term $p'$,
such that $APPTC\vdash p=p'$, so, we only need to consider closed basic $APPTC$ terms.

The basic terms (see Definition \ref{BTAPPTC}) modulo associativity and commutativity (AC) of conflict $+$ (defined by axioms $A1$ and $A2$ in Table \ref{AxiomsForBAPTC} and these
equivalences is denoted by $=_{AC}$. Then, each equivalence class $s$ modulo AC of $+$ has the following normal form

$$s_1\boxplus_{\pi_1}\cdots\boxplus_{\pi_{k-1}} s_k$$

with each $s_i$ has the following form

$$t_1+\cdots+ t_l$$

with each $t_j$ either an atomic event or of the form

$$u_1\cdot\cdots\cdot u_m$$

with each $u_l$ either an atomic event or of the form

$$v_1\leftmerge\cdots\leftmerge v_m$$

with each $v_m$ an atomic event, and each $t_j$ is called the summand of $s$.

Now, we prove that for normal forms $n$ and $n'$, if $n\sim_{ps} n'$ then $n=_{AC}n'$. It is sufficient to induct on the sizes of $n$ and $n'$.

\begin{itemize}
  \item Consider a summand $e$ of $n$. Then $n\rightsquigarrow\breve{e}\xrightarrow{e}\surd$, so $n\sim_{ps} n'$ implies $n'\rightsquigarrow\breve{e}\xrightarrow{e}\surd$, meaning that
  $n'$ also contains the summand $e$.
  \item Consider a summand $t_1\cdot t_2$ of $n$,
  \begin{itemize}
    \item if $t_1\equiv e'$, then $n\rightsquigarrow\breve{e'}\xrightarrow{e'}t_2$, so $n\sim_{ps} n'$ implies $n'\rightsquigarrow\breve{e'}\xrightarrow{e'}t_2'$ with $t_2\sim_{ps} t_2'$,
    meaning that $n'$ contains a summand $e'\cdot t_2'$. Since $t_2$ and $t_2'$ are normal forms and have sizes smaller than $n$ and $n'$, by the induction hypotheses if
    $t_2\sim_{ps} t_2'$ then $t_2=_{AC} t_2'$;
    \item if $t_1\equiv e_1\leftmerge\cdots\leftmerge e_m$, then $n\rightsquigarrow\breve{e_1}\leftmerge\cdots\leftmerge\breve{e_m}\xrightarrow{\{e_1,\cdots,e_m\}}t_2$,
    so $n\sim_{ps} n'$ implies $n'\rightsquigarrow\breve{e_1}\leftmerge\cdots\leftmerge\breve{e_m}\xrightarrow{\{e_1,\cdots,e_m\}}t_2'$ with $t_2\sim_{ps} t_2'$, meaning that $n'$
    contains a summand $(e_1\leftmerge\cdots\leftmerge e_m)\cdot t_2'$. Since $t_2$ and $t_2'$ are normal forms and have sizes smaller than $n$ and $n'$, by the induction hypotheses
    if $t_2\sim_{ps} t_2'$ then $t_2=_{AC} t_2'$.
  \end{itemize}
\end{itemize}

So, we get $n=_{AC} n'$.

Finally, let $s$ and $t$ be basic $APPTC$ terms, and $s\sim_{ps} t$, there are normal forms $n$ and $n'$, such that $s=n$ and $t=n'$. The soundness theorem of parallelism modulo
probabilistic step bisimulation equivalence (see Theorem \ref{SPPSBE}) yields $s\sim_{ps} n$ and $t\sim_{ps} n'$, so $n\sim_{ps} s\sim_{ps} t\sim_{ps} n'$. Since if $n\sim_{ps} n'$
then $n=_{AC}n'$, $s=n=_{AC}n'=t$, as desired.
\end{proof}

\begin{theorem}[Congruence of $APPTC$ with respect to probabilistic hp-bisimulation equivalence]
Probabilistic hp-bisimulation equivalence $\sim_{php}$ is a congruence with respect to $APPTC$.
\end{theorem}

\begin{proof}
It is easy to see that probabilistic hp-bisimulation is an equivalent relation on $APPTC$ terms, we only need to prove that $\sim_{php}$ is preserved by the operators
$\between$, $\parallel$, $\leftmerge$, $\mid$, $\Theta$ and $\triangleleft$.
That is, if $x\sim_{php} x'$ and $y\sim_{php}y'$, we need to prove that $x\between y\sim_{php}x'\between y'$, $x\parallel y\sim_{php}x'\parallel y'$, $x\leftmerge y\sim_{php}x'\leftmerge y'$,
$x\mid y\sim_{php}x'\mid y'$, $x\triangleleft y\sim_{php}x'\triangleleft y'$, and $\Theta(x)\sim_{php}\Theta(x')$. The proof is quite trivial and we omit it.
\end{proof}

\begin{theorem}[Soundness of parallelism modulo probabilistic hp-bisimulation equivalence]\label{SPPHPBE}
Let $x$ and $y$ be $APPTC$ terms. If $APPTC\vdash x=y$, then $x\sim_{php} y$.
\end{theorem}

\begin{proof}
Since probabilistic hp-bisimulation $\sim_{php}$ is both an equivalent and a congruent relation with respect to the operators $\between$, $\parallel$, $\leftmerge$, $\mid$, $\Theta$ and
$\triangleleft$, we only need to check if each axiom in Table \ref{AxiomsForPParallelism} is sound modulo probabilistic hp-bisimulation equivalence. The proof is quite trivial, and we omit it.
\end{proof}

\begin{theorem}[Completeness of parallelism modulo probabilistic hp-bisimulation equivalence]\label{CPPHPBE}
Let $p$ and $q$ be closed $APPTC$ terms, if $p\sim_{php} q$ then $p=q$.
\end{theorem}

\begin{proof}
Firstly, by the elimination theorem of $APPTC$ (see Theorem \ref{ETPParallelism}), we know that for each closed $APPTC$ term $p$, there exists a closed basic $APPTC$ term $p'$,
such that $APPTC\vdash p=p'$, so, we only need to consider closed basic $APPTC$ terms.

The basic terms (see Definition \ref{BTAPPTC}) modulo associativity and commutativity (AC) of conflict $+$ (defined by axioms $A1$ and $A2$ in Table \ref{AxiomsForBAPTC} and these
equivalences is denoted by $=_{AC}$. Then, each equivalence class $s$ modulo AC of $+$ has the following normal form

$$s_1\boxplus_{\pi_1}\cdots\boxplus_{\pi_{k-1}} s_k$$

with each $s_i$ has the following form

$$t_1+\cdots+ t_l$$

with each $t_j$ either an atomic event or of the form

$$u_1\cdot\cdots\cdot u_m$$

with each $u_l$ either an atomic event or of the form

$$v_1\leftmerge\cdots\leftmerge v_m$$

with each $v_m$ an atomic event, and each $t_j$ is called the summand of $s$.

Now, we prove that for normal forms $n$ and $n'$, if $n\sim_{php} n'$ then $n=_{AC}n'$. It is sufficient to induct on the sizes of $n$ and $n'$.

\begin{itemize}
  \item Consider a summand $e$ of $n$. Then $n\rightsquigarrow\breve{e}\xrightarrow{e}\surd$, so $n\sim_{php} n'$ implies $n'\rightsquigarrow\breve{e}\xrightarrow{e}\surd$, meaning that
  $n'$ also contains the summand $e$.
  \item Consider a summand $t_1\cdot t_2$ of $n$,
  \begin{itemize}
    \item if $t_1\equiv e'$, then $n\rightsquigarrow\breve{e'}\xrightarrow{e'}t_2$, so $n\sim_{php} n'$ implies $n'\rightsquigarrow\breve{e'}\xrightarrow{e'}t_2'$ with $t_2\sim_{php} t_2'$,
    meaning that $n'$ contains a summand $e'\cdot t_2'$. Since $t_2$ and $t_2'$ are normal forms and have sizes smaller than $n$ and $n'$, by the induction hypotheses if
    $t_2\sim_{php} t_2'$ then $t_2=_{AC} t_2'$;
    \item if $t_1\equiv e_1\leftmerge\cdots\leftmerge e_m$, then $n\rightsquigarrow\breve{e_1}\leftmerge\cdots\leftmerge\breve{e_m}\xrightarrow{\{e_1,\cdots,e_m\}}t_2$,
    so $n\sim_{php} n'$ implies $n'\rightsquigarrow\breve{e_1}\leftmerge\cdots\leftmerge\breve{e_m}\xrightarrow{\{e_1,\cdots,e_m\}}t_2'$ with $t_2\sim_{php} t_2'$, meaning that $n'$
    contains a summand $(e_1\leftmerge\cdots\leftmerge e_m)\cdot t_2'$. Since $t_2$ and $t_2'$ are normal forms and have sizes smaller than $n$ and $n'$, by the induction hypotheses
    if $t_2\sim_{php} t_2'$ then $t_2=_{AC} t_2'$.
  \end{itemize}
\end{itemize}

So, we get $n=_{AC} n'$.

Finally, let $s$ and $t$ be basic $APPTC$ terms, and $s\sim_{php} t$, there are normal forms $n$ and $n'$, such that $s=n$ and $t=n'$. The soundness theorem of parallelism modulo
probabilistic hp-bisimulation equivalence (see Theorem \ref{SPPHPBE}) yields $s\sim_{php} n$ and $t\sim_{php} n'$, so $n\sim_{php} s\sim_{php} t\sim_{php} n'$. Since if $n\sim_{php} n'$
then $n=_{AC}n'$, $s=n=_{AC}n'=t$, as desired.
\end{proof}

\begin{theorem}[Congruence of $APPTC$ with respect to probabilistic hhp-bisimulation equivalence]
Probabilistic hhp-bisimulation equivalence $\sim_{phhp}$ is a congruence with respect to $APPTC$.
\end{theorem}

\begin{proof}
It is easy to see that probabilistic hhp-bisimulation is an equivalent relation on $APPTC$ terms, we only need to prove that $\sim_{phhp}$ is preserved by the operators
$\between$, $\parallel$, $\leftmerge$, $\mid$, $\Theta$ and $\triangleleft$.
That is, if $x\sim_{phhp} x'$ and $y\sim_{phhp}y'$, we need to prove that $x\between y\sim_{phhp}x'\between y'$, $x\parallel y\sim_{phhp}x'\parallel y'$, $x\leftmerge y\sim_{phhp}x'\leftmerge y'$,
$x\mid y\sim_{phhp}x'\mid y'$, $x\triangleleft y\sim_{phhp}x'\triangleleft y'$, and $\Theta(x)\sim_{phhp}\Theta(x')$. The proof is quite trivial and we omit it.
\end{proof}

\begin{theorem}[Soundness of parallelism modulo probabilistic hhp-bisimulation equivalence]\label{SPPHHPBE}
Let $x$ and $y$ be $APPTC$ terms. If $APPTC\vdash x=y$, then $x\sim_{phhp} y$.
\end{theorem}

\begin{proof}
Since probabilistic hhp-bisimulation $\sim_{phhp}$ is both an equivalent and a congruent relation with respect to the operators $\between$, $\parallel$, $\leftmerge$, $\mid$, $\Theta$ and
$\triangleleft$, we only need to check if each axiom in Table \ref{AxiomsForPParallelism} is sound modulo probabilistic hhp-bisimulation equivalence. The proof is quite trivial, and we omit it.
\end{proof}

\begin{theorem}[Completeness of parallelism modulo probabilistic hhp-bisimulation equivalence]\label{CPPHHPBE}
Let $p$ and $q$ be closed $APPTC$ terms, if $p\sim_{phhp} q$ then $p=q$.
\end{theorem}

\begin{proof}
Firstly, by the elimination theorem of $APPTC$ (see Theorem \ref{ETPParallelism}), we know that for each closed $APPTC$ term $p$, there exists a closed basic $APPTC$ term $p'$,
such that $APPTC\vdash p=p'$, so, we only need to consider closed basic $APPTC$ terms.

The basic terms (see Definition \ref{BTAPPTC}) modulo associativity and commutativity (AC) of conflict $+$ (defined by axioms $A1$ and $A2$ in Table \ref{AxiomsForBAPTC} and these
equivalences is denoted by $=_{AC}$. Then, each equivalence class $s$ modulo AC of $+$ has the following normal form

$$s_1\boxplus_{\pi_1}\cdots\boxplus_{\pi_{k-1}} s_k$$

with each $s_i$ has the following form

$$t_1+\cdots+ t_l$$

with each $t_j$ either an atomic event or of the form

$$u_1\cdot\cdots\cdot u_m$$

with each $u_l$ either an atomic event or of the form

$$v_1\leftmerge\cdots\leftmerge v_m$$

with each $v_m$ an atomic event, and each $t_j$ is called the summand of $s$.

Now, we prove that for normal forms $n$ and $n'$, if $n\sim_{phhp} n'$ then $n=_{AC}n'$. It is sufficient to induct on the sizes of $n$ and $n'$.

\begin{itemize}
  \item Consider a summand $e$ of $n$. Then $n\rightsquigarrow\breve{e}\xrightarrow{e}\surd$, so $n\sim_{phhp} n'$ implies $n'\rightsquigarrow\breve{e}\xrightarrow{e}\surd$, meaning that
  $n'$ also contains the summand $e$.
  \item Consider a summand $t_1\cdot t_2$ of $n$,
  \begin{itemize}
    \item if $t_1\equiv e'$, then $n\rightsquigarrow\breve{e'}\xrightarrow{e'}t_2$, so $n\sim_{phhp} n'$ implies $n'\rightsquigarrow\breve{e'}\xrightarrow{e'}t_2'$ with $t_2\sim_{phhp} t_2'$,
    meaning that $n'$ contains a summand $e'\cdot t_2'$. Since $t_2$ and $t_2'$ are normal forms and have sizes smaller than $n$ and $n'$, by the induction hypotheses if
    $t_2\sim_{phhp} t_2'$ then $t_2=_{AC} t_2'$;
    \item if $t_1\equiv e_1\leftmerge\cdots\leftmerge e_m$, then $n\rightsquigarrow\breve{e_1}\leftmerge\cdots\leftmerge\breve{e_m}\xrightarrow{\{e_1,\cdots,e_m\}}t_2$,
    so $n\sim_{phhp} n'$ implies $n'\rightsquigarrow\breve{e_1}\leftmerge\cdots\leftmerge\breve{e_m}\xrightarrow{\{e_1,\cdots,e_m\}}t_2'$ with $t_2\sim_{phhp} t_2'$, meaning that $n'$
    contains a summand $(e_1\leftmerge\cdots\leftmerge e_m)\cdot t_2'$. Since $t_2$ and $t_2'$ are normal forms and have sizes smaller than $n$ and $n'$, by the induction hypotheses
    if $t_2\sim_{phhp} t_2'$ then $t_2=_{AC} t_2'$.
  \end{itemize}
\end{itemize}

So, we get $n=_{AC} n'$.

Finally, let $s$ and $t$ be basic $APPTC$ terms, and $s\sim_{phhp} t$, there are normal forms $n$ and $n'$, such that $s=n$ and $t=n'$. The soundness theorem of parallelism modulo
probabilistic hhp-bisimulation equivalence (see Theorem \ref{SPPHHPBE}) yields $s\sim_{phhp} n$ and $t\sim_{phhp} n'$, so $n\sim_{phhp} s\sim_{phhp} t\sim_{phhp} n'$. Since if $n\sim_{phhp} n'$
then $n=_{AC}n'$, $s=n=_{AC}n'=t$, as desired.
\end{proof}

\subsubsection{Encapsulation}

The mismatch of two communicating events in different parallel branches can cause deadlock, so the deadlocks in the concurrent processes should be eliminated. Like $ACP$ \cite{ACP}, we
also introduce the unary encapsulation operator $\partial_H$ for set $H$ of atomic events, which renames all atomic events in $H$ into $\delta$. The whole algebra including parallelism
for true concurrency in the above subsections, deadlock $\delta$ and encapsulation operator $\partial_H$, is also called Algebra for Parallelism in Probabilistic True Concurrency, abbreviated $APPTC$.

Firstly, we give the definition of PDFs in Table \ref{PDFEncap}.

\begin{center}
    \begin{table}
        $$\mu(\partial_H(x),\partial_H(x'))=\mu(x,x')$$
        $$\mu(x,y)=0,\textrm{otherwise}$$
        \caption{PDF definitions of $\partial_H$}
        \label{PDFEncap}
    \end{table}
\end{center}

The transition rules of encapsulation operator $\partial_H$ are shown in Table \ref{TRForPEncapsulation}.

\begin{center}
    \begin{table}
        $$\frac{x\rightsquigarrow x'}{\partial_H(x)\rightsquigarrow\partial_H(x')}$$
        $$\frac{x\xrightarrow{e}\surd}{\partial_H(x)\xrightarrow{e}\surd}\quad (e\notin H)\quad\quad\frac{x\xrightarrow{e}x'}{\partial_H(x)\xrightarrow{e}\partial_H(x')}\quad(e\notin H)$$
        \caption{Transition rules of encapsulation operator $\partial_H$}
        \label{TRForPEncapsulation}
    \end{table}
\end{center}

Based on the transition rules for encapsulation operator $\partial_H$ in Table \ref{TRForPEncapsulation}, we design the axioms as Table \ref{AxiomsForPEncapsulation} shows.

\begin{center}
    \begin{table}
        \begin{tabular}{@{}ll@{}}
            \hline No. &Axiom\\
            $D1$ & $e\notin H\quad\partial_H(e) = e$\\
            $D2$ & $e\in H\quad \partial_H(e) = \delta$\\
            $D3$ & $\partial_H(\delta) = \delta$\\
            $D4$ & $\partial_H(x+ y) = \partial_H(x)+\partial_H(y)$\\
            $D5$ & $\partial_H(x\cdot y) = \partial_H(x)\cdot\partial_H(y)$\\
            $D6$ & $\partial_H(x\leftmerge y) = \partial_H(x)\leftmerge\partial_H(y)$\\
            $PD1$ & $\partial_H(x\boxplus_{\pi}y)=\partial_H(x)\boxplus_{\pi}\partial_H(y)$\\
        \end{tabular}
        \caption{Axioms of encapsulation operator}
        \label{AxiomsForPEncapsulation}
    \end{table}
\end{center}

The axioms $D1-D3$ are the defining laws for the encapsulation operator $\partial_H$, $D1$ leaves atomic events outside $H$ unchanged, $D2$ renames atomic events in $H$ into $\delta$,
and $D3$ says that it leaves $\delta$ unchanged. $D4-D6$ and $PD1$ say that in term $\partial_H(t)$, all transitions of $t$ labeled with atomic events in $H$ are blocked.

\begin{theorem}[Conservativity of $APPTC$ with respect to the algebra for parallelism]
$APPTC$ is a conservative extension of the algebra for parallelism.
\end{theorem}

\begin{proof}
It follows from the following two facts (see Theorem \ref{TCE}).

\begin{enumerate}
  \item The transition rules of the algebra for parallelism in the above subsections are all source-dependent;
  \item The sources of the transition rules for the encapsulation operator contain an occurrence of $\partial_H$.
\end{enumerate}

So, $APPTC$ is a conservative extension of the algebra for parallelism, as desired.
\end{proof}

\begin{theorem}[Elimination theorem of $APPTC$]\label{ETPEncapsulation}
Let $p$ be a closed $APPTC$ term including the encapsulation operator $\partial_H$. Then there is a basic $APPTC$ term $q$ such that $APPTC\vdash p=q$.
\end{theorem}

\begin{proof}
(1) Firstly, suppose that the following ordering on the signature of $APPTC$ is defined: $\leftmerge > \cdot > +>\boxplus_{\pi}$ and the symbol $\leftmerge$ is given the
lexicographical status for the first argument, then for each rewrite rule $p\rightarrow q$ in Table \ref{TRSForPEncapsulation} relation $p>_{lpo} q$ can easily be proved. We obtain
that the term rewrite system shown in Table \ref{TRSForPEncapsulation} is strongly normalizing, for it has finitely many rewriting rules, and $>$ is a well-founded ordering on the
signature of $APPTC$, and if $s>_{lpo} t$, for each rewriting rule $s\rightarrow t$ is in Table \ref{TRSForPEncapsulation} (see Theorem \ref{SN}).

\begin{center}
    \begin{table}
        \begin{tabular}{@{}ll@{}}
            \hline No. &Rewriting Rule\\
            $RD1$ & $e\notin H\quad\partial_H(e) \rightarrow e$\\
            $RD2$ & $e\in H\quad \partial_H(e) \rightarrow \delta$\\
            $RD3$ & $\partial_H(\delta) \rightarrow \delta$\\
            $RD4$ & $\partial_H(x+ y) \rightarrow \partial_H(x)+\partial_H(y)$\\
            $RD5$ & $\partial_H(x\cdot y) \rightarrow \partial_H(x)\cdot\partial_H(y)$\\
            $RD6$ & $\partial_H(x\leftmerge y) \rightarrow \partial_H(x)\leftmerge\partial_H(y)$\\
            $RPD1$ & $\partial_H(x\boxplus_{\pi}y)\rightarrow\partial_H(x)\boxplus_{\pi}\partial_H(y)$\\
        \end{tabular}
        \caption{Term rewrite system of encapsulation operator $\partial_H$}
        \label{TRSForPEncapsulation}
    \end{table}
\end{center}

(2) Then we prove that the normal forms of closed $APPTC$ terms including encapsulation operator $\partial_H$ are basic $APPTC$ terms.

Suppose that $p$ is a normal form of some closed $APPTC$ term and suppose that $p$ is not a basic $APPTC$ term. Let $p'$ denote the smallest sub-term of $p$ which is not a
basic $APPTC$ term. It implies that each sub-term of $p'$ is a basic $APPTC$ term. Then we prove that $p$ is not a term in normal form. It is sufficient to induct on the structure of
$p'$, following from Theorem \ref{ETPParallelism}, we only prove the new case $p'\equiv \partial_H(p_1)$:

\begin{itemize}
  \item Case $p_1\equiv e$. The transition rules $RD1$ or $RD2$ can be applied, so $p$ is not a normal form;
  \item Case $p_1\equiv \delta$. The transition rules $RD3$ can be applied, so $p$ is not a normal form;
  \item Case $p_1\equiv p_1'+ p_1''$. The transition rules $RD4$ can be applied, so $p$ is not a normal form;
  \item Case $p_1\equiv p_1'\cdot p_1''$. The transition rules $RD5$ can be applied, so $p$ is not a normal form;
  \item Case $p_1\equiv p_1'\leftmerge p_1''$. The transition rules $RD6$ can be applied, so $p$ is not a normal form;
  \item Case $p_1\equiv p_1'\boxplus_{\pi} p_1''$. The transition rules $RPD1$ can be applied, so $p$ is not a normal form.
\end{itemize}
\end{proof}

\begin{theorem}[Congruence theorem of encapsulation operator $\partial_H$ with respect to probabilistic pomset bisimulation equivalence]
Probabilistic pomset bisimulation equivalence $\sim_{pp}$ is a congruence with respect to encapsulation operator $\partial_H$.
\end{theorem}

\begin{proof}
It is easy to see that probabilistic pomset bisimulation is an equivalent relation on $APPTC$ terms, we only need to prove that $\sim_{pp}$ is preserved by the operators $\partial_H$.
That is, if $x\sim_{pp} x'$, we need to prove that $\partial_H(x)\sim_{pp}\partial_H(x')$. The proof is quite trivial and we omit it.
\end{proof}

\begin{theorem}[Soundness of $APPTC$ modulo probabilistic pomset bisimulation equivalence]\label{SAPPTCPPBE}
Let $x$ and $y$ be $APPTC$ terms including encapsulation operator $\partial_H$. If $APPTC\vdash x=y$, then $x\sim_{pp} y$.
\end{theorem}

\begin{proof}
Since probabilistic pomset bisimulation $\sim_{pp}$ is both an equivalent and a congruent relation with respect to the operator $\partial_H$, we only need to check if each axiom in
Table \ref{AxiomsForPEncapsulation} is sound modulo probabilistic pomset bisimulation equivalence. The proof is quite trivial and we omit it.
\end{proof}

\begin{theorem}[Completeness of $APPTC$ modulo probabilistic pomset bisimulation equivalence]\label{CAPPTCPPBE}
Let $p$ and $q$ be closed $APPTC$ terms including encapsulation operator $\partial_H$, if $p\sim_{pp} q$ then $p=q$.
\end{theorem}

\begin{proof}
Firstly, by the elimination theorem of $APPTC$ (see Theorem \ref{ETPEncapsulation}), we know that the normal form of $APPTC$ does not contain $\partial_H$, and for each closed $APPTC$
term $p$, there exists a closed basic $APPTC$ term $p'$, such that $APPTC\vdash p=p'$, so, we only need to consider closed basic $APPTC$ terms.

Similarly to Theorem \ref{CPPPBE}, we can prove that for normal forms $n$ and $n'$, if $n\sim_{pp} n'$ then $n=_{AC}n'$.

Finally, let $s$ and $t$ be basic $APPTC$ terms, and $s\sim_{pp} t$, there are normal forms $n$ and $n'$, such that $s=n$ and $t=n'$. The soundness theorem of $APPTC$ modulo probabilistic
pomset bisimulation equivalence (see Theorem \ref{SAPPTCPPBE}) yields $s\sim_{pp} n$ and $t\sim_{pp} n'$, so $n\sim_{pp} s\sim_{pp} t\sim_{pp} n'$. Since if $n\sim_{pp} n'$ then
$n=_{AC}n'$, $s=n=_{AC}n'=t$, as desired.
\end{proof}

\begin{theorem}[Congruence theorem of encapsulation operator $\partial_H$ with respect to probabilistic step bisimulation equivalence]
Probabilistic step bisimulation equivalence $\sim_{ps}$ is a congruence with respect to encapsulation operator $\partial_H$.
\end{theorem}

\begin{proof}
It is easy to see that probabilistic step bisimulation is an equivalent relation on $APPTC$ terms, we only need to prove that $\sim_{ps}$ is preserved by the operators $\partial_H$.
That is, if $x\sim_{ps} x'$, we need to prove that $\partial_H(x)\sim_{ps}\partial_H(x')$. The proof is quite trivial and we omit it.
\end{proof}

\begin{theorem}[Soundness of $APPTC$ modulo probabilistic step bisimulation equivalence]\label{SAPPTCPSBE}
Let $x$ and $y$ be $APPTC$ terms including encapsulation operator $\partial_H$. If $APPTC\vdash x=y$, then $x\sim_{ps} y$.
\end{theorem}

\begin{proof}
Since probabilistic step bisimulation $\sim_{ps}$ is both an equivalent and a congruent relation with respect to the operator $\partial_H$, we only need to check if each axiom in
Table \ref{AxiomsForPEncapsulation} is sound modulo probabilistic step bisimulation equivalence. The proof is quite trivial and we omit it.
\end{proof}

\begin{theorem}[Completeness of $APPTC$ modulo probabilistic step bisimulation equivalence]\label{CAPPTCPSBE}
Let $p$ and $q$ be closed $APPTC$ terms including encapsulation operator $\partial_H$, if $p\sim_{ps} q$ then $p=q$.
\end{theorem}

\begin{proof}
Firstly, by the elimination theorem of $APPTC$ (see Theorem \ref{ETPEncapsulation}), we know that the normal form of $APPTC$ does not contain $\partial_H$, and for each closed $APPTC$
term $p$, there exists a closed basic $APPTC$ term $p'$, such that $APPTC\vdash p=p'$, so, we only need to consider closed basic $APPTC$ terms.

Similarly to Theorem \ref{CPPSBE}, we can prove that for normal forms $n$ and $n'$, if $n\sim_{ps} n'$ then $n=_{AC}n'$.

Finally, let $s$ and $t$ be basic $APPTC$ terms, and $s\sim_{ps} t$, there are normal forms $n$ and $n'$, such that $s=n$ and $t=n'$. The soundness theorem of $APPTC$ modulo probabilistic
step bisimulation equivalence (see Theorem \ref{SAPPTCPSBE}) yields $s\sim_{ps} n$ and $t\sim_{ps} n'$, so $n\sim_{ps} s\sim_{ps} t\sim_{ps} n'$. Since if $n\sim_{ps} n'$ then
$n=_{AC}n'$, $s=n=_{AC}n'=t$, as desired.
\end{proof}

\begin{theorem}[Congruence theorem of encapsulation operator $\partial_H$ with respect to probabilistic hp-bisimulation equivalence]
Probabilistic hp-bisimulation equivalence $\sim_{php}$ is a congruence with respect to encapsulation operator $\partial_H$.
\end{theorem}

\begin{proof}
It is easy to see that probabilistic hp-bisimulation is an equivalent relation on $APPTC$ terms, we only need to prove that $\sim_{php}$ is preserved by the operators $\partial_H$.
That is, if $x\sim_{php} x'$, we need to prove that $\partial_H(x)\sim_{php}\partial_H(x')$. The proof is quite trivial and we omit it.
\end{proof}

\begin{theorem}[Soundness of $APPTC$ modulo probabilistic hp-bisimulation equivalence]\label{SAPPTCPHPBE}
Let $x$ and $y$ be $APPTC$ terms including encapsulation operator $\partial_H$. If $APPTC\vdash x=y$, then $x\sim_{php} y$.
\end{theorem}

\begin{proof}
Since probabilistic hp-bisimulation $\sim_{php}$ is both an equivalent and a congruent relation with respect to the operator $\partial_H$, we only need to check if each axiom in
Table \ref{AxiomsForPEncapsulation} is sound modulo probabilistic hp-bisimulation equivalence. The proof is quite trivial and we omit it.
\end{proof}

\begin{theorem}[Completeness of $APPTC$ modulo probabilistic hp-bisimulation equivalence]\label{CAPPTCPHPBE}
Let $p$ and $q$ be closed $APPTC$ terms including encapsulation operator $\partial_H$, if $p\sim_{php} q$ then $p=q$.
\end{theorem}

\begin{proof}
Firstly, by the elimination theorem of $APPTC$ (see Theorem \ref{ETPEncapsulation}), we know that the normal form of $APPTC$ does not contain $\partial_H$, and for each closed $APPTC$
term $p$, there exists a closed basic $APPTC$ term $p'$, such that $APPTC\vdash p=p'$, so, we only need to consider closed basic $APPTC$ terms.

Similarly to Theorem \ref{CPPHPBE}, we can prove that for normal forms $n$ and $n'$, if $n\sim_{php} n'$ then $n=_{AC}n'$.

Finally, let $s$ and $t$ be basic $APPTC$ terms, and $s\sim_{php} t$, there are normal forms $n$ and $n'$, such that $s=n$ and $t=n'$. The soundness theorem of $APPTC$ modulo probabilistic
hp-bisimulation equivalence (see Theorem \ref{SAPPTCPHPBE}) yields $s\sim_{php} n$ and $t\sim_{php} n'$, so $n\sim_{php} s\sim_{php} t\sim_{php} n'$. Since if $n\sim_{php} n'$ then
$n=_{AC}n'$, $s=n=_{AC}n'=t$, as desired.
\end{proof}

\begin{theorem}[Congruence theorem of encapsulation operator $\partial_H$ with respect to probabilistic hhp-bisimulation equivalence]
Probabilistic hhp-bisimulation equivalence $\sim_{phhp}$ is a congruence with respect to encapsulation operator $\partial_H$.
\end{theorem}

\begin{proof}
It is easy to see that probabilistic hhp-bisimulation is an equivalent relation on $APPTC$ terms, we only need to prove that $\sim_{phhp}$ is preserved by the operators $\partial_H$.
That is, if $x\sim_{phhp} x'$, we need to prove that $\partial_H(x)\sim_{phhp}\partial_H(x')$. The proof is quite trivial and we omit it.
\end{proof}

\begin{theorem}[Soundness of $APPTC$ modulo probabilistic hhp-bisimulation equivalence]\label{SAPPTCPHHPBE}
Let $x$ and $y$ be $APPTC$ terms including encapsulation operator $\partial_H$. If $APPTC\vdash x=y$, then $x\sim_{phhp} y$.
\end{theorem}

\begin{proof}
Since probabilistic hhp-bisimulation $\sim_{phhp}$ is both an equivalent and a congruent relation with respect to the operator $\partial_H$, we only need to check if each axiom in
Table \ref{AxiomsForPEncapsulation} is sound modulo probabilistic hhp-bisimulation equivalence. The proof is quite trivial and we omit it.
\end{proof}

\begin{theorem}[Completeness of $APPTC$ modulo probabilistic hhp-bisimulation equivalence]\label{CAPPTCPHHPBE}
Let $p$ and $q$ be closed $APPTC$ terms including encapsulation operator $\partial_H$, if $p\sim_{phhp} q$ then $p=q$.
\end{theorem}

\begin{proof}
Firstly, by the elimination theorem of $APPTC$ (see Theorem \ref{ETPEncapsulation}), we know that the normal form of $APPTC$ does not contain $\partial_H$, and for each closed $APPTC$
term $p$, there exists a closed basic $APPTC$ term $p'$, such that $APPTC\vdash p=p'$, so, we only need to consider closed basic $APPTC$ terms.

Similarly to Theorem \ref{CPPHHPBE}, we can prove that for normal forms $n$ and $n'$, if $n\sim_{phhp} n'$ then $n=_{AC}n'$.

Finally, let $s$ and $t$ be basic $APPTC$ terms, and $s\sim_{phhp} t$, there are normal forms $n$ and $n'$, such that $s=n$ and $t=n'$. The soundness theorem of $APPTC$ modulo probabilistic
hhp-bisimulation equivalence (see Theorem \ref{SAPPTCPHHPBE}) yields $s\sim_{phhp} n$ and $t\sim_{phhp} n'$, so $n\sim_{phhp} s\sim_{phhp} t\sim_{phhp} n'$. Since if $n\sim_{phhp} n'$ then
$n=_{AC}n'$, $s=n=_{AC}n'=t$, as desired.
\end{proof}

\subsection{Recursion}\label{rec}

In this section, we introduce recursion to capture infinite processes based on $APPTC$. Since in $APPTC$, there are four basic operators $\cdot$, $+$, $\boxplus_{\pi}$ and $\leftmerge$,
the recursion must be adapted this situation to include $\boxplus_{\pi}$ and $\leftmerge$.

In the following, $E,F,G$ are recursion specifications, $X,Y,Z$ are recursive variables.

\subsubsection{Guarded Recursive Specifications}

\begin{definition}[Guarded recursive specification]
A recursive specification

$$X_1=t_1(X_1,\cdots,X_n)$$
$$...$$
$$X_n=t_n(X_1,\cdots,X_n)$$

is guarded if the right-hand sides of its recursive equations can be adapted to the form by applications of the axioms in $APPTC$ and replacing recursion variables by the right-hand
sides of their recursive equations,

$((a_{111}\leftmerge\cdots\leftmerge a_{11i_1})\cdot s_1(X_1,\cdots,X_n)+\cdots+(a_{1k1}\leftmerge\cdots\leftmerge a_{1ki_k})\cdot s_k(X_1,\cdots,X_n)+(b_{111}\leftmerge\cdots\leftmerge
b_{11j_1})+\cdots+(b_{11j_1}\leftmerge\cdots\leftmerge b_{1lj_l}))\boxplus_{\pi_1}\cdots\boxplus_{\pi_{m-1}}((a_{m11}\leftmerge\cdots\leftmerge a_{m1i_1})\cdot s_1(X_1,\cdots,X_n)+
\cdots+(a_{mk1}\leftmerge\cdots\leftmerge a_{mki_k})\cdot s_k(X_1,\cdots,X_n)+(b_{m11}\leftmerge\cdots\leftmerge b_{m1j_1})+\cdots+(b_{m1j_1}\leftmerge\cdots\leftmerge b_{mlj_l}))$

where $a_{111},\cdots,a_{11i_1},a_{1k1},\cdots,a_{1ki_k},b_{111},\cdots,b_{11j_1},b_{11j_1},\cdots,b_{1lj_l},\cdots, a_{m11},\cdots,a_{m1i_1},a_{mk1},\cdots,a_{mki_k},\\b_{m11},\cdots,
b_{m1j_1},b_{m1j_1},\cdots,b_{mlj_l}\in \mathbb{E}$, and the sum above is allowed to be empty, in which case it represents the deadlock $\delta$.
\end{definition}

\begin{definition}[Linear recursive specification]\label{LRS}
A recursive specification is linear if its recursive equations are of the form

$((a_{111}\leftmerge\cdots\leftmerge a_{11i_1})X_1+\cdots+(a_{1k1}\leftmerge\cdots\leftmerge a_{1ki_k})X_k+(b_{111}\leftmerge\cdots\leftmerge b_{11j_1})+\cdots+(b_{11j_1}\leftmerge\cdots
\leftmerge b_{1lj_l}))\boxplus_{\pi_1}\cdots\boxplus_{\pi_{m-1}}((a_{m11}\leftmerge\cdots\leftmerge a_{m1i_1})X_1+\cdots+(a_{mk1}\leftmerge\cdots\leftmerge a_{mki_k})X_k+
(b_{m11}\leftmerge\cdots\leftmerge b_{m1j_1})+\cdots+(b_{m1j_1}\leftmerge\cdots\leftmerge b_{mlj_l}))$

where $a_{111},\cdots,a_{11i_1},a_{1k1},\cdots,a_{1ki_k},b_{111},\cdots,b_{11j_1},b_{11j_1},\cdots,b_{1lj_l},\cdots,a_{m11},\cdots,a_{m1i_1},a_{mk1},\cdots,a_{mki_k},\\b_{m11},\cdots,
b_{m1j_1},b_{m1j_1},\cdots,b_{mlj_l}\in \mathbb{E}$, and the sum above is allowed to be empty, in which case it represents the deadlock $\delta$.
\end{definition}

Firstly, we give the definition of PDFs in Table \ref{PDFGR}.

\begin{center}
    \begin{table}
        $$\mu(\langle X|E\rangle,y)=\mu(\langle t_X|E\rangle,y)$$
        $$\mu(x,y)=0,\textrm{otherwise}$$
        \caption{PDF definitions of recursion}
        \label{PDFGR}
    \end{table}
\end{center}

For a guarded recursive specifications $E$ with the form

$$X_1=t_1(X_1,\cdots,X_n)$$
$$\cdots$$
$$X_n=t_n(X_1,\cdots,X_n)$$

the behavior of the solution $\langle X_i|E\rangle$ for the recursion variable $X_i$ in $E$, where $i\in\{1,\cdots,n\}$, is exactly the behavior of their right-hand sides
$t_i(X_1,\cdots,X_n)$, which is captured by the two transition rules in Table \ref{TRForPGR}.

\begin{center}
    \begin{table}
        $$\frac{t_i(\langle X_1|E\rangle,\cdots,\langle X_n|E\rangle)\rightsquigarrow y}{\langle X_i|E\rangle\rightsquigarrow y}$$
        $$\frac{t_i(\langle X_1|E\rangle,\cdots,\langle X_n|E\rangle)\xrightarrow{\{e_1,\cdots,e_k\}}\surd}{\langle X_i|E\rangle\xrightarrow{\{e_1,\cdots,e_k\}}\surd}$$
        $$\frac{t_i(\langle X_1|E\rangle,\cdots,\langle X_n|E\rangle)\xrightarrow{\{e_1,\cdots,e_k\}} y}{\langle X_i|E\rangle\xrightarrow{\{e_1,\cdots,e_k\}} y}$$
        \caption{Transition rules of guarded recursion}
        \label{TRForPGR}
    \end{table}
\end{center}

\begin{theorem}[Conservitivity of $APPTC$ with guarded recursion]
$APPTC$ with guarded recursion is a conservative extension of $APPTC$.
\end{theorem}

\begin{proof}
Since the transition rules of $APPTC$ are source-dependent, and the transition rules for guarded recursion in Table \ref{TRForPGR} contain only a fresh constant in their source, so the
transition rules of $APPTC$ with guarded recursion are a conservative extension of those of $APPTC$.
\end{proof}

\begin{theorem}[Congruence theorem of $APPTC$ with guarded recursion]
Probabilistic truly concurrent bisimulation equivalences $\sim_{pp}$, $\sim_{ps}$, $\sim_{php}$ and $\sim_{phhp}$ are all congruences with respect to $APPTC$ with guarded recursion.
\end{theorem}

\begin{proof}
It follows the following two facts:
\begin{enumerate}
  \item in a guarded recursive specification, right-hand sides of its recursive equations can be adapted to the form by applications of the axioms in $APPTC$ and replacing recursion
  variables by the right-hand sides of their recursive equations;
  \item truly concurrent bisimulation equivalences $\sim_{pp}$, $\sim_{ps}$, $\sim_{php}$ and $\sim_{phhp}$ are all congruences with respect to all operators of $APPTC$.
\end{enumerate}
\end{proof}

\subsubsection{Recursive Definition and Specification Principles}

The $RDP$ (Recursive Definition Principle) and the $RSP$ (Recursive Specification Principle) are shown in Table \ref{PRDPRSP}.

\begin{center}
\begin{table}
  \begin{tabular}{@{}ll@{}}
\hline No. &Axiom\\
  $RDP$ & $\langle X_i|E\rangle = t_i(\langle X_1|E\rangle,\cdots,\langle X_n|E\rangle)\quad (i\in\{1,\cdots,n\})$\\
  $RSP$ & if $y_i=t_i(y_1,\cdots,y_n)$ for $i\in\{1,\cdots,n\}$, then $y_i=\langle X_i|E\rangle \quad(i\in\{1,\cdots,n\})$\\
\end{tabular}
\caption{Recursive definition and specification principle}
\label{PRDPRSP}
\end{table}
\end{center}

$RDP$ follows immediately from the two transition rules for guarded recursion, which express that $\langle X_i|E\rangle$ and $t_i(\langle X_1|E\rangle,\cdots,\langle X_n|E\rangle)$
have the same initial transitions for $i\in\{1,\cdots,n\}$. $RSP$ follows from the fact that guarded recursive specifications have only one solution.

\begin{theorem}[Elimination theorem of $APPTC$ with linear recursion]\label{ETPRecursion}
Each process term in $APPTC$ with linear recursion is equal to a process term $\langle X_1|E\rangle$ with $E$ a linear recursive specification.
\end{theorem}

\begin{proof}
By applying structural induction with respect to term size, each process term $t_1$ in $APPTC$ with linear recursion generates a process can be expressed in the form of equations

$t_i=((a_{1i11}\leftmerge\cdots\leftmerge a_{1i1i_1})t_{i1}+\cdots+(a_{1ik_i1}\leftmerge\cdots\leftmerge a_{1ik_ii_k})t_{ik_i}+(b_{1i11}\leftmerge\cdots\leftmerge b_{1i1i_1})+\cdots+
(b_{1il_i1}\leftmerge\cdots\leftmerge b_{1il_ii_l}))\boxplus_{\pi_1}\cdots\boxplus_{\pi_{m-1}}((a_{mi11}\leftmerge\cdots\leftmerge a_{mi1i_1})t_{i1}+\cdots+(a_{mik_i1}\leftmerge\cdots
\leftmerge a_{mik_ii_k})t_{ik_i}+(b_{mi11}\leftmerge\cdots\leftmerge b_{mi1i_1})+\cdots+(b_{mil_i1}\leftmerge\cdots\leftmerge b_{mil_ii_l}))$

for $i\in\{1,\cdots,n\}$. Let the linear recursive specification $E$ consist of the recursive equations

$X_i=((a_{1i11}\leftmerge\cdots\leftmerge a_{1i1i_1})X_{i1}+\cdots+(a_{1ik_i1}\leftmerge\cdots\leftmerge a_{1ik_ii_k})X_{ik_i}+(b_{1i11}\leftmerge\cdots\leftmerge b_{1i1i_1})+\cdots+
(b_{1il_i1}\leftmerge\cdots\leftmerge b_{1il_ii_l}))\boxplus_{\pi_1}\cdots\boxplus_{\pi_{m-1}}((a_{mi11}\leftmerge\cdots\leftmerge a_{mi1i_1})X_{i1}+\cdots+(a_{mik_i1}\leftmerge\cdots
\leftmerge a_{mik_ii_k})X_{ik_i}+(b_{mi11}\leftmerge\cdots\leftmerge b_{mi1i_1})+\cdots+(b_{mil_i1}\leftmerge\cdots\leftmerge b_{mil_ii_l}))$

for $i\in\{1,\cdots,n\}$. Replacing $X_i$ by $t_i$ for $i\in\{1,\cdots,n\}$ is a solution for $E$, $RSP$ yields $t_1=\langle X_1|E\rangle$.
\end{proof}

\begin{theorem}[Soundness of $APPTC$ with guarded recursion]\label{SAPPTCR}
Let $x$ and $y$ be $APPTC$ with guarded recursion terms. If $APPTC\textrm{ with guarded recursion}\vdash x=y$, then
\begin{enumerate}
  \item $x\sim_{pp} y$;
  \item $x\sim_{ps} y$;
  \item $x\sim_{php} y$;
  \item $x\sim_{phhp} y$.
\end{enumerate}
\end{theorem}

\begin{proof}
Since $\sim_{pp}$, $\sim_{ps}$, $\sim_{php}$ and $\sim_{phhp}$ are all both an equivalent and a congruent relation with respect to $APPTC$ with guarded recursion, we only need to
check if each axiom in Table \ref{PRDPRSP} is sound modulo $\sim_{pp}$, $\sim_{ps}$, $\sim_{php}$ and $\sim_{phhp}$. The proof is quite trivial and we omit it.
\end{proof}

\begin{theorem}[Completeness of $APPTC$ with linear recursion]\label{CAPPTCR}
Let $p$ and $q$ be closed $APPTC$ with linear recursion terms, then,
\begin{enumerate}
  \item if $p\sim_{pp} q$ then $p=q$;
  \item if $p\sim_{ps} q$ then $p=q$;
  \item if $p\sim_{php} q$ then $p=q$;
  \item if $p\sim_{phhp} q$ then $p=q$.
\end{enumerate}
\end{theorem}

\begin{proof}
Firstly, by the elimination theorem of $APPTC$ with guarded recursion (see Theorem \ref{ETPRecursion}), we know that each process term in $APPTC$ with linear recursion is equal to a
process term $\langle X_1|E\rangle$ with $E$ a linear recursive specification.

It remains to prove the following cases.

(1) If $\langle X_1|E_1\rangle \sim_{pp} \langle Y_1|E_2\rangle$ for linear recursive specification $E_1$ and $E_2$, then $\langle X_1|E_1\rangle = \langle Y_1|E_2\rangle$.

Let $E_1$ consist of recursive equations $X=t_X$ for $X\in \mathcal{X}$ and $E_2$
consists of recursion equations $Y=t_Y$ for $Y\in\mathcal{Y}$. Let the linear recursive specification $E$ consist of recursion equations $Z_{XY}=t_{XY}$, and
$\langle X|E_1\rangle\sim_{pp}\langle Y|E_2\rangle$, and $t_{XY}$ consists of the following summands:

\begin{enumerate}
  \item $t_{XY}$ contains a summand $(a_1\leftmerge\cdots\leftmerge a_m)Z_{X'Y'}$ iff $t_X$ contains the summand $(a_1\leftmerge\cdots\leftmerge a_m)X'$ and $t_Y$ contains the
  summand $(a_1\leftmerge\cdots\leftmerge a_m)Y'$ such that $\langle X'|E_1\rangle\sim_{pp}\langle Y'|E_2\rangle$;
  \item $t_{XY}$ contains a summand $b_1\leftmerge\cdots\leftmerge b_n$ iff $t_X$ contains the summand $b_1\leftmerge\cdots\leftmerge b_n$ and $t_Y$ contains the summand
  $b_1\leftmerge\cdots\leftmerge b_n$.
\end{enumerate}

Let $\sigma$ map recursion variable $X$ in $E_1$ to $\langle X|E_1\rangle$, and let $\psi$ map recursion variable $Z_{XY}$ in $E$ to $\langle X|E_1\rangle$. So,
$\sigma((a_1\leftmerge\cdots\leftmerge a_m)X')\equiv(a_1\leftmerge\cdots\leftmerge a_m)\langle X'|E_1\rangle\equiv\psi((a_1\leftmerge\cdots\leftmerge a_m)Z_{X'Y'})$,
so by $RDP$, we get $\langle X|E_1\rangle=\sigma(t_X)=\psi(t_{XY})$. Then by $RSP$, $\langle X|E_1\rangle=\langle Z_{XY}|E\rangle$, particularly,
$\langle X_1|E_1\rangle=\langle Z_{X_1Y_1}|E\rangle$. Similarly, we can obtain $\langle Y_1|E_2\rangle=\langle Z_{X_1Y_1}|E\rangle$. Finally,
$\langle X_1|E_1\rangle=\langle Z_{X_1Y_1}|E\rangle=\langle Y_1|E_2\rangle$, as desired.

(2) If $\langle X_1|E_1\rangle \sim_{ps} \langle Y_1|E_2\rangle$ for linear recursive specification $E_1$ and $E_2$, then $\langle X_1|E_1\rangle = \langle Y_1|E_2\rangle$.

It can be proven similarly to (1), we omit it.

(3) If $\langle X_1|E_1\rangle \sim_{php} \langle Y_1|E_2\rangle$ for linear recursive specification $E_1$ and $E_2$, then $\langle X_1|E_1\rangle = \langle Y_1|E_2\rangle$.

It can be proven similarly to (1), we omit it.

(4) If $\langle X_1|E_1\rangle \sim_{phhp} \langle Y_1|E_2\rangle$ for linear recursive specification $E_1$ and $E_2$, then $\langle X_1|E_1\rangle = \langle Y_1|E_2\rangle$.

It can be proven similarly to (1), we omit it.
\end{proof}

\subsubsection{Approximation Induction Principle}

In this subsection, we introduce approximation induction principle ($AIP$) and try to explain that $AIP$ is still valid in probabilistic true concurrency. $AIP$ can be used to try and
equate probabilistic truly concurrent bisimilar guarded recursive specifications. $AIP$ says that if two process terms are probabilistic truly concurrent bisimilar up to any finite depth,
then they are probabilistic truly concurrent bisimilar.

Also, we need the auxiliary unary projection operator $\Pi_n$ for $n\in\mathbb{N}$ and $\mathbb{N}\triangleq\{0,1,2,\cdots\}$.

Firstly, we give the definition of PDFs in Table \ref{PDFAIP}.

\begin{center}
    \begin{table}
        $$\mu(\Pi_n(x),\Pi_n(x'))=\mu(x,x')\quad n\geq 1$$
        $$\mu(x,y)=0,\textrm{otherwise}$$
        \caption{PDF definitions of approximation induction principle}
        \label{PDFAIP}
    \end{table}
\end{center}

The transition rules of $\Pi_n$ are expressed in Table \ref{TRForPProjection}.

\begin{center}
    \begin{table}
        $$\frac{x\rightsquigarrow x'}{\Pi_n(x)\rightsquigarrow\Pi_n(x')}$$
        $$\frac{x\xrightarrow{\{e_1,\cdots,e_k\}}\surd}{\Pi_{n+1}(x)\xrightarrow{\{e_1,\cdots,e_k\}}\surd}
        \quad\frac{x\xrightarrow{\{e_1,\cdots,e_k\}}x'}{\Pi_{n+1}(x)\xrightarrow{\{e_1,\cdots,e_k\}}\Pi_n(x')}$$
        \caption{Transition rules of encapsulation operator $\partial_H$}
        \label{TRForPProjection}
    \end{table}
\end{center}

Based on the transition rules for projection operator $\Pi_n$ in Table \ref{TRForPProjection}, we design the axioms as Table \ref{AxiomsForPProjection} shows.

\begin{center}
    \begin{table}
        \begin{tabular}{@{}ll@{}}
            \hline No. &Axiom\\
            $PR1$ & $\Pi_n(x+y)=\Pi_n(x)+\Pi_n(y)$\\
            $PPR1$ & $\Pi_n(x\boxplus_{\rho} y)=\Pi_n(x)\boxplus_{\rho}\Pi_n(y)$\\
            $PR2$ & $\Pi_n(x\leftmerge y)=\Pi_n(x)\leftmerge \Pi_n(y)$\\
            $PR3$ & $\Pi_{n+1}(e_1\leftmerge\cdots\leftmerge e_k)=e_1\leftmerge\cdots\leftmerge e_k$\\
            $PR4$ & $\Pi_{n+1}((e_1\leftmerge\cdots\leftmerge e_k)\cdot x)=(e_1\leftmerge\cdots\leftmerge e_k)\cdot\Pi_n(x)$\\
            $PR5$ & $\Pi_0(x)=\delta$\\
            $PR6$ & $\Pi_n(\delta)=\delta$\\
        \end{tabular}
        \caption{Axioms of projection operator}
        \label{AxiomsForPProjection}
    \end{table}
\end{center}

The axioms $PR1-PR2$ and $PPR1$ say that $\Pi_n(s+t)$, $\Pi_n(s\leftmerge t)$ and $\Pi_n(s\boxplus_{\rho} t)$ can execute transitions of $s$ and $t$ up to depth $n$.
$PR3$ says that $\Pi_{n+1}(e_1\leftmerge\cdots\leftmerge e_k)$ executes $\{e_1,\cdots,e_{k}\}$ and terminates successfully. $PR4$ says that
$\Pi_{n+1}((e_1\leftmerge\cdots\leftmerge e_k)\cdot t)$ executes $\{e_1,\cdots,e_{k}\}$ and then executes transitions of $t$ up to depth $n$. $PR5$ and $PR6$ say that
$\Pi_0(t)$ and $\Pi_n(\delta)$ exhibit no actions.

\begin{theorem}[Conservativity of $APPTC$ with projection operator and guarded recursion]
$APPTC$  with projection operator and guarded recursion is a conservative extension of $APPTC$ with guarded recursion.
\end{theorem}

\begin{proof}
It follows from the following two facts (see Theorem \ref{TCE}).

\begin{enumerate}
  \item The transition rules of $APPTC$ with guarded recursion are all source-dependent;
  \item The sources of the transition rules for the projection operator contain an occurrence of $\Pi_n$.
\end{enumerate}
\end{proof}

\begin{theorem}[Congruence theorem of projection operator $\Pi_n$]
Probabilistic truly concurrent bisimulation equivalences $\sim_{pp}$, $\sim_{ps}$, $\sim_{php}$ and $\sim_{phhp}$ are all congruences with respect to projection operator $\Pi_n$.
\end{theorem}

\begin{proof}
It is easy to see that $\sim_{pp}$, $\sim_{ps}$, $\sim_{php}$ and $\sim_{phhp}$ are all an equivalent relation with respect to projection operator $\Pi_n$, we only need to prove that
$\sim_{pp}$, $\sim_{ps}$, $\sim_{php}$ and $\sim_{phhp}$ are preserved by the operators $\Pi_n$.
That is, if $x\sim_{pp} x'$, $x\sim_{ps} x'$, $x\sim_{php} x'$ and $x\sim_{phhp} x'$, we need to prove that $\Pi_n(x)\sim_{pp} \Pi_n(x')$, $\Pi_n(x)\sim_{ps} \Pi_n(x')$,
$\Pi_n(x)\sim_{php} \Pi_n(x')$ and $\Pi_n(x)\sim_{phhp} \Pi_n(x')$. The proof is quite trivial and we omit it.
\end{proof}

\begin{theorem}[Elimination theorem of $APPTC$ with linear recursion and projection operator]\label{ETPProjection}
Each process term in $APPTC$ with linear recursion and projection operator is equal to a process term $\langle X_1|E\rangle$ with $E$ a linear recursive specification.
\end{theorem}

\begin{proof}
By applying structural induction with respect to term size, each process term $t_1$ in $APPTC$ with linear recursion and projection operator $\Pi_n$ generates a process can be
expressed in the form of equations

$t_i=((a_{1i11}\leftmerge\cdots\leftmerge a_{1i1i_1})t_{i1}+\cdots+(a_{1ik_i1}\leftmerge\cdots\leftmerge a_{1ik_ii_k})t_{ik_i}+(b_{1i11}\leftmerge\cdots\leftmerge b_{1i1i_1})+\cdots+
(b_{1il_i1}\leftmerge\cdots\leftmerge b_{1il_ii_l}))\boxplus_{\pi_1}\cdots\boxplus_{\pi_{m-1}}((a_{mi11}\leftmerge\cdots\leftmerge a_{mi1i_1})t_{i1}+\cdots+(a_{mik_i1}\leftmerge\cdots
\leftmerge a_{mik_ii_k})t_{ik_i}+(b_{mi11}\leftmerge\cdots\leftmerge b_{mi1i_1})+\cdots+(b_{mil_i1}\leftmerge\cdots\leftmerge b_{mil_ii_l}))$

for $i\in\{1,\cdots,n\}$. Let the linear recursive specification $E$ consist of the recursive equations

$X_i=((a_{1i11}\leftmerge\cdots\leftmerge a_{1i1i_1})X_{i1}+\cdots+(a_{1ik_i1}\leftmerge\cdots\leftmerge a_{1ik_ii_k})X_{ik_i}+(b_{1i11}\leftmerge\cdots\leftmerge b_{1i1i_1})+\cdots+
(b_{1il_i1}\leftmerge\cdots\leftmerge b_{1il_ii_l}))\boxplus_{\pi_1}\cdots\boxplus_{\pi_{m-1}}((a_{mi11}\leftmerge\cdots\leftmerge a_{mi1i_1})X_{i1}+\cdots+(a_{mik_i1}\leftmerge\cdots
\leftmerge a_{mik_ii_k})X_{ik_i}+(b_{mi11}\leftmerge\cdots\leftmerge b_{mi1i_1})+\cdots+(b_{mil_i1}\leftmerge\cdots\leftmerge b_{mil_ii_l}))$

for $i\in\{1,\cdots,n\}$. Replacing $X_i$ by $t_i$ for $i\in\{1,\cdots,n\}$ is a solution for $E$, $RSP$ yields $t_1=\langle X_1|E\rangle$.

That is, in $E$, there is not the occurrence of projection operator $\Pi_n$.
\end{proof}

\begin{theorem}[Soundness of $APPTC$ with projection operator and guarded recursion]\label{SAPPTCR}
Let $x$ and $y$ be $APPTC$ with projection operator and guarded recursion terms. If $APPTC$ with projection operator and guarded recursion $\vdash x=y$, then
\begin{enumerate}
  \item $x\sim_{pp} y$;
  \item $x\sim_{ps} y$;
  \item $x\sim_{php} y$;
  \item $x\sim_{phhp} y$.
\end{enumerate}
\end{theorem}

\begin{proof}
Since $\sim_{pp}$, $\sim_{ps}$, $\sim_{php}$ and $\sim_{phhp}$ are all both an equivalent and a congruent relation with respect to $APPTC$ with guarded recursion, we only need to
check if each axiom in Table \ref{TRForPProjection} is sound modulo $\sim_{pp}$, $\sim_{ps}$, $\sim_{php}$ and $\sim_{phhp}$. The proof is quite trivial and we omit it.
\end{proof}

Then $AIP$ is given in Table \ref{PAIP}.

\begin{center}
    \begin{table}
        \begin{tabular}{@{}ll@{}}
            \hline No. &Axiom\\
            $AIP$ & if $\Pi_n(x)=\Pi_n(y)$ for $n\in\mathbb{N}$, then $x=y$\\
        \end{tabular}
        \caption{$AIP$}
        \label{PAIP}
    \end{table}
\end{center}

\begin{theorem}[Soundness of $AIP$]\label{SPAIP}
Let $x$ and $y$ be $APPTC$ with projection operator and guarded recursion terms.

\begin{enumerate}
  \item If $\Pi_n(x)\sim_{pp}\Pi_n(y)$ for $n\in\mathbb{N}$, then $x\sim_{pp} y$;
  \item If $\Pi_n(x)\sim_{ps}\Pi_n(y)$ for $n\in\mathbb{N}$, then $x\sim_{ps} y$;
  \item If $\Pi_n(x)\sim_{php}\Pi_n(y)$ for $n\in\mathbb{N}$, then $x\sim_{php} y$;
  \item If $\Pi_n(x)\sim_{phhp}\Pi_n(y)$ for $n\in\mathbb{N}$, then $x\sim_{phhp} y$.
\end{enumerate}
\end{theorem}

\begin{proof}
(1) If $\Pi_n(x)\sim_{pp}\Pi_n(y)$ for $n\in\mathbb{N}$, then $x\sim_{pp} y$.

Since $\sim_{pp}$ is both an equivalent and a congruent relation with respect to $APPTC$ with guarded recursion and projection operator, we only need to check if $AIP$ in Table
\ref{PAIP} is sound modulo $\sim_{pp}$.

Let $p,p_0$ and $q,q_0$ be closed $APPTC$ with projection operator and guarded recursion terms such that $\Pi_n(p_0)\sim_{pp}\Pi_n(q_0)$ for $n\in\mathbb{N}$. We define a relation
$R$ such that $p R q$ iff $\Pi_n(p)\sim_{pp} \Pi_n(q)$. Obviously, $p_0 R q_0$, next, we prove that $R\in\sim_{pp}$.

Let $p R q$ and $p\rightsquigarrow\xrightarrow{\{e_1,\cdots,e_k\}}\surd$, then $\Pi_1(p)\rightsquigarrow\xrightarrow{\{e_1,\cdots,e_k\}}\surd$, $\Pi_1(p)\sim_{pp}\Pi_1(q)$ yields
$\Pi_1(q)\rightsquigarrow\xrightarrow{\{e_1,\cdots,e_k\}}\surd$. Similarly, $q\rightsquigarrow\xrightarrow{\{e_1,\cdots,e_k\}}\surd$ implies
$p\rightsquigarrow\xrightarrow{\{e_1,\cdots,e_k\}}\surd$.

Let $p R q$ and $p\rightsquigarrow\xrightarrow{\{e_1,\cdots,e_k\}}p'$. We define the set of process terms

$$S_n\triangleq\{q'|q\rightsquigarrow\xrightarrow{\{e_1,\cdots,e_k\}}q'\textrm{ and }\Pi_n(p')\sim_{pp}\Pi_n(q')\}$$

\begin{enumerate}
  \item Since $\Pi_{n+1}(p)\sim_{pp}\Pi_{n+1}(q)$ and $\Pi_{n+1}(p)\rightsquigarrow\xrightarrow{\{e_1,\cdots,e_k\}}\Pi_n(p')$, there exist $q'$ such that
  $\Pi_{n+1}(q)\rightsquigarrow\xrightarrow{\{e_1,\cdots,e_k\}}\Pi_n(q')$ and $\Pi_{n}(p')\sim_{pp}\Pi_{n}(q')$. So, $S_n$ is not empty.
  \item There are only finitely many $q'$ such that $q\rightsquigarrow\xrightarrow{\{e_1,\cdots,e_k\}}q'$, so, $S_n$ is finite.
  \item $\Pi_{n+1}(p)\sim_{pp}\Pi_{n+1}(q)$ implies $\Pi_{n}(p')\sim_{pp}\Pi_{n}(q')$, so $S_n\supseteq S_{n+1}$.
\end{enumerate}

So, $S_n$ has a non-empty intersection, and let $q'$ be in this intersection, then $q\rightsquigarrow\xrightarrow{\{e_1,\cdots,e_k\}}q'$ and $\Pi_n(p')\sim_{pp}\Pi_n(q')$,
so $p' R q'$. Similarly, let $p\mathcal{q}q$, we can obtain $q\rightsquigarrow\xrightarrow{\{e_1,\cdots,e_k\}}q'$ implies $p\rightsquigarrow\xrightarrow{\{e_1,\cdots,e_k\}}p'$
such that $p' R q'$.

Finally, $R\in\sim_{pp}$ and $p_0\sim_{pp} q_0$, as desired.

(2) If $\Pi_n(x)\sim_{ps}\Pi_n(y)$ for $n\in\mathbb{N}$, then $x\sim_{ps} y$.

It can be proven similarly to (1).

(3) If $\Pi_n(x)\sim_{php}\Pi_n(y)$ for $n\in\mathbb{N}$, then $x\sim_{php} y$.

It can be proven similarly to (1).

(4) If $\Pi_n(x)\sim_{phhp}\Pi_n(y)$ for $n\in\mathbb{N}$, then $x\sim_{phhp} y$.

It can be proven similarly to (1).
\end{proof}

\begin{theorem}[Completeness of $AIP$]\label{CPAIP}
Let $p$ and $q$ be closed $APPTC$ with linear recursion and projection operator terms, then,
\begin{enumerate}
  \item if $p\sim_{pp} q$ then $\Pi_n(p)=\Pi_n(q)$;
  \item if $p\sim_{ps} q$ then $\Pi_n(p)=\Pi_n(q)$;
  \item if $p\sim_{php} q$ then $\Pi_n(p)=\Pi_n(q)$;
  \item if $p\sim_{phhp} q$ then $\Pi_n(p)=\Pi_n(q)$.
\end{enumerate}
\end{theorem}

\begin{proof}
Firstly, by the elimination theorem of $APPTC$ with guarded recursion and projection operator (see Theorem \ref{ETPProjection}), we know that each process term in $APPTC$ with linear
recursion and projection operator is equal to a process term $\langle X_1|E\rangle$ with $E$ a linear recursive specification:

$X_i=((a_{1i11}\leftmerge\cdots\leftmerge a_{1i1i_1})X_{i1}+\cdots+(a_{1ik_i1}\leftmerge\cdots\leftmerge a_{1ik_ii_k})X_{ik_i}+(b_{1i11}\leftmerge\cdots\leftmerge b_{1i1i_1})+\cdots+
(b_{1il_i1}\leftmerge\cdots\leftmerge b_{1il_ii_l}))\boxplus_{\pi_1}\cdots\boxplus_{\pi_{m-1}}((a_{mi11}\leftmerge\cdots\leftmerge a_{mi1i_1})X_{i1}+\cdots+(a_{mik_i1}\leftmerge\cdots
\leftmerge a_{mik_ii_k})X_{ik_i}+(b_{mi11}\leftmerge\cdots\leftmerge b_{mi1i_1})+\cdots+(b_{mil_i1}\leftmerge\cdots\leftmerge b_{mil_ii_l}))$

for $i\in\{1,\cdots,n\}$.

It remains to prove the following cases.

(1) if $p\sim_{pp} q$ then $\Pi_n(p)=\Pi_n(q)$.

Let $p\sim_{pp} q$, and fix an $n\in\mathbb{N}$, there are $p',q'$ in basic $APPTC$ terms such that $p'=\Pi_n(p)$ and $q'=\Pi_n(q)$. Since $\sim_{pp}$ is a congruence with respect to
$APPTC$, if $p\sim_{pp} q$ then $\Pi_n(p)\sim_{pp}\Pi_n(q)$. The soundness theorem yields $p'\sim_{pp}\Pi_n(p)\sim_{pp}\Pi_n(q)\sim_{pp} q'$. Finally, the completeness of $APPTC$
modulo $\sim_{pp}$ (see Theorem \ref{SAPPTCPPBE}) ensures $p'=q'$, and $\Pi_n(p)=p'=q'=\Pi_n(q)$, as desired.

(2) if $p\sim_{ps} q$ then $\Pi_n(p)=\Pi_n(q)$.

Let $p\sim_{ps} q$, and fix an $n\in\mathbb{N}$, there are $p',q'$ in basic $APPTC$ terms such that $p'=\Pi_n(p)$ and $q'=\Pi_n(q)$. Since $\sim_{ps}$ is a congruence with respect to
$APPTC$, if $p\sim_{ps} q$ then $\Pi_n(p)\sim_{ps}\Pi_n(q)$. The soundness theorem yields $p'\sim_{ps}\Pi_n(p)\sim_{ps}\Pi_n(q)\sim_{ps} q'$. Finally, the completeness of $APPTC$
modulo $\sim_{ps}$ (see Theorem \ref{SAPPTCPSBE}) ensures $p'=q'$, and $\Pi_n(p)=p'=q'=\Pi_n(q)$, as desired.

(3) if $p\sim_{php} q$ then $\Pi_n(p)=\Pi_n(q)$.

Let $p\sim_{php} q$, and fix an $n\in\mathbb{N}$, there are $p',q'$ in basic $APPTC$ terms such that $p'=\Pi_n(p)$ and $q'=\Pi_n(q)$. Since $\sim_{php}$ is a congruence with respect to
$APPTC$, if $p\sim_{php} q$ then $\Pi_n(p)\sim_{php}\Pi_n(q)$. The soundness theorem yields $p'\sim_{php}\Pi_n(p)\sim_{php}\Pi_n(q)\sim_{php} q'$. Finally, the completeness of $APPTC$
modulo $\sim_{php}$ (see Theorem \ref{SAPPTCPHPBE}) ensures $p'=q'$, and $\Pi_n(p)=p'=q'=\Pi_n(q)$, as desired.

(4) if $p\sim_{phhp} q$ then $\Pi_n(p)=\Pi_n(q)$.

Let $p\sim_{phhp} q$, and fix an $n\in\mathbb{N}$, there are $p',q'$ in basic $APPTC$ terms such that $p'=\Pi_n(p)$ and $q'=\Pi_n(q)$. Since $\sim_{phhp}$ is a congruence with respect to
$APPTC$, if $p\sim_{phhp} q$ then $\Pi_n(p)\sim_{phhp}\Pi_n(q)$. The soundness theorem yields $p'\sim_{phhp}\Pi_n(p)\sim_{phhp}\Pi_n(q)\sim_{phhp} q'$. Finally, the completeness of $APPTC$
modulo $\sim_{phhp}$ (see Theorem \ref{SAPPTCPHHPBE}) ensures $p'=q'$, and $\Pi_n(p)=p'=q'=\Pi_n(q)$, as desired.
\end{proof}

\subsection{Abstraction}\label{abs}

To abstract away from the internal implementations of a program, and verify that the program exhibits the desired external behaviors, the silent step $\tau$ and abstraction operator
$\tau_I$ are introduced, where $I\subseteq \mathbb{E}$ denotes the internal events. The silent step $\tau$ represents the internal events, when we consider the external behaviors of a
process, $\tau$ events can be removed, that is, $\tau$ events must keep silent. The transition rule of $\tau$ is shown in Table \ref{TRForPTau}. In the following, let the atomic event
$e$ range over $\mathbb{E}\cup\{\delta\}\cup\{\tau\}$, and let the communication function $\gamma:\mathbb{E}\cup\{\tau\}\times \mathbb{E}\cup\{\tau\}\rightarrow \mathbb{E}\cup\{\delta\}$,
with each communication involved $\tau$ resulting into $\delta$.

\begin{center}
    \begin{table}
        $$\frac{}{\tau\rightsquigarrow\breve{\tau}}$$
        $$\frac{}{\tau\xrightarrow{\tau}\surd}$$
        \caption{Transition rules of the silent step}
        \label{TRForPTau}
    \end{table}
\end{center}

In this section, we try to find the algebraic laws of $\tau$ and $\tau_I$ in probabilistic true concurrency.

\subsubsection{Guarded Linear Recursion}

The silent step $\tau$ as an atomic event, is introduced into $E$. Considering the recursive specification $X=\tau X$, $\tau s$, $\tau\tau s$, and $\tau\cdots s$ are all its solutions,
that is, the solutions make the existence of $\tau$-loops which cause unfairness. To prevent $\tau$-loops, we extend the definition of linear recursive specification
to the guarded one.

\begin{definition}[Guarded linear recursive specification]\label{GLRS}
A recursive specification is linear if its recursive equations are of the form

$((a_{111}\parallel\cdots\parallel a_{11i_1})X_1+\cdots+(a_{1k1}\parallel\cdots\parallel a_{1ki_k})X_k+(b_{111}\parallel\cdots\parallel b_{11j_1})+\cdots+
(b_{11j_1}\parallel\cdots\parallel b_{1lj_l}))\boxplus_{\pi_1}\cdots\boxplus_{\pi_{m-1}}((a_{m11}\parallel\cdots\parallel a_{m1i_1})X_1+\cdots+(a_{mk1}\parallel\cdots\parallel a_{mki_k})X_k+
(b_{m11}\parallel\cdots\parallel b_{m1j_1})+\cdots+(b_{m1j_1}\parallel\cdots\parallel b_{mlj_l}))$

where $a_{111},\cdots,a_{11i_1},a_{1k1},\cdots,a_{1ki_k},b_{111},\cdots,b_{11j_1},b_{11j_1},\cdots,b_{1lj_l}\cdots\\
a_{m11},\cdots,a_{m1i_1},a_{mk1},\cdots,a_{mki_k},b_{m11},\cdots,b_{m1j_1},b_{m1j_1},\cdots,b_{mlj_l}\in \mathbb{E}\cup\{\tau\}$, and the sum above is allowed to be empty, in which
case it represents the deadlock $\delta$.

A linear recursive specification $E$ is guarded if there does not exist an infinite sequence of $\tau$-transitions $\langle X|E\rangle\rightsquigarrow\xrightarrow{\tau}
\langle X'|E\rangle\rightsquigarrow\xrightarrow{\tau}\langle X''|E\rangle\rightsquigarrow\xrightarrow{\tau}\cdots$.
\end{definition}

\begin{theorem}[Conservitivity of $APPTC$ with silent step and guarded linear recursion]
$APPTC$ with silent step and guarded linear recursion is a conservative extension of $APPTC$ with linear recursion.
\end{theorem}

\begin{proof}
Since the transition rules of $APPTC$ with linear recursion are source-dependent, and the transition rules for silent step in Table \ref{TRForPTau} contain only a fresh constant $\tau$
in their source, so the transition rules of $APPTC$ with silent step and guarded linear recursion is a conservative extension of those of $APPTC$ with linear recursion.
\end{proof}

\begin{theorem}[Congruence theorem of $APPTC$ with silent step and guarded linear recursion]
Probabilistic rooted branching truly concurrent bisimulation equivalences $\approx_{prbp}$, $\approx_{prbs}$, $\approx_{prbhp}$ and $\approx_{prbhhp}$ are all congruences with respect
to $APPTC$ with silent step and guarded linear recursion.
\end{theorem}

\begin{proof}
It follows the following three facts:
\begin{enumerate}
  \item in a guarded linear recursive specification, right-hand sides of its recursive equations can be adapted to the form by applications of the axioms in $APPTC$ and replacing
  recursion variables by the right-hand sides of their recursive equations;
  \item probabilistic truly concurrent bisimulation equivalences $\sim_{pp}$, $\sim_{ps}$, $\sim_{php}$ and $\sim_{phhp}$ are all congruences with respect to all operators of $APPTC$,
  while probabilistic truly concurrent bisimulation equivalences $\sim_{pp}$, $\sim_{ps}$, $\sim_{php}$ and $\sim_{phhp}$ imply the corresponding probabilistic rooted branching truly
  concurrent bisimulations $\approx_{prbp}$, $\approx_{prbs}$, $\approx_{prbhp}$ and $\approx_{prbhhp}$, so probabilistic rooted branching truly concurrent bisimulations $\approx_{prbp}$,
  $\approx_{prbs}$, $\approx_{prbhp}$ and $\approx_{prbhhp}$ are all congruences with respect to all operators of $APPTC$;
  \item While $\mathbb{E}$ is extended to $\mathbb{E}\cup\{\tau\}$, it can be proved that probabilistic rooted branching truly concurrent bisimulations $\approx_{prbp}$, $\approx_{prbs}$,
  $\approx_{prbhp}$ and $\approx_{prbhhp}$ are all congruences with respect to all operators of $APPTC$, we omit it.
\end{enumerate}
\end{proof}

\subsubsection{Algebraic Laws for the Silent Step}

We design the axioms for the silent step $\tau$ in Table \ref{AxiomsForPTau}.

\begin{center}
\begin{table}
  \begin{tabular}{@{}ll@{}}
\hline No. &Axiom\\
  $B1$ & $(y=y+y,z=z+z)\quad x\cdot((y+\tau\cdot(y+z))\boxplus_{\pi}w)=x\cdot((y+z)\boxplus_{\pi}w)$\\
  $B2$ & $(y=y+y,z=z+z)\quad x\leftmerge((y+\tau\leftmerge(y+z))\boxplus_{\pi}w)=x\leftmerge((y+z)\boxplus_{\pi}w)$\\
\end{tabular}
\caption{Axioms of silent step}
\label{AxiomsForPTau}
\end{table}
\end{center}

The axioms $B1$ and $B2$ are the conditions in which $\tau$ really keeps silent to act with the operators $\cdot$, $+$, $\boxplus_{\pi}$ and $\leftmerge$.

\begin{theorem}[Elimination theorem of $APPTC$ with silent step and guarded linear recursion]\label{ETTau}
Each process term in $APPTC$ with silent step and guarded linear recursion is equal to a process term $\langle X_1|E\rangle$ with $E$ a guarded linear recursive specification.
\end{theorem}

\begin{proof}
By applying structural induction with respect to term size, each process term $t_1$ in $APPTC$ with silent step and guarded linear recursion generates a process can be expressed in the
form of equations

$t_i=((a_{1i11}\leftmerge\cdots\leftmerge a_{1i1i_1})t_{i1}+\cdots+(a_{1ik_i1}\leftmerge\cdots\leftmerge a_{1ik_ii_k})t_{ik_i}+(b_{1i11}\leftmerge\cdots\leftmerge b_{1i1i_1})+\cdots+
(b_{1il_i1}\leftmerge\cdots\leftmerge b_{1il_ii_l}))\boxplus_{\pi_1}\cdots\boxplus_{\pi_{m-1}}((a_{mi11}\leftmerge\cdots\leftmerge a_{mi1i_1})t_{i1}+\cdots+(a_{mik_i1}\leftmerge\cdots
\leftmerge a_{mik_ii_k})t_{ik_i}+(b_{mi11}\leftmerge\cdots\leftmerge b_{mi1i_1})+\cdots+(b_{mil_i1}\leftmerge\cdots\leftmerge b_{m1il_ii_l}))$

for $i\in\{1,\cdots,n\}$. Let the linear recursive specification $E$ consist of the recursive equations

$X_i=((a_{1i11}\leftmerge\cdots\leftmerge a_{1i1i_1})X_{i1}+\cdots+(a_{1ik_i1}\leftmerge\cdots\leftmerge a_{1ik_ii_k})X_{ik_i}+(b_{1i11}\leftmerge\cdots\leftmerge b_{1i1i_1})+\cdots+
(b_{1il_i1}\leftmerge\cdots\leftmerge b_{1il_ii_l}))\boxplus_{\pi_1}\cdots\boxplus_{\pi_{m-1}}((a_{mi11}\leftmerge\cdots\leftmerge a_{mi1i_1})X_{i1}+\cdots+(a_{mik_i1}\leftmerge\cdots
\leftmerge a_{mik_ii_k})X_{ik_i}+(b_{mi11}\leftmerge\cdots\leftmerge b_{mi1i_1})+\cdots+(b_{mil_i1}\leftmerge\cdots\leftmerge b_{mil_ii_l}))$

for $i\in\{1,\cdots,n\}$. Replacing $X_i$ by $t_i$ for $i\in\{1,\cdots,n\}$ is a solution for $E$, $RSP$ yields $t_1=\langle X_1|E\rangle$.
\end{proof}

\begin{theorem}[Soundness of $APPTC$ with silent step and guarded linear recursion]\label{SAPPTCTAU}
Let $x$ and $y$ be $APPTC$ with silent step and guarded linear recursion terms. If $APPTC$ with silent step and guarded linear recursion $\vdash x=y$, then
\begin{enumerate}
  \item $x\approx_{prbp} y$;
  \item $x\approx_{prbs} y$;
  \item $x\approx_{prbhp} y$;
  \item $x\approx_{prbhhp}y$.
\end{enumerate}
\end{theorem}

\begin{proof}
Since probabilistic truly concurrent rooted branching bisimulation $\approx_{prbp}$, $\approx_{prbs}$, $\approx_{prbhp}$ and $\approx_{prbhhp}$ are all both an equivalent and a
congruent relation with respect to $APPTC$ with silent step and guarded linear recursion, we only need to check if each axiom in Table \ref{AxiomsForPTau} is sound modulo
probabilistic truly concurrent rooted branching bisimulation $\approx_{prbp}$, $\approx_{prbs}$, $\approx_{prbhp}$ and $\approx_{prbhhp}$. The proof is quite trivial and we omit it.
\end{proof}

\begin{theorem}[Completeness of $APPTC$ with silent step and guarded linear recursion]\label{CAPPTCTAU}
Let $p$ and $q$ be closed $APPTC$ with silent step and guarded linear recursion terms, then,
\begin{enumerate}
  \item if $p\approx_{prbp} q$ then $p=q$;
  \item if $p\approx_{prbs} q$ then $p=q$;
  \item if $p\approx_{prbhp} q$ then $p=q$;
  \item if $p\approx_{prbhhp} q$ then $p=q$.
\end{enumerate}
\end{theorem}

\begin{proof}
Firstly, by the elimination theorem of $APPTC$ with silent step and guarded linear recursion (see Theorem \ref{ETTau}), we know that each process term in $APPTC$ with silent step and guarded linear recursion is equal to a process term $\langle X_1|E\rangle$ with $E$ a guarded linear recursive specification.

It remains to prove the following cases.

(1) If $\langle X_1|E_1\rangle \approx_{prbp} \langle Y_1|E_2\rangle$ for guarded linear recursive specification $E_1$ and $E_2$, then $\langle X_1|E_1\rangle = \langle Y_1|E_2\rangle$.

Firstly, the recursive equation $W=\tau+\cdots+\tau$ with $W\nequiv X_1$ in $E_1$ and $E_2$, can be removed, and the corresponding summands $aW$ are replaced by $a$, to get $E_1'$ and
$E_2'$, by use of the axioms $RDP$, $A3$ and $B1-2$, and $\langle X|E_1\rangle = \langle X|E_1'\rangle$, $\langle Y|E_2\rangle = \langle Y|E_2'\rangle$.

Let $E_1$ consists of recursive equations $X=t_X$ for $X\in \mathcal{X}$ and $E_2$
consists of recursion equations $Y=t_Y$ for $Y\in\mathcal{Y}$, and are not the form $\tau+\cdots+\tau$. Let the guarded linear recursive specification $E$ consists of recursion
equations $Z_{XY}=t_{XY}$, and $\langle X|E_1\rangle\approx_{prbp}\langle Y|E_2\rangle$, and $t_{XY}$ consists of the following summands:

\begin{enumerate}
  \item $t_{XY}$ contains a summand $(a_1\leftmerge\cdots\leftmerge a_m)Z_{X'Y'}$ iff $t_X$ contains the summand $(a_1\leftmerge\cdots\leftmerge a_m)X'$ and $t_Y$ contains the summand
  $(a_1\leftmerge\cdots\leftmerge a_m)Y'$ such that $\langle X'|E_1\rangle\approx_{prbp}\langle Y'|E_2\rangle$;
  \item $t_{XY}$ contains a summand $b_1\leftmerge\cdots\leftmerge b_n$ iff $t_X$ contains the summand $b_1\leftmerge\cdots\leftmerge b_n$ and $t_Y$ contains the summand
  $b_1\leftmerge\cdots\leftmerge b_n$;
  \item $t_{XY}$ contains a summand $\tau Z_{X'Y}$ iff $XY\nequiv X_1Y_1$, $t_X$ contains the summand $\tau X'$, and $\langle X'|E_1\rangle\approx_{prbp}\langle Y|E_2\rangle$;
  \item $t_{XY}$ contains a summand $\tau Z_{XY'}$ iff $XY\nequiv X_1Y_1$, $t_Y$ contains the summand $\tau Y'$, and $\langle X|E_1\rangle\approx_{prbp}\langle Y'|E_2\rangle$.
\end{enumerate}

Since $E_1$ and $E_2$ are guarded, $E$ is guarded. Constructing the process term $u_{XY}$ consist of the following summands:

\begin{enumerate}
  \item $u_{XY}$ contains a summand $(a_1\leftmerge\cdots\leftmerge a_m)\langle X'|E_1\rangle$ iff $t_X$ contains the summand $(a_1\leftmerge\cdots\leftmerge a_m)X'$ and $t_Y$ contains
  the summand $(a_1\leftmerge\cdots\leftmerge a_m)Y'$ such that $\langle X'|E_1\rangle\approx_{prbp}\langle Y'|E_2\rangle$;
  \item $u_{XY}$ contains a summand $b_1\leftmerge\cdots\leftmerge b_n$ iff $t_X$ contains the summand $b_1\leftmerge\cdots\leftmerge b_n$ and $t_Y$ contains the summand
  $b_1\leftmerge\cdots\leftmerge b_n$;
  \item $u_{XY}$ contains a summand $\tau \langle X'|E_1\rangle$ iff $XY\nequiv X_1Y_1$, $t_X$ contains the summand $\tau X'$, and
  $\langle X'|E_1\rangle\approx_{rbs}\langle Y|E_2\rangle$.
\end{enumerate}

Let the process term $s_{XY}$ be defined as follows:

\begin{enumerate}
  \item $s_{XY}\triangleq\tau\langle X|E_1\rangle + u_{XY}$ iff $XY\nequiv X_1Y_1$, $t_Y$ contains the summand $\tau Y'$, and $\langle X|E_1\rangle\approx_{prbp}\langle Y'|E_2\rangle$;
  \item $s_{XY}\triangleq\langle X|E_1\rangle$, otherwise.
\end{enumerate}

So, $\langle X|E_1\rangle=\langle X|E_1\rangle+u_{XY}$, and $(a_1\leftmerge\cdots\leftmerge a_m)(\tau\langle X|E_1\rangle+u_{XY})
=(a_1\leftmerge\cdots\leftmerge a_m)((\tau\langle X|E_1\rangle+u_{XY})+u_{XY})=(a_1\leftmerge\cdots\leftmerge a_m)(\langle X|E_1\rangle+u_{XY})
=(a_1\leftmerge\cdots\leftmerge a_m)\langle X|E_1\rangle$, hence, $(a_1\leftmerge\cdots\leftmerge a_m)s_{XY}=(a_1\leftmerge\cdots\leftmerge a_m)\langle X|E_1\rangle$.

Let $\sigma$ map recursion variable $X$ in $E_1$ to $\langle X|E_1\rangle$, and let $\psi$ map recursion variable $Z_{XY}$ in $E$ to $s_{XY}$. It is sufficient to prove
$s_{XY}=\psi(t_{XY})$ for recursion variables $Z_{XY}$ in $E$. Either $XY\equiv X_1Y_1$ or $XY\nequiv X_1Y_1$, we all can get $s_{XY}=\psi(t_{XY})$.
So, $s_{XY}=\langle Z_{XY}|E\rangle$ for recursive variables $Z_{XY}$ in $E$ is a solution for $E$. Then by $RSP$, particularly, $\langle X_1|E_1\rangle=\langle Z_{X_1Y_1}|E\rangle$.
Similarly, we can obtain $\langle Y_1|E_2\rangle=\langle Z_{X_1Y_1}|E\rangle$. Finally, $\langle X_1|E_1\rangle=\langle Z_{X_1Y_1}|E\rangle=\langle Y_1|E_2\rangle$, as desired.

(2) If $\langle X_1|E_1\rangle \approx_{prbs} \langle Y_1|E_2\rangle$ for guarded linear recursive specification $E_1$ and $E_2$, then $\langle X_1|E_1\rangle = \langle Y_1|E_2\rangle$.

It can be proven similarly to (1), we omit it.

(3) If $\langle X_1|E_1\rangle \approx_{prbhp} \langle Y_1|E_2\rangle$ for guarded linear recursive specification $E_1$ and $E_2$, then $\langle X_1|E_1\rangle = \langle Y_1|E_2\rangle$.

It can be proven similarly to (1), we omit it.

(4) If $\langle X_1|E_1\rangle \approx_{prbhhp} \langle Y_1|E_2\rangle$ for guarded linear recursive specification $E_1$ and $E_2$, then $\langle X_1|E_1\rangle = \langle Y_1|E_2\rangle$.

It can be proven similarly to (1), we omit it.
\end{proof}

\subsubsection{Abstraction}

The unary abstraction operator $\tau_I$ ($I\subseteq \mathbb{E}$) renames all atomic events in $I$ into $\tau$. $APPTC$ with silent step and abstraction operator is called $APPTC_{\tau}$.
The transition rules of operator $\tau_I$ are shown in Table \ref{TRForPAbstraction}.

\begin{center}
    \begin{table}
        $$\frac{x\rightsquigarrow x'}{\tau_I(x)\rightsquigarrow\tau_I(x')}$$
        $$\frac{x\xrightarrow{e}\surd}{\tau_I(x)\xrightarrow{e}\surd}\quad e\notin I
        \quad\frac{x\xrightarrow{e}x'}{\tau_I(x)\xrightarrow{e}\tau_I(x')}\quad e\notin I$$

        $$\frac{x\xrightarrow{e}\surd}{\tau_I(x)\xrightarrow{\tau}\surd}\quad e\in I
        \quad\frac{x\xrightarrow{e}x'}{\tau_I(x)\xrightarrow{\tau}\tau_I(x')}\quad e\in I$$
        \caption{Transition rules of the abstraction operator}
        \label{TRForPAbstraction}
    \end{table}
\end{center}

\begin{theorem}[Conservitivity of $APPTC_{\tau}$ with guarded linear recursion]
$APPTC_{\tau}$ with guarded linear recursion is a conservative extension of $APPTC$ with silent step and guarded linear recursion.
\end{theorem}

\begin{proof}
Since the transition rules of $APPTC$ with silent step and guarded linear recursion are source-dependent, and the transition rules for abstraction operator in Table
\ref{TRForPAbstraction} contain only a fresh operator $\tau_I$ in their source, so the transition rules of $APPTC_{\tau}$ with guarded linear recursion is a conservative extension of
those of $APPTC$ with silent step and guarded linear recursion.
\end{proof}

\begin{theorem}[Congruence theorem of $APPTC_{\tau}$ with guarded linear recursion]
Probabilistic rooted branching truly concurrent bisimulation equivalences $\approx_{prbp}$, $\approx_{prbs}$, $\approx_{prbhp}$ and $\approx_{prbhhp}$ are all congruences with respect
to $APPTC_{\tau}$ with guarded linear recursion.
\end{theorem}

\begin{proof}
It is easy to see that probabilistic rooted branching truly concurrent bisimulations $\approx_{prbp}$, $\approx_{prbs}$, $\approx_{prbhp}$ and $\approx_{prbhhp}$ are all equivalent
relations on $APPTC$ terms, we only need to prove that $\approx_{prbp}$, $\approx_{prbs}$, $\approx_{prbhp}$ and $\approx_{prbhhp}$ are preserved by the operators $\tau_I$.
That is, if $x\approx_{prbp} x'$, $x\approx_{prbs} x'$, $x\approx_{prbhp} x'$ and $x\approx_{prbhhp} x'$, we need to prove that $\tau_I(x)\approx_{prbp}\tau_I(x')$,
$\tau_I(x)\approx_{prbs}\tau_I(x')$, $\tau_I(x)\approx_{prbhp}\tau_I(x')$ and $\tau_I(x)\approx_{prbhhp}\tau_I(x')$. The proof is quite trivial and we omit it.
\end{proof}

We design the axioms for the abstraction operator $\tau_I$ in Table \ref{AxiomsForPAbstraction}.

\begin{center}
\begin{table}
  \begin{tabular}{@{}ll@{}}
\hline No. &Axiom\\
  $TI1$ & $e\notin I\quad \tau_I(e)=e$\\
  $TI2$ & $e\in I\quad \tau_I(e)=\tau$\\
  $TI3$ & $\tau_I(\delta)=\delta$\\
  $TI4$ & $\tau_I(x+y)=\tau_I(x)+\tau_I(y)$\\
  $PTI1$ & $\tau_I(x\boxplus_{\pi}y)=\tau_I(x)\boxplus_{\pi}\tau_I(y)$\\
  $TI5$ & $\tau_I(x\cdot y)=\tau_I(x)\cdot\tau_I(y)$\\
  $TI6$ & $\tau_I(x\leftmerge y)=\tau_I(x)\leftmerge\tau_I(y)$\\
\end{tabular}
\caption{Axioms of abstraction operator}
\label{AxiomsForPAbstraction}
\end{table}
\end{center}

The axioms $TI1-TI3$ are the defining laws for the abstraction operator $\tau_I$; $TI4-TI6$ and $PTI1$ say that in process term $\tau_I(t)$, all transitions of $t$ labeled with atomic
events from $I$ are renamed into $\tau$.

\begin{theorem}[Soundness of $APPTC_{\tau}$ with guarded linear recursion]\label{SAPPTCABS}
Let $x$ and $y$ be $APPTC_{\tau}$ with guarded linear recursion terms. If $APPTC_{\tau}$ with guarded linear recursion $\vdash x=y$, then
\begin{enumerate}
  \item $x\approx_{prbp} y$;
  \item $x\approx_{prbs} y$;
  \item $x\approx_{prbhp} y$;
  \item $x\approx_{prbhhp} y$.
\end{enumerate}
\end{theorem}

\begin{proof}
Since probabilistic rooted branching step bisimulations $\approx_{prbp}$, $\approx_{prbs}$, $\approx_{prbhp}$ and $\approx_{prbhhp}$ are all both equivalent and congruent relations
with respect to $APPTC_{\tau}$ with guarded linear recursion, we only need to check if each axiom in Table \ref{AxiomsForPAbstraction} is sound modulo $\approx_{prbp}$, $\approx_{prbs}$, $\approx_{prbhp}$ and $\approx_{prbhhp}$.
The proof is quite trivial and we omit it.
\end{proof}

Though $\tau$-loops are prohibited in guarded linear recursive specifications in a specifiable way, they can be constructed using the abstraction operator, for example, there exist
$\tau$-loops in the process term $\tau_{\{a\}}(\langle X|X=aX\rangle)$. To avoid $\tau$-loops caused by $\tau_I$ and ensure fairness, we introduce the following recursive verification
rules as Table \ref{RVR} shows, note that $i_1,\cdots, i_m,j_1,\cdots,j_n\in I\subseteq \mathbb{E}\setminus\{\tau\}$.

\begin{center}
\begin{table}
    $$VR_1\quad \frac{x=y+(i_1\leftmerge\cdots\leftmerge i_m)\cdot x, y=y+y}{\tau\cdot\tau_I(x)=\tau\cdot \tau_I(y)}$$
    $$VR_2\quad \frac{x=z\boxplus_{\pi}(u+(i_1\leftmerge\cdots\leftmerge i_m)\cdot x),z=z+u,z=z+z}{\tau\cdot\tau_I(x)=\tau\cdot\tau_I(z)}$$
    $$VR_3\quad \frac{x=z+(i_1\leftmerge\cdots\leftmerge i_m)\cdot y,y=z\boxplus_{\pi}(u+(j_1\leftmerge\cdots\leftmerge j_n)\cdot x), z=z+u,z=z+z}{\tau\cdot\tau_I(x)=\tau\cdot\tau_I(y')\textrm{ for }y'=z\boxplus_{\pi}(u+(i_1\leftmerge\cdots\leftmerge i_m)\cdot y')}$$
\caption{Recursive verification rules}
\label{RVR}
\end{table}
\end{center}

\begin{theorem}[Soundness of $VR_1,VR_2,VR_3$]
$VR_1$, $VR_2$ and $VR_3$ are sound modulo probabilistic rooted branching truly concurrent bisimulation equivalences $\approx_{prbp}$, $\approx_{prbs}$, $\approx_{prbhp}$ and $\approx_{prbhhp}$.
\end{theorem}

\newpage\section{Mobility}\label{piptc}

In this chapter, we design a calculus of probabilistic truly concurrent mobile processes ($\pi_{ptc}$). This chapter is organized as follows. We introduce the syntax and operational
semantics of $\pi_{ptc}$ in section \ref{sos5}, its properties for strongly probabilistic truly concurrent bisimulations in section \ref{stcb5}, its axiomatization in section \ref{at5}.

\subsection{Syntax and Operational Semantics}\label{sos5}

We assume an infinite set $\mathcal{N}$ of (action or event) names, and use $a,b,c,\cdots$ to range over $\mathcal{N}$, use $x,y,z,w,u,v$ as meta-variables over names. We denote by 
$\overline{\mathcal{N}}$ the set of co-names and let $\overline{a},\overline{b},\overline{c},\cdots$ range over $\overline{\mathcal{N}}$. Then we set 
$\mathcal{L}=\mathcal{N}\cup\overline{\mathcal{N}}$ as the set of labels, and use $l,\overline{l}$ to range over $\mathcal{L}$. We extend complementation to $\mathcal{L}$ such that 
$\overline{\overline{a}}=a$. Let $\tau$ denote the silent step (internal action or event) and define $Act=\mathcal{L}\cup\{\tau\}$ to be the set of actions, $\alpha,\beta$ range over 
$Act$. And $K,L$ are used to stand for subsets of $\mathcal{L}$ and $\overline{L}$ is used for the set of complements of labels in $L$.

Further, we introduce a set $\mathcal{X}$ of process variables, and a set $\mathcal{K}$ of process constants, and let $X,Y,\cdots$ range over $\mathcal{X}$, and $A,B,\cdots$ range over 
$\mathcal{K}$. For each process constant $A$, a nonnegative arity $ar(A)$ is assigned to it. Let $\widetilde{x}=x_1,\cdots,x_{ar(A)}$ be a tuple of distinct name variables, then 
$A(\widetilde{x})$ is called a process constant. $\widetilde{X}$ is a tuple of distinct process variables, and also $E,F,\cdots$ range over the recursive expressions. We write 
$\mathcal{P}$ for the set of processes. Sometimes, we use $I,J$ to stand for an indexing set, and we write $E_i:i\in I$ for a family of expressions indexed by $I$. $Id_D$ is the 
identity function or relation over set $D$. The symbol $\equiv_{\alpha}$ denotes equality under standard alpha-convertibility, note that the subscript $\alpha$ has no relation to the 
action $\alpha$.

\subsubsection{Syntax}

We use the Prefix $.$ to model the causality relation $\leq$ in true concurrency, the Summation $+$ to model the conflict relation $\sharp$, and $\boxplus_{\pi}$ to model the probabilistic
conflict relation $\sharp_{\pi}$ in probabilistic true concurrency, and the Composition $\parallel$ to explicitly model concurrent relation in true concurrency. And we follow the 
conventions of process algebra.

\begin{definition}[Syntax]\label{syntax5}
A truly concurrent process $P$ is defined inductively by the following formation rules:

\begin{enumerate}
  \item $A(\widetilde{x})\in\mathcal{P}$;
  \item $\textbf{nil}\in\mathcal{P}$;
  \item if $P\in\mathcal{P}$, then the Prefix $\tau.P\in\mathcal{P}$, for $\tau\in Act$ is the silent action;
  \item if $P\in\mathcal{P}$, then the Output $\overline{x}y.P\in\mathcal{P}$, for $x,y\in Act$;
  \item if $P\in\mathcal{P}$, then the Input $x(y).P\in\mathcal{P}$, for $x,y\in Act$;
  \item if $P\in\mathcal{P}$, then the Restriction $(x)P\in\mathcal{P}$, for $x\in Act$;
%  \item if $P\in\mathcal{P}$, then the Match $[x=y]P\in\mathcal{P}$, for $x,y\in Act$;
  \item if $P,Q\in\mathcal{P}$, then the Summation $P+Q\in\mathcal{P}$;
  \item if $P,Q\in\mathcal{P}$, then the Summation $P\boxplus_{\pi}Q\in\mathcal{P}$;
  \item if $P,Q\in\mathcal{P}$, then the Composition $P\parallel Q\in\mathcal{P}$;
\end{enumerate}

The standard BNF grammar of syntax of $\pi_{tc}$ can be summarized as follows:

$$P::=A(\widetilde{x})\quad|\quad\textbf{nil}\quad|\quad\tau.P\quad|\quad \overline{x}y.P\quad |\quad x(y).P \quad|\quad (x)P \quad | \quad P+P\quad|\quad P\boxplus_{\pi}P\quad |\quad P\parallel P.$$
\end{definition}

In $\overline{x}y$, $x(y)$ and $\overline{x}(y)$, $x$ is called the subject, $y$ is called the object and it may be free or bound.

\begin{definition}[Free variables]
The free names of a process $P$, $fn(P)$, are defined as follows.

\begin{enumerate}
  \item $fn(A(\widetilde{x}))\subseteq\{\widetilde{x}\}$;
  \item $fn(\textbf{nil})=\emptyset$;
  \item $fn(\tau.P)=fn(P)$;
  \item $fn(\overline{x}y.P)=fn(P)\cup\{x\}\cup\{y\}$;
  \item $fn(x(y).P)=fn(P)\cup\{x\}-\{y\}$;
  \item $fn((x)P)=fn(P)-\{x\}$;
%  \item $fn([x=y]P)=fn(P)$;
  \item $fn(P+Q)=fn(P)\cup fn(Q)$;
  \item $fn(P\boxplus_{\pi}Q)=fn(P)\cup fn(Q)$;
  \item $fn(P\parallel Q)=fn(P)\cup fn(Q)$.
\end{enumerate}
\end{definition}

\begin{definition}[Bound variables]
Let $n(P)$ be the names of a process $P$, then the bound names $bn(P)=n(P)-fn(P)$.
\end{definition}

For each process constant schema $A(\widetilde{x})$, a defining equation of the form

$$A(\widetilde{x})\overset{\text{def}}{=}P$$

is assumed, where $P$ is a process with $fn(P)\subseteq \{\widetilde{x}\}$.

\begin{definition}[Substitutions]\label{subs5}
A substitution is a function $\sigma:\mathcal{N}\rightarrow\mathcal{N}$. For $x_i\sigma=y_i$ with $1\leq i\leq n$, we write $\{y_1/x_1,\cdots,y_n/x_n\}$ or 
$\{\widetilde{y}/\widetilde{x}\}$ for $\sigma$. For a process $P\in\mathcal{P}$, $P\sigma$ is defined inductively as follows:
\begin{enumerate}
  \item if $P$ is a process constant $A(\widetilde{x})=A(x_1,\cdots,x_n)$, then $P\sigma=A(x_1\sigma,\cdots,x_n\sigma)$;
  \item if $P=\textbf{nil}$, then $P\sigma=\textbf{nil}$;
  \item if $P=\tau.P'$, then $P\sigma=\tau.P'\sigma$;
  \item if $P=\overline{x}y.P'$, then $P\sigma=\overline{x\sigma}y\sigma.P'\sigma$;
  \item if $P=x(y).P'$, then $P\sigma=x\sigma(y).P'\sigma$;
  \item if $P=(x)P'$, then $P\sigma=(x\sigma)P'\sigma$;
%  \item if $P=[x=y]P'$, then $P\sigma=[x\sigma=y\sigma]P'\sigma$;
  \item if $P=P_1+P_2$, then $P\sigma=P_1\sigma+P_2\sigma$;
  \item if $P=P_1\boxplus_{\pi}P_2$, then $P\sigma=P_1\sigma\boxplus_{\pi}P_2\sigma$;
  \item if $P=P_1\parallel P_2$, then $P\sigma=P_1\sigma \parallel P_2\sigma$.
\end{enumerate}
\end{definition}

\subsubsection{Operational Semantics}

The operational semantics is defined by LTSs (labelled transition systems), and it is detailed by the following definition.

\begin{definition}[Semantics]\label{semantics5}
The operational semantics of $\pi_{tc}$ corresponding to the syntax in Definition \ref{syntax5} is defined by a series of transition rules, named $\textbf{PACT}$, $\textbf{PSUM}$, $\textbf{PBOX-SUM}$
$\textbf{PIDE}$, $\textbf{PPAR}$, $\textbf{PRES}$ and named $\textbf{ACT}$, $\textbf{SUM}$, 
$\textbf{IDE}$, $\textbf{PAR}$, $\textbf{COM}$, $\textbf{CLOSE}$, $\textbf{RES}$, $\textbf{OPEN}$ indicate that the rules are associated respectively with Prefix, Summation, Box-Summation,
Identity, Parallel Composition, Communication, and Restriction in Definition \ref{syntax5}. They are shown in Table \ref{PTRForPITC5} and \ref{TRForPITC5}.

\begin{center}
    \begin{table}
        \[\textbf{PTAU-ACT}\quad \frac{}{\tau.P\rightsquigarrow \breve{\tau}.P} \quad \textbf{POUTPUT-ACT}\quad \frac{}{\overline{x}y.P\rightsquigarrow \breve{\overline{x}y}.P}\]

        \[\textbf{PINPUT-ACT}\quad \frac{}{x(z).P\rightsquigarrow \breve{x(z)}.P}\]

%        \[\textbf{PPAR}_1\quad \frac{P\rightsquigarrow P'\quad Q\xnrsquigarrow{ }}{P\parallel Q\rightsquigarrow P'\parallel Q} \quad \textbf{PPAR}_2\quad \frac{Q\rightsquigarrow Q'\quad P\xnrsquigarrow{ }}{P\parallel Q\rightsquigarrow P\parallel Q'}\]

        \[\textbf{PPAR}\quad \frac{P\rightsquigarrow P'\quad Q\rightsquigarrow Q'}{P\parallel Q\rightsquigarrow P'\parallel Q'}\]

        \[\textbf{PSUM}\quad \frac{P\rightsquigarrow P'\quad Q\rightsquigarrow Q'}{P+Q\rightsquigarrow P'+Q'}\]
        
        \[\textbf{PBOX-SUM}\quad \frac{P\rightsquigarrow P'}{P\boxplus_{\pi}Q\rightsquigarrow P'}\]
%        \[\textbf{MATCH}_1\quad \frac{P\xrightarrow{\alpha}P'}{[x=x]P\xrightarrow{\alpha}P'} \quad
%        \textbf{MATCH}_2\quad \frac{P\xrightarrow{\{\alpha_1,\cdots,\alpha_n\}}P'}{[x=x]P\xrightarrow{\{\alpha_1,\cdots,\alpha_n\}}P'}\]

        \[\textbf{PIDE}\quad\frac{P\{\widetilde{y}/\widetilde{x}\}\rightsquigarrow P'}{A(\widetilde{y})\rightsquigarrow P'}\quad (A(\widetilde{x})\overset{\text{def}}{=}P)\]

        \[\textbf{PRES}\quad \frac{P\rightsquigarrow P'}{(y)P\rightsquigarrow (y)P'}\quad (y\notin n(\alpha))\]

        \caption{Probabilistic transition rules of $\pi_{ptc}$}
        \label{PTRForPITC5}
    \end{table}
\end{center}

\begin{center}
    \begin{table}
        \[\textbf{TAU-ACT}\quad \frac{}{\tau.P\xrightarrow{\tau}P} \quad \textbf{OUTPUT-ACT}\quad \frac{}{\overline{x}y.P\xrightarrow{\overline{x}y}P}\]

        \[\textbf{INPUT-ACT}\quad \frac{}{x(z).P\xrightarrow{x(w)}P\{w/z\}}\quad (w\notin fn((z)P))\]

        \[\textbf{PAR}_1\quad \frac{P\xrightarrow{\alpha}P'\quad Q\nrightarrow}{P\parallel Q\xrightarrow{\alpha}P'\parallel Q}\quad (bn(\alpha)\cap fn(Q)=\emptyset) \quad \textbf{PAR}_2\quad \frac{Q\xrightarrow{\alpha}Q'\quad P\nrightarrow}{P\parallel Q\xrightarrow{\alpha}P\parallel Q'}\quad (bn(\alpha)\cap fn(P)=\emptyset)\]

        \[\textbf{PAR}_3\quad \frac{P\xrightarrow{\alpha}P'\quad Q\xrightarrow{\beta}Q'}{P\parallel Q\xrightarrow{\{\alpha,\beta\}}P'\parallel Q'}\quad (\beta\neq\overline{\alpha}, bn(\alpha)\cap bn(\beta)=\emptyset, bn(\alpha)\cap fn(Q)=\emptyset,bn(\beta)\cap fn(P)=\emptyset)\]

        \[\textbf{PAR}_4\quad \frac{P\xrightarrow{x_1(z)}P'\quad Q\xrightarrow{x_2(z)}Q'}{P\parallel Q\xrightarrow{\{x_1(w),x_2(w)\}}P'\{w/z\}\parallel Q'\{w/z\}}\quad (w\notin fn((z)P)\cup fn((z)Q))\]

        \[\textbf{COM}\quad \frac{P\xrightarrow{\overline{x}y}P'\quad Q\xrightarrow{x(z)}Q'}{P\parallel Q\xrightarrow{\tau}P'\parallel Q'\{y/z\}}\]

        \[\textbf{CLOSE}\quad \frac{P\xrightarrow{\overline{x}(w)}P'\quad Q\xrightarrow{x(w)}Q'}{P\parallel Q\xrightarrow{\tau}(w)(P'\parallel Q')}\]

        \[\textbf{SUM}_1\quad \frac{P\xrightarrow{\alpha}P'}{P+Q\xrightarrow{\alpha}P'} \quad \textbf{SUM}_2\quad \frac{P\xrightarrow{\{\alpha_1,\cdots,\alpha_n\}}P'}{P+Q\xrightarrow{\{\alpha_1,\cdots,\alpha_n\}}P'}\]

%        \[\textbf{MATCH}_1\quad \frac{P\xrightarrow{\alpha}P'}{[x=x]P\xrightarrow{\alpha}P'} \quad
%        \textbf{MATCH}_2\quad \frac{P\xrightarrow{\{\alpha_1,\cdots,\alpha_n\}}P'}{[x=x]P\xrightarrow{\{\alpha_1,\cdots,\alpha_n\}}P'}\]

        \[\textbf{IDE}_1\quad\frac{P\{\widetilde{y}/\widetilde{x}\}\xrightarrow{\alpha}P'}{A(\widetilde{y})\xrightarrow{\alpha}P'}\quad (A(\widetilde{x})\overset{\text{def}}{=}P) \quad \textbf{IDE}_2\quad\frac{P\{\widetilde{y}/\widetilde{x}\}\xrightarrow{\{\alpha_1,\cdots,\alpha_n\}}P'} {A(\widetilde{y})\xrightarrow{\{\alpha_1,\cdots,\alpha_n\}}P'}\quad (A(\widetilde{x})\overset{\text{def}}{=}P)\]

        \[\textbf{RES}_1\quad \frac{P\xrightarrow{\alpha}P'}{(y)P\xrightarrow{\alpha}(y)P'}\quad (y\notin n(\alpha)) \quad \textbf{RES}_2\quad \frac{P\xrightarrow{\{\alpha_1,\cdots,\alpha_n\}}P'}{(y)P\xrightarrow{\{\alpha_1,\cdots,\alpha_n\}}(y)P'}\quad (y\notin n(\alpha_1)\cup\cdots\cup n(\alpha_n))\]

        \[\textbf{OPEN}_1\quad \frac{P\xrightarrow{\overline{x}y}P'}{(y)P\xrightarrow{\overline{x}(w)}P'\{w/y\}} \quad (y\neq x, w\notin fn((y)P'))\]

        \[\textbf{OPEN}_2\quad \frac{P\xrightarrow{\{\overline{x}_1 y,\cdots,\overline{x}_n y\}}P'}{(y)P\xrightarrow{\{\overline{x}_1(w),\cdots,\overline{x}_n(w)\}}P'\{w/y\}} \quad (y\neq x_1\neq\cdots\neq x_n, w\notin fn((y)P'))\]

        \caption{Action transition rules of $\pi_{ptc}$}
        \label{TRForPITC5}
    \end{table}
\end{center}
\end{definition}

\subsubsection{Properties of Transitions}

\begin{proposition}
\begin{enumerate}
  \item If $P\xrightarrow{\alpha}P'$ then
  \begin{enumerate}
    \item $fn(\alpha)\subseteq fn(P)$;
    \item $fn(P')\subseteq fn(P)\cup bn(\alpha)$;
  \end{enumerate}
  \item If $P\xrightarrow{\{\alpha_1,\cdots,\alpha_n\}}P'$ then
  \begin{enumerate}
    \item $fn(\alpha_1)\cup\cdots\cup fn(\alpha_n)\subseteq fn(P)$;
    \item $fn(P')\subseteq fn(P)\cup bn(\alpha_1)\cup\cdots\cup bn(\alpha_n)$.
  \end{enumerate}
\end{enumerate}
\end{proposition}

\begin{proof}
By induction on the depth of inference.
\end{proof}

\begin{proposition}
Suppose that $P\xrightarrow{\alpha(y)}P'$, where $\alpha=x$ or $\alpha=\overline{x}$, and $x\notin n(P)$, then there exists some $P''\equiv_{\alpha}P'\{z/y\}$, 
$P\xrightarrow{\alpha(z)}P''$.
\end{proposition}

\begin{proof}
By induction on the depth of inference.
\end{proof}

\begin{proposition}
If $P\rightarrow P'$, $bn(\alpha)\cap fn(P'\sigma)=\emptyset$, and $\sigma\lceil bn(\alpha)=id$, then there exists some $P''\equiv_{\alpha}P'\sigma$, 
$P\sigma\xrightarrow{\alpha\sigma}P''$.
\end{proposition}

\begin{proof}
By the definition of substitution (Definition \ref{subs5}) and induction on the depth of inference.
\end{proof}

\begin{proposition}
\begin{enumerate}
  \item If $P\{w/z\}\xrightarrow{\alpha}P'$, where $w\notin fn(P)$ and $bn(\alpha)\cap fn(P,w)=\emptyset$, then there exist some $Q$ and $\beta$ with $Q\{w/z\}\equiv_{\alpha}P'$ and 
  $\beta\sigma=\alpha$, $P\xrightarrow{\beta}Q$;
  \item If $P\{w/z\}\xrightarrow{\{\alpha_1,\cdots,\alpha_n\}}P'$, where $w\notin fn(P)$ and $bn(\alpha_1)\cap\cdots\cap bn(\alpha_n)\cap fn(P,w)=\emptyset$, then there exist some $Q$ 
  and $\beta_1,\cdots,\beta_n$ with $Q\{w/z\}\equiv_{\alpha}P'$ and $\beta_1\sigma=\alpha_1,\cdots,\beta_n\sigma=\alpha_n$, $P\xrightarrow{\{\beta_1,\cdots,\beta_n\}}Q$.
\end{enumerate}

\end{proposition}

\begin{proof}
By the definition of substitution (Definition \ref{subs5}) and induction on the depth of inference.
\end{proof}

\subsection{Strongly Probabilistic Truly Concurrent Bisimilarities}\label{stcb5}

\subsubsection{Basic Definitions}\label{STCC5}

Firstly, in this subsection, we introduce concepts of (strongly) probabilistic truly concurrent bisimilarities, including probabilistic pomset bisimilarity, probabilistic step 
bisimilarity, probabilistic history-preserving (hp-)bisimilarity and probabilistic hereditary history-preserving (hhp-)bisimilarity. In contrast to traditional probabilistic truly 
concurrent bisimilarities in section \ref{bg}, these versions in $\pi_{ptc}$ must take care of actions with bound objects. Note that, these probabilistic truly concurrent bisimilarities 
are defined as late bisimilarities, but not early bisimilarities, as defined in $\pi$-calculus \cite{PI1} \cite{PI2}. Note that, here, a PES $\mathcal{E}$ is deemed as a process.

\begin{definition}[Strongly probabilistic pomset, step bisimilarity 2]\label{PSB5}
Let $\mathcal{E}_1$, $\mathcal{E}_2$ be PESs. A strongly probabilistic pomset bisimulation is a relation $R\subseteq\mathcal{C}(\mathcal{E}_1)\times\mathcal{C}(\mathcal{E}_2)$, such that 
(1) if $(C_1,C_2)\in R$, and $C_1\xrightarrow{X_1}C_1'$ (with $\mathcal{E}_1\xrightarrow{X_1}\mathcal{E}_1'$) then $C_2\xrightarrow{X_2}C_2'$ (with 
$\mathcal{E}_2\xrightarrow{X_2}\mathcal{E}_2'$), with $X_1\subseteq \mathbb{E}_1$, $X_2\subseteq \mathbb{E}_2$, $X_1\sim X_2$ and $(C_1',C_2')\in R$:
\begin{enumerate}
  \item for each fresh action $\alpha\in X_1$, if $C_1''\xrightarrow{\alpha}C_1'''$ (with $\mathcal{E}_1''\xrightarrow{\alpha}\mathcal{E}_1'''$), then for some $C_2''$ and $C_2'''$, $C_2''\xrightarrow{\alpha}C_2'''$ (with $\mathcal{E}_2''\xrightarrow{\alpha}\mathcal{E}_2'''$), such that if $(C_1'',C_2'')\in R$ then $(C_1''',C_2''')\in R$;
  \item for each $x(y)\in X_1$ with ($y\notin n(\mathcal{E}_1, \mathcal{E}_2)$), if $C_1''\xrightarrow{x(y)}C_1'''$ (with $\mathcal{E}_1''\xrightarrow{x(y)}\mathcal{E}_1'''\{w/y\}$) for all $w$, then for some $C_2''$ and $C_2'''$, $C_2''\xrightarrow{x(y)}C_2'''$ (with $\mathcal{E}_2''\xrightarrow{x(y)}\mathcal{E}_2'''\{w/y\}$) for all $w$, such that if $(C_1'',C_2'')\in R$ then $(C_1''',C_2''')\in R$;
  \item for each two $x_1(y),x_2(y)\in X_1$ with ($y\notin n(\mathcal{E}_1, \mathcal{E}_2)$), if $C_1''\xrightarrow{\{x_1(y),x_2(y)\}}C_1'''$ (with $\mathcal{E}_1''\xrightarrow{\{x_1(y),x_2(y)\}}\mathcal{E}_1'''\{w/y\}$) for all $w$, then for some $C_2''$ and $C_2'''$, $C_2''\xrightarrow{\{x_1(y),x_2(y)\}}C_2'''$ (with $\mathcal{E}_2''\xrightarrow{\{x_1(y),x_2(y)\}}\mathcal{E}_2'''\{w/y\}$) for all $w$, such that if $(C_1'',C_2'')\in R$ then $(C_1''',C_2''')\in R$;
  \item for each $\overline{x}(y)\in X_1$ with $y\notin n(\mathcal{E}_1, \mathcal{E}_2)$, if $C_1''\xrightarrow{\overline{x}(y)}C_1'''$ (with $\mathcal{E}_1''\xrightarrow{\overline{x}(y)}\mathcal{E}_1'''$), then for some $C_2''$ and $C_2'''$, $C_2''\xrightarrow{\overline{x}(y)}C_2'''$ (with $\mathcal{E}_2''\xrightarrow{\overline{x}(y)}\mathcal{E}_2'''$), such that if $(C_1'',C_2'')\in R$ then $(C_1''',C_2''')\in R$.
\end{enumerate}
 and vice-versa; (2) if $(C_1,C_2)\in R$, and $C_1\xrsquigarrow{\pi}C_1^{\pi}$ then $C_2\xrsquigarrow{\pi}C_2^{\pi}$ and $(C_1^{\pi},C_2^{\pi})\in R$, and vice-versa; (3) if $(C_1,C_2)\in R$,
then $\mu(C_1,C)=\mu(C_2,C)$ for each $C\in\mathcal{C}(\mathcal{E})/R$; (4) $[\surd]_R=\{\surd\}$.

We say that $\mathcal{E}_1$, $\mathcal{E}_2$ are strongly probabilistic pomset bisimilar, written $\mathcal{E}_1\sim_{pp}\mathcal{E}_2$, if there exists a strongly probabilistic pomset 
bisimulation $R$, such that $(\emptyset,\emptyset)\in R$. By replacing probabilistic pomset transitions with steps, we can get the definition of strongly probabilistic step bisimulation. 
When PESs $\mathcal{E}_1$ and $\mathcal{E}_2$ are strongly probabilistic step bisimilar, we write $\mathcal{E}_1\sim_{ps}\mathcal{E}_2$.

%Note that, we do not admit multiple input with the same input variable in concurrency, that is, $\xrightarrow{\{x_1(y),x_2(y)\}}$ is forbidden. We assume that in each $x(y)\in X_1$, all $y$ are distinct, or all $x_1(y), x_2(y)\in X_1$ are in causality. For all $w_1$ and $w_2$, $P\nrightarrow^{\{x_1(y),x_2(y)\}}P'\{w_1/y\}\{w_2/y\}$, makes $P'\{w_1/y\}\{w_2/y\}$ do not occur. Actually, if we assume that all concurrent actions are put into parallel explicitly, then the side condition of the transition rule $\textbf{PAR}_3$ in Table \ref{TRForPITC} ($bn(\alpha)\cap fn(Q)=\emptyset,bn(\beta)\cap fn(P)=\emptyset$) also ensures that $\nrightarrow^{\{x_1(y),x_2(y)\}}$. This principle is throughout this paper, and we do not mention any more.
\end{definition}

\begin{definition}[Strongly probabilistic (hereditary) history-preserving bisimilarity 2]\label{HHPB5}
A strongly probabilistic history-preserving (hp-) bisimulation is a posetal relation $R\subseteq\mathcal{C}(\mathcal{E}_1)\overline{\times}\mathcal{C}(\mathcal{E}_2)$ such that (1) if 
$(C_1,f,C_2)\in R$, and
\begin{enumerate}
  \item for $e_1=\alpha$ a fresh action, if $C_1\xrightarrow{\alpha}C_1'$ (with $\mathcal{E}_1\xrightarrow{\alpha}\mathcal{E}_1'$), then for some $C_2'$ and $e_2=\alpha$, $C_2\xrightarrow{\alpha}C_2'$ (with $\mathcal{E}_2\xrightarrow{\alpha}\mathcal{E}_2'$), such that $(C_1',f[e_1\mapsto e_2],C_2')\in R$;
  \item for $e_1=x(y)$ with ($y\notin n(\mathcal{E}_1, \mathcal{E}_2)$), if $C_1\xrightarrow{x(y)}C_1'$ (with $\mathcal{E}_1\xrightarrow{x(y)}\mathcal{E}_1'\{w/y\}$) for all $w$, then for some $C_2'$ and $e_2=x(y)$, $C_2\xrightarrow{x(y)}C_2'$ (with $\mathcal{E}_2\xrightarrow{x(y)}\mathcal{E}_2'\{w/y\}$) for all $w$, such that $(C_1',f[e_1\mapsto e_2],C_2')\in R$;
  \item for $e_1=\overline{x}(y)$ with $y\notin n(\mathcal{E}_1, \mathcal{E}_2)$, if $C_1\xrightarrow{\overline{x}(y)}C_1'$ (with $\mathcal{E}_1\xrightarrow{\overline{x}(y)}\mathcal{E}_1'$), then for some $C_2'$ and $e_2=\overline{x}(y)$, $C_2\xrightarrow{\overline{x}(y)}C_2'$ (with $\mathcal{E}_2\xrightarrow{\overline{x}(y)}\mathcal{E}_2'$), such that $(C_1',f[e_1\mapsto e_2],C_2')\in R$.
\end{enumerate}

and vice-versa; (2) if $(C_1,f,C_2)\in R$, and $C_1\xrsquigarrow{\pi}C_1^{\pi}$ then $C_2\xrsquigarrow{\pi}C_2^{\pi}$ and $(C_1^{\pi},f,C_2^{\pi})\in R$, and vice-versa; (3) if 
$(C_1,f,C_2)\in R$, then $\mu(C_1,C)=\mu(C_2,C)$ for each $C\in\mathcal{C}(\mathcal{E})/R$; (4) $[\surd]_R=\{\surd\}$. $\mathcal{E}_1,\mathcal{E}_2$ are strongly probabilistic 
history-preserving (hp-)bisimilar and are written $\mathcal{E}_1\sim_{php}\mathcal{E}_2$ if there exists a strongly probabilistic hp-bisimulation $R$ such that 
$(\emptyset,\emptyset,\emptyset)\in R$.

A strongly probabilistic hereditary history-preserving (hhp-)bisimulation is a downward closed strongly probabilistic hp-bisimulation. $\mathcal{E}_1,\mathcal{E}_2$ are strongly probabilistic 
hereditary history-preserving (hhp-)bisimilar and are written $\mathcal{E}_1\sim_{phhp}\mathcal{E}_2$.
\end{definition}

\begin{theorem}
$\equiv_{\alpha}$ are strongly probabilistic truly concurrent bisimulations. That is, if $P\equiv_{\alpha}Q$, then,
\begin{enumerate}
  \item $P\sim_{pp} Q$;
  \item $P\sim_{ps} Q$;
  \item $P\sim_{php} Q$;
  \item $P\sim_{phhp} Q$.
\end{enumerate}
\end{theorem}

\begin{proof}
By induction on the depth of inference (see Table \ref{TRForPITC5}), we can get the following facts:

\begin{enumerate}
  \item If $\alpha$ is a free action and $P\rightsquigarrow\xrightarrow{\alpha}P'$, then equally for some $Q'$ with $P'\equiv_{\alpha}Q'$, $Q\rightsquigarrow\xrightarrow{\alpha}Q'$;
  \item If $P\rightsquigarrow\xrightarrow{a(y)}P'$ with $a=x$ or $a=\overline{x}$ and $z\notin n(Q)$, then equally for some $Q'$ with $P'\{z/y\}\equiv_{\alpha}Q'$, 
  $Q\rightsquigarrow\xrightarrow{a(z)}Q'$.
\end{enumerate}

Then, we can get:

\begin{enumerate}
  \item by the definition of strongly probabilistic pomset bisimilarity (Definition \ref{PSB5}), $P\sim_{pp} Q$;
  \item by the definition of strongly probabilistic step bisimilarity (Definition \ref{PSB5}), $P\sim_{ps} Q$;
  \item by the definition of strongly probabilistic hp-bisimilarity (Definition \ref{HHPB5}), $P\sim_{php} Q$;
  \item by the definition of strongly probabilistic hhp-bisimilarity (Definition \ref{HHPB5}), $P\sim_{phhp} Q$.
\end{enumerate}
\end{proof}

\subsubsection{Laws and Congruence}

Similarly to CPTC, we can obtain the following laws with respect to probabilistic truly concurrent bisimilarities.

\begin{theorem}[Summation laws for strongly probabilistic pomset bisimilarity]
The summation laws for strongly probabilistic pomset bisimilarity are as follows.

\begin{enumerate}
  \item $P+\textbf{nil}\sim_{pp} P$;
  \item $P+P\sim_{pp} P$;
  \item $P_1+P_2\sim_{pp} P_2+P_1$;
  \item $P_1+(P_2+P_3)\sim_{pp} (P_1+P_2)+P_3$.
\end{enumerate}
\end{theorem}

\begin{proof}
According to the definition of strongly probabilistic pomset bisimulation, we can easily prove the above equations, and we omit the proof.
\end{proof}

\begin{theorem}[Summation laws for strongly probabilistic step bisimilarity]
The summation laws for strongly probabilistic step bisimilarity are as follows.

\begin{enumerate}
  \item $P+\textbf{nil}\sim_{ps} P$;
  \item $P+P\sim_{ps} P$;
  \item $P_1+P_2\sim_{ps} P_2+P_1$;
  \item $P_1+(P_2+P_3)\sim_{ps} (P_1+P_2)+P_3$.
\end{enumerate}
\end{theorem}

\begin{proof}
According to the definition of strongly probabilistic step bisimulation, we can easily prove the above equations, and we omit the proof.
\end{proof}

\begin{theorem}[Summation laws for strongly probabilistic hp-bisimilarity]
The summation laws for strongly probabilistic hp-bisimilarity are as follows.

\begin{enumerate}
  \item $P+\textbf{nil}\sim_{php} P$;
  \item $P+P\sim_{php} P$;
  \item $P_1+P_2\sim_{php} P_2+P_1$;
  \item $P_1+(P_2+P_3)\sim_{php} (P_1+P_2)+P_3$.
\end{enumerate}
\end{theorem}

\begin{proof}
According to the definition of strongly probabilistic hp-bisimulation, we can easily prove the above equations, and we omit the proof.
\end{proof}

\begin{theorem}[Summation laws for strongly probabilistic hhp-bisimilarity]
The summation laws for strongly probabilistic hhp-bisimilarity are as follows.

\begin{enumerate}
  \item $P+\textbf{nil}\sim_{phhp} P$;
  \item $P+P\sim_{phhp} P$;
  \item $P_1+P_2\sim_{phhp} P_2+P_1$;
  \item $P_1+(P_2+P_3)\sim_{phhp} (P_1+P_2)+P_3$.
\end{enumerate}
\end{theorem}

\begin{proof}
According to the definition of strongly probabilistic hhp-bisimulation, we can easily prove the above equations, and we omit the proof.
\end{proof}

\begin{proposition}[Box-Summation laws for strongly probabilistic pomset bisimulation] 
The box-summation laws for strongly probabilistic pomset bisimulation are as follows.

\begin{enumerate}
  \item $P\boxplus_{\pi}\textbf{nil}\sim_{pp} P$.
  \item $P\boxplus_{\pi}P\sim_{pp} P$\\
  \item $P_1\boxplus_{\pi} P_2\sim_{pp} P_2\boxplus_{1-\pi} P_1$\\
  \item $P_1\boxplus_{\pi}(P_2\boxplus_{\rho} P_3)\sim_{pp} (P_1\boxplus_{\frac{\pi}{\pi+\rho-\pi\rho}}P_2)\boxplus_{\pi+\rho-\pi\rho} P_3$\\
\end{enumerate}
\end{proposition}

\begin{proof}
According to the definition of strongly probabilistic pomset bisimulation, we can easily prove the above equations, and we omit the proof.
\end{proof}

\begin{proposition}[Box-Summation laws for strongly probabilistic step bisimulation]
The box-summation laws for strongly probabilistic step bisimulation are as follows.

\begin{enumerate}
  \item $P\boxplus_{\pi}\textbf{nil}\sim_{ps} P$.
  \item $P\boxplus_{\pi}P\sim_{ps} P$\\
  \item $P_1\boxplus_{\pi} P_2\sim_{ps} P_2\boxplus_{1-\pi} P_1$\\
  \item $P_1\boxplus_{\pi}(P_2\boxplus_{\rho} P_3)\sim_{ps} (P_1\boxplus_{\frac{\pi}{\pi+\rho-\pi\rho}}P_2)\boxplus_{\pi+\rho-\pi\rho} P_3$\\
\end{enumerate}
\end{proposition}

\begin{proof}
According to the definition of strongly probabilistic step bisimulation, we can easily prove the above equations, and we omit the proof.
\end{proof}

\begin{proposition}[Box-Summation laws for strongly probabilistic hp-bisimulation]
The box-summation laws for strongly probabilistic hp-bisimulation are as follows.

\begin{enumerate}
  \item $P\boxplus_{\pi}\textbf{nil}\sim_{php} P$.
  \item $P\boxplus_{\pi}P\sim_{php} P$\\
  \item $P_1\boxplus_{\pi} P_2\sim_{php} P_2\boxplus_{1-\pi} P_1$\\
  \item $P_1\boxplus_{\pi}(P_2\boxplus_{\rho} P_3)\sim_{php} (P_1\boxplus_{\frac{\pi}{\pi+\rho-\pi\rho}}P_2)\boxplus_{\pi+\rho-\pi\rho} P_3$\\
\end{enumerate}
\end{proposition}

\begin{proof}
According to the definition of strongly probabilistic hp-bisimulation, we can easily prove the above equations, and we omit the proof.
\end{proof}

\begin{proposition}[Box-Summation laws for strongly probabilistic hhp-bisimulation]
The box-summation laws for strongly probabilistic hhp-bisimulation are as follows.

\begin{enumerate}
  \item $P\boxplus_{\pi}\textbf{nil}\sim_{phhp} P$.
  \item $P\boxplus_{\pi}P\sim_{phhp} P$\\
  \item $P_1\boxplus_{\pi} P_2\sim_{phhp} P_2\boxplus_{1-\pi} P_1$\\
  \item $P_1\boxplus_{\pi}(P_2\boxplus_{\rho} P_3)\sim_{phhp} (P_1\boxplus_{\frac{\pi}{\pi+\rho-\pi\rho}}P_2)\boxplus_{\pi+\rho-\pi\rho} P_3$\\
\end{enumerate}
\end{proposition}

\begin{proof}
According to the definition of strongly probabilistic hhp-bisimulation, we can easily prove the above equations, and we omit the proof.
\end{proof}

\begin{theorem}[Identity law for probabilistic truly concurrent bisimilarities]
If $A(\widetilde{x})\overset{\text{def}}{=}P$, then

\begin{enumerate}
  \item $A(\widetilde{y})\sim_{pp} P\{\widetilde{y}/\widetilde{x}\}$;
  \item $A(\widetilde{y})\sim_{ps} P\{\widetilde{y}/\widetilde{x}\}$;
  \item $A(\widetilde{y})\sim_{php} P\{\widetilde{y}/\widetilde{x}\}$;
  \item $A(\widetilde{y})\sim_{phhp} P\{\widetilde{y}/\widetilde{x}\}$.
\end{enumerate}
\end{theorem}

\begin{proof}
According to the definition of strongly probabilistic truly concurrent bisimulations, we can easily prove the above equations, and we omit the proof.
\end{proof}

\begin{theorem}[Restriction Laws for strongly probabilistic pomset bisimilarity]
The restriction laws for strongly probabilistic pomset bisimilarity are as follows.

\begin{enumerate}
  \item $(y)P\sim_{pp} P$, if $y\notin fn(P)$;
  \item $(y)(z)P\sim_{pp} (z)(y)P$;
  \item $(y)(P+Q)\sim_{pp} (y)P+(y)Q$;
  \item $(y)(P\boxplus_{\pi}Q)\sim_{pp} (y)P\boxplus_{\pi}(y)Q$;
  \item $(y)\alpha.P\sim_{pp} \alpha.(y)P$ if $y\notin n(\alpha)$;
  \item $(y)\alpha.P\sim_{pp} \textbf{nil}$ if $y$ is the subject of $\alpha$.
\end{enumerate}
\end{theorem}

\begin{proof}
According to the definition of strongly probabilistic pomset bisimulation, we can easily prove the above equations, and we omit the proof.
\end{proof}

\begin{theorem}[Restriction Laws for strongly probabilistic step bisimilarity]
The restriction laws for strongly probabilistic step bisimilarity are as follows.

\begin{enumerate}
  \item $(y)P\sim_{ps} P$, if $y\notin fn(P)$;
  \item $(y)(z)P\sim_{ps} (z)(y)P$;
  \item $(y)(P+Q)\sim_{ps} (y)P+(y)Q$;
  \item $(y)(P\boxplus_{\pi}Q)\sim_{ps} (y)P\boxplus_{\pi}(y)Q$;
  \item $(y)\alpha.P\sim_{ps} \alpha.(y)P$ if $y\notin n(\alpha)$;
  \item $(y)\alpha.P\sim_{ps} \textbf{nil}$ if $y$ is the subject of $\alpha$.
\end{enumerate}
\end{theorem}

\begin{proof}
According to the definition of strongly probabilistic step bisimulation, we can easily prove the above equations, and we omit the proof.
\end{proof}

\begin{theorem}[Restriction Laws for strongly probabilistic hp-bisimilarity]
The restriction laws for strongly probabilistic hp-bisimilarity are as follows.

\begin{enumerate}
  \item $(y)P\sim_{php} P$, if $y\notin fn(P)$;
  \item $(y)(z)P\sim_{php} (z)(y)P$;
  \item $(y)(P+Q)\sim_{php} (y)P+(y)Q$;
  \item $(y)(P\boxplus_{\pi}Q)\sim_{php} (y)P\boxplus_{\pi}(y)Q$;
  \item $(y)\alpha.P\sim_{php} \alpha.(y)P$ if $y\notin n(\alpha)$;
  \item $(y)\alpha.P\sim_{php} \textbf{nil}$ if $y$ is the subject of $\alpha$.
\end{enumerate}
\end{theorem}

\begin{proof}
According to the definition of strongly probabilistic hp-bisimulation, we can easily prove the above equations, and we omit the proof.
\end{proof}

\begin{theorem}[Restriction Laws for strongly probabilistic hhp-bisimilarity]
The restriction laws for strongly probabilistic hhp-bisimilarity are as follows.

\begin{enumerate}
  \item $(y)P\sim_{phhp} P$, if $y\notin fn(P)$;
  \item $(y)(z)P\sim_{phhp} (z)(y)P$;
  \item $(y)(P+Q)\sim_{phhp} (y)P+(y)Q$;
  \item $(y)(P\boxplus_{\pi}Q)\sim_{phhp} (y)P\boxplus_{\pi}(y)Q$;
  \item $(y)\alpha.P\sim_{phhp} \alpha.(y)P$ if $y\notin n(\alpha)$;
  \item $(y)\alpha.P\sim_{phhp} \textbf{nil}$ if $y$ is the subject of $\alpha$.
\end{enumerate}
\end{theorem}

\begin{proof}
According to the definition of strongly probabilistic hhp-bisimulation, we can easily prove the above equations, and we omit the proof.
\end{proof}

\begin{theorem}[Parallel laws for strongly probabilistic pomset bisimilarity]
The parallel laws for strongly probabilistic pomset bisimilarity are as follows.

\begin{enumerate}
  \item $P\parallel \textbf{nil}\sim_{pp} P$;
  \item $P_1\parallel P_2\sim_{pp} P_2\parallel P_1$;
  \item $(P_1\parallel P_2)\parallel P_3\sim_{pp} P_1\parallel (P_2\parallel P_3)$;
  \item $(y)(P_1\parallel P_2)\sim_{pp} (y)P_1\parallel (y)P_2$, if $y\notin fn(P_1)\cap fn(P_2)$.
\end{enumerate}
\end{theorem}

\begin{proof}
According to the definition of strongly probabilistic pomset bisimulation, we can easily prove the above equations, and we omit the proof.
\end{proof}

\begin{theorem}[Parallel laws for strongly probabilistic step bisimilarity]
The parallel laws for strongly probabilistic step bisimilarity are as follows.

\begin{enumerate}
  \item $P\parallel \textbf{nil}\sim_{ps} P$;
  \item $P_1\parallel P_2\sim_{ps} P_2\parallel P_1$;
  \item $(P_1\parallel P_2)\parallel P_3\sim_{ps} P_1\parallel (P_2\parallel P_3)$;
  \item $(y)(P_1\parallel P_2)\sim_{ps} (y)P_1\parallel (y)P_2$, if $y\notin fn(P_1)\cap fn(P_2)$.
\end{enumerate}
\end{theorem}

\begin{proof}
According to the definition of strongly probabilistic step bisimulation, we can easily prove the above equations, and we omit the proof.
\end{proof}

\begin{theorem}[Parallel laws for strongly probabilistic hp-bisimilarity]
The parallel laws for strongly probabilistic hp-bisimilarity are as follows.

\begin{enumerate}
  \item $P\parallel \textbf{nil}\sim_{php} P$;
  \item $P_1\parallel P_2\sim_{php} P_2\parallel P_1$;
  \item $(P_1\parallel P_2)\parallel P_3\sim_{php} P_1\parallel (P_2\parallel P_3)$;
  \item $(y)(P_1\parallel P_2)\sim_{php} (y)P_1\parallel (y)P_2$, if $y\notin fn(P_1)\cap fn(P_2)$.
\end{enumerate}
\end{theorem}

\begin{proof}
According to the definition of strongly probabilistic hp-bisimulation, we can easily prove the above equations, and we omit the proof.
\end{proof}

\begin{theorem}[Parallel laws for strongly probabilistic hhp-bisimilarity]
The parallel laws for strongly probabilistic hhp-bisimilarity are as follows.

\begin{enumerate}
  \item $P\parallel \textbf{nil}\sim_{phhp} P$;
  \item $P_1\parallel P_2\sim_{phhp} P_2\parallel P_1$;
  \item $(P_1\parallel P_2)\parallel P_3\sim_{phhp} P_1\parallel (P_2\parallel P_3)$;
  \item $(y)(P_1\parallel P_2)\sim_{phhp} (y)P_1\parallel (y)P_2$, if $y\notin fn(P_1)\cap fn(P_2)$.
\end{enumerate}
\end{theorem}

\begin{proof}
According to the definition of strongly probabilistic hhp-bisimulation, we can easily prove the above equations, and we omit the proof.
\end{proof}

\begin{theorem}[Expansion law for truly concurrent bisimilarities]
Let $P\equiv\boxplus_i\sum_j \alpha_{ij}.P_{ij}$ and $Q\equiv\boxplus_{k}\sum_l\beta_{kl}.Q_{kl}$, where $bn(\alpha_{ij})\cap fn(Q)=\emptyset$ for all $i,j$, and
  $bn(\beta_{kl})\cap fn(P)=\emptyset$ for all $k,l$. Then,

\begin{enumerate}
  \item $P\parallel Q\sim_{pp} \boxplus_{i}\boxplus_{k}\sum_j\sum_l (\alpha_{ij}\parallel \beta_{kl}).(P_{ij}\parallel Q_{kl})+\boxplus_i\boxplus_k\sum_{\alpha_{ij} \textrm{ comp }\beta_{kl}}\tau.R_{ijkl}$;
  \item $P\parallel Q\sim_{ps} \boxplus_{i}\boxplus_{k}\sum_j\sum_l (\alpha_{ij}\parallel \beta_{kl}).(P_{ij}\parallel Q_{kl})+\boxplus_i\boxplus_k\sum_{\alpha_{ij} \textrm{ comp }\beta_{kl}}\tau.R_{ijkl}$;
  \item $P\parallel Q\sim_{php} \boxplus_{i}\boxplus_{k}\sum_j\sum_l (\alpha_{ij}\parallel \beta_{kl}).(P_{ij}\parallel Q_{kl})+\boxplus_i\boxplus_k\sum_{\alpha_{ij} \textrm{ comp }\beta_{kl}}\tau.R_{ijkl}$;
  \item $P\parallel Q\nsim_{phhp} \boxplus_{i}\boxplus_{k}\sum_j\sum_l (\alpha_{ij}\parallel \beta_{kl}).(P_{ij}\parallel Q_{kl})+\boxplus_i\boxplus_k\sum_{\alpha_{ij} \textrm{ comp }\beta_{kl}}\tau.R_{ijkl}$.
\end{enumerate}

Where $\alpha_i$ comp $\beta_j$ and $R_{ij}$ are defined as follows:
\begin{enumerate}
  \item $\alpha_{ij}$ is $\overline{x}u$ and $\beta_{kl}$ is $x(v)$, then $R_{ijkl}=P_{ij}\parallel Q_{kl}\{u/v\}$;
  \item $\alpha_{ij}$ is $\overline{x}(u)$ and $\beta_{kl}$ is $x(v)$, then $R_{ijkl}=(w)(P_{ij}\{w/u\}\parallel Q_{kl}\{w/v\})$, if $w\notin fn((u)P_{ij})\cup fn((v)Q_{kl})$;
  \item $\alpha_{ij}$ is $x(v)$ and $\beta_{kl}$ is $\overline{x}u$, then $R_{ijkl}=P_{ij}\{u/v\}\parallel Q_{kl}$;
  \item $\alpha_{ij}$ is $x(v)$ and $\beta_{kl}$ is $\overline{x}(u)$, then $R_{ijkl}=(w)(P_{ij}\{w/v\}\parallel Q_{kl}\{w/u\})$, if $w\notin fn((v)P_{ij})\cup fn((u)Q_{kl})$.
\end{enumerate}
\end{theorem}

\begin{proof}
According to the definition of strongly probabilistic truly concurrent bisimulations, we can easily prove the above equations, and we omit the proof.
\end{proof}

\begin{theorem}[Equivalence and congruence for strongly probabilistic pomset bisimilarity]
\begin{enumerate}
  \item $\sim_{pp}$ is an equivalence relation;
  \item If $P\sim_{pp} Q$ then
  \begin{enumerate}
    \item $\alpha.P\sim_{pp} \alpha.Q$, $\alpha$ is a free action;
    \item $P+R\sim_{pp} Q+R$;
    \item $P\boxplus_{\pi} R\sim_{pp}Q\boxplus_{\pi}R$
    \item $P\parallel R\sim_{pp} Q\parallel R$;
    \item $(w)P\sim_{pp} (w)Q$;
    \item $x(y).P\sim_{pp} x(y).Q$.
  \end{enumerate}
\end{enumerate}
\end{theorem}

\begin{proof}
According to the definition of strongly probabilistic pomset bisimulation, we can easily prove the above equations, and we omit the proof.
\end{proof}

\begin{theorem}[Equivalence and congruence for strongly probabilistic step bisimilarity]
\begin{enumerate}
  \item $\sim_{ps}$ is an equivalence relation;
  \item If $P\sim_{ps} Q$ then
  \begin{enumerate}
    \item $\alpha.P\sim_{ps} \alpha.Q$, $\alpha$ is a free action;
    \item $P+R\sim_{ps} Q+R$;
    \item $P\boxplus_{\pi} R\sim_{ps}Q\boxplus_{\pi}R$
    \item $P\parallel R\sim_{ps} Q\parallel R$;
    \item $(w)P\sim_{ps} (w)Q$;
    \item $x(y).P\sim_{ps} x(y).Q$.
  \end{enumerate}
\end{enumerate}
\end{theorem}

\begin{proof}
According to the definition of strongly probabilistic step bisimulation, we can easily prove the above equations, and we omit the proof.
\end{proof}

\begin{theorem}[Equivalence and congruence for strongly probabilistic hp-bisimilarity]
\begin{enumerate}
  \item $\sim_{php}$ is an equivalence relation;
  \item If $P\sim_{php} Q$ then
  \begin{enumerate}
    \item $\alpha.P\sim_{php} \alpha.Q$, $\alpha$ is a free action;
    \item $P+R\sim_{php} Q+R$;
    \item $P\boxplus_{\pi} R\sim_{php}Q\boxplus_{\pi}R$
    \item $P\parallel R\sim_{php} Q\parallel R$;
    \item $(w)P\sim_{php} (w)Q$;
    \item $x(y).P\sim_{php} x(y).Q$.
  \end{enumerate}
\end{enumerate}
\end{theorem}

\begin{proof}
According to the definition of strongly probabilistic hp-bisimulation, we can easily prove the above equations, and we omit the proof.
\end{proof}

\begin{theorem}[Equivalence and congruence for strongly probabilistic hhp-bisimilarity]
\begin{enumerate}
  \item $\sim_{phhp}$ is an equivalence relation;
  \item If $P\sim_{phhp} Q$ then
  \begin{enumerate}
    \item $\alpha.P\sim_{phhp} \alpha.Q$, $\alpha$ is a free action;
    \item $P+R\sim_{phhp} Q+R$;
    \item $P\boxplus_{\pi} R\sim_{phhp}Q\boxplus_{\pi}R$
    \item $P\parallel R\sim_{phhp} Q\parallel R$;
    \item $(w)P\sim_{phhp} (w)Q$;
    \item $x(y).P\sim_{phhp} x(y).Q$.
  \end{enumerate}
\end{enumerate}
\end{theorem}

\begin{proof}
According to the definition of strongly probabilistic hhp-bisimulation, we can easily prove the above equations, and we omit the proof.
\end{proof}

\subsubsection{Recursion}

\begin{definition}
Let $X$ have arity $n$, and let $\widetilde{x}=x_1,\cdots,x_n$ be distinct names, and $fn(P)\subseteq\{x_1,\cdots,x_n\}$. The replacement of $X(\widetilde{x})$ by $P$ in $E$, written 
$E\{X(\widetilde{x}):=P\}$, means the result of replacing each subterm $X(\widetilde{y})$ in $E$ by $P\{\widetilde{y}/\widetilde{x}\}$.
\end{definition}

\begin{definition}
Let $E$ and $F$ be two process expressions containing only $X_1,\cdots,X_m$ with associated name sequences $\widetilde{x}_1,\cdots,\widetilde{x}_m$. Then,
\begin{enumerate}
  \item $E\sim_{pp} F$ means $E(\widetilde{P})\sim_{pp} F(\widetilde{P})$;
  \item $E\sim_{ps} F$ means $E(\widetilde{P})\sim_{ps} F(\widetilde{P})$;
  \item $E\sim_{php} F$ means $E(\widetilde{P})\sim_{php} F(\widetilde{P})$;
  \item $E\sim_{phhp} F$ means $E(\widetilde{P})\sim_{phhp} F(\widetilde{P})$;
\end{enumerate}

for all $\widetilde{P}$ such that $fn(P_i)\subseteq \widetilde{x}_i$ for each $i$.
\end{definition}

\begin{definition}
A term or identifier is weakly guarded in $P$ if it lies within some subterm $\alpha.Q$ or $(\alpha_1\parallel\cdots\parallel \alpha_n).Q$ of $P$.
\end{definition}

\begin{theorem}
Assume that $\widetilde{E}$ and $\widetilde{F}$ are expressions containing only $X_i$ with $\widetilde{x}_i$, and $\widetilde{A}$ and $\widetilde{B}$ are identifiers with $A_i$, $B_i$. Then, for all $i$,
\begin{enumerate}
  \item $E_i\sim_{ps} F_i$, $A_i(\widetilde{x}_i)\overset{\text{def}}{=}E_i(\widetilde{A})$, $B_i(\widetilde{x}_i)\overset{\text{def}}{=}F_i(\widetilde{B})$, then 
  $A_i(\widetilde{x}_i)\sim_{ps} B_i(\widetilde{x}_i)$;
  \item $E_i\sim_{pp} F_i$, $A_i(\widetilde{x}_i)\overset{\text{def}}{=}E_i(\widetilde{A})$, $B_i(\widetilde{x}_i)\overset{\text{def}}{=}F_i(\widetilde{B})$, then 
  $A_i(\widetilde{x}_i)\sim_{pp} B_i(\widetilde{x}_i)$;
  \item $E_i\sim_{php} F_i$, $A_i(\widetilde{x}_i)\overset{\text{def}}{=}E_i(\widetilde{A})$, $B_i(\widetilde{x}_i)\overset{\text{def}}{=}F_i(\widetilde{B})$, then 
  $A_i(\widetilde{x}_i)\sim_{php} B_i(\widetilde{x}_i)$;
  \item $E_i\sim_{phhp} F_i$, $A_i(\widetilde{x}_i)\overset{\text{def}}{=}E_i(\widetilde{A})$, $B_i(\widetilde{x}_i)\overset{\text{def}}{=}F_i(\widetilde{B})$, then 
  $A_i(\widetilde{x}_i)\sim_{phhp} B_i(\widetilde{x}_i)$.
\end{enumerate}
\end{theorem}

\begin{proof}
\begin{enumerate}
  \item $E_i\sim_{ps} F_i$, $A_i(\widetilde{x}_i)\overset{\text{def}}{=}E_i(\widetilde{A})$, $B_i(\widetilde{x}_i)\overset{\text{def}}{=}F_i(\widetilde{B})$, then 
  $A_i(\widetilde{x}_i)\sim_{ps} B_i(\widetilde{x}_i)$.

      We will consider the case $I=\{1\}$ with loss of generality, and show the following relation $R$ is a strongly probabilistic step bisimulation.

      $$R=\{(G(A),G(B)):G\textrm{ has only identifier }X\}.$$

      By choosing $G\equiv X(\widetilde{y})$, it follows that $A(\widetilde{y})\sim_{ps} B(\widetilde{y})$. It is sufficient to prove the following:
      \begin{enumerate}
        \item If $G(A)\rightsquigarrow\xrightarrow{\{\alpha_1,\cdots,\alpha_n\}}P'$, where $\alpha_i(1\leq i\leq n)$ is a free action or bound output action with 
        $bn(\alpha_1)\cap\cdots\cap bn(\alpha_n)\cap n(G(A),G(B))=\emptyset$, then $G(B)\rightsquigarrow\xrightarrow{\{\alpha_1,\cdots,\alpha_n\}}Q''$ such that $P'\sim_{ps} Q''$;
        \item If $G(A)\rightsquigarrow\xrightarrow{x(y)}P'$ with $x\notin n(G(A),G(B))$, then $G(B)\rightsquigarrow\xrightarrow{x(y)}Q''$, such that for all $u$, 
        $P'\{u/y\}\sim_{ps} Q''\{u/y\}$.
      \end{enumerate}

      To prove the above properties, it is sufficient to induct on the depth of inference and quite routine, we omit it.
  \item $E_i\sim_{pp} F_i$, $A_i(\widetilde{x}_i)\overset{\text{def}}{=}E_i(\widetilde{A})$, $B_i(\widetilde{x}_i)\overset{\text{def}}{=}F_i(\widetilde{B})$, then 
  $A_i(\widetilde{x}_i)\sim_{pp} B_i(\widetilde{x}_i)$. It can be proven similarly to the above case.
  \item $E_i\sim_{php} F_i$, $A_i(\widetilde{x}_i)\overset{\text{def}}{=}E_i(\widetilde{A})$, $B_i(\widetilde{x}_i)\overset{\text{def}}{=}F_i(\widetilde{B})$, then 
  $A_i(\widetilde{x}_i)\sim_{php} B_i(\widetilde{x}_i)$. It can be proven similarly to the above case.
  \item $E_i\sim_{phhp} F_i$, $A_i(\widetilde{x}_i)\overset{\text{def}}{=}E_i(\widetilde{A})$, $B_i(\widetilde{x}_i)\overset{\text{def}}{=}F_i(\widetilde{B})$, then 
  $A_i(\widetilde{x}_i)\sim_{phhp} B_i(\widetilde{x}_i)$. It can be proven similarly to the above case.
\end{enumerate}
\end{proof}

\begin{theorem}[Unique solution of equations]
Assume $\widetilde{E}$ are expressions containing only $X_i$ with $\widetilde{x}_i$, and each $X_i$ is weakly guarded in each $E_j$. Assume that $\widetilde{P}$ and $\widetilde{Q}$ are 
processes such that $fn(P_i)\subseteq \widetilde{x}_i$ and $fn(Q_i)\subseteq \widetilde{x}_i$. Then, for all $i$,
\begin{enumerate}
  \item if $P_i\sim_{pp} E_i(\widetilde{P})$, $Q_i\sim_{pp} E_i(\widetilde{Q})$, then $P_i\sim_{pp} Q_i$;
  \item if $P_i\sim_{ps} E_i(\widetilde{P})$, $Q_i\sim_{ps} E_i(\widetilde{Q})$, then $P_i\sim_{ps} Q_i$;
  \item if $P_i\sim_{php} E_i(\widetilde{P})$, $Q_i\sim_{php} E_i(\widetilde{Q})$, then $P_i\sim_{php} Q_i$;
  \item if $P_i\sim_{phhp} E_i(\widetilde{P})$, $Q_i\sim_{phhp} E_i(\widetilde{Q})$, then $P_i\sim_{phhp} Q_i$.
\end{enumerate}
\end{theorem}

\begin{proof}
\begin{enumerate}
  \item It is similar to the proof of unique solution of equations for strongly probabilistic pomset bisimulation in CPTC, we omit it;
  \item It is similar to the proof of unique solution of equations for strongly probabilistic step bisimulation in CPTC, we omit it;
  \item It is similar to the proof of unique solution of equations for strongly probabilistic hp-bisimulation in CPTC, we omit it;
  \item It is similar to the proof of unique solution of equations for strongly probabilistic hhp-bisimulation in CPTC, we omit it.
\end{enumerate}
\end{proof}

\subsection{Algebraic Theory}\label{at5}

In this section, we will try to axiomatize $\pi_{ptc}$, the theory is \textbf{SPTC} (for strongly probabilistic true concurrency).

\begin{definition}[SPTC]
The theory \textbf{SPTC} is consisted of the following axioms and inference rules:

\begin{enumerate}
  \item Alpha-conversion $\textbf{A}$.
  \[\textrm{if } P\equiv Q, \textrm{ then } P=Q\]
  \item Congruence $\textbf{C}$. If $P=Q$, then,
  \[\tau.P=\tau.Q\quad \overline{x}y.P=\overline{x}y.Q\]
  \[P+R=Q+R\quad P\parallel R=Q\parallel R\]
  \[(x)P=(x)Q\quad x(y).P=x(y).Q\]
  \item Summation $\textbf{S}$.
  \[\textbf{S0}\quad P+\textbf{nil}=P\]
  \[\textbf{S1}\quad P+P=P\]
  \[\textbf{S2}\quad P+Q=Q+P\]
  \[\textbf{S3}\quad P+(Q+R)=(P+Q)+R\]
  \item Box-Summation $\textbf(BS)$.
  \[\textbf{BS0}\quad P\boxplus_{\pi}\textbf{nil}= P\]
  \[\textbf{BS1}\quad P\boxplus_{\pi}P= P\]
  \[\textbf{BS2}\quad P\boxplus_{\pi} Q= Q\boxplus_{1-\pi} P\]
  \[\textbf{BS3}\quad P\boxplus_{\pi}(Q\boxplus_{\rho} R)= (P\boxplus_{\frac{\pi}{\pi+\rho-\pi\rho}}Q)\boxplus_{\pi+\rho-\pi\rho} R\]
  \item Restriction $\textbf{R}$.
  \[\textbf{R0}\quad (x)P=P\quad \textrm{ if }x\notin fn(P)\]
  \[\textbf{R1}\quad (x)(y)P=(y)(x)P\]
  \[\textbf{R2}\quad (x)(P+Q)=(x)P+(x)Q\]
  \[\textbf{R3}\quad (x)\alpha.P=\alpha.(x)P\quad \textrm{ if }x\notin n(\alpha)\]
  \[\textbf{R4}\quad (x)\alpha.P=\textbf{nil}\quad \textrm{ if }x\textrm{is the subject of }\alpha\]
  \item Expansion $\textbf{E}$.
  Let $P\equiv\boxplus_i\sum_j \alpha_{ij}.P_{ij}$ and $Q\equiv\boxplus_{k}\sum_l\beta_{kl}.Q_{kl}$, where $bn(\alpha_{ij})\cap fn(Q)=\emptyset$ for all $i,j$, and 
  $bn(\beta_{kl})\cap fn(P)=\emptyset$ for all $k,l$. Then
  $P\parallel Q= \boxplus_{i}\boxplus_{k}\sum_j\sum_l (\alpha_{ij}\parallel \beta_{kl}).(P_{ij}\parallel Q_{kl})+\boxplus_i\boxplus_k\sum_{\alpha_{ij} \textrm{ comp }\beta_{kl}}\tau.R_{ijkl}$.

  Where $\alpha_{ij}$ comp $\beta_{kl}$ and $R_{ijkl}$ are defined as follows:
\begin{enumerate}
  \item $\alpha_{ij}$ is $\overline{x}u$ and $\beta_{kl}$ is $x(v)$, then $R_{ijkl}=P_{ij}\parallel Q_{kl}\{u/v\}$;
  \item $\alpha_{ij}$ is $\overline{x}(u)$ and $\beta_{kl}$ is $x(v)$, then $R_{ijkl}=(w)(P_{ij}\{w/u\}\parallel Q_{kl}\{w/v\})$, if $w\notin fn((u)P_{ij})\cup fn((v)Q_{kl})$;
  \item $\alpha_{ij}$ is $x(v)$ and $\beta_{kl}$ is $\overline{x}u$, then $R_{ijkl}=P_{ij}\{u/v\}\parallel Q_{kl}$;
  \item $\alpha_{ij}$ is $x(v)$ and $\beta_{kl}$ is $\overline{x}(u)$, then $R_{ijkl}=(w)(P_{ij}\{w/v\}\parallel Q_{kl}\{w/u\})$, if $w\notin fn((v)P_{ij})\cup fn((u)Q_{kl})$.
\end{enumerate}
  \item Identifier $\textbf{I}$.
  \[\textrm{If }A(\widetilde{x})\overset{\text{def}}{=}P,\textrm{ then }A(\widetilde{y})= P\{\widetilde{y}/\widetilde{x}\}.\]
\end{enumerate}
\end{definition}

\begin{theorem}[Soundness]
If $\textbf{STC}\vdash P=Q$ then
\begin{enumerate}
  \item $P\sim_{pp} Q$;
  \item $P\sim_{ps} Q$;
  \item $P\sim_{php} Q$;
  \item $P\sim_{phhp} Q$.
\end{enumerate}
\end{theorem}

\begin{proof}
The soundness of these laws modulo strongly truly concurrent bisimilarities is already proven in Section \ref{stcb5}.
\end{proof}

\begin{definition}
The agent identifier $A$ is weakly guardedly defined if every agent identifier is weakly guarded in the right-hand side of the definition of $A$.
\end{definition}

\begin{definition}[Head normal form]
A Process $P$ is in head normal form if it is a sum of the prefixes:

$$P\equiv \boxplus_{i}\sum_j(\alpha_{ij1}\parallel\cdots\parallel\alpha_{ijn}).P_{ij}.$$
\end{definition}

\begin{proposition}
If every agent identifier is weakly guardedly defined, then for any process $P$, there is a head normal form $H$ such that

$$\textbf{STC}\vdash P=H.$$
\end{proposition}

\begin{proof}
It is sufficient to induct on the structure of $P$ and quite obvious.
\end{proof}

\begin{theorem}[Completeness]
For all processes $P$ and $Q$,
\begin{enumerate}
  \item if $P\sim_{pp} Q$, then $\textbf{STC}\vdash P=Q$;
  \item if $P\sim_{ps} Q$, then $\textbf{STC}\vdash P=Q$;
  \item if $P\sim_{php} Q$, then $\textbf{STC}\vdash P=Q$.
\end{enumerate}
\end{theorem}

\begin{proof}
\begin{enumerate}
  \item if $P\sim_{ps} Q$, then $\textbf{STC}\vdash P=Q$.
Since $P$ and $Q$ all have head normal forms, let $P\equiv\boxplus_{j=1}^l\sum_{i=1}^k\alpha_{ji}.P_{ji}$ and $Q\equiv\boxplus_{j=1}^l\sum_{i=1}^k\beta_{ji}.Q_{ji}$. Then the depth of 
$P$, denoted as $d(P)=0$, if $k=0$; $d(P)=1+max\{d(P_{ji})\}$ for $1\leq j,i\leq k$. The depth $d(Q)$ can be defined similarly.

It is sufficient to induct on $d=d(P)+d(Q)$. When $d=0$, $P\equiv\textbf{nil}$ and $Q\equiv\textbf{nil}$, $P=Q$, as desired.

Suppose $d>0$.

\begin{itemize}
  \item If $(\alpha_1\parallel\cdots\parallel\alpha_n).M$ with $\alpha_{ji}(1\leq j,i\leq n)$ free actions is a summand of $P$, then $P\rightsquigarrow\xrightarrow{\{\alpha_1,\cdots,\alpha_n\}}M$. 
  Since $Q$ is in head normal form and has a summand $(\alpha_1\parallel\cdots\parallel\alpha_n).N$ such that $M\sim_{ps} N$, by the induction hypothesis $\textbf{STC}\vdash M=N$, 
  $\textbf{STC}\vdash (\alpha_1\parallel\cdots\parallel\alpha_n).M= (\alpha_1\parallel\cdots\parallel\alpha_n).N$;
  \item If $x(y).M$ is a summand of $P$, then for $z\notin n(P, Q)$, $P\rightsquigarrow\xrightarrow{x(z)}M'\equiv M\{z/y\}$. Since $Q$ is in head normal form and has a summand 
  $x(w).N$ such that for all $v$, $M'\{v/z\}\sim_{ps} N'\{v/z\}$ where $N'\equiv N\{z/w\}$, by the induction hypothesis $\textbf{STC}\vdash M'\{v/z\}=N'\{v/z\}$, by the axioms 
  $\textbf{C}$ and $\textbf{A}$, $\textbf{STC}\vdash x(y).M=x(w).N$;
  \item If $\overline{x}(y).M$ is a summand of $P$, then for $z\notin n(P,Q)$, $P\rightsquigarrow\xrightarrow{\overline{x}(z)}M'\equiv M\{z/y\}$. Since $Q$ is in head normal form and 
  has a summand $\overline{x}(w).N$ such that $M'\sim_{ps} N'$ where $N'\equiv N\{z/w\}$, by the induction hypothesis $\textbf{STC}\vdash M'=N'$, by the axioms 
  $\textbf{A}$ and $\textbf{C}$, $\textbf{STC}\vdash \overline{x}(y).M= \overline{x}(w).N$;
\end{itemize}

  \item if $P\sim_{pp} Q$, then $\textbf{STC}\vdash P=Q$. It can be proven similarly to the above case.
  \item if $P\sim_{php} Q$, then $\textbf{STC}\vdash P=Q$. It can be proven similarly to the above case.
\end{enumerate}
\end{proof}

\newpage\section{Guards}\label{pguards}

In this chapter, we introduce guards into probabilistic process algebra in chapter 4 based on the work on guards for process algebra \cite{HLPA}. This chapter is organized as follows. 
We introduce the operational semantics of guards in
section \ref{os2}, $BAPTC$ with Guards in section \ref{batcg}, $APPTC$ with Guards \ref{aptcg}, recursion in section \ref{recg}, abstraction in section \ref{absg}, Hoare Logic for
$APPTC_G$ in section \ref{hl}. Note that, all the definitions of PDFs are the same as those in chapter 4, and we do not repeat any more.

\subsection{Operational Semantics}{\label{os2}}

In this section, we extend probabilistic truly concurrent bisimilarities to the ones containing data states.

\begin{definition}[Prime event structure with silent event and empty event]\label{PESG}
Let $\Lambda$ be a fixed set of labels, ranged over $a,b,c,\cdots$ and $\tau,\epsilon$. A ($\Lambda$-labelled) prime event structure with silent event $\tau$ and empty event $\epsilon$ 
is a tuple $\mathcal{E}=\langle \mathbb{E}, \leq, \sharp, \sharp_{\pi} \lambda\rangle$, where $\mathbb{E}$ is a denumerable set of events, including the silent event $\tau$ and 
empty event $\epsilon$. Let $\hat{\mathbb{E}}=\mathbb{E}\backslash\{\tau,\epsilon\}$, exactly excluding $\tau$ and $\epsilon$, it is obvious that $\hat{\tau^*}=\epsilon$. Let 
$\lambda:\mathbb{E}\rightarrow\Lambda$ be a labelling function and let $\lambda(\tau)=\tau$ and $\lambda(\epsilon)=\epsilon$. And $\leq$, $\sharp$, $\sharp_{\pi}$ are binary relations 
on $\mathbb{E}$, called causality, conflict and probabilistic conflict respectively, such that:

\begin{enumerate}
  \item $\leq$ is a partial order and $\lceil e \rceil = \{e'\in \mathbb{E}|e'\leq e\}$ is finite for all $e\in \mathbb{E}$. It is easy to see that 
  $e\leq\tau^*\leq e'=e\leq\tau\leq\cdots\leq\tau\leq e'$, then $e\leq e'$.
  \item $\sharp$ is irreflexive, symmetric and hereditary with respect to $\leq$, that is, for all $e,e',e''\in \mathbb{E}$, if $e\sharp e'\leq e''$, then $e\sharp e''$;
  \item $\sharp_{\pi}$ is irreflexive, symmetric and hereditary with respect to $\leq$, that is, for all $e,e',e''\in \mathbb{E}$, if $e\sharp_{\pi} e'\leq e''$, then $e\sharp_{\pi} e''$.
\end{enumerate}

Then, the concepts of consistency and concurrency can be drawn from the above definition:

\begin{enumerate}
  \item $e,e'\in \mathbb{E}$ are consistent, denoted as $e\frown e'$, if $\neg(e\sharp e')$ and $\neg(e\sharp_{\pi} e')$. A subset $X\subseteq \mathbb{E}$ is called consistent, if 
  $e\frown e'$ for all $e,e'\in X$.
  \item $e,e'\in \mathbb{E}$ are concurrent, denoted as $e\parallel e'$, if $\neg(e\leq e')$, $\neg(e'\leq e)$, and $\neg(e\sharp e')$ and $\neg(e\sharp_{\pi} e')$.
\end{enumerate}
\end{definition}

\begin{definition}[Configuration]
Let $\mathcal{E}$ be a PES. A (finite) configuration in $\mathcal{E}$ is a (finite) consistent subset of events $C\subseteq \mathcal{E}$, closed with respect to causality 
(i.e. $\lceil C\rceil=C$), and a data state $s\in S$ with $S$ the set of all data states, denoted $\langle C, s\rangle$. The set of finite configurations of $\mathcal{E}$ is denoted by 
$\langle\mathcal{C}(\mathcal{E}), S\rangle$. We let $\hat{C}=C\backslash\{\tau\}\cup\{\epsilon\}$.
\end{definition}

A consistent subset of $X\subseteq \mathbb{E}$ of events can be seen as a pomset. Given $X, Y\subseteq \mathbb{E}$, $\hat{X}\sim \hat{Y}$ if $\hat{X}$ and $\hat{Y}$ are isomorphic as 
pomsets. In the following of the paper, we say $C_1\sim C_2$, we mean $\hat{C_1}\sim\hat{C_2}$.

\begin{definition}[Pomset transitions and step]
Let $\mathcal{E}$ be a PES and let $C\in\mathcal{C}(\mathcal{E})$, and $\emptyset\neq X\subseteq \mathbb{E}$, if $C\cap X=\emptyset$ and $C'=C\cup X\in\mathcal{C}(\mathcal{E})$, then 
$\langle C,s\rangle\xrightarrow{X} \langle C',s'\rangle$ is called a pomset transition from $\langle C,s\rangle$ to $\langle C',s'\rangle$. When the events in $X$ are pairwise 
concurrent, we say that $\langle C,s\rangle\xrightarrow{X}\langle C',s'\rangle$ is a step. It is obvious that $\rightarrow^*\xrightarrow{X}\rightarrow^*=\xrightarrow{X}$ and 
$\rightarrow^*\xrightarrow{e}\rightarrow^*=\xrightarrow{e}$ for any $e\in\mathbb{E}$ and $X\subseteq\mathbb{E}$.
\end{definition}

\begin{definition}[Probabilistic transitions]
Let $\mathcal{E}$ be a PES and let $C\in\mathcal{C}(\mathcal{E})$, the transition $\langle C,s\rangle\xrsquigarrow{\pi} \langle C^{\pi},s\rangle$ is called a probabilistic transition 
from $\langle C,s\rangle$ to $\langle C^{\pi},s\rangle$.
\end{definition}

\begin{definition}[Weak pomset transitions and weak step]
Let $\mathcal{E}$ be a PES and let $C\in\mathcal{C}(\mathcal{E})$, and $\emptyset\neq X\subseteq \hat{\mathbb{E}}$, if $C\cap X=\emptyset$ and 
$\hat{C'}=\hat{C}\cup X\in\mathcal{C}(\mathcal{E})$, then $\langle C,s\rangle\xRightarrow{X} \langle C',s'\rangle$ is called a weak pomset transition from $\langle C,s\rangle$ to 
$\langle C',s'\rangle$, where we define $\xRightarrow{e}\triangleq\xrightarrow{\tau^*}\xrightarrow{e}\xrightarrow{\tau^*}$. And 
$\xRightarrow{X}\triangleq\xrightarrow{\tau^*}\xrightarrow{e}\xrightarrow{\tau^*}$, for every $e\in X$. When the events in $X$ are pairwise concurrent, we say that 
$\langle C,s\rangle\xRightarrow{X}\langle C',s'\rangle$ is a weak step.
\end{definition}

We will also suppose that all the PESs in this chapter are image finite, that is, for any PES $\mathcal{E}$ and $C\in \mathcal{C}(\mathcal{E})$ and $a\in \Lambda$, 
$\{\langle C,s\rangle\xrsquigarrow{\pi} \langle C^{\pi},s\rangle\}$,
$\{e\in \mathbb{E}|\langle C,s\rangle\xrightarrow{e} \langle C',s'\rangle\wedge \lambda(e)=a\}$ and 
$\{e\in\hat{\mathbb{E}}|\langle C,s\rangle\xRightarrow{e} \langle C',s'\rangle\wedge \lambda(e)=a\}$ is finite.

\begin{definition}[Probabilistic pomset, step bisimulation]\label{PSBG}
Let $\mathcal{E}_1$, $\mathcal{E}_2$ be PESs. A probabilistic pomset bisimulation is a relation $R\subseteq\langle\mathcal{C}(\mathcal{E}_1),S\rangle\times\langle\mathcal{C}(\mathcal{E}_2),S\rangle$, 
such that (1) if $(\langle C_1,s\rangle,\langle C_2,s\rangle)\in R$, and $\langle C_1,s\rangle\xrightarrow{X_1}\langle C_1',s'\rangle$ then 
$\langle C_2,s\rangle\xrightarrow{X_2}\langle C_2',s'\rangle$, with $X_1\subseteq \mathbb{E}_1$, $X_2\subseteq \mathbb{E}_2$, $X_1\sim X_2$ and 
$(\langle C_1',s'\rangle,\langle C_2',s'\rangle)\in R$ for all $s,s'\in S$, and vice-versa; (2) if $(\langle C_1,s\rangle,\langle C_2,s\rangle)\in R$, and $\langle C_1,s\rangle\xrsquigarrow{\pi}\langle C_1^{\pi},s\rangle$ 
then $\langle C_2,s\rangle\xrsquigarrow{\pi}\langle C_2^{\pi},s\rangle$ and $(\langle C_1^{\pi},s\rangle,\langle C_2^{\pi},s\rangle)\in R$, and vice-versa; (3) if $(\langle C_1,s\rangle,\langle C_2,s\rangle)\in R$,
then $\mu(C_1,C)=\mu(C_2,C)$ for each $C\in\mathcal{C}(\mathcal{E})/R$; (4) $[\surd]_R=\{\surd\}$. We say that $\mathcal{E}_1$, $\mathcal{E}_2$ are probabilistic pomset bisimilar, written 
$\mathcal{E}_1\sim_{pp}\mathcal{E}_2$, if there exists a probabilistic pomset bisimulation $R$, such that $(\langle\emptyset,\emptyset\rangle,\langle\emptyset,\emptyset\rangle)\in R$. 
By replacing probabilistic pomset transitions with probabilistic steps, we can get the definition of probabilistic step bisimulation. When PESs $\mathcal{E}_1$ and $\mathcal{E}_2$ are 
probabilistic step bisimilar, we write $\mathcal{E}_1\sim_{ps}\mathcal{E}_2$.
\end{definition}

\begin{definition}[Weakly probabilistic pomset, step bisimulation]\label{WPSBG}
Let $\mathcal{E}_1$, $\mathcal{E}_2$ be PESs. A weakly probabilistic pomset bisimulation is a relation $R\subseteq\langle\mathcal{C}(\mathcal{E}_1),S\rangle\times\langle\mathcal{C}(\mathcal{E}_2),S\rangle$, 
such that (1) if $(\langle C_1,s\rangle,\langle C_2,s\rangle)\in R$, and $\langle C_1,s\rangle\xRightarrow{X_1}\langle C_1',s'\rangle$ then 
$\langle C_2,s\rangle\xRightarrow{X_2}\langle C_2',s'\rangle$, with $X_1\subseteq \hat{\mathbb{E}_1}$, $X_2\subseteq \hat{\mathbb{E}_2}$, $X_1\sim X_2$ and 
$(\langle C_1',s'\rangle,\langle C_2',s'\rangle)\in R$ for all $s,s'\in S$, and vice-versa; (2) if $(\langle C_1,s\rangle,\langle C_2,s\rangle)\in R$, and $\langle C_1,s\rangle\xrsquigarrow{\pi}\langle C_1^{\pi},s\rangle$
then $\langle C_2,s\rangle\xrsquigarrow{\pi}\langle C_2^{\pi},s\rangle$ and $(\langle C_1^{\pi},s\rangle,\langle C_2^{\pi},s\rangle)\in R$, and vice-versa; (3) if $(\langle C_1,s\rangle,\langle C_2,s\rangle)\in R$,
then $\mu(C_1,C)=\mu(C_2,C)$ for each $C\in\mathcal{C}(\mathcal{E})/R$; (4) $[\surd]_R=\{\surd\}$. We say that $\mathcal{E}_1$, $\mathcal{E}_2$ are weakly probabilistic pomset bisimilar, 
written $\mathcal{E}_1\approx_{pp}\mathcal{E}_2$, if there exists a weakly probabilistic pomset bisimulation $R$, such that 
$(\langle\emptyset,\emptyset\rangle,\langle\emptyset,\emptyset\rangle)\in R$. By replacing weakly probabilistic pomset transitions with weakly probabilistic steps, we can get the 
definition of weakly probabilistic step bisimulation. When PESs $\mathcal{E}_1$ and $\mathcal{E}_2$ are weakly probabilistic step bisimilar, we write 
$\mathcal{E}_1\approx_{ps}\mathcal{E}_2$.
\end{definition}

\begin{definition}[Posetal product]
Given two PESs $\mathcal{E}_1$, $\mathcal{E}_2$, the posetal product of their configurations, denoted 
$\langle\mathcal{C}(\mathcal{E}_1),S\rangle\overline{\times}\langle\mathcal{C}(\mathcal{E}_2),S\rangle$, is defined as

$$\{(\langle C_1,s\rangle,f,\langle C_2,s\rangle)|C_1\in\mathcal{C}(\mathcal{E}_1),C_2\in\mathcal{C}(\mathcal{E}_2),f:C_1\rightarrow C_2 \textrm{ isomorphism}\}.$$

A subset $R\subseteq\langle\mathcal{C}(\mathcal{E}_1),S\rangle\overline{\times}\langle\mathcal{C}(\mathcal{E}_2),S\rangle$ is called a posetal relation. We say that $R$ is downward 
closed when for any $(\langle C_1,s\rangle,f,\langle C_2,s\rangle),(\langle C_1',s'\rangle,f',\langle C_2',s'\rangle)\in \langle\mathcal{C}(\mathcal{E}_1),S\rangle\overline{\times}\langle\mathcal{C}(\mathcal{E}_2),S\rangle$, 
if $(\langle C_1,s\rangle,f,\langle C_2,s\rangle)\subseteq (\langle C_1',s'\rangle,f',\langle C_2',s'\rangle)$ pointwise and 
$(\langle C_1',s'\rangle,f',\langle C_2',s'\rangle)\in R$, then $(\langle C_1,s\rangle,f,\langle C_2,s\rangle)\in R$.

For $f:X_1\rightarrow X_2$, we define $f[x_1\mapsto x_2]:X_1\cup\{x_1\}\rightarrow X_2\cup\{x_2\}$, $z\in X_1\cup\{x_1\}$,(1)$f[x_1\mapsto x_2](z)=
x_2$,if $z=x_1$;(2)$f[x_1\mapsto x_2](z)=f(z)$, otherwise. Where $X_1\subseteq \mathbb{E}_1$, $X_2\subseteq \mathbb{E}_2$, $x_1\in \mathbb{E}_1$, $x_2\in \mathbb{E}_2$.
\end{definition}

\begin{definition}[Weakly posetal product]
Given two PESs $\mathcal{E}_1$, $\mathcal{E}_2$, the weakly posetal product of their configurations, denoted 
$\langle\mathcal{C}(\mathcal{E}_1),S\rangle\overline{\times}\langle\mathcal{C}(\mathcal{E}_2),S\rangle$, is defined as

$$\{(\langle C_1,s\rangle,f,\langle C_2,s\rangle)|C_1\in\mathcal{C}(\mathcal{E}_1),C_2\in\mathcal{C}(\mathcal{E}_2),f:\hat{C_1}\rightarrow \hat{C_2} \textrm{ isomorphism}\}.$$

A subset $R\subseteq\langle\mathcal{C}(\mathcal{E}_1),S\rangle\overline{\times}\langle\mathcal{C}(\mathcal{E}_2),S\rangle$ is called a weakly posetal relation. We say that $R$ is 
downward closed when for any $(\langle C_1,s\rangle,f,\langle C_2,s\rangle),(\langle C_1',s'\rangle,f,\langle C_2',s'\rangle)\in \langle\mathcal{C}(\mathcal{E}_1),S\rangle\overline{\times}\langle\mathcal{C}(\mathcal{E}_2),S\rangle$, 
if $(\langle C_1,s\rangle,f,\langle C_2,s\rangle)\subseteq (\langle C_1',s'\rangle,f',\langle C_2',s'\rangle)$ pointwise and 
$(\langle C_1',s'\rangle,f',\langle C_2',s'\rangle)\in R$, then $(\langle C_1,s\rangle,f,\langle C_2,s\rangle)\in R$.

For $f:X_1\rightarrow X_2$, we define $f[x_1\mapsto x_2]:X_1\cup\{x_1\}\rightarrow X_2\cup\{x_2\}$, $z\in X_1\cup\{x_1\}$,(1)$f[x_1\mapsto x_2](z)=
x_2$,if $z=x_1$;(2)$f[x_1\mapsto x_2](z)=f(z)$, otherwise. Where $X_1\subseteq \hat{\mathbb{E}_1}$, $X_2\subseteq \hat{\mathbb{E}_2}$, $x_1\in \hat{\mathbb{E}}_1$, 
$x_2\in \hat{\mathbb{E}}_2$. Also, we define $f(\tau^*)=f(\tau^*)$.
\end{definition}

\begin{definition}[Probabilistic (hereditary) history-preserving bisimulation]\label{HHPBG}
A probabilistic history-preserving (hp-) bisimulation is a posetal relation 
$R\subseteq\langle\mathcal{C}(\mathcal{E}_1),S\rangle\overline{\times}\langle\mathcal{C}(\mathcal{E}_2),S\rangle$ such that (1) if $(\langle C_1,s\rangle,f,\langle C_2,s\rangle)\in R$, 
and $\langle C_1,s\rangle\xrightarrow{e_1} \langle C_1',s'\rangle$, then $\langle C_2,s\rangle\xrightarrow{e_2} \langle C_2',s'\rangle$, with 
$(\langle C_1',s'\rangle,f[e_1\mapsto e_2],\langle C_2',s'\rangle)\in R$ for all $s,s'\in S$, and vice-versa; (2) if $(\langle C_1,s\rangle,f,\langle C_2,s\rangle)\in R$, and
$\langle C_1,s\rangle\xrsquigarrow{\pi}\langle C_1^{\pi},s\rangle$ then $\langle C_2,s\rangle\xrsquigarrow{\pi}\langle C_2^{\pi},s\rangle$ and $(\langle C_1^{\pi},s\rangle,f,\langle C_2^{\pi},s\rangle)\in R$, 
and vice-versa; (3) if $(C_1,f,C_2)\in R$, then $\mu(C_1,C)=\mu(C_2,C)$ for each $C\in\mathcal{C}(\mathcal{E})/R$; (4) $[\surd]_R=\{\surd\}$. $\mathcal{E}_1,\mathcal{E}_2$ are 
probabilistic history-preserving (hp-)bisimilar and are written $\mathcal{E}_1\sim_{php}\mathcal{E}_2$ if there exists a probabilistic hp-bisimulation $R$ such that 
$(\langle\emptyset,\emptyset\rangle,\emptyset,\langle\emptyset,\emptyset\rangle)\in R$.

A probabilistic hereditary history-preserving (hhp-)bisimulation is a downward closed probabilistic hp-bisimulation. $\mathcal{E}_1,\mathcal{E}_2$ are probabilistic hereditary 
history-preserving (hhp-)bisimilar and are written $\mathcal{E}_1\sim_{phhp}\mathcal{E}_2$.
\end{definition}

\begin{definition}[Weakly probabilistic (hereditary) history-preserving bisimulation]\label{WHHPBG}
A weakly probabilistic history-preserving (hp-) bisimulation is a weakly posetal relation 
$R\subseteq\langle\mathcal{C}(\mathcal{E}_1),S\rangle\overline{\times}\langle\mathcal{C}(\mathcal{E}_2),S\rangle$ such that (1) if $(\langle C_1,s\rangle,f,\langle C_2,s\rangle)\in R$, 
and $\langle C_1,s\rangle\xRightarrow{e_1} \langle C_1',s'\rangle$, then $\langle C_2,s\rangle\xRightarrow{e_2} \langle C_2',s'\rangle$, with 
$(\langle C_1',s'\rangle,f[e_1\mapsto e_2],\langle C_2',s'\rangle)\in R$ for all $s,s'\in S$, and vice-versa; (2) if $(\langle C_1,s\rangle,f,\langle C_2,s\rangle)\in R$, and
$\langle C_1,s\rangle\xrsquigarrow{\pi}\langle C_1^{\pi},s\rangle$ then $\langle C_2,s\rangle\xrsquigarrow{\pi}\langle C_2^{\pi},s\rangle$ and 
$(\langle C_1^{\pi},s\rangle,f,\langle C_2^{\pi},s\rangle)\in R$, and vice-versa; (3) if $(C_1,f,C_2)\in R$, then $\mu(C_1,C)=\mu(C_2,C)$ for each $C\in\mathcal{C}(\mathcal{E})/R$; 
(4) $[\surd]_R=\{\surd\}$. $\mathcal{E}_1,\mathcal{E}_2$ are weakly probabilistic history-preserving (hp-)bisimilar and are written $\mathcal{E}_1\approx_{php}\mathcal{E}_2$ if there 
exists a weakly probabilistic hp-bisimulation $R$ such that $(\langle\emptyset,\emptyset\rangle,\emptyset,\langle\emptyset,\emptyset\rangle)\in R$.

A weakly probabilistic hereditary history-preserving (hhp-)bisimulation is a downward closed weakly probabilistic hp-bisimulation. $\mathcal{E}_1,\mathcal{E}_2$ are weakly 
probabilistic hereditary history-preserving (hhp-)bisimilar and are written $\mathcal{E}_1\approx_{phhp}\mathcal{E}_2$.
\end{definition}

\begin{definition}[Probabilistic branching pomset, step bisimulation]\label{BPSBG}
Assume a special termination predicate $\downarrow$, and let $\surd$ represent a state with $\surd\downarrow$. Let $\mathcal{E}_1$, $\mathcal{E}_2$ be PESs. A probabilistic branching 
pomset bisimulation is a relation $R\subseteq\langle\mathcal{C}(\mathcal{E}_1),S\rangle\times\langle\mathcal{C}(\mathcal{E}_2),S\rangle$, such that:

 \begin{enumerate}
   \item if $(\langle C_1,s\rangle,\langle C_2,s\rangle)\in R$, and $\langle C_1,s\rangle\xrightarrow{X}\langle C_1',s'\rangle$ then
   \begin{itemize}
     \item either $X\equiv \tau^*$, and $(\langle C_1',s'\rangle,\langle C_2,s\rangle)\in R$ with $s'\in \tau(s)$;
     \item or there is a sequence of (zero or more) probabilistic transitions and $\tau$-transitions $\langle C_2,s\rangle\rightsquigarrow^*\xrightarrow{\tau^*} \langle C_2^0,s^0\rangle$, such that 
     $(\langle C_1,s\rangle,\langle C_2^0,s^0\rangle)\in R$ and $\langle C_2^0,s^0\rangle\xRightarrow{X}\langle C_2',s'\rangle$ with 
     $(\langle C_1',s'\rangle,\langle C_2',s'\rangle)\in R$;
   \end{itemize}
   \item if $(\langle C_1,s\rangle,\langle C_2,s\rangle)\in R$, and $\langle C_2,s\rangle\xrightarrow{X}\langle C_2',s'\rangle$ then
   \begin{itemize}
     \item either $X\equiv \tau^*$, and $(\langle C_1,s\rangle,\langle C_2',s'\rangle)\in R$;
     \item or there is a sequence of (zero or more) probabilistic transitions and $\tau$-transitions $\langle C_1,s\rangle\rightsquigarrow^*\xrightarrow{\tau^*} \langle C_1^0,s^0\rangle$, such that 
     $(\langle C_1^0,s^0\rangle,\langle C_2,s\rangle)\in R$ and $\langle C_1^0,s^0\rangle\xRightarrow{X}\langle C_1',s'\rangle$ with 
     $(\langle C_1',s'\rangle,\langle C_2',s'\rangle)\in R$;
   \end{itemize}
   \item if $(\langle C_1,s\rangle,\langle C_2,s\rangle)\in R$ and $\langle C_1,s\rangle\downarrow$, then there is a sequence of (zero or more) probabilistic transitions and $\tau$-transitions 
   $\langle C_2,s\rangle\rightsquigarrow^*\xrightarrow{\tau^*}\langle C_2^0,s^0\rangle$ such that $(\langle C_1,s\rangle,\langle C_2^0,s^0\rangle)\in R$ and 
   $\langle C_2^0,s^0\rangle\downarrow$;
   \item if $(\langle C_1,s\rangle,\langle C_2,s\rangle)\in R$ and $\langle C_2,s\rangle\downarrow$, then there is a sequence of (zero or more) probabilistic transitions and $\tau$-transitions 
   $\langle C_1,s\rangle\rightsquigarrow^*\xrightarrow{\tau^*}\langle C_1^0,s^0\rangle$ such that $(\langle C_1^0,s^0\rangle,\langle C_2,s\rangle)\in R$ and 
   $\langle C_1^0,s^0\rangle\downarrow$;
   \item if $(C_1,C_2)\in R$,then $\mu(C_1,C)=\mu(C_2,C)$ for each $C\in\mathcal{C}(\mathcal{E})/R$;
   \item $[\surd]_R=\{\surd\}$.
 \end{enumerate}

We say that $\mathcal{E}_1$, $\mathcal{E}_2$ are probabilistic branching pomset bisimilar, written $\mathcal{E}_1\approx_{pbp}\mathcal{E}_2$, if there exists a probabilistic branching 
pomset bisimulation $R$, such that $(\langle\emptyset,\emptyset\rangle,\langle\emptyset,\emptyset\rangle)\in R$.

By replacing probabilistic pomset transitions with steps, we can get the definition of probabilistic branching step bisimulation. When PESs $\mathcal{E}_1$ and $\mathcal{E}_2$ are 
probabilistic branching step bisimilar, we write $\mathcal{E}_1\approx_{pbs}\mathcal{E}_2$.
\end{definition}

\begin{definition}[Probabilistic rooted branching pomset, step bisimulation]\label{RBPSBG}
Assume a special termination predicate $\downarrow$, and let $\surd$ represent a state with $\surd\downarrow$. Let $\mathcal{E}_1$, $\mathcal{E}_2$ be PESs. A probabilistic rooted 
branching pomset bisimulation is a relation $R\subseteq\langle\mathcal{C}(\mathcal{E}_1),S\rangle\times\langle\mathcal{C}(\mathcal{E}_2),S\rangle$, such that:

 \begin{enumerate}
   \item if $(\langle C_1,s\rangle,\langle C_2,s\rangle)\in R$, and $\langle C_1,s\rangle\rightsquigarrow\xrightarrow{X}\langle C_1',s'\rangle$ then 
   $\langle C_2,s\rangle\rightsquigarrow\xrightarrow{X}\langle C_2',s'\rangle$ with $\langle C_1',s'\rangle\approx_{pbp}\langle C_2',s'\rangle$;
   \item if $(\langle C_1,s\rangle,\langle C_2,s\rangle)\in R$, and $\langle C_2,s\rangle\rightsquigarrow\xrightarrow{X}\langle C_2',s'\rangle$ then 
   $\langle C_1,s\rangle\rightsquigarrow\xrightarrow{X}\langle C_1',s'\rangle$ with $\langle C_1',s'\rangle\approx_{pbp}\langle C_2',s'\rangle$;
   \item if $(\langle C_1,s\rangle,\langle C_2,s\rangle)\in R$ and $\langle C_1,s\rangle\downarrow$, then $\langle C_2,s\rangle\downarrow$;
   \item if $(\langle C_1,s\rangle,\langle C_2,s\rangle)\in R$ and $\langle C_2,s\rangle\downarrow$, then $\langle C_1,s\rangle\downarrow$.
 \end{enumerate}

We say that $\mathcal{E}_1$, $\mathcal{E}_2$ are probabilistic rooted branching pomset bisimilar, written $\mathcal{E}_1\approx_{prbp}\mathcal{E}_2$, if there exists a probabilistic 
rooted branching pomset bisimulation $R$, such that $(\langle\emptyset,\emptyset\rangle,\langle\emptyset,\emptyset\rangle)\in R$.

By replacing pomset transitions with steps, we can get the definition of probabilistic rooted branching step bisimulation. When PESs $\mathcal{E}_1$ and $\mathcal{E}_2$ are probabilistic 
rooted branching step bisimilar, we write $\mathcal{E}_1\approx_{prbs}\mathcal{E}_2$.
\end{definition}

\begin{definition}[Probabilistic branching (hereditary) history-preserving bisimulation]\label{BHHPBG}
Assume a special termination predicate $\downarrow$, and let $\surd$ represent a state with $\surd\downarrow$. A probabilistic branching history-preserving (hp-) bisimulation is a 
weakly posetal relation $R\subseteq\langle\mathcal{C}(\mathcal{E}_1),S\rangle\overline{\times}\langle\mathcal{C}(\mathcal{E}_2),S\rangle$ such that:

 \begin{enumerate}
   \item if $(\langle C_1,s\rangle,f,\langle C_2,s\rangle)\in R$, and $\langle C_1,s\rangle\xrightarrow{e_1}\langle C_1',s'\rangle$ then
   \begin{itemize}
     \item either $e_1\equiv \tau$, and $(\langle C_1',s'\rangle,f[e_1\mapsto \tau],\langle C_2,s\rangle)\in R$;
     \item or there is a sequence of (zero or more) probabilistic transitions and $\tau$-transitions $\langle C_2,s\rangle\rightsquigarrow^*\xrightarrow{\tau^*} \langle C_2^0,s^0\rangle$, such that 
     $(\langle C_1,s\rangle,f,\langle C_2^0,s^0\rangle)\in R$ and $\langle C_2^0,s^0\rangle\xrightarrow{e_2}\langle C_2',s'\rangle$ with 
     $(\langle C_1',s'\rangle,f[e_1\mapsto e_2],\langle C_2',s'\rangle)\in R$;
   \end{itemize}
   \item if $(\langle C_1,s\rangle,f,\langle C_2,s\rangle)\in R$, and $\langle C_2,s\rangle\xrightarrow{e_2}\langle C_2',s'\rangle$ then
   \begin{itemize}
     \item either $e_2\equiv \tau$, and $(\langle C_1,s\rangle,f[e_2\mapsto \tau],\langle C_2',s'\rangle)\in R$;
     \item or there is a sequence of (zero or more) probabilistic transitions and $\tau$-transitions $\langle C_1,s\rangle\rightsquigarrow^*\xrightarrow{\tau^*} \langle C_1^0,s^0\rangle$, such that 
     $(\langle C_1^0,s^0\rangle,f,\langle C_2,s\rangle)\in R$ and $\langle C_1^0,s^0\rangle\xrightarrow{e_1}\langle C_1',s'\rangle$ with 
     $(\langle C_1',s'\rangle,f[e_2\mapsto e_1],\langle C_2',s'\rangle)\in R$;
   \end{itemize}
   \item if $(\langle C_1,s\rangle,f,\langle C_2,s\rangle)\in R$ and $\langle C_1,s\rangle\downarrow$, then there is a sequence of (zero or more) probabilistic transitions and $\tau$-transitions 
   $\langle C_2,s\rangle\rightsquigarrow^*\xrightarrow{\tau^*}\langle C_2^0,s^0\rangle$ such that $(\langle C_1,s\rangle,f,\langle C_2^0,s^0\rangle)\in R$ and 
   $\langle C_2^0,s^0\rangle\downarrow$;
   \item if $(\langle C_1,s\rangle,f,\langle C_2,s\rangle)\in R$ and $\langle C_2,s\rangle\downarrow$, then there is a sequence of (zero or more) probabilistic transitions and $\tau$-transitions 
   $\langle C_1,s\rangle\rightsquigarrow^*\xrightarrow{\tau^*}\langle C_1^0,s^0\rangle$ such that $(\langle C_1^0,s^0\rangle,f,\langle C_2,s\rangle)\in R$ and 
   $\langle C_1^0,s^0\rangle\downarrow$;
   \item if $(C_1,C_2)\in R$,then $\mu(C_1,C)=\mu(C_2,C)$ for each $C\in\mathcal{C}(\mathcal{E})/R$;
   \item $[\surd]_R=\{\surd\}$.
 \end{enumerate}

$\mathcal{E}_1,\mathcal{E}_2$ are probabilistic branching history-preserving (hp-)bisimilar and are written $\mathcal{E}_1\approx_{pbhp}\mathcal{E}_2$ if there exists a probabilistic 
branching hp-bisimulation $R$ such that $(\langle\emptyset,\emptyset\rangle,\emptyset,\langle\emptyset,\emptyset\rangle)\in R$.

A probabilistic branching hereditary history-preserving (hhp-)bisimulation is a downward closed probabilistic branching hp-bisimulation. $\mathcal{E}_1,\mathcal{E}_2$ are probabilistic 
branching hereditary history-preserving (hhp-)bisimilar and are written $\mathcal{E}_1\approx_{pbhhp}\mathcal{E}_2$.
\end{definition}

\begin{definition}[Probabilistic rooted branching (hereditary) history-preserving bisimulation]\label{RBHHPBG}
Assume a special termination predicate $\downarrow$, and let $\surd$ represent a state with $\surd\downarrow$. A probabilistic rooted branching history-preserving (hp-) bisimulation is 
a weakly posetal relation $R\subseteq\langle\mathcal{C}(\mathcal{E}_1),S\rangle\overline{\times}\langle\mathcal{C}(\mathcal{E}_2),S\rangle$ such that:

 \begin{enumerate}
   \item if $(\langle C_1,s\rangle,f,\langle C_2,s\rangle)\in R$, and $\langle C_1,s\rangle\rightsquigarrow\xrightarrow{e_1}\langle C_1',s'\rangle$, then 
   $\langle C_2,s\rangle\rightsquigarrow\xrightarrow{e_2}\langle C_2',s'\rangle$ with $\langle C_1',s'\rangle\approx_{pbhp}\langle C_2',s'\rangle$;
   \item if $(\langle C_1,s\rangle,f,\langle C_2,s\rangle)\in R$, and $\langle C_2,s\rangle\rightsquigarrow\xrightarrow{e_2}\langle C_2',s'\rangle$, then 
   $\langle C_1,s\rangle\rightsquigarrow\xrightarrow{e_1}\langle C_1',s'\rangle$ with $\langle C_1',s'\rangle\approx_{pbhp}\langle C_2',s'\rangle$;
   \item if $(\langle C_1,s\rangle,f,\langle C_2,s\rangle)\in R$ and $\langle C_1,s\rangle\downarrow$, then $\langle C_2,s\rangle\downarrow$;
   \item if $(\langle C_1,s\rangle,f,\langle C_2,s\rangle)\in R$ and $\langle C_2,s\rangle\downarrow$, then $\langle C_1,s\rangle\downarrow$.
 \end{enumerate}

$\mathcal{E}_1,\mathcal{E}_2$ are probabilistic rooted branching history-preserving (hp-)bisimilar and are written $\mathcal{E}_1\approx_{prbhp}\mathcal{E}_2$ if there exists a probabilistic 
rooted branching hp-bisimulation $R$ such that $(\langle\emptyset,\emptyset\rangle,\emptyset,\langle\emptyset,\emptyset\rangle)\in R$.

A probabilistic rooted branching hereditary history-preserving (hhp-)bisimulation is a downward closed probabilistic rooted branching hp-bisimulation. $\mathcal{E}_1,\mathcal{E}_2$ are 
probabilistic rooted branching hereditary history-preserving (hhp-)bisimilar and are written $\mathcal{E}_1\approx_{prbhhp}\mathcal{E}_2$.
\end{definition}

\subsection{$BAPTC$ with Guards}{\label{batcg}}

In this subsection, we will discuss the guards for $BAPTC$, which is denoted as $BAPTC_G$. Let $\mathbb{E}$ be the set of atomic events (actions), $G_{at}$ be the set of atomic guards,
$\delta$ be the deadlock constant, and $\epsilon$ be the empty event. We extend $G_{at}$ to the set of basic guards $G$ with element $\phi,\psi,\cdots$, which is generated by the
following formation rules:

$$\phi::=\delta|\epsilon|\neg\phi|\psi\in G_{at}|\phi+\psi|\phi\boxplus_{\pi}\psi|\phi\cdot\psi$$

In the following, let $e_1, e_2, e_1', e_2'\in \mathbb{E}$, $\phi,\psi\in G$ and let variables $x,y,z$ range over the set of terms for true concurrency, $p,q,s$ range over the set of
closed terms. The predicate $test(\phi,s)$ represents that $\phi$ holds in the state $s$, and $test(\epsilon,s)$ holds and $test(\delta,s)$ does not hold. $effect(e,s)\in S$ denotes
$s'$ in $s\xrightarrow{e}s'$. The predicate weakest precondition $wp(e,\phi)$ denotes that $\forall s,s'\in S, test(\phi,effect(e,s))$ holds.

The set of axioms of $BAPTC_G$ consists of the laws given in Table \ref{AxiomsForBAPTCG}.

\begin{center}
    \begin{table}
        \begin{tabular}{@{}ll@{}}
            \hline No. &Axiom\\
            $A1$ & $x+ y = y+ x$\\
            $A2$ & $(x+ y)+ z = x+ (y+ z)$\\
            $A3$ & $e+ e = e$\\
            $A4$ & $(x+ y)\cdot z = x\cdot z + y\cdot z$\\
            $A5$ & $(x\cdot y)\cdot z = x\cdot(y\cdot z)$\\
            $A6$ & $x+\delta = x$\\
            $A7$ & $\delta\cdot x = \delta$\\
            $A8$ & $\epsilon\cdot x = x$\\
            $A9$ & $x\cdot\epsilon = x$\\
            $PA1$ & $x\boxplus_{\pi} y=y\boxplus_{1-\pi} x$\\
            $PA2$ & $x\boxplus_{\pi}(y\boxplus_{\rho} z)=(x\boxplus_{\frac{\pi}{\pi+\rho-\pi\rho}}y)\boxplus_{\pi+\rho-\pi\rho} z$\\
            $PA3$ & $x\boxplus_{\pi}x=x$\\
            $PA4$ & $(x\boxplus_{\pi}y)\cdot z=x\cdot z\boxplus_{\pi}y\cdot z$\\
            $PA5$ & $(x\boxplus_{\pi}y)+z=(x+z)\boxplus_{\pi}(y+z)$\\
            $G1$ & $\phi\cdot\neg\phi = \delta$\\
            $G2$ & $\phi+\neg\phi = \epsilon$\\
            $PG1$ & $\phi\boxplus_{\pi}\neg\phi = \epsilon$\\
            $G3$ & $\phi\delta = \delta$\\
            $G4$ & $\phi(x+y)=\phi x+\phi y$\\
            $PG2$ & $\phi(x\boxplus_{\pi}y)=\phi x\boxplus_{\pi}\phi y$\\
            $G5$ & $\phi(x\cdot y)= \phi x\cdot y$\\
            $G6$ & $(\phi+\psi)x = \phi x + \psi x$\\
            $PG3$ & $(\phi\boxplus_{\pi}\psi)x = \phi x \boxplus_{\pi} \psi x$\\
            $G7$ & $(\phi\cdot \psi)\cdot x = \phi\cdot(\psi\cdot x)$\\
            $G8$ & $\phi=\epsilon$ if $\forall s\in S.test(\phi,s)$\\
            $G9$ & $\phi_0\cdot\cdots\cdot\phi_n = \delta$ if $\forall s\in S,\exists i\leq n.test(\neg\phi_i,s)$\\
            $G10$ & $wp(e,\phi)e\phi=wp(e,\phi)e$\\
            $G11$ & $\neg wp(e,\phi)e\neg\phi=\neg wp(e,\phi)e$\\
        \end{tabular}
        \caption{Axioms of $BAPTC_G$}
        \label{AxiomsForBAPTCG}
    \end{table}
\end{center}

Note that, by eliminating atomic event from the process terms, the axioms in Table \ref{AxiomsForBAPTCG} will lead to a Boolean Algebra. And $G8$ and $G9$ are preconditions of $e$ and
$\phi$, $G10$ is the weakest precondition of $e$ and $\phi$. A data environment with $effect$ function is sufficiently deterministic, and it is obvious that if the weakest precondition
is expressible and $G10$, $G11$ are sound, then the related data environment is sufficiently deterministic.

\begin{definition}[Basic terms of $BAPTC_G$]\label{BTBAPTCG}
The set of basic terms of $BAPTC_G$, $\mathcal{B}(BAPTC_G)$, is inductively defined as follows:

\begin{enumerate}
  \item $\mathbb{E}\subset\mathcal{B}(BAPTC_G)$;
  \item $G\subset\mathcal{B}(BAPTC_G)$;
  \item if $e\in \mathbb{E}, t\in\mathcal{B}(BAPTC_G)$ then $e\cdot t\in\mathcal{B}(BAPTC_G)$;
  \item if $\phi\in G, t\in\mathcal{B}(BAPTC_G)$ then $\phi\cdot t\in\mathcal{B}(BAPTC_G)$;
  \item if $t,s\in\mathcal{B}(BAPTC_G)$ then $t+ s\in\mathcal{B}(BAPTC_G)$;
  \item if $t,s\in\mathcal{B}(BAPTC_G)$ then $t\boxplus_{\pi} s\in\mathcal{B}(BAPTC_G)$.
\end{enumerate}
\end{definition}

\begin{theorem}[Elimination theorem of $BAPTC_G$]\label{ETBAPTCG}
Let $p$ be a closed $BAPTC_G$ term. Then there is a basic $BAPTC_G$ term $q$ such that $BAPTC_G\vdash p=q$.
\end{theorem}

\begin{proof}
(1) Firstly, suppose that the following ordering on the signature of $BAPTC_G$ is defined: $\cdot > +>\boxplus_{\pi}$ and the symbol $\cdot$ is given the lexicographical status for
the first argument, then for each rewrite rule $p\rightarrow q$ in Table \ref{TRSForBAPTCG} relation $p>_{lpo} q$ can easily be proved. We obtain that the term rewrite system shown
in Table \ref{TRSForBAPTCG} is strongly normalizing, for it has finitely many rewriting rules, and $>$ is a well-founded ordering on the signature of $BAPTC_G$, and if $s>_{lpo} t$,
for each rewriting rule $s\rightarrow t$ is in Table \ref{TRSForBAPTCG} (see Theorem \ref{SN}).

\begin{center}
    \begin{table}
        \begin{tabular}{@{}ll@{}}
            \hline No. &Rewriting Rule\\
            $RA3$ & $e+ e \rightarrow e$\\
            $RA4$ & $(x+ y)\cdot z \rightarrow x\cdot z + y\cdot z$\\
            $RA5$ & $(x\cdot y)\cdot z \rightarrow x\cdot(y\cdot z)$\\
            $RA6$ & $x+\delta \rightarrow x$\\
            $RA7$ & $\delta\cdot x \rightarrow \delta$\\
            $RA8$ & $\epsilon\cdot x \rightarrow x$\\
            $RA9$ & $x\cdot\epsilon \rightarrow x$\\
            $RPA1$ & $x\boxplus_{\pi} y\rightarrow y\boxplus_{1-\pi} x$\\
            $RPA2$ & $x\boxplus_{\pi}(y\boxplus_{\rho} z)\rightarrow(x\boxplus_{\frac{\pi}{\pi+\rho-\pi\rho}}y)\boxplus_{\pi+\rho-\pi\rho} z$\\
            $RPA3$ & $x\boxplus_{\pi}x\rightarrow x$\\
            $RPA4$ & $(x\boxplus_{\pi}y)\cdot z\rightarrow x\cdot z\boxplus_{\pi}y\cdot z$\\
            $RPA5$ & $(x\boxplus_{\pi}y)+z\rightarrow (x+z)\boxplus_{\pi}(y+z)$\\
            $RG1$ & $\phi\cdot\neg\phi \rightarrow \delta$\\
            $RG2$ & $\phi+\neg\phi \rightarrow \epsilon$\\
            $RPG1$ & $\phi\boxplus_{\pi}\neg\phi \rightarrow \epsilon$\\
            $RG3$ & $\phi\delta \rightarrow \delta$\\
            $RG4$ & $\phi(x+y)\rightarrow\phi x+\phi y$\\
            $RPG2$ & $\phi(x\boxplus_{\pi}y)\rightarrow\phi x\boxplus_{\pi}\phi y$\\
            $RG5$ & $\phi(x\cdot y)\rightarrow \phi x\cdot y$\\
            $RG6$ & $(\phi+\psi)x \rightarrow \phi x + \psi x$\\
            $RPG3$ & $(\phi\boxplus_{\pi}\psi)x \rightarrow \phi x \boxplus_{\pi} \psi x$\\
            $RG7$ & $(\phi\cdot \psi)\cdot x \rightarrow \phi\cdot(\psi\cdot x)$\\
            $RG8$ & $\phi\rightarrow \epsilon$ if $\forall s\in S.test(\phi,s)$\\
            $RG9$ & $\phi_0\cdot\cdots\cdot\phi_n \rightarrow \delta$ if $\forall s\in S,\exists i\leq n.test(\neg\phi_i,s)$\\
            $RG10$ & $wp(e,\phi)e\phi\rightarrow wp(e,\phi)e$\\
            $RG11$ & $\neg wp(e,\phi)e\neg\phi\rightarrow\neg wp(e,\phi)e$\\
        \end{tabular}
        \caption{Term rewrite system of $BAPTC_G$}
        \label{TRSForBAPTCG}
    \end{table}
\end{center}

(2) Then we prove that the normal forms of closed $BAPTC_G$ terms are basic $BAPTC_G$ terms.

Suppose that $p$ is a normal form of some closed $BAPTC_G$ term and suppose that $p$ is not a basic term. Let $p'$ denote the smallest sub-term of $p$ which is not a basic term. It
implies that each sub-term of $p'$ is a basic term. Then we prove that $p$ is not a term in normal form. It is sufficient to induct on the structure of $p'$:

\begin{itemize}
  \item Case $p'\equiv e, e\in \mathbb{E}$. $p'$ is a basic term, which contradicts the assumption that $p'$ is not a basic term, so this case should not occur.
  \item Case $p'\equiv \phi, \phi\in G$. $p'$ is a basic term, which contradicts the assumption that $p'$ is not a basic term, so this case should not occur.
  \item Case $p'\equiv p_1\cdot p_2$. By induction on the structure of the basic term $p_1$:
      \begin{itemize}
        \item Subcase $p_1\in \mathbb{E}$. $p'$ would be a basic term, which contradicts the assumption that $p'$ is not a basic term;
        \item Subcase $p_1\in G$. $p'$ would be a basic term, which contradicts the assumption that $p'$ is not a basic term;
        \item Subcase $p_1\equiv e\cdot p_1'$. $RA5$ or $RA9$ rewriting rule can be applied. So $p$ is not a normal form;
        \item Subcase $p_1\equiv \phi\cdot p_1'$. $RG1$, $RG3$, $RG4$, $RG5$, $RG7$, or $RG8-9$ rewriting rules can be applied. So $p$ is not a normal form;
        \item Subcase $p_1\equiv p_1'+ p_1''$. $RA4$, $RA6$, $RG2$, or $RG6$ rewriting rules can be applied. So $p$ is not a normal form;
        \item Subcase $p_1\equiv p_1'\boxplus_{\pi} p_1''$. $RPG3$ rewriting rule can be applied. So $p$ is not a normal form.
      \end{itemize}
  \item Case $p'\equiv p_1+ p_2$. By induction on the structure of the basic terms both $p_1$ and $p_2$, all subcases will lead to that $p'$ would be a basic term, which contradicts the
  assumption that $p'$ is not a basic term.
  \item Case $p'\equiv p_1\boxplus_{\pi} p_2$. By induction on the structure of the basic terms both $p_1$ and $p_2$, all subcases will lead to that $p'$ would be a basic term, which contradicts the
  assumption that $p'$ is not a basic term.
\end{itemize}
\end{proof}

We will define a term-deduction system which gives the operational semantics of $BAPTC_G$. We give the operational transition rules for $\epsilon$, atomic guard $\phi\in G_{at}$,
atomic event $e\in\mathbb{E}$, operators $\cdot$ and $+$ as Table \ref{SETRForBAPTCG} shows. And the predicate $\xrightarrow{e}\surd$ represents successful termination after execution
of the event $e$.

\begin{center}
    \begin{table}
        $$\frac{}{\langle\epsilon,s\rangle\rightsquigarrow\langle\breve{\epsilon},s\rangle}$$
        $$\frac{}{\langle e,s\rangle\rightsquigarrow\langle\breve{e},s\rangle}$$
        $$\frac{}{\langle\phi,s\rangle\rightsquigarrow\langle\breve{\phi},s\rangle}$$
        $$\frac{\langle x,s\rangle\rightsquigarrow \langle x',s\rangle}{\langle x\cdot y,s\rangle\rightsquigarrow \langle x'\cdot y,s\rangle}$$
        $$\frac{\langle x,s\rangle\rightsquigarrow \langle x',s\rangle\quad \langle y,s\rangle\rightsquigarrow \langle y',s\rangle}{\langle x+y,s\rangle\rightsquigarrow \langle x'+y',s\rangle}$$
        $$\frac{\langle x,s\rangle\rightsquigarrow \langle x',s\rangle}{\langle x\boxplus_{\pi}y,s\rangle\rightsquigarrow \langle x',s\rangle}\quad \frac{\langle y,s\rangle\rightsquigarrow \langle y',s\rangle}{\langle x\boxplus_{\pi}y,s\rangle\rightsquigarrow \langle y',s\rangle}$$

        $$\frac{}{\langle\breve{\epsilon},s\rangle\rightarrow\langle\surd,s\rangle}$$
        $$\frac{}{\langle \breve{e},s\rangle\xrightarrow{e}\langle\surd,s'\rangle}\textrm{ if }s'\in effect(e,s)$$
        $$\frac{}{\langle\breve{\phi},s\rangle\rightarrow\langle\surd,s\rangle}\textrm{ if }test(\phi,s)$$
        $$\frac{\langle x,s\rangle\xrightarrow{e}\langle\surd,s'\rangle}{\langle x+ y,s\rangle\xrightarrow{e}\langle\surd,s'\rangle} \quad\frac{\langle x,s\rangle\xrightarrow{e}\langle x',s'\rangle}{\langle x+ y,s\rangle\xrightarrow{e}\langle x',s'\rangle}$$
        $$\frac{\langle y,s\rangle\xrightarrow{e}\langle\surd,s'\rangle}{\langle x+ y,s\rangle\xrightarrow{e}\langle\surd,s'\rangle} \quad\frac{\langle y,s\rangle\xrightarrow{e}\langle y',s'\rangle}{\langle x+ y,s\rangle\xrightarrow{e}\langle y',s'\rangle}$$
        $$\frac{\langle x,s\rangle\xrightarrow{e}\langle\surd,s'\rangle}{\langle x\cdot y,s\rangle\xrightarrow{e} \langle y,s'\rangle} \quad\frac{\langle x,s\rangle\xrightarrow{e}\langle x',s'\rangle}{\langle x\cdot y,s\rangle\xrightarrow{e}\langle x'\cdot y,s'\rangle}$$
        \caption{Single event transition rules of $BAPTC_G$}
        \label{SETRForBAPTCG}
    \end{table}
\end{center}

Note that, we replace the single atomic event $e\in\mathbb{E}$ by $X\subseteq\mathbb{E}$, we can obtain the pomset transition rules of $BAPTC_G$, and omit them.

\begin{theorem}[Congruence of $BAPTC_G$ with respect to probabilistic truly concurrent bisimulation equivalences]\label{CBAPTCG}
(1) Probabilistic pomset bisimulation equivalence $\sim_{pp}$ is a congruence with respect to $BAPTC_G$.

(2) Probabilistic step bisimulation equivalence $\sim_{ps}$ is a congruence with respect to $BAPTC_G$.

(3) Probabilistic hp-bisimulation equivalence $\sim_{php}$ is a congruence with respect to $BAPTC_G$.

(4) Probabilistic hhp-bisimulation equivalence $\sim_{phhp}$ is a congruence with respect to $BAPTC_G$.
\end{theorem}

\begin{proof}
(1) It is easy to see that probabilistic pomset bisimulation is an equivalent relation on $BAPTC_G$ terms, we only need to prove that $\sim_{pp}$ is preserved by the operators $\cdot$
, $+$ and $\boxplus_{\pi}$. It is trivial and we leave the proof as an exercise for the readers.

(2) It is easy to see that probabilistic step bisimulation is an equivalent relation on $BAPTC_G$ terms, we only need to prove that $\sim_{ps}$ is preserved by the operators $\cdot$,
$+$ and $\boxplus_{\pi}$. It is trivial and we leave the proof as an exercise for the readers.

(3) It is easy to see that probabilistic hp-bisimulation is an equivalent relation on $BAPTC_G$ terms, we only need to prove that $\sim_{php}$ is preserved by the operators $\cdot$,
$+$, and $\boxplus_{\pi}$. It is trivial and we leave the proof as an exercise for the readers.

(4) It is easy to see that probabilistic hhp-bisimulation is an equivalent relation on $BAPTC_G$ terms, we only need to prove that $\sim_{phhp}$ is preserved by the operators $\cdot$,
$+$, and $\boxplus_{\pi}$. It is trivial and we leave the proof as an exercise for the readers.
\end{proof}

\begin{theorem}[Soundness of $BAPTC_G$ modulo probabilistic truly concurrent bisimulation equivalences]\label{SBAPTCG}
(1) Let $x$ and $y$ be $BAPTC_G$ terms. If $BAPTC_G\vdash x=y$, then $x\sim_{pp} y$.

(2) Let $x$ and $y$ be $BAPTC_G$ terms. If $BAPTC_G\vdash x=y$, then $x\sim_{ps} y$.

(3) Let $x$ and $y$ be $BAPTC_G$ terms. If $BAPTC_G\vdash x=y$, then $x\sim_{php} y$.

(4) Let $x$ and $y$ be $BAPTC_G$ terms. If $BAPTC_G\vdash x=y$, then $x\sim_{phhp} y$.
\end{theorem}

\begin{proof}
(1) Since probabilistic pomset bisimulation $\sim_{pp}$ is both an equivalent and a congruent relation, we only need to check if each axiom in Table \ref{AxiomsForBAPTCG} is sound
modulo probabilistic pomset bisimulation equivalence. We leave the proof as an exercise for the readers.

(2) Since probabilistic step bisimulation $\sim_{ps}$ is both an equivalent and a congruent relation, we only need to check if each axiom in Table \ref{AxiomsForBAPTCG} is sound modulo
probabilistic step bisimulation equivalence. We leave the proof as an exercise for the readers.

(3) Since probabilistic hp-bisimulation $\sim_{php}$ is both an equivalent and a congruent relation, we only need to check if each axiom in Table \ref{AxiomsForBAPTCG} is sound modulo
probabilistic hp-bisimulation equivalence. We leave the proof as an exercise for the readers.

(4) Since probabilistic hhp-bisimulation $\sim_{phhp}$ is both an equivalent and a congruent relation, we only need to check if each axiom in Table \ref{AxiomsForBAPTCG} is sound modulo
probabilistic hhp-bisimulation equivalence. We leave the proof as an exercise for the readers.
\end{proof}

\begin{theorem}[Completeness of $BAPTC_G$ modulo probabilistic truly concurrent bisimulation equivalences]\label{CBAPTCG}
(1) Let $p$ and $q$ be closed $BAPTC_G$ terms, if $p\sim_{pp} q$ then $p=q$.

(2) Let $p$ and $q$ be closed $BAPTC_G$ terms, if $p\sim_{ps} q$ then $p=q$.

(3) Let $p$ and $q$ be closed $BAPTC_G$ terms, if $p\sim_{php} q$ then $p=q$.

(4) Let $p$ and $q$ be closed $BAPTC_G$ terms, if $p\sim_{phhp} q$ then $p=q$.
\end{theorem}

\begin{proof}
(1) Firstly, by the elimination theorem of $BAPTC_G$, we know that for each closed $BAPTC_G$ term $p$, there exists a closed basic $BAPTC_G$ term $p'$, such that $BAPTC_G\vdash p=p'$,
so, we only need to consider closed basic $BAPTC_G$ terms.

The basic terms (see Definition \ref{BTBAPTCG}) modulo associativity and commutativity (AC) of conflict $+$ (defined by axioms $A1$ and $A2$ in Table \ref{AxiomsForBAPTCG}), and this
equivalence is denoted by $=_{AC}$. Then, each equivalence class $s$ modulo AC of $+$ has the following normal form

$$s_1\boxplus_{\pi_1}\cdots\boxplus_{\pi_{k-1}} s_k$$

with each $s_i$ has the following form

$$t_1+\cdots+ t_l$$

with each $t_j$ either an atomic event or of the form $u_1\cdot u_2$, and each $t_j$ is called the summand of $s$.

Now, we prove that for normal forms $n$ and $n'$, if $n\sim_{pp} n'$ then $n=_{AC}n'$. It is sufficient to induct on the sizes of $n$ and $n'$.

\begin{itemize}
  \item Consider a summand $e$ of $n$. Then $\langle n,s\rangle\rightsquigarrow\xrightarrow{e}\langle \surd,s'\rangle$, so $n\sim_{pp} n'$ implies $\langle n',s\rangle\rightsquigarrow
  \xrightarrow{e}\langle \surd,s\rangle$, meaning that $n'$ also contains the summand $e$.
  \item Consider a summand $\phi$ of $n$. Then $\langle n,s\rangle\rightsquigarrow\rightarrow\langle \surd,s\rangle$, if $test(\phi,s)$ holds, so $n\sim_{pp} n'$ implies
  $\langle n',s\rangle\rightsquigarrow\rightarrow\langle \surd,s\rangle$, if $test(\phi,s)$ holds, meaning that $n'$ also contains the summand $\phi$.
  \item Consider a summand $t_1\cdot t_2$ of $n$. Then $\langle n,s\rangle\rightsquigarrow\xrightarrow{t_1}\langle t_2,s'\rangle$, so $n\sim_{pp} n'$ implies $\langle n',s\rangle
  \rightsquigarrow\xrightarrow{t_1}\langle t_2',s'\rangle$ with $t_2\sim_{pp} t_2'$, meaning that $n'$ contains a summand $t_1\cdot t_2'$. Since $t_2$ and $t_2'$ are normal forms and
  have sizes smaller than $n$ and $n'$, by the induction hypotheses $t_2\sim_{pp} t_2'$ implies $t_2=_{AC} t_2'$.
\end{itemize}

So, we get $n=_{AC} n'$.

Finally, let $s$ and $t$ be basic terms, and $s\sim_{pp} t$, there are normal forms $n$ and $n'$, such that $s=n$ and $t=n'$. The soundness theorem of $BAPTC_G$ modulo probabilistic
pomset bisimulation equivalence (see Theorem \ref{SBAPTCG}) yields $s\sim_{pp} n$ and $t\sim_{pp} n'$, so $n\sim_{pp} s\sim_{pp} t\sim_{pp} n'$. Since if $n\sim_{pp} n'$ then $n=_{AC}n'$,
$s=n=_{AC}n'=t$, as desired.

(2) It can be proven similarly as (1).

(3) It can be proven similarly as (1).

(4) It can be proven similarly as (1).
\end{proof}

\begin{theorem}[Sufficient determinism]
All related data environments with respect to $BAPTC_G$ can be sufficiently deterministic.
\end{theorem}

\begin{proof}
It only needs to check $effect(t,s)$ function is deterministic, and is sufficient to induct on the structure of term $t$. The only matter are the cases $t=t_1+t_2$ and
$t=t_1\boxplus_{\pi}t_2$, with the help of guards, we can make $t_1=\phi_1\cdot t_1'$ and $t_2=\phi_2\cdot t_2'$, and $effect(t)$ is sufficiently deterministic.
\end{proof}

\subsection{$APPTC$ with Guards}{\label{aptcg}}

In this subsection, we will extend $APPTC$ with guards, which is abbreviated $APPTC_G$. The set of basic guards $G$ with element $\phi,\psi,\cdots$, which is extended by the following
formation rules:

$$\phi::=\delta|\epsilon|\neg\phi|\psi\in G_{at}|\phi+\psi|\phi\boxplus_{\pi}\psi|\phi\cdot\psi|\phi\leftmerge\psi$$

The set of axioms of $APPTC_G$ including axioms of $BAPTC_G$ in Table \ref{AxiomsForBAPTCG} and the axioms are shown in Table \ref{AxiomsForAPPTCG}.

\begin{center}
    \begin{table}
        \begin{tabular}{@{}ll@{}}
            \hline No. &Axiom\\
            $P1$ & $(x+x=x,y+y=y)\quad x\between y = x\parallel y + x\mid y$\\
            $P2$ & $x\parallel y = y \parallel x$\\
            $P3$ & $(x\parallel y)\parallel z = x\parallel (y\parallel z)$\\
            $P4$ & $(x+x=x,y+y=y)\quad x\parallel y = x\leftmerge y + y\leftmerge x$\\
            $P5$ & $(e_1\leq e_2)\quad e_1\leftmerge (e_2\cdot y) = (e_1\leftmerge e_2)\cdot y$\\
            $P6$ & $(e_1\leq e_2)\quad (e_1\cdot x)\leftmerge e_2 = (e_1\leftmerge e_2)\cdot x$\\
            $P7$ & $(e_1\leq e_2)\quad (e_1\cdot x)\leftmerge (e_2\cdot y) = (e_1\leftmerge e_2)\cdot (x\between y)$\\
            $P8$ & $(x+ y)\leftmerge z = (x\leftmerge z)+ (y\leftmerge z)$\\
            $P9$ & $\delta\leftmerge x = \delta$\\
            $P10$ & $\epsilon\leftmerge x = x$\\
            $P11$ & $x\leftmerge \epsilon = x$\\
            $C1$ & $e_1\mid e_2 = \gamma(e_1,e_2)$\\
            $C2$ & $e_1\mid (e_2\cdot y) = \gamma(e_1,e_2)\cdot y$\\
            $C3$ & $(e_1\cdot x)\mid e_2 = \gamma(e_1,e_2)\cdot x$\\
            $C4$ & $(e_1\cdot x)\mid (e_2\cdot y) = \gamma(e_1,e_2)\cdot (x\between y)$\\
            $C5$ & $(x+ y)\mid z = (x\mid z) + (y\mid z)$\\
            $C6$ & $x\mid (y+ z) = (x\mid y)+ (x\mid z)$\\
            $C7$ & $\delta\mid x = \delta$\\
            $C8$ & $x\mid\delta = \delta$\\
            $C9$ & $\epsilon\mid x = \delta$\\
            $C10$ & $x\mid\epsilon = \delta$\\
            $PM1$ & $x\parallel (y\boxplus_{\pi} z)=(x\parallel y)\boxplus_{\pi}(x\parallel z)$\\
            $PM2$ & $(x\boxplus_{\pi} y)\parallel z=(x\parallel z)\boxplus_{\pi}(y\parallel z)$\\
            $PM3$ & $x\mid (y\boxplus_{\pi} z)=(x\mid y)\boxplus_{\pi}(x\mid z)$\\
            $PM4$ & $(x\boxplus_{\pi} y)\mid z=(x\mid z)\boxplus_{\pi}(y\mid z)$\\
            $CE1$ & $\Theta(e) = e$\\
            $CE2$ & $\Theta(\delta) = \delta$\\
            $CE3$ & $\Theta(\epsilon) = \epsilon$\\
            $CE4$ & $\Theta(x+ y) = \Theta(x)\triangleleft y + \Theta(y)\triangleleft x$\\
            $PCE1$ & $\Theta(x\boxplus_{\pi} y) = \Theta(x)\triangleleft y \boxplus_{\pi} \Theta(y)\triangleleft x$\\
            $CE5$ & $\Theta(x\cdot y)=\Theta(x)\cdot\Theta(y)$\\
            $CE6$ & $\Theta(x\leftmerge y) = ((\Theta(x)\triangleleft y)\leftmerge y)+ ((\Theta(y)\triangleleft x)\leftmerge x)$\\
            $CE7$ & $\Theta(x\mid y) = ((\Theta(x)\triangleleft y)\mid y)+ ((\Theta(y)\triangleleft x)\mid x)$\\
        \end{tabular}
        \caption{Axioms of $APPTC_G$}
        \label{AxiomsForAPPTCG}
    \end{table}
\end{center}

\begin{center}
    \begin{table}
        \begin{tabular}{@{}ll@{}}
            \hline No. &Axiom\\
            $U1$ & $(\sharp(e_1,e_2))\quad e_1\triangleleft e_2 = \tau$\\
            $U2$ & $(\sharp(e_1,e_2),e_2\leq e_3)\quad e_1\triangleleft e_3 = e_1$\\
            $U3$ & $(\sharp(e_1,e_2),e_2\leq e_3)\quad e3\triangleleft e_1 = \tau$\\
            $PU1$ & $(\sharp_{\pi}(e_1,e_2))\quad e_1\triangleleft e_2 = \tau$\\
            $PU2$ & $(\sharp_{\pi}(e_1,e_2),e_2\leq e_3)\quad e_1\triangleleft e_3 = e_1$\\
            $PU3$ & $(\sharp_{\pi}(e_1,e_2),e_2\leq e_3)\quad e_3\triangleleft e_1 = \tau$\\
            $U4$ & $e\triangleleft \delta = e$\\
            $U5$ & $\delta \triangleleft e = \delta$\\
            $U6$ & $e\triangleleft \epsilon = e$\\
            $U7$ & $\epsilon \triangleleft e = e$\\
            $U8$ & $(x+ y)\triangleleft z = (x\triangleleft z)+ (y\triangleleft z)$\\
            $PU4$ & $(x\boxplus_{\pi} y)\triangleleft z = (x\triangleleft z)\boxplus_{\pi} (y\triangleleft z)$\\
            $U9$ & $(x\cdot y)\triangleleft z = (x\triangleleft z)\cdot (y\triangleleft z)$\\
            $U10$ & $(x\leftmerge y)\triangleleft z = (x\triangleleft z)\leftmerge (y\triangleleft z)$\\
            $U11$ & $(x\mid y)\triangleleft z = (x\triangleleft z)\mid (y\triangleleft z)$\\
            $U12$ & $x\triangleleft (y+ z) = (x\triangleleft y)\triangleleft z$\\
            $PU5$ & $x\triangleleft (y\boxplus_{\pi} z) = (x\triangleleft y)\triangleleft z$\\
            $U13$ & $x\triangleleft (y\cdot z)=(x\triangleleft y)\triangleleft z$\\
            $U14$ & $x\triangleleft (y\leftmerge z) = (x\triangleleft y)\triangleleft z$\\
            $U15$ & $x\triangleleft (y\mid z) = (x\triangleleft y)\triangleleft z$\\
            $D1$ & $e\notin H\quad\partial_H(e) = e$\\
            $D2$ & $e\in H\quad \partial_H(e) = \delta$\\
            $D3$ & $\partial_H(\delta) = \delta$\\
            $D4$ & $\partial_H(x+ y) = \partial_H(x)+\partial_H(y)$\\
            $D5$ & $\partial_H(x\cdot y) = \partial_H(x)\cdot\partial_H(y)$\\
            $D6$ & $\partial_H(x\leftmerge y) = \partial_H(x)\leftmerge\partial_H(y)$\\
            $PD1$ & $\partial_H(x\boxplus_{\pi}y)=\partial_H(x)\boxplus_{\pi}\partial_H(y)$\\
            $G12$ & $\phi(x\leftmerge y) =\phi x\leftmerge \phi y$\\
            $G13$ & $\phi(x\mid y) =\phi x\mid \phi y$\\
            $G14$ & $\delta\leftmerge \phi = \delta$\\
            $G15$ & $\phi\mid \delta = \delta$\\
            $G16$ & $\delta\mid \phi = \delta$\\
            $G17$ & $\phi\leftmerge \epsilon = \phi$\\
            $G18$ & $\epsilon\leftmerge \phi = \phi$\\
            $G19$ & $\phi\mid \epsilon = \delta$\\
            $G20$ & $\epsilon\mid \phi = \delta$\\
            $G21$ & $\phi\leftmerge\neg\phi = \delta$\\
            $G22$ & $\Theta(\phi) = \phi$\\
            $G23$ & $\partial_H(\phi) = \phi$\\
            $G24$ & $\phi_0\leftmerge\cdots\leftmerge\phi_n = \delta$ if $\forall s_0,\cdots,s_n\in S,\exists i\leq n.test(\neg\phi_i,s_0\cup\cdots\cup s_n)$\\        \end{tabular}
        \caption{Axioms of $APPTC_G$ (continuing)}
        \label{AxiomsForAPPTCG2}
    \end{table}
\end{center}

\begin{definition}[Basic terms of $APPTC_G$]\label{BTAPPTCG}
The set of basic terms of $APPTC_G$, $\mathcal{B}(APPTC_G)$, is inductively defined as follows:

\begin{enumerate}
    \item $\mathbb{E}\subset\mathcal{B}(APPTC_G)$;
    \item $G\subset\mathcal{B}(APPTC_G)$;
    \item if $e\in \mathbb{E}, t\in\mathcal{B}(APPTC_G)$ then $e\cdot t\in\mathcal{B}(APPTC_G)$;
    \item if $\phi\in G, t\in\mathcal{B}(APPTC_G)$ then $\phi\cdot t\in\mathcal{B}(APPTC_G)$;
    \item if $t,s\in\mathcal{B}(APPTC_G)$ then $t+ s\in\mathcal{B}(APPTC_G)$;
    \item if $t,s\in\mathcal{B}(APPTC_G)$ then $t\boxplus_{\pi} s\in\mathcal{B}(APPTC_G)$
    \item if $t,s\in\mathcal{B}(APPTC_G)$ then $t\leftmerge s\in\mathcal{B}(APPTC_G)$.
\end{enumerate}
\end{definition}

Based on the definition of basic terms for $APPTC_G$ (see Definition \ref{BTAPPTCG}) and axioms of $APPTC_G$, we can prove the elimination theorem of $APPTC_G$.

\begin{theorem}[Elimination theorem of $APPTC_G$]\label{ETAPPTCG}
Let $p$ be a closed $APPTC_G$ term. Then there is a basic $APPTC_G$ term $q$ such that $APPTC_G\vdash p=q$.
\end{theorem}

\begin{proof}
(1) Firstly, suppose that the following ordering on the signature of $APPTC_G$ is defined: $\leftmerge > \cdot > +>\boxplus_{\pi}$ and the symbol $\leftmerge$ is given the
lexicographical status for the first argument, then for each rewrite rule $p\rightarrow q$ in Table \ref{TRSForAPPTCG} relation $p>_{lpo} q$ can easily be proved. We obtain that the
term rewrite system shown in Table \ref{TRSForAPPTCG} is strongly normalizing, for it has finitely many rewriting rules, and $>$ is a well-founded ordering on the signature of
$APPTC_G$, and if $s>_{lpo} t$, for each rewriting rule $s\rightarrow t$ is in Table \ref{TRSForAPPTCG} (see Theorem \ref{SN}).

\begin{center}
    \begin{table}
        \begin{tabular}{@{}ll@{}}
            \hline No. &Rewriting Rule\\
            $RP1$ & $(x+x=x,y+y=y)\quad x\between y \rightarrow x\parallel y + x\mid y$\\
            $RP2$ & $x\parallel y \rightarrow y \parallel x$\\
            $RP3$ & $(x\parallel y)\parallel z \rightarrow x\parallel (y\parallel z)$\\
            $RP4$ & $(x+x=x,y+y=y)\quad x\parallel y \rightarrow x\leftmerge y + y\leftmerge x$\\
            $RP5$ & $(e_1\leq e_2)\quad e_1\leftmerge (e_2\cdot y) \rightarrow (e_1\leftmerge e_2)\cdot y$\\
            $RP6$ & $(e_1\leq e_2)\quad (e_1\cdot x)\leftmerge e_2 \rightarrow (e_1\leftmerge e_2)\cdot x$\\
            $RP7$ & $(e_1\leq e_2)\quad (e_1\cdot x)\leftmerge (e_2\cdot y) \rightarrow (e_1\leftmerge e_2)\cdot (x\between y)$\\
            $RP8$ & $(x+ y)\leftmerge z \rightarrow (x\leftmerge z)+ (y\leftmerge z)$\\
            $RP9$ & $\delta\leftmerge x \rightarrow \delta$\\
            $RP10$ & $\epsilon\leftmerge x \rightarrow x$\\
            $RP11$ & $x\leftmerge \epsilon \rightarrow x$\\
            $RC1$ & $e_1\mid e_2 \rightarrow \gamma(e_1,e_2)$\\
            $RC2$ & $e_1\mid (e_2\cdot y) \rightarrow \gamma(e_1,e_2)\cdot y$\\
            $RC3$ & $(e_1\cdot x)\mid e_2 \rightarrow \gamma(e_1,e_2)\cdot x$\\
            $RC4$ & $(e_1\cdot x)\mid (e_2\cdot y) \rightarrow \gamma(e_1,e_2)\cdot (x\between y)$\\
            $RC5$ & $(x+ y)\mid z \rightarrow (x\mid z) + (y\mid z)$\\
            $RC6$ & $x\mid (y+ z) \rightarrow (x\mid y)+ (x\mid z)$\\
            $RC7$ & $\delta\mid x \rightarrow \delta$\\
            $RC8$ & $x\mid\delta \rightarrow \delta$\\
            $RC9$ & $\epsilon\mid x \rightarrow \delta$\\
            $RC10$ & $x\mid\epsilon \rightarrow \delta$\\
            $RPM1$ & $x\parallel (y\boxplus_{\pi} z)\rightarrow(x\parallel y)\boxplus_{\pi}(x\parallel z)$\\
            $RPM2$ & $(x\boxplus_{\pi} y)\parallel z\rightarrow(x\parallel z)\boxplus_{\pi}(y\parallel z)$\\
            $RPM3$ & $x\mid (y\boxplus_{\pi} z)\rightarrow(x\mid y)\boxplus_{\pi}(x\mid z)$\\
            $RPM4$ & $(x\boxplus_{\pi} y)\mid z\rightarrow(x\mid z)\boxplus_{\pi}(y\mid z)$\\
            $RCE1$ & $\Theta(e) \rightarrow e$\\
            $RCE2$ & $\Theta(\delta) \rightarrow \delta$\\
            $RCE3$ & $\Theta(\epsilon) \rightarrow \epsilon$\\
            $RCE4$ & $\Theta(x+ y) \rightarrow \Theta(x)\triangleleft y + \Theta(y)\triangleleft x$\\
            $RPCE1$ & $\Theta(x\boxplus_{\pi} y) \rightarrow \Theta(x)\triangleleft y \boxplus_{\pi} \Theta(y)\triangleleft x$\\
            $RCE5$ & $\Theta(x\cdot y)\rightarrow\Theta(x)\cdot\Theta(y)$\\
            $RCE6$ & $\Theta(x\leftmerge y) \rightarrow ((\Theta(x)\triangleleft y)\leftmerge y)+ ((\Theta(y)\triangleleft x)\leftmerge x)$\\
            $RCE7$ & $\Theta(x\mid y) \rightarrow ((\Theta(x)\triangleleft y)\mid y)+ ((\Theta(y)\triangleleft x)\mid x)$\\
        \end{tabular}
        \caption{Term rewrite system of $APPTC_G$}
        \label{TRSForAPPTCG}
    \end{table}
\end{center}

\begin{center}
    \begin{table}
        \begin{tabular}{@{}ll@{}}
            \hline No. &Rewriting Rule\\
            $RU1$ & $(\sharp(e_1,e_2))\quad e_1\triangleleft e_2 \rightarrow \tau$\\
            $RU2$ & $(\sharp(e_1,e_2),e_2\leq e_3)\quad e_1\triangleleft e_3 \rightarrow e_1$\\
            $RU3$ & $(\sharp(e_1,e_2),e_2\leq e_3)\quad e3\triangleleft e_1 \rightarrow \tau$\\
            $RPU1$ & $(\sharp_{\pi}(e_1,e_2))\quad e_1\triangleleft e_2 \rightarrow \tau$\\
            $RPU2$ & $(\sharp_{\pi}(e_1,e_2),e_2\leq e_3)\quad e_1\triangleleft e_3 \rightarrow e_1$\\
            $RPU3$ & $(\sharp_{\pi}(e_1,e_2),e_2\leq e_3)\quad e_3\triangleleft e_1 \rightarrow \tau$\\
            $RU4$ & $e\triangleleft \delta \rightarrow e$\\
            $RU5$ & $\delta \triangleleft e \rightarrow \delta$\\
            $RU6$ & $e\triangleleft \epsilon \rightarrow e$\\
            $RU7$ & $\epsilon \triangleleft e \rightarrow e$\\
            $RU8$ & $(x+ y)\triangleleft z \rightarrow (x\triangleleft z)+ (y\triangleleft z)$\\
            $RPU4$ & $(x\boxplus_{\pi} y)\triangleleft z \rightarrow (x\triangleleft z)\boxplus_{\pi} (y\triangleleft z)$\\
            $RU9$ & $(x\cdot y)\triangleleft z \rightarrow (x\triangleleft z)\cdot (y\triangleleft z)$\\
            $RU10$ & $(x\leftmerge y)\triangleleft z \rightarrow (x\triangleleft z)\leftmerge (y\triangleleft z)$\\
            $RU11$ & $(x\mid y)\triangleleft z \rightarrow (x\triangleleft z)\mid (y\triangleleft z)$\\
            $RU12$ & $x\triangleleft (y+ z) \rightarrow (x\triangleleft y)\triangleleft z$\\
            $RPU5$ & $x\triangleleft (y\boxplus_{\pi} z) \rightarrow (x\triangleleft y)\triangleleft z$\\
            $RU13$ & $x\triangleleft (y\cdot z)\rightarrow(x\triangleleft y)\triangleleft z$\\
            $RU14$ & $x\triangleleft (y\leftmerge z) \rightarrow (x\triangleleft y)\triangleleft z$\\
            $RU15$ & $x\triangleleft (y\mid z) \rightarrow (x\triangleleft y)\triangleleft z$\\
            $RD1$ & $e\notin H\quad\partial_H(e) \rightarrow e$\\
            $RD2$ & $e\in H\quad \partial_H(e) \rightarrow \delta$\\
            $RD3$ & $\partial_H(\delta) \rightarrow \delta$\\
            $RD4$ & $\partial_H(x+ y) \rightarrow \partial_H(x)+\partial_H(y)$\\
            $RD5$ & $\partial_H(x\cdot y) \rightarrow \partial_H(x)\cdot\partial_H(y)$\\
            $RD6$ & $\partial_H(x\leftmerge y) \rightarrow \partial_H(x)\leftmerge\partial_H(y)$\\
            $RPD1$ & $\partial_H(x\boxplus_{\pi}y)\rightarrow\partial_H(x)\boxplus_{\pi}\partial_H(y)$\\
            $RG12$ & $\phi(x\leftmerge y) \rightarrow\phi x\leftmerge \phi y$\\
            $RG13$ & $\phi(x\mid y) \rightarrow\phi x\mid \phi y$\\
            $RG14$ & $\delta\leftmerge \phi \rightarrow \delta$\\
            $RG15$ & $\phi\mid \delta \rightarrow \delta$\\
            $RG16$ & $\delta\mid \phi \rightarrow \delta$\\
            $RG17$ & $\phi\leftmerge \epsilon \rightarrow \phi$\\
            $RG18$ & $\epsilon\leftmerge \phi \rightarrow \phi$\\
            $RG19$ & $\phi\mid \epsilon \rightarrow \delta$\\
            $RG20$ & $\epsilon\mid \phi \rightarrow \delta$\\
            $RG21$ & $\phi\leftmerge\neg\phi \rightarrow \delta$\\
            $RG22$ & $\Theta(\phi) \rightarrow \phi$\\
            $RG23$ & $\partial_H(\phi) \rightarrow \phi$\\
            $RG24$ & $\phi_0\leftmerge\cdots\leftmerge\phi_n \rightarrow \delta$ if $\forall s_0,\cdots,s_n\in S,\exists i\leq n.test(\neg\phi_i,s_0\cup\cdots\cup s_n)$\\
        \end{tabular}
        \caption{Term rewrite system of $APPTC_G$ (continuing)}
        \label{TRSForAPPTCG2}
    \end{table}
\end{center}

(2) Then we prove that the normal forms of closed $APPTC_G$ terms are basic $APPTC_G$ terms.

Suppose that $p$ is a normal form of some closed $APPTC_G$ term and suppose that $p$ is not a basic $APPTC_G$ term. Let $p'$ denote the smallest sub-term of $p$ which is not a basic
$APPTC_G$ term. It implies that each sub-term of $p'$ is a basic $APPTC_G$ term. Then we prove that $p$ is not a term in normal form. It is sufficient to induct on the structure of
$p'$:

\begin{itemize}
  \item Case $p'\equiv e, e\in \mathbb{E}$. $p'$ is a basic $APPTC_G$ term, which contradicts the assumption that $p'$ is not a basic $APPTC_G$ term, so this case should not occur.
  \item Case $p'\equiv \phi, \phi\in G$. $p'$ is a basic term, which contradicts the assumption that $p'$ is not a basic term, so this case should not occur.
  \item Case $p'\equiv p_1\cdot p_2$. By induction on the structure of the basic $APPTC_G$ term $p_1$:
      \begin{itemize}
        \item Subcase $p_1\in \mathbb{E}$. $p'$ would be a basic $APPTC_G$ term, which contradicts the assumption that $p'$ is not a basic $APPTC_G$ term;
        \item Subcase $p_1\in G$. $p'$ would be a basic term, which contradicts the assumption that $p'$ is not a basic term;
        \item Subcase $p_1\equiv e\cdot p_1'$. $RA5$ or $RA9$ rewriting rules in Table \ref{TRSForBAPTCG} can be applied. So $p$ is not a normal form;
        \item Subcase $p_1\equiv \phi\cdot p_1'$. $RG1$, $RG3$, $RG4$, $RG5$, $RG7$, or $RG8-9$ rewriting rules can be applied. So $p$ is not a normal form;
        \item Subcase $p_1\equiv p_1'+ p_1''$. $RA4$, $RA6$, $RG2$, or $RG6$ rewriting rules in Table \ref{TRSForBAPTCG} can be applied. So $p$ is not a normal form;
        \item Subcase $p_1\equiv p_1'\boxplus_{\pi} p_1''$. $RRA1-5$ rewriting rules in Table \ref{TRSForBAPTCG} can be applied. So $p$ is not a normal form;
        \item Subcase $p_1\equiv p_1'\leftmerge p_1''$. $RP5$-$RP11$ rewrite rules in Table \ref{TRSForAPPTCG} can be applied. So $p$ is not a normal form;
        \item Subcase $p_1\equiv p_1'\mid p_1''$. $RC1$-$RC10$ rewrite rules in Table \ref{TRSForAPPTCG} can be applied. So $p$ is not a normal form;
        \item Subcase $p_1\equiv \Theta(p_1')$. $RCE1$-$RCE7$ rewrite rules in Table \ref{TRSForAPPTCG} can be applied. So $p$ is not a normal form;
        \item Subcase $p_1\equiv \partial_H(p_1')$. $RD1$-$RD7$ rewrite rules in Table \ref{TRSForAPPTCG} can be applied. So $p$ is not a normal form.
      \end{itemize}
  \item Case $p'\equiv p_1+ p_2$. By induction on the structure of the basic $APPTC_G$ terms both $p_1$ and $p_2$, all subcases will lead to that $p'$ would be a basic $APPTC_G$ term,
  which contradicts the assumption that $p'$ is not a basic $APPTC_G$ term.
  \item Case $p'\equiv p_1\boxplus_{\pi} p_2$. By induction on the structure of the basic $APPTC_G$ terms both $p_1$ and $p_2$, all subcases will lead to that $p'$ would be a basic $APPTC_G$ term,
  which contradicts the assumption that $p'$ is not a basic $APPTC_G$ term.
  \item Case $p'\equiv p_1\leftmerge p_2$. By induction on the structure of the basic $APPTC_G$ terms both $p_1$ and $p_2$, all subcases will lead to that $p'$ would be a basic
  $APPTC_G$ term, which contradicts the assumption that $p'$ is not a basic $APPTC_G$ term.
  \item Case $p'\equiv p_1\mid p_2$. By induction on the structure of the basic $APPTC_G$ terms both $p_1$ and $p_2$, all subcases will lead to that $p'$ would be a basic $APPTC_G$
  term, which contradicts the assumption that $p'$ is not a basic $APPTC_G$ term.
  \item Case $p'\equiv \Theta(p_1)$. By induction on the structure of the basic $APPTC_G$ term $p_1$, $RCE1-RCE7$ rewrite rules in Table \ref{TRSForAPPTCG} can be applied. So $p$ is not
  a normal form.
  \item Case $p'\equiv p_1\triangleleft p_2$. By induction on the structure of the basic $APPTC_G$ terms both $p_1$ and $p_2$, all subcases will lead to that $p'$ would be a basic
  $APPTC_G$ term, which contradicts the assumption that $p'$ is not a basic $APPTC_G$ term.
  \item Case $p'\equiv \partial_H(p_1)$. By induction on the structure of the basic $APPTC_G$ terms of $p_1$, all subcases will lead to that $p'$ would be a basic $APPTC_G$ term,
  which contradicts the assumption that $p'$ is not a basic $APPTC_G$ term.
\end{itemize}
\end{proof}

We will define a term-deduction system which gives the operational semantics of $APPTC_G$. Two atomic events $e_1$ and $e_2$ are in race condition, which are denoted $e_1\% e_2$.

\begin{center}
    \begin{table}
        $$\frac{x\rightsquigarrow x'\quad y\rightsquigarrow y'}{x\between y\rightsquigarrow x'\parallel y'+x'\mid y'}$$
        $$\frac{x\rightsquigarrow x'\quad y\rightsquigarrow y'}{x\parallel y\rightsquigarrow x'\leftmerge y+y'\leftmerge x}$$
        $$\frac{x\rightsquigarrow x'}{x\leftmerge y\rightsquigarrow x'\leftmerge y}$$
        $$\frac{x\rightsquigarrow x'\quad y\rightsquigarrow y'}{x\mid y\rightsquigarrow x'\mid y'}$$
        $$\frac{x\rightsquigarrow x'}{\Theta(x)\rightsquigarrow \Theta(x')}$$
        $$\frac{x\rightsquigarrow x'}{x\triangleleft y\rightsquigarrow x'\triangleleft y}$$
        \caption{Probabilistic transition rules of $APPTC_G$}
        \label{TRForAPPTCG1}
    \end{table}
\end{center}

\begin{center}
    \begin{table}
        $$\frac{}{\langle \breve{e_1}\parallel\cdots \parallel \breve{e_n},s\rangle\xrightarrow{\{e_1,\cdots,e_n\}}\langle\surd,s'\rangle}\textrm{ if }s'\in effect(e_1,s)\cup\cdots\cup effect(e_n,s)$$

        $$\frac{}{\langle\breve{\phi_1}\parallel\cdots\parallel \breve{\phi_n},s\rangle\rightarrow\langle\surd,s\rangle}\textrm{ if }test(\phi_1,s),\cdots,test(\phi_n,s)$$

        $$\frac{\langle x,s\rangle\xrightarrow{e_1}\langle\surd,s'\rangle\quad \langle y,s\rangle\xrightarrow{e_2}\langle\surd,s''\rangle}{\langle x\parallel y,s\rangle\xrightarrow{\{e_1,e_2\}}\langle\surd,s'\cup s''\rangle} \quad\frac{\langle x,s\rangle\xrightarrow{e_1}\langle x',s'\rangle\quad \langle y,s\rangle\xrightarrow{e_2}\langle\surd,s''\rangle}{\langle x\parallel y,s\rangle\xrightarrow{\{e_1,e_2\}}\langle x',s'\cup s''\rangle}$$

        $$\frac{\langle x,s\rangle\xrightarrow{e_1}\langle\surd,s'\rangle\quad \langle y,s\rangle\xrightarrow{e_2}\langle y',s''\rangle}{\langle x\parallel y,s\rangle\xrightarrow{\{e_1,e_2\}}\langle y',s'\cup s''\rangle} \quad\frac{\langle x,s\rangle\xrightarrow{e_1}\langle x',s'\rangle\quad \langle y,s\rangle\xrightarrow{e_2}\langle y',s''\rangle}{\langle x\parallel y,s\rangle\xrightarrow{\{e_1,e_2\}}\langle x'\between y',s'\cup s''\rangle}$$

        $$\frac{\langle x,s\rangle\xrightarrow{e_1}\langle\surd,s'\rangle\quad \langle y,s\rangle\xnrightarrow{e_2}\quad(e_1\%e_2)}{\langle x\parallel y,s\rangle\xrightarrow{e_1}\langle y,s'\rangle} \quad\frac{\langle x,s\rangle\xrightarrow{e_1}\langle x',s'\rangle\quad \langle y,s\rangle\xnrightarrow{e_2}\quad(e_1\%e_2)}{\langle x\parallel y,s\rangle\xrightarrow{e_1}\langle x'\between y,s'\rangle}$$

        $$\frac{\langle x,s\rangle\xnrightarrow{e_1}\quad \langle y,s\rangle\xrightarrow{e_2}\langle\surd,s''\rangle\quad(e_1\%e_2)}{\langle x\parallel y,s\rangle\xrightarrow{e_2}\langle x,s''\rangle} \quad\frac{\langle x,s\rangle\xnrightarrow{e_1}\quad \langle y,s\rangle\xrightarrow{e_2}\langle y',s''\rangle\quad(e_1\%e_2)}{\langle x\parallel y,s\rangle\xrightarrow{e_2}\langle x\between y',s''\rangle}$$

        $$\frac{\langle x,s\rangle\xrightarrow{e_1}\langle\surd,s'\rangle\quad \langle y,s\rangle\xrightarrow{e_2}\langle\surd,s''\rangle \quad(e_1\leq e_2)}{\langle x\leftmerge y,s\rangle\xrightarrow{\{e_1,e_2\}}\langle \surd,s'\cup s''\rangle} \quad\frac{\langle x,s\rangle\xrightarrow{e_1}\langle x',s'\rangle\quad \langle y,s\rangle\xrightarrow{e_2}\langle\surd,s''\rangle \quad(e_1\leq e_2)}{\langle x\leftmerge y,s\rangle\xrightarrow{\{e_1,e_2\}}\langle x',s'\cup s''\rangle}$$

        $$\frac{\langle x,s\rangle\xrightarrow{e_1}\langle\surd,s'\rangle\quad \langle y,s\rangle\xrightarrow{e_2}\langle y',s''\rangle \quad(e_1\leq e_2)}{\langle x\leftmerge y,s\rangle\xrightarrow{\{e_1,e_2\}}\langle y',s'\cup s''\rangle} \quad\frac{\langle x,s\rangle\xrightarrow{e_1}\langle x',s'\rangle\quad \langle y,s\rangle\xrightarrow{e_2}\langle y',s''\rangle \quad(e_1\leq e_2)}{\langle x\leftmerge y,s\rangle\xrightarrow{\{e_1,e_2\}}\langle x'\between y',s'\cup s''\rangle}$$

        $$\frac{\langle x,s\rangle\xrightarrow{e_1}\langle\surd,s'\rangle\quad \langle y,s\rangle\xrightarrow{e_2}\langle\surd,s''\rangle}{\langle x\mid y,s\rangle\xrightarrow{\gamma(e_1,e_2)}\langle\surd,effect(\gamma(e_1,e_2),s)\rangle} \quad\frac{\langle x,s\rangle\xrightarrow{e_1}\langle x',s'\rangle\quad \langle y,s\rangle\xrightarrow{e_2}\langle\surd,s''\rangle}{\langle x\mid y,s\rangle\xrightarrow{\gamma(e_1,e_2)}\langle x',effect(\gamma(e_1,e_2),s)\rangle}$$

        $$\frac{\langle x,s\rangle\xrightarrow{e_1}\langle\surd,s'\rangle\quad \langle y,s\rangle\xrightarrow{e_2}\langle y',s''\rangle}{\langle x\mid y,s\rangle\xrightarrow{\gamma(e_1,e_2)}\langle y',effect(\gamma(e_1,e_2),s)\rangle} \quad\frac{\langle x,s\rangle\xrightarrow{e_1}\langle x',s'\rangle\quad \langle y,s\rangle\xrightarrow{e_2}\langle y',s''\rangle}{\langle x\mid y,s\rangle\xrightarrow{\gamma(e_1,e_2)}\langle x'\between y',effect(\gamma(e_1,e_2),s)\rangle}$$

        \caption{Action transition rules of $APPTC_G$}
        \label{TRForAPPTCG}
    \end{table}
\end{center}

\begin{center}
    \begin{table}
        $$\frac{\langle x,s\rangle\xrightarrow{e_1}\langle\surd,s'\rangle\quad (\sharp(e_1,e_2))}{\langle \Theta(x),s\rangle\xrightarrow{e_1}\langle\surd,s'\rangle} \quad\frac{\langle x,s\rangle\xrightarrow{e_2}\langle\surd,s''\rangle\quad (\sharp(e_1,e_2))}{\langle\Theta(x),s\rangle\xrightarrow{e_2}\langle\surd,s''\rangle}$$

        $$\frac{\langle x,s\rangle\xrightarrow{e_1}\langle x',s'\rangle\quad (\sharp(e_1,e_2))}{\langle\Theta(x),s\rangle\xrightarrow{e_1}\langle\Theta(x'),s'\rangle} \quad\frac{\langle x,s\rangle\xrightarrow{e_2}\langle x'',s''\rangle\quad (\sharp(e_1,e_2))}{\langle\Theta(x),s\rangle\xrightarrow{e_2}\langle\Theta(x''),s''\rangle}$$

        $$\frac{\langle x,s\rangle\xrightarrow{e_1}\langle\surd,s'\rangle \quad \langle y,s\rangle\nrightarrow^{e_2}\quad (\sharp(e_1,e_2))}{\langle x\triangleleft y,s\rangle\xrightarrow{\tau}\langle\surd,s'\rangle}
        \quad\frac{\langle x,s\rangle\xrightarrow{e_1}\langle x',s'\rangle \quad \langle y,s\rangle\nrightarrow^{e_2}\quad (\sharp(e_1,e_2))}{\langle x\triangleleft y,s\rangle\xrightarrow{\tau}\langle x',s'\rangle}$$

        $$\frac{\langle x,s\rangle\xrightarrow{e_1}\langle\surd,s\rangle \quad \langle y,s\rangle\nrightarrow^{e_3}\quad (\sharp(e_1,e_2),e_2\leq e_3)}{\langle x\triangleleft y,s\rangle\xrightarrow{e_1}\langle\surd,s'\rangle}
        \quad\frac{\langle x,s\rangle\xrightarrow{e_1}\langle x',s'\rangle \quad \langle y,s\rangle\nrightarrow^{e_3}\quad (\sharp(e_1,e_2),e_2\leq e_3)}{\langle x\triangleleft y,s\rangle\xrightarrow{e_1}\langle x',s'\rangle}$$

        $$\frac{\langle x,s\rangle\xrightarrow{e_3}\langle\surd,s'\rangle \quad \langle y,s\rangle\nrightarrow^{e_2}\quad (\sharp(e_1,e_2),e_1\leq e_3)}{\langle x\triangleleft y,s\rangle\xrightarrow{\tau}\langle\surd,s'\rangle}
        \quad\frac{\langle x,s\rangle\xrightarrow{e_3}\langle x',s'\rangle \quad \langle y,s\rangle\nrightarrow^{e_2}\quad (\sharp(e_1,e_2),e_1\leq e_3)}{\langle x\triangleleft y,s\rangle\xrightarrow{\tau}\langle x',s'\rangle}$$

        $$\frac{\langle x,s\rangle\xrightarrow{e_1}\langle\surd,s'\rangle\quad (\sharp_{\pi}(e_1,e_2))}{\langle \Theta(x),s\rangle\xrightarrow{e_1}\langle\surd,s'\rangle} \quad\frac{\langle x,s\rangle\xrightarrow{e_2}\langle\surd,s''\rangle\quad (\sharp_{\pi}(e_1,e_2))}{\langle\Theta(x),s\rangle\xrightarrow{e_2}\langle\surd,s''\rangle}$$

        $$\frac{\langle x,s\rangle\xrightarrow{e_1}\langle x',s'\rangle\quad (\sharp_{\pi}(e_1,e_2))}{\langle\Theta(x),s\rangle\xrightarrow{e_1}\langle\Theta(x'),s'\rangle} \quad\frac{\langle x,s\rangle\xrightarrow{e_2}\langle x'',s''\rangle\quad (\sharp_{\pi}(e_1,e_2))}{\langle\Theta(x),s\rangle\xrightarrow{e_2}\langle\Theta(x''),s''\rangle}$$

        $$\frac{\langle x,s\rangle\xrightarrow{e_1}\langle\surd,s'\rangle \quad \langle y,s\rangle\nrightarrow^{e_2}\quad (\sharp_{\pi}(e_1,e_2))}{\langle x\triangleleft y,s\rangle\xrightarrow{\tau}\langle\surd,s'\rangle}
        \quad\frac{\langle x,s\rangle\xrightarrow{e_1}\langle x',s'\rangle \quad \langle y,s\rangle\nrightarrow^{e_2}\quad (\sharp_{\pi}(e_1,e_2))}{\langle x\triangleleft y,s\rangle\xrightarrow{\tau}\langle x',s'\rangle}$$

        $$\frac{\langle x,s\rangle\xrightarrow{e_1}\langle\surd,s\rangle \quad \langle y,s\rangle\nrightarrow^{e_3}\quad (\sharp_{\pi}(e_1,e_2),e_2\leq e_3)}{\langle x\triangleleft y,s\rangle\xrightarrow{e_1}\langle\surd,s'\rangle}
        \quad\frac{\langle x,s\rangle\xrightarrow{e_1}\langle x',s'\rangle \quad \langle y,s\rangle\nrightarrow^{e_3}\quad (\sharp_{\pi}(e_1,e_2),e_2\leq e_3)}{\langle x\triangleleft y,s\rangle\xrightarrow{e_1}\langle x',s'\rangle}$$

        $$\frac{\langle x,s\rangle\xrightarrow{e_3}\langle\surd,s'\rangle \quad \langle y,s\rangle\nrightarrow^{e_2}\quad (\sharp_{\pi}(e_1,e_2),e_1\leq e_3)}{\langle x\triangleleft y,s\rangle\xrightarrow{\tau}\langle\surd,s'\rangle}
        \quad\frac{\langle x,s\rangle\xrightarrow{e_3}\langle x',s'\rangle \quad \langle y,s\rangle\nrightarrow^{e_2}\quad (\sharp_{\pi}(e_1,e_2),e_1\leq e_3)}{\langle x\triangleleft y,s\rangle\xrightarrow{\tau}\langle x',s'\rangle}$$

        $$\frac{\langle x,s\rangle\xrightarrow{e}\langle\surd,s'\rangle}{\langle\partial_H(x),s\rangle\xrightarrow{e}\langle\surd,s'\rangle}\quad (e\notin H)\quad\frac{\langle x,s\rangle\xrightarrow{e}\langle x',s'\rangle}{\langle\partial_H(x),s\rangle\xrightarrow{e}\langle\partial_H(x'),s'\rangle}\quad(e\notin H)$$
        \caption{Action transition rules of $APPTC_G$ (continuing)}
        \label{TRForAPPTCG2}
    \end{table}
\end{center}

\begin{theorem}[Generalization of $APPTC_G$ with respect to $BAPTC_G$]
$APPTC_G$ is a generalization of $BAPTC_G$.
\end{theorem}

\begin{proof}
It follows from the following three facts.

\begin{enumerate}
  \item The transition rules of $BAPTC_G$ in section \ref{batcg} are all source-dependent;
  \item The sources of the transition rules $APPTC_G$ contain an occurrence of $\between$, or $\parallel$, or $\leftmerge$, or $\mid$, or $\Theta$, or $\triangleleft$;
  \item The transition rules of $APPTC_G$ are all source-dependent.
\end{enumerate}

So, $APPTC_G$ is a generalization of $BAPTC_G$, that is, $BAPTC_G$ is an embedding of $APPTC_G$, as desired.
\end{proof}

\begin{theorem}[Congruence of $APPTC_G$ with respect to probabilistic truly concurrent bisimulation equivalences]\label{CAPPTCG}
(1) Probabilistic pomset bisimulation equivalence $\sim_{pp}$ is a congruence with respect to $APPTC_G$.

(2) Probabilistic step bisimulation equivalence $\sim_{ps}$ is a congruence with respect to $APPTC_G$.

(3) Probabilistic hp-bisimulation equivalence $\sim_{php}$ is a congruence with respect to $APPTC_G$.

(4) Probabilistic hhp-bisimulation equivalence $\sim_{phhp}$ is a congruence with respect to $APPTC_G$.
\end{theorem}

\begin{proof}
(1) It is easy to see that probabilistic pomset bisimulation is an equivalent relation on $APPTC_G$ terms, we only need to prove that $\sim_{pp}$ is preserved by the operators
$\parallel$, $\leftmerge$, $\mid$, $\Theta$, $\triangleleft$, $\partial_H$. It is trivial and we leave the proof as an exercise for the readers.

(2) It is easy to see that probabilistic step bisimulation is an equivalent relation on $APPTC_G$ terms, we only need to prove that $\sim_{ps}$ is preserved by the operators
$\parallel$, $\leftmerge$, $\mid$, $\Theta$, $\triangleleft$, $\partial_H$. It is trivial and we leave the proof as an exercise for the readers.

(3) It is easy to see that probabilistic hp-bisimulation is an equivalent relation on $APPTC_G$ terms, we only need to prove that $\sim_{php}$ is preserved by the operators
$\parallel$, $\leftmerge$, $\mid$, $\Theta$, $\triangleleft$, $\partial_H$. It is trivial and we leave the proof as an exercise for the readers.

(4) It is easy to see that probabilistic hhp-bisimulation is an equivalent relation on $APPTC_G$ terms, we only need to prove that $\sim_{phhp}$ is preserved by the operators
$\parallel$, $\leftmerge$, $\mid$, $\Theta$, $\triangleleft$, $\partial_H$. It is trivial and we leave the proof as an exercise for the readers.
\end{proof}

\begin{theorem}[Soundness of $APPTC_G$ modulo probabilistic truly concurrent bisimulation equivalences]\label{SAPPTCG}
(1) Let $x$ and $y$ be $APPTC_G$ terms. If $APPTC_G\vdash x=y$, then $x\sim_{pp} y$.

(2) Let $x$ and $y$ be $APPTC_G$ terms. If $APPTC_G\vdash x=y$, then $x\sim_{ps} y$.

(3) Let $x$ and $y$ be $APPTC_G$ terms. If $APPTC_G\vdash x=y$, then $x\sim_{php} y$;

(3) Let $x$ and $y$ be $APPTC_G$ terms. If $APPTC_G\vdash x=y$, then $x\sim_{phhp} y$.
\end{theorem}

\begin{proof}
(1) Since probabilistic pomset bisimulation $\sim_{pp}$ is both an equivalent and a congruent relation, we only need to check if each axiom in Table \ref{AxiomsForAPPTCG} is sound modulo
probabilistic pomset bisimulation equivalence. We leave the proof as an exercise for the readers.

(2) Since probabilistic step bisimulation $\sim_{ps}$ is both an equivalent and a congruent relation, we only need to check if each axiom in Table \ref{AxiomsForAPPTCG} is sound modulo
probabilistic step bisimulation equivalence. We leave the proof as an exercise for the readers.

(3) Since probabilistic hp-bisimulation $\sim_{php}$ is both an equivalent and a congruent relation, we only need to check if each axiom in Table \ref{AxiomsForAPPTCG} is sound modulo
probabilistic hp-bisimulation equivalence. We leave the proof as an exercise for the readers.

(4) Since probabilistic hhp-bisimulation $\sim_{phhp}$ is both an equivalent and a congruent relation, we only need to check if each axiom in Table \ref{AxiomsForAPPTCG} is sound modulo
probabilistic hhp-bisimulation equivalence. We leave the proof as an exercise for the readers.
\end{proof}

\begin{theorem}[Completeness of $APPTC_G$ modulo probabilistic truly concurrent bisimulation equivalences]\label{CAPPTCG}
(1) Let $p$ and $q$ be closed $APPTC_G$ terms, if $p\sim_{pp} q$ then $p=q$.

(2) Let $p$ and $q$ be closed $APPTC_G$ terms, if $p\sim_{ps} q$ then $p=q$.

(3) Let $p$ and $q$ be closed $APPTC_G$ terms, if $p\sim_{php} q$ then $p=q$.

(3) Let $p$ and $q$ be closed $APPTC_G$ terms, if $p\sim_{phhp} q$ then $p=q$.
\end{theorem}

\begin{proof}
(1) Firstly, by the elimination theorem of $APPTC_G$ (see Theorem \ref{ETAPPTCG}), we know that for each closed $APPTC_G$ term $p$, there exists a closed basic $APPTC_G$ term $p'$, such
that $APPTC\vdash p=p'$, so, we only need to consider closed basic $APPTC_G$ terms.

The basic terms (see Definition \ref{BTAPPTCG}) modulo associativity and commutativity (AC) of conflict $+$ (defined by axioms $A1$ and $A2$ in Table \ref{AxiomsForBAPTCG}), and these
equivalences is denoted by $=_{AC}$. Then, each equivalence class $s$ modulo AC of $+$ has the following normal form

$$s_1\boxplus_{\pi_1}\cdots\boxplus_{\pi_{k-1}} s_k$$

with each $s_i$ has the following form

$$t_1+\cdots+ t_l$$

with each $t_j$ either an atomic event or of the form

$$u_1\cdot\cdots\cdot u_m$$

with each $u_l$ either an atomic event or of the form

$$v_1\leftmerge\cdots\leftmerge v_m$$

with each $v_m$ an atomic event, and each $t_j$ is called the summand of $s$.

Now, we prove that for normal forms $n$ and $n'$, if $n\sim_{pp} n'$ then $n=_{AC}n'$. It is sufficient to induct on the sizes of $n$ and $n'$.

\begin{itemize}
  \item Consider a summand $e$ of $n$. Then $\langle n,s\rangle\rightsquigarrow\xrightarrow{e}\langle \surd,s'\rangle$, so $n\sim_{pp} n'$ implies $\langle n',s\rangle
  \rightsquigarrow\xrightarrow{e}\langle \surd,s\rangle$, meaning that $n'$ also contains the summand $e$.
  \item Consider a summand $\phi$ of $n$. Then $\langle n,s\rangle\rightsquigarrow\rightarrow\langle \surd,s\rangle$, if $test(\phi,s)$ holds, so $n\sim_{pp} n'$ implies
  $\langle n',s\rangle\rightsquigarrow\rightarrow\langle \surd,s\rangle$, if $test(\phi,s)$ holds, meaning that $n'$ also contains the summand $\phi$.
  \item Consider a summand $t_1\cdot t_2$ of $n$,
  \begin{itemize}
    \item if $t_1\equiv e'$, then $\langle n,s\rangle\rightsquigarrow\xrightarrow{e'}\langle t_2,s'\rangle$, so $n\sim_{pp} n'$ implies $\langle n',s\rangle\rightsquigarrow
    \xrightarrow{e'}\langle t_2',s'\rangle$ with $t_2\sim_{pp} t_2'$, meaning that $n'$ contains a summand $e'\cdot t_2'$. Since $t_2$ and $t_2'$ are normal forms and have sizes
    smaller than $n$ and $n'$, by the induction hypotheses if $t_2\sim_{pp} t_2'$ then $t_2=_{AC} t_2'$;
    \item if $t_1\equiv \phi'$, then $\langle n,s\rangle\rightsquigarrow\rightarrow\langle t_2,s\rangle$, if $test(\phi',s)$ holds, so $n\sim_{pp} n'$ implies $\langle n',s\rangle
    \rightsquigarrow\rightarrow\langle t_2',s\rangle$ with $t_2\sim_{pp} t_2'$, if $test(\phi',s)$ holds, meaning that $n'$ contains a summand $\phi'\cdot t_2'$. Since $t_2$ and
    $t_2'$ are normal forms and have sizes smaller than $n$ and $n'$, by the induction hypotheses if $t_2\sim_{pp} t_2'$ then $t_2=_{AC} t_2'$;
    \item if $t_1\equiv e_1\leftmerge\cdots\leftmerge e_l$, then $\langle n,s\rangle\rightsquigarrow\xrightarrow{\{e_1,\cdots,e_l\}}\langle t_2,s'\rangle$, so $n\sim_{pp} n'$ implies
    $\langle n',s\rangle\rightsquigarrow\xrightarrow{\{e_1,\cdots,e_l\}}\langle t_2',s'\rangle$ with $t_2\sim_{pp} t_2'$, meaning that $n'$ contains a summand
    $(e_1\leftmerge\cdots\leftmerge e_l)\cdot t_2'$. Since $t_2$ and $t_2'$ are normal forms and have sizes smaller than $n$ and $n'$, by the induction hypotheses if $t_2\sim_{pp} t_2'$
    then $t_2=_{AC} t_2'$;
    \item if $t_1\equiv \phi_1\leftmerge\cdots\leftmerge \phi_l$, then $\langle n,s\rangle\rightsquigarrow\rightarrow\langle t_2,s\rangle$, if $test(\phi_1,s),\cdots,test(\phi_l,s)$
    hold, so $n\sim_p n'$ implies $\langle n',s\rangle\rightsquigarrow\rightarrow\langle t_2',s\rangle$ with $t_2\sim_{pp} t_2'$, if $test(\phi_1,s),\cdots,test(\phi_l,s)$ hold,
    meaning that $n'$ contains a summand $(\phi_1\leftmerge\cdots\leftmerge \phi_l)\cdot t_2'$. Since $t_2$ and $t_2'$ are normal forms and have sizes smaller than $n$ and $n'$, by
    the induction hypotheses if $t_2\sim_{pp} t_2'$ then $t_2=_{AC} t_2'$.
  \end{itemize}
\end{itemize}

So, we get $n=_{AC} n'$.

Finally, let $s$ and $t$ be basic $APPTC_G$ terms, and $s\sim_{pp} t$, there are normal forms $n$ and $n'$, such that $s=n$ and $t=n'$. The soundness theorem of $APPTC_G$ modulo
probabilistic pomset bisimulation equivalence (see Theorem \ref{SAPPTCG}) yields $s\sim_{pp} n$ and $t\sim_{pp} n'$, so $n\sim_{pp} s\sim_{pp} t\sim_{pp} n'$. Since if $n\sim_{pp} n'$
then $n=_{AC}n'$, $s=n=_{AC}n'=t$, as desired.

(2) It can be proven similarly as (1).

(3) It can be proven similarly as (1).

(4) It can be proven similarly as (1).
\end{proof}

\begin{theorem}[Sufficient determinism]
All related data environments with respect to $APPTC_G$ can be sufficiently deterministic.
\end{theorem}

\begin{proof}
It only needs to check $effect(t,s)$ function is deterministic, and is sufficient to induct on the structure of term $t$. The new matter is the case $t=t_1\between t_2$, the whole
thing is $t_1\between t_2 = t_1\leftmerge t_2 + t_2\leftmerge t_1 + t_1\mid t_2$. We can make $effect(t)$ be sufficiently deterministic: eliminating non-determinism
during modeling time by use of empty event $\epsilon$. We can make $t=t_1\parallel t_2$ be $t=(\epsilon \cdot t_1)\parallel t_2$ or $t=t_1\parallel (\epsilon\cdot t_2)$ during
modeling phase, and then $effect(t,s)$ becomes sufficiently deterministic.
\end{proof}

\subsection{Recursion}{\label{recg}}

In this subsection, we introduce recursion to capture infinite processes based on $APPTC_G$. In the following, $E,F,G$ are recursion specifications, $X,Y,Z$ are recursive variables.

\begin{definition}[Guarded recursive specification]
A recursive specification

$$X_1=t_1(X_1,\cdots,X_n)$$
$$...$$
$$X_n=t_n(X_1,\cdots,X_n)$$

is guarded if the right-hand sides of its recursive equations can be adapted to the form by applications of the axioms in $APPTC$ and replacing recursion variables by the right-hand
sides of their recursive equations,

$((a_{111}\leftmerge\cdots\leftmerge a_{11i_1})\cdot s_1(X_1,\cdots,X_n)+\cdots+(a_{1k1}\leftmerge\cdots\leftmerge a_{1ki_k})\cdot s_k(X_1,\cdots,X_n)+(b_{111}\leftmerge\cdots\leftmerge
b_{11j_1})+\cdots+(b_{11j_1}\leftmerge\cdots\leftmerge b_{1lj_l}))\boxplus_{\pi_1}\cdots\boxplus_{\pi_{m-1}}((a_{m11}\leftmerge\cdots\leftmerge a_{m1i_1})\cdot s_1(X_1,\cdots,X_n)+
\cdots+(a_{mk1}\leftmerge\cdots\leftmerge a_{mki_k})\cdot s_k(X_1,\cdots,X_n)+(b_{m11}\leftmerge\cdots\leftmerge b_{m1j_1})+\cdots+(b_{m1j_1}\leftmerge\cdots\leftmerge b_{mlj_l}))$

where $a_{111},\cdots,a_{11i_1},a_{1k1},\cdots,a_{1ki_k},b_{111},\cdots,b_{11j_1},b_{11j_1},\cdots,b_{1lj_l},\cdots, a_{m11},\cdots,a_{m1i_1},a_{1k1},\cdots,a_{mki_k},\\b_{111},\cdots,
b_{m1j_1},b_{m1j_1},\cdots,b_{mlj_l}\in \mathbb{E}$, and the sum above is allowed to be empty, in which case it represents the deadlock $\delta$. And there does not exist an infinite
sequence of $\epsilon$-transitions $\langle X|E\rangle\rightarrow\langle X'|E\rangle\rightarrow\langle X''|E\rangle\rightarrow\cdots$.
\end{definition}

\begin{center}
    \begin{table}
        $$\frac{\langle t_i(\langle X_1|E\rangle,\cdots,\langle X_n|E\rangle),s\rangle\rightsquigarrow \langle y,s\rangle}{\langle\langle X_i|E\rangle,s\rangle\rightsquigarrow \langle y,s\rangle}$$
        $$\frac{\langle t_i(\langle X_1|E\rangle,\cdots,\langle X_n|E\rangle),s\rangle\xrightarrow{\{e_1,\cdots,e_k\}}\langle\surd,s'\rangle}{\langle\langle X_i|E\rangle,s\rangle\xrightarrow{\{e_1,\cdots,e_k\}}\langle\surd,s'\rangle}$$
        $$\frac{\langle t_i(\langle X_1|E\rangle,\cdots,\langle X_n|E\rangle),s\rangle\xrightarrow{\{e_1,\cdots,e_k\}} \langle y,s'\rangle}{\langle\langle X_i|E\rangle,s\rangle\xrightarrow{\{e_1,\cdots,e_k\}} \langle y,s'\rangle}$$
        \caption{Transition rules of guarded recursion}
        \label{TRForGRG}
    \end{table}
\end{center}

\begin{theorem}[Conservitivity of $APPTC_G$ with guarded recursion]
$APPTC_G$ with guarded recursion is a conservative extension of $APPTC_G$.
\end{theorem}

\begin{proof}
Since the transition rules of $APPTC_G$ are source-dependent, and the transition rules for guarded recursion in Table \ref{TRForGRG} contain only a fresh constant in their source, so
the transition rules of $APPTC_G$ with guarded recursion are a conservative extension of those of $APPTC_G$.
\end{proof}

\begin{theorem}[Congruence theorem of $APPTC_G$ with guarded recursion]
Probabilistic truly concurrent bisimulation equivalences $\sim_{pp}$, $\sim_{p}$, $\sim_{php}$ and $\sim_{phhp}$ are all congruences with respect to $APPTC_G$ with guarded recursion.
\end{theorem}

\begin{proof}
It follows the following two facts:
\begin{enumerate}
  \item in a guarded recursive specification, right-hand sides of its recursive equations can be adapted to the form by applications of the axioms in $APPTC_G$ and replacing recursion
  variables by the right-hand sides of their recursive equations;
  \item probabilistic truly concurrent bisimulation equivalences $\sim_{pp}$, $\sim_{ps}$, $\sim_{php}$ and $\sim_{phhp}$ are all congruences with respect to all operators of $APPTC_G$.
\end{enumerate}
\end{proof}

\begin{theorem}[Elimination theorem of $APPTC_G$ with linear recursion]\label{ETRecursionG}
Each process term in $APPTC_G$ with linear recursion is equal to a process term $\langle X_1|E\rangle$ with $E$ a linear recursive specification.
\end{theorem}

\begin{proof}
By applying structural induction with respect to term size, each process term $t_1$ in $APPTC$ with linear recursion generates a process can be expressed in the form of equations

$t_i=((a_{1i11}\leftmerge\cdots\leftmerge a_{1i1i_1})t_{i1}+\cdots+(a_{1ik_i1}\leftmerge\cdots\leftmerge a_{1ik_ii_k})t_{ik_i}+(b_{1i11}\leftmerge\cdots\leftmerge b_{1i1i_1})+\cdots+
(b_{1il_i1}\leftmerge\cdots\leftmerge b_{1il_ii_l}))\boxplus_{\pi_1}\cdots\boxplus_{\pi_{m-1}}((a_{mi11}\leftmerge\cdots\leftmerge a_{mi1i_1})t_{i1}+\cdots+(a_{mik_i1}\leftmerge\cdots
\leftmerge a_{mik_ii_k})t_{ik_i}+(b_{mi11}\leftmerge\cdots\leftmerge b_{mi1i_1})+\cdots+(b_{mil_i1}\leftmerge\cdots\leftmerge b_{mil_ii_l}))$

for $i\in\{1,\cdots,n\}$. Let the linear recursive specification $E$ consist of the recursive equations

$X_i=((a_{1i11}\leftmerge\cdots\leftmerge a_{1i1i_1})X_{i1}+\cdots+(a_{1ik_i1}\leftmerge\cdots\leftmerge a_{1ik_ii_k})X_{ik_i}+(b_{1i11}\leftmerge\cdots\leftmerge b_{1i1i_1})+\cdots+
(b_{1il_i1}\leftmerge\cdots\leftmerge b_{1il_ii_l}))\boxplus_{\pi_1}\cdots\boxplus_{\pi_{m-1}}((a_{mi11}\leftmerge\cdots\leftmerge a_{mi1i_1})X_{i1}+\cdots+(a_{mik_i1}\leftmerge\cdots
\leftmerge a_{mik_ii_k})X_{ik_i}+(b_{mi11}\leftmerge\cdots\leftmerge b_{mi1i_1})+\cdots+(b_{mil_i1}\leftmerge\cdots\leftmerge b_{mil_ii_l}))$

for $i\in\{1,\cdots,n\}$. Replacing $X_i$ by $t_i$ for $i\in\{1,\cdots,n\}$ is a solution for $E$, $RSP$ yields $t_1=\langle X_1|E\rangle$.
\end{proof}

\begin{theorem}[Soundness of $APPTC_G$ with guarded recursion]\label{SAPPTC_GRG}
Let $x$ and $y$ be $APPTC_G$ with guarded recursion terms. If $APPTC_G\textrm{ with guarded recursion}\vdash x=y$, then

(1) $x\sim_{ps} y$.

(2) $x\sim_{pp} y$.

(3) $x\sim_{php} y$.

(4) $x\sim_{phhp} y$.
\end{theorem}

\begin{proof}
(1) Since probabilistic step bisimulation $\sim_{ps}$ is both an equivalent and a congruent relation with respect to $APPTC_G$ with guarded recursion, we only need to check if each
axiom in Table \ref{RDPRSP} is sound modulo probabilistic step bisimulation equivalence. We leave them as exercises to the readers.

(2) Since probabilistic pomset bisimulation $\sim_{pp}$ is both an equivalent and a congruent relation with respect to the guarded recursion, we only need to check if each axiom in
Table \ref{RDPRSP} is sound modulo probabilistic pomset bisimulation equivalence. We leave them as exercises to the readers.

(3) Since probabilistic hp-bisimulation $\sim_{php}$ is both an equivalent and a congruent relation with respect to guarded recursion, we only need to check if each axiom in Table
\ref{RDPRSP} is sound modulo probabilistic hp-bisimulation equivalence. We leave them as exercises to the readers.

(4) Since probabilistic hhp-bisimulation $\sim_{phhp}$ is both an equivalent and a congruent relation with respect to guarded recursion, we only need to check if each axiom in Table
\ref{RDPRSP} is sound modulo probabilistic hhp-bisimulation equivalence. We leave them as exercises to the readers.
\end{proof}

\begin{theorem}[Completeness of $APPTC_G$ with linear recursion]\label{CAPPTC_GRG}
Let $p$ and $q$ be closed $APPTC_G$ with linear recursion terms, then,

(1) if $p\sim_{ps} q$ then $p=q$.

(2) if $p\sim_{pp} q$ then $p=q$.

(3) if $p\sim_{php} q$ then $p=q$.

(4) if $p\sim_{phhp} q$ then $p=q$.
\end{theorem}

\begin{proof}
Firstly, by the elimination theorem of $APPTC_G$ with guarded recursion (see Theorem \ref{ETRecursionG}), we know that each process term in $APPTC_G$ with linear recursion is equal to
a process term $\langle X_1|E\rangle$ with $E$ a linear recursive specification. And for the simplicity, without loss of generalization, we do not consider empty event $\epsilon$,
just because recursion with $\epsilon$ are similar to that with silent event $\tau$.

It remains to prove the following cases.

(1) If $\langle X_1|E_1\rangle \sim_{ps} \langle Y_1|E_2\rangle$ for linear recursive specification $E_1$ and $E_2$, then $\langle X_1|E_1\rangle = \langle Y_1|E_2\rangle$.

Let $E_1$ consist of recursive equations $X=t_X$ for $X\in \mathcal{X}$ and $E_2$
consists of recursion equations $Y=t_Y$ for $Y\in\mathcal{Y}$. Let the linear recursive specification $E$ consist of recursion equations $Z_{XY}=t_{XY}$, and
$\langle X|E_1\rangle\sim_s\langle Y|E_2\rangle$, and $t_{XY}$ consists of the following summands:

\begin{enumerate}
  \item $t_{XY}$ contains a summand $(a_1\leftmerge\cdots\leftmerge a_m)Z_{X'Y'}$ iff $t_X$ contains the summand $(a_1\leftmerge\cdots\leftmerge a_m)X'$ and $t_Y$ contains the
  summand $(a_1\leftmerge\cdots\leftmerge a_m)Y'$ such that $\langle X'|E_1\rangle\sim_s\langle Y'|E_2\rangle$;
  \item $t_{XY}$ contains a summand $b_1\leftmerge\cdots\leftmerge b_n$ iff $t_X$ contains the summand $b_1\leftmerge\cdots\leftmerge b_n$ and $t_Y$ contains the summand
  $b_1\leftmerge\cdots\leftmerge b_n$.
\end{enumerate}

Let $\sigma$ map recursion variable $X$ in $E_1$ to $\langle X|E_1\rangle$, and let $\pi$ map recursion variable $Z_{XY}$ in $E$ to $\langle X|E_1\rangle$. So,
$\sigma((a_1\leftmerge\cdots\leftmerge a_m)X')\equiv(a_1\leftmerge\cdots\leftmerge a_m)\langle X'|E_1\rangle\equiv\pi((a_1\leftmerge\cdots\leftmerge a_m)Z_{X'Y'})$, so by $RDP$, we
get $\langle X|E_1\rangle=\sigma(t_X)=\pi(t_{XY})$. Then by $RSP$, $\langle X|E_1\rangle=\langle Z_{XY}|E\rangle$, particularly, $\langle X_1|E_1\rangle=\langle Z_{X_1Y_1}|E\rangle$.
Similarly, we can obtain $\langle Y_1|E_2\rangle=\langle Z_{X_1Y_1}|E\rangle$. Finally, $\langle X_1|E_1\rangle=\langle Z_{X_1Y_1}|E\rangle=\langle Y_1|E_2\rangle$, as desired.

(2) If $\langle X_1|E_1\rangle \sim_{pp} \langle Y_1|E_2\rangle$ for linear recursive specification $E_1$ and $E_2$, then $\langle X_1|E_1\rangle = \langle Y_1|E_2\rangle$.

It can be proven similarly to (1), we omit it.

(3) If $\langle X_1|E_1\rangle \sim_{php} \langle Y_1|E_2\rangle$ for linear recursive specification $E_1$ and $E_2$, then $\langle X_1|E_1\rangle = \langle Y_1|E_2\rangle$.

It can be proven similarly to (1), we omit it.

(4) If $\langle X_1|E_1\rangle \sim_{phhp} \langle Y_1|E_2\rangle$ for linear recursive specification $E_1$ and $E_2$, then $\langle X_1|E_1\rangle = \langle Y_1|E_2\rangle$.

It can be proven similarly to (1), we omit it.
\end{proof}

\subsection{Abstraction}{\label{absg}}

To abstract away from the internal implementations of a program, and verify that the program exhibits the desired external behaviors, the silent step $\tau$ and abstraction operator
$\tau_I$ are introduced, where $I\subseteq \mathbb{E}\cup G_{at}$ denotes the internal events or guards. The silent step $\tau$ represents the internal events or guards, when we
consider the external behaviors of a process, $\tau$ steps can be removed, that is, $\tau$ steps must keep silent. The transition rule of $\tau$ is shown in Table \ref{TRForTauG}. In
the following, let the atomic event $e$ range over $\mathbb{E}\cup\{\epsilon\}\cup\{\delta\}\cup\{\tau\}$, and $\phi$ range over $G\cup \{\tau\}$, and let the communication function
$\gamma:\mathbb{E}\cup\{\tau\}\times \mathbb{E}\cup\{\tau\}\rightarrow \mathbb{E}\cup\{\delta\}$, with each communication involved $\tau$ resulting in $\delta$. We use $\tau(s)$ to
denote $effect(\tau,s)$, for the fact that $\tau$ only change the state of internal data environment, that is, for the external data environments, $s=\tau(s)$.

\begin{center}
    \begin{table}
        $$\frac{}{\tau\rightsquigarrow\breve{\tau}}$$
        $$\frac{}{\langle\tau,s\rangle\rightarrow\langle\surd,s\rangle}\textrm{ if }test(\tau,s)$$
        $$\frac{}{\langle\tau,s\rangle\xrightarrow{\tau}\langle\surd,\tau(s)\rangle}$$
        \caption{Transition rule of the silent step}
        \label{TRForTauG}
    \end{table}
\end{center}

\begin{definition}[Guarded linear recursive specification]\label{GLRSG}
A linear recursive specification $E$ is guarded if there does not exist an infinite sequence of $\tau$-transitions
$\langle X|E\rangle\xrightarrow{\tau}\langle X'|E\rangle\xrightarrow{\tau}\langle X''|E\rangle\xrightarrow{\tau}\cdots$, and there does not exist an infinite sequence of
$\epsilon$-transitions $\langle X|E\rangle\rightarrow\langle X'|E\rangle\rightarrow\langle X''|E\rangle\rightarrow\cdots$.
\end{definition}

\begin{theorem}[Conservitivity of $APPTC_G$ with silent step and guarded linear recursion]
$APPTC_G$ with silent step and guarded linear recursion is a conservative extension of $APPTC_G$ with linear recursion.
\end{theorem}

\begin{proof}
Since the transition rules of $APPTC_G$ with linear recursion are source-dependent, and the transition rules for silent step in Table \ref{TRForTauG} contain only a fresh constant
$\tau$ in their source, so the transition rules of $APPTC_G$ with silent step and guarded linear recursion is a conservative extension of those of $APPTC_G$ with linear recursion.
\end{proof}

\begin{theorem}[Congruence theorem of $APPTC_G$ with silent step and guarded linear recursion]
Probabilistic rooted branching truly concurrent bisimulation equivalences $\approx_{prbp}$, $\approx_{prbs}$, $\approx_{prbhp}$ and $\approx_{rbhhp}$ are all congruences with respect
to $APPTC_G$ with silent step and guarded linear recursion.
\end{theorem}

\begin{proof}
It follows the following three facts:
\begin{enumerate}
  \item in a guarded linear recursive specification, right-hand sides of its recursive equations can be adapted to the form by applications of the axioms in $APPTC_G$ and replacing
  recursion variables by the right-hand sides of their recursive equations;
  \item probabilistic truly concurrent bisimulation equivalences $\sim_{pp}$, $\sim_{ps}$, $\sim_{php}$ and $\sim_{phhp}$ are all congruences with respect to all operators of
  $APPTC_G$, while probabilistic truly concurrent bisimulation equivalences $\sim_{pp}$, $\sim_{ps}$, $\sim_{php}$ and $\sim_{phhp}$ imply the corresponding probabilistic rooted
  branching truly concurrent bisimulations $\approx_{prbp}$, $\approx_{prbs}$, $\approx_{prbhp}$ and $\approx_{prbhhp}$, so probabilistic rooted branching truly concurrent
  bisimulations $\approx_{prbp}$, $\approx_{prbs}$, $\approx_{prbhp}$ and $\approx_{prbhhp}$ are all congruences with respect to all operators of $APPTC_G$;
  \item While $\mathbb{E}$ is extended to $\mathbb{E}\cup\{\tau\}$, and $G$ is extended to $G\cup\{\tau\}$, it can be proved that probabilistic rooted branching truly concurrent
  bisimulations $\approx_{prbp}$, $\approx_{prbs}$, $\approx_{prbhp}$ and $\approx_{prbhhp}$ are all congruences with respect to all operators of $APPTC_G$, we omit it.
\end{enumerate}
\end{proof}

We design the axioms for the silent step $\tau$ in Table \ref{AxiomsForTauG}.

\begin{center}
\begin{table}
  \begin{tabular}{@{}ll@{}}
  \hline No. &Axiom\\
  $B1$ & $(y=y+y,z=z+z)\quad x\cdot((y+\tau\cdot(y+z))\boxplus_{\pi}w)=x\cdot((y+z)\boxplus_{\pi}w)$\\
  $B2$ & $(y=y+y,z=z+z)\quad x\leftmerge((y+\tau\leftmerge(y+z))\boxplus_{\pi}w)=x\leftmerge((y+z)\boxplus_{\pi}w)$\\
\end{tabular}
\caption{Axioms of silent step}
\label{AxiomsForTauG}
\end{table}
\end{center}

\begin{theorem}[Elimination theorem of $APPTC_G$ with silent step and guarded linear recursion]\label{ETTauG}
Each process term in $APPTC_G$ with silent step and guarded linear recursion is equal to a process term $\langle X_1|E\rangle$ with $E$ a guarded linear recursive specification.
\end{theorem}

\begin{proof}
By applying structural induction with respect to term size, each process term $t_1$ in $APPTC_G$ with silent step and guarded linear recursion generates a process can be expressed in the
form of equations

$t_i=((a_{1i11}\leftmerge\cdots\leftmerge a_{1i1i_1})t_{i1}+\cdots+(a_{1ik_i1}\leftmerge\cdots\leftmerge a_{1ik_ii_k})t_{ik_i}+(b_{1i11}\leftmerge\cdots\leftmerge b_{1i1i_1})+\cdots+
(b_{1il_i1}\leftmerge\cdots\leftmerge b_{1il_ii_l}))\boxplus_{\pi_1}\cdots\boxplus_{\pi_{m-1}}((a_{mi11}\leftmerge\cdots\leftmerge a_{mi1i_1})t_{i1}+\cdots+(a_{mik_i1}\leftmerge\cdots
\leftmerge a_{mik_ii_k})t_{ik_i}+(b_{mi11}\leftmerge\cdots\leftmerge b_{mi1i_1})+\cdots+(b_{mil_i1}\leftmerge\cdots\leftmerge b_{m1il_ii_l}))$

for $i\in\{1,\cdots,n\}$. Let the linear recursive specification $E$ consist of the recursive equations

$X_i=((a_{1i11}\leftmerge\cdots\leftmerge a_{1i1i_1})X_{i1}+\cdots+(a_{1ik_i1}\leftmerge\cdots\leftmerge a_{1ik_ii_k})X_{ik_i}+(b_{1i11}\leftmerge\cdots\leftmerge b_{1i1i_1})+\cdots+
(b_{1il_i1}\leftmerge\cdots\leftmerge b_{1il_ii_l}))\boxplus_{\pi_1}\cdots\boxplus_{\pi_{m-1}}((a_{mi11}\leftmerge\cdots\leftmerge a_{mi1i_1})X_{i1}+\cdots+(a_{mik_i1}\leftmerge\cdots
\leftmerge a_{mik_ii_k})X_{ik_i}+(b_{mi11}\leftmerge\cdots\leftmerge b_{mi1i_1})+\cdots+(b_{mil_i1}\leftmerge\cdots\leftmerge b_{mil_ii_l}))$

for $i\in\{1,\cdots,n\}$. Replacing $X_i$ by $t_i$ for $i\in\{1,\cdots,n\}$ is a solution for $E$, $RSP$ yields $t_1=\langle X_1|E\rangle$.
\end{proof}

\begin{theorem}[Soundness of $APPTC_G$ with silent step and guarded linear recursion]\label{SAPPTC_GTAUG}
Let $x$ and $y$ be $APPTC_G$ with silent step and guarded linear recursion terms. If $APPTC_G$ with silent step and guarded linear recursion $\vdash x=y$, then

(1) $x\approx_{prbs} y$.

(2) $x\approx_{prbp} y$.

(3) $x\approx_{prbhp} y$.

(4) $x\approx_{prbhhp} y$.
\end{theorem}

\begin{proof}
(1) Since probabilistic rooted branching step bisimulation $\approx_{prbs}$ is both an equivalent and a congruent relation with respect to $APPTC_G$ with silent step and guarded
linear recursion, we only need to check if each axiom in Table \ref{AxiomsForTauG} is sound modulo probabilistic rooted branching step bisimulation $\approx_{prbs}$. We leave them as
exercises to the readers.

(2) Since probabilistic rooted branching pomset bisimulation $\approx_{prbp}$ is both an equivalent and a congruent relation with respect to $APPTC_G$ with silent step and guarded
linear recursion, we only need to check if each axiom in Table \ref{AxiomsForTauG} is sound modulo probabilistic rooted branching pomset bisimulation $\approx_{prbp}$. We leave them
as exercises to the readers.

(3) Since probabilistic rooted branching hp-bisimulation $\approx_{prbhp}$ is both an equivalent and a congruent relation with respect to $APPTC_G$ with silent step and guarded linear
recursion, we only need to check if each axiom in Table \ref{AxiomsForTauG} is sound modulo probabilistic rooted branching hp-bisimulation $\approx_{prbhp}$. We leave them as exercises
to the readers.

(4) Since probabilistic rooted branching hhp-bisimulation $\approx_{prbhhp}$ is both an equivalent and a congruent relation with respect to $APPTC_G$ with silent step and guarded linear
recursion, we only need to check if each axiom in Table \ref{AxiomsForTauG} is sound modulo probabilistic rooted branching hhp-bisimulation $\approx_{prbhhp}$. We leave them as exercises
to the readers.
\end{proof}

\begin{theorem}[Completeness of $APPTC_G$ with silent step and guarded linear recursion]\label{CAPPTC_GTAUG}
Let $p$ and $q$ be closed $APPTC_G$ with silent step and guarded linear recursion terms, then,

(1) if $p\approx_{prbs} q$ then $p=q$.

(2) if $p\approx_{prbp} q$ then $p=q$.

(3) if $p\approx_{prbhp} q$ then $p=q$.

(3) if $p\approx_{prbhhp} q$ then $p=q$.
\end{theorem}

\begin{proof}
Firstly, by the elimination theorem of $APPTC_G$ with silent step and guarded linear recursion (see Theorem \ref{ETTauG}), we know that each process term in $APPTC_G$ with silent step
and guarded linear recursion is equal to a process term $\langle X_1|E\rangle$ with $E$ a guarded linear recursive specification.

It remains to prove the following cases.

(1) If $\langle X_1|E_1\rangle \approx_{prbs} \langle Y_1|E_2\rangle$ for guarded linear recursive specification $E_1$ and $E_2$, then $\langle X_1|E_1\rangle = \langle Y_1|E_2\rangle$.

Firstly, the recursive equation $W=\tau+\cdots+\tau$ with $W\nequiv X_1$ in $E_1$ and $E_2$, can be removed, and the corresponding summands $aW$ are replaced by $a$, to get $E_1'$ and
$E_2'$, by use of the axioms $RDP$, $A3$ and $B1$, and $\langle X|E_1\rangle = \langle X|E_1'\rangle$, $\langle Y|E_2\rangle = \langle Y|E_2'\rangle$.

Let $E_1$ consists of recursive equations $X=t_X$ for $X\in \mathcal{X}$ and $E_2$
consists of recursion equations $Y=t_Y$ for $Y\in\mathcal{Y}$, and are not the form $\tau+\cdots+\tau$. Let the guarded linear recursive specification $E$ consists of recursion
equations $Z_{XY}=t_{XY}$, and $\langle X|E_1\rangle\approx_{rbs}\langle Y|E_2\rangle$, and $t_{XY}$ consists of the following summands:

\begin{enumerate}
  \item $t_{XY}$ contains a summand $(a_1\leftmerge\cdots\leftmerge a_m)Z_{X'Y'}$ iff $t_X$ contains the summand $(a_1\leftmerge\cdots\leftmerge a_m)X'$ and $t_Y$ contains the
  summand $(a_1\leftmerge\cdots\leftmerge a_m)Y'$ such that $\langle X'|E_1\rangle\approx_{rbs}\langle Y'|E_2\rangle$;
  \item $t_{XY}$ contains a summand $b_1\leftmerge\cdots\leftmerge b_n$ iff $t_X$ contains the summand $b_1\leftmerge\cdots\leftmerge b_n$ and $t_Y$ contains the summand
  $b_1\leftmerge\cdots\leftmerge b_n$;
  \item $t_{XY}$ contains a summand $\tau Z_{X'Y}$ iff $XY\nequiv X_1Y_1$, $t_X$ contains the summand $\tau X'$, and $\langle X'|E_1\rangle\approx_{prbs}\langle Y|E_2\rangle$;
  \item $t_{XY}$ contains a summand $\tau Z_{XY'}$ iff $XY\nequiv X_1Y_1$, $t_Y$ contains the summand $\tau Y'$, and $\langle X|E_1\rangle\approx_{prbs}\langle Y'|E_2\rangle$.
\end{enumerate}

Since $E_1$ and $E_2$ are guarded, $E$ is guarded. Constructing the process term $u_{XY}$ consist of the following summands:

\begin{enumerate}
  \item $u_{XY}$ contains a summand $(a_1\leftmerge\cdots\leftmerge a_m)\langle X'|E_1\rangle$ iff $t_X$ contains the summand $(a_1\leftmerge\cdots\leftmerge a_m)X'$ and $t_Y$
  contains the summand $(a_1\leftmerge\cdots\leftmerge a_m)Y'$ such that $\langle X'|E_1\rangle\approx_{prbs}\langle Y'|E_2\rangle$;
  \item $u_{XY}$ contains a summand $b_1\leftmerge\cdots\leftmerge b_n$ iff $t_X$ contains the summand $b_1\leftmerge\cdots\leftmerge b_n$ and $t_Y$ contains the summand
  $b_1\leftmerge\cdots\leftmerge b_n$;
  \item $u_{XY}$ contains a summand $\tau \langle X'|E_1\rangle$ iff $XY\nequiv X_1Y_1$, $t_X$ contains the summand $\tau X'$, and
  $\langle X'|E_1\rangle\approx_{prbs}\langle Y|E_2\rangle$.
\end{enumerate}

Let the process term $s_{XY}$ be defined as follows:

\begin{enumerate}
  \item $s_{XY}\triangleq\tau\langle X|E_1\rangle + u_{XY}$ iff $XY\nequiv X_1Y_1$, $t_Y$ contains the summand $\tau Y'$, and $\langle X|E_1\rangle\approx_{prbs}\langle Y'|E_2\rangle$;
  \item $s_{XY}\triangleq\langle X|E_1\rangle$, otherwise.
\end{enumerate}

So, $\langle X|E_1\rangle=\langle X|E_1\rangle+u_{XY}$, and
$(a_1\leftmerge\cdots\leftmerge a_m)(\tau\langle X|E_1\rangle+u_{XY})=(a_1\leftmerge\cdots\leftmerge a_m)((\tau\langle X|E_1\rangle+u_{XY})+u_{XY})=(a_1\leftmerge\cdots\leftmerge a_m)(\langle X|E_1\rangle+u_{XY})=(a_1\leftmerge\cdots\leftmerge a_m)\langle X|E_1\rangle$,
hence, $(a_1\leftmerge\cdots\leftmerge a_m)s_{XY}=(a_1\leftmerge\cdots\leftmerge a_m)\langle X|E_1\rangle$.

Let $\sigma$ map recursion variable $X$ in $E_1$ to $\langle X|E_1\rangle$, and let $\pi$ map recursion variable $Z_{XY}$ in $E$ to $s_{XY}$. It is sufficient to prove
$s_{XY}=\pi(t_{XY})$ for recursion variables $Z_{XY}$ in $E$. Either $XY\equiv X_1Y_1$ or $XY\nequiv X_1Y_1$, we all can get $s_{XY}=\pi(t_{XY})$. So,
$s_{XY}=\langle Z_{XY}|E\rangle$ for recursive variables $Z_{XY}$ in $E$ is a solution for $E$. Then by $RSP$, particularly,
$\langle X_1|E_1\rangle=\langle Z_{X_1Y_1}|E\rangle$. Similarly, we can obtain $\langle Y_1|E_2\rangle=\langle Z_{X_1Y_1}|E\rangle$. Finally,
$\langle X_1|E_1\rangle=\langle Z_{X_1Y_1}|E\rangle=\langle Y_1|E_2\rangle$, as desired.

(2) If $\langle X_1|E_1\rangle \approx_{prbp} \langle Y_1|E_2\rangle$ for guarded linear recursive specification $E_1$ and $E_2$, then $\langle X_1|E_1\rangle = \langle Y_1|E_2\rangle$.

It can be proven similarly to (1), we omit it.

(3) If $\langle X_1|E_1\rangle \approx_{prbhb} \langle Y_1|E_2\rangle$ for guarded linear recursive specification $E_1$ and $E_2$, then $\langle X_1|E_1\rangle = \langle Y_1|E_2\rangle$.

It can be proven similarly to (1), we omit it.

(4) If $\langle X_1|E_1\rangle \approx_{prbhhb} \langle Y_1|E_2\rangle$ for guarded linear recursive specification $E_1$ and $E_2$, then $\langle X_1|E_1\rangle = \langle Y_1|E_2\rangle$.

It can be proven similarly to (1), we omit it.
\end{proof}

The unary abstraction operator $\tau_I$ ($I\subseteq \mathbb{E}\cup G_{at}$) renames all atomic events or atomic guards in $I$ into $\tau$. $APPTC_G$ with silent step and abstraction
operator is called $APPTC_{G_{\tau}}$. The transition rules of operator $\tau_I$ are shown in Table \ref{TRForAbstractionG}.

\begin{center}
    \begin{table}
        $$\frac{\langle x,s\rangle\rightsquigarrow \langle x',s\rangle}{\langle \tau_I(x),s\rangle\rightsquigarrow\langle\tau_I(x'),s\rangle}$$
        $$\frac{\langle x,s\rangle\xrightarrow{e}\langle\surd,s'\rangle}{\langle\tau_I(x),s\rangle\xrightarrow{e}\langle\surd,s'\rangle}\quad e\notin I
        \quad\quad\frac{\langle x,s\rangle\xrightarrow{e}\langle x',s'\rangle}{\langle\tau_I(x),s\rangle\xrightarrow{e}\langle\tau_I(x'),s'\rangle}\quad e\notin I$$

        $$\frac{\langle x,s\rangle\xrightarrow{e}\langle\surd,s'\rangle}{\langle\tau_I(x),s\rangle\xrightarrow{\tau}\langle\surd,\tau(s)\rangle}\quad e\in I
        \quad\quad\frac{\langle x,s\rangle\xrightarrow{e}\langle x',s'\rangle}{\langle\tau_I(x),s\rangle\xrightarrow{\tau}\langle\tau_I(x'),\tau(s)\rangle}\quad e\in I$$
        \caption{Transition rule of the abstraction operator}
        \label{TRForAbstractionG}
    \end{table}
\end{center}

\begin{theorem}[Conservitivity of $APPTC_{G_{\tau}}$ with guarded linear recursion]
$APPTC_{G_{\tau}}$ with guarded linear recursion is a conservative extension of $APPTC_G$ with silent step and guarded linear recursion.
\end{theorem}

\begin{proof}
Since the transition rules of $APPTC_G$ with silent step and guarded linear recursion are source-dependent, and the transition rules for abstraction operator in Table
\ref{TRForAbstractionG} contain only a fresh operator $\tau_I$ in their source, so the transition rules of $APPTC_{G_{\tau}}$ with guarded linear recursion is a conservative extension
of those of $APPTC_G$ with silent step and guarded linear recursion.
\end{proof}

\begin{theorem}[Congruence theorem of $APPTC_{G_{\tau}}$ with guarded linear recursion]
Probabilistic rooted branching truly concurrent bisimulation equivalences $\approx_{prbp}$, $\approx_{prbs}$, $\approx_{prbhp}$ and $\approx_{prbhhp}$ are all congruences with respect
to $APPTC_{G_{\tau}}$ with guarded linear recursion.
\end{theorem}

\begin{proof}
(1) It is easy to see that probabilistic rooted branching pomset bisimulation is an equivalent relation on $APPTC_{G_{\tau}}$ with guarded linear recursion terms, we only need to
prove that $\approx_{prbp}$ is preserved by the operators $\tau_I$. It is trivial and we leave the proof as an exercise for the readers.

(2) It is easy to see that probabilistic rooted branching step bisimulation is an equivalent relation on $APPTC_{G_{\tau}}$ with guarded linear recursion terms, we only need to
prove that $\approx_{prbs}$ is preserved by the operators $\tau_I$. It is trivial and we leave the proof as an exercise for the readers.

(3) It is easy to see that probabilistic rooted branching hp-bisimulation is an equivalent relation on $APPTC_{G_{\tau}}$ with guarded linear recursion terms, we only need to
prove that $\approx_{prbhp}$ is preserved by the operators $\tau_I$. It is trivial and we leave the proof as an exercise for the readers.

(4) It is easy to see that probabilistic rooted branching hhp-bisimulation is an equivalent relation on $APPTC_{G_{\tau}}$ with guarded linear recursion terms, we only need to
prove that $\approx_{prbhhp}$ is preserved by the operators $\tau_I$. It is trivial and we leave the proof as an exercise for the readers.
\end{proof}

We design the axioms for the abstraction operator $\tau_I$ in Table \ref{AxiomsForAbstractionG}.

\begin{center}
\begin{table}
  \begin{tabular}{@{}ll@{}}
\hline No. &Axiom\\
  $TI1$ & $e\notin I\quad \tau_I(e)=e$\\
  $TI2$ & $e\in I\quad \tau_I(e)=\tau$\\
  $TI3$ & $\tau_I(\delta)=\delta$\\
  $TI4$ & $\tau_I(x+y)=\tau_I(x)+\tau_I(y)$\\
  $PTI1$ & $\tau_I(x\boxplus_{\pi}y)=\tau_I(x)\boxplus_{\pi}\tau_I(y)$\\
  $TI5$ & $\tau_I(x\cdot y)=\tau_I(x)\cdot\tau_I(y)$\\
  $TI6$ & $\tau_I(x\leftmerge y)=\tau_I(x)\leftmerge\tau_I(y)$\\
  $G28$ & $\phi\notin I\quad \tau_I(\phi)=\phi$\\
  $G29$ & $\phi\in I\quad \tau_I(\phi)=\tau$\\
\end{tabular}
\caption{Axioms of abstraction operator}
\label{AxiomsForAbstractionG}
\end{table}
\end{center}

\begin{theorem}[Soundness of $APPTC_{G_{\tau}}$ with guarded linear recursion]\label{SAPPTC_GABSG}
Let $x$ and $y$ be $APPTC_{G_{\tau}}$ with guarded linear recursion terms. If $APPTC_{G_{\tau}}$ with guarded linear recursion $\vdash x=y$, then

(1) $x\approx_{prbs} y$.

(2) $x\approx_{prbp} y$.

(3) $x\approx_{prbhp} y$.

(4) $x\approx_{prbhhp} y$.
\end{theorem}

\begin{proof}
(1) Since probabilistic rooted branching step bisimulation $\approx_{prbs}$ is both an equivalent and a congruent relation with respect to $APPTC_{G_{\tau}}$ with guarded linear
recursion, we only need to check if each axiom in Table \ref{AxiomsForAbstractionG} is sound modulo probabilistic rooted branching step bisimulation $\approx_{prbs}$. We leave them as
exercises to the readers.

(2) Since probabilistic rooted branching pomset bisimulation $\approx_{prbp}$ is both an equivalent and a congruent relation with respect to $APPTC_{G_{\tau}}$ with guarded linear
recursion, we only need to check if each axiom in Table \ref{AxiomsForAbstractionG} is sound modulo probabilistic rooted branching pomset bisimulation $\approx_{prbp}$. We leave them
as exercises to the readers.

(3) Since probabilistic rooted branching hp-bisimulation $\approx_{prbhp}$ is both an equivalent and a congruent relation with respect to $APPTC_{G_{\tau}}$ with guarded linear
recursion, we only need to check if each axiom in Table \ref{AxiomsForAbstractionG} is sound modulo probabilistic rooted branching hp-bisimulation $\approx_{prbhp}$. We leave them as
exercises to the readers.

(4) Since probabilistic rooted branching hhp-bisimulation $\approx_{prbhhp}$ is both an equivalent and a congruent relation with respect to $APPTC_{G_{\tau}}$ with guarded linear
recursion, we only need to check if each axiom in Table \ref{AxiomsForAbstractionG} is sound modulo probabilistic rooted branching hhp-bisimulation $\approx_{prbhhp}$. We leave them as
exercises to the readers.
\end{proof}

Though $\tau$-loops are prohibited in guarded linear recursive specifications in a specifiable way, they can be constructed using the abstraction operator, for example, there exist
$\tau$-loops in the process term $\tau_{\{a\}}(\langle X|X=aX\rangle)$. To avoid $\tau$-loops caused by $\tau_I$ and ensure fairness, we introduce the following recursive verification
rules as Table \ref{RVR} shows, note that $i_1,\cdots, i_m,j_1,\cdots,j_n\in I\subseteq \mathbb{E}\setminus\{\tau\}$.

\begin{center}
\begin{table}
    $$VR_1\quad \frac{x=y+(i_1\leftmerge\cdots\leftmerge i_m)\cdot x, y=y+y}{\tau\cdot\tau_I(x)=\tau\cdot \tau_I(y)}$$
    $$VR_2\quad \frac{x=z\boxplus_{\pi}(u+(i_1\leftmerge\cdots\leftmerge i_m)\cdot x),z=z+u,z=z+z}{\tau\cdot\tau_I(x)=\tau\cdot\tau_I(z)}$$
    $$VR_3\quad \frac{x=z+(i_1\leftmerge\cdots\leftmerge i_m)\cdot y,y=z\boxplus_{\pi}(u+(j_1\leftmerge\cdots\leftmerge j_n)\cdot x), z=z+u,z=z+z}{\tau\cdot\tau_I(x)=\tau\cdot\tau_I(y')\textrm{ for }y'=z\boxplus_{\pi}(u+(i_1\leftmerge\cdots\leftmerge i_m)\cdot y')}$$
\caption{Recursive verification rules}
\label{RVR}
\end{table}
\end{center}

\begin{theorem}[Soundness of $VR_1,VR_2,VR_3$]
$VR_1$, $VR_2$ and $VR_3$ are sound modulo probabilistic rooted branching truly concurrent bisimulation equivalences $\approx_{prbp}$, $\approx_{prbs}$, $\approx_{prbhp}$ and $\approx_{prbhhp}$.
\end{theorem}

\subsection{Hoare Logic for $APPTC_G$}{\label{hl}}

In this section, we introduce Hoare logic for $APPTC_G$. We do not introduce the preliminaries of Hoare logic, please refer to \cite{HL} for details.

A partial correct formula has the form

$$\{pre\}P\{post\}$$

where $pre$ are preconditions, $post$ are postconditions, and $P$ are programs. $\{pre\}P\{post\}$ means that $pre$ hold, then $P$ are executed and $post$ hold. We take the guards
$G$ of $APPTC_G$ as the language of conditions, and closed terms of $APPTC_G$ as programs. For some condition $\alpha\in G$ and some data state $s\in S$, we denote
$S\models \alpha[s]$ for $\langle \alpha,s\rangle\rightarrow\langle\surd,s\rangle$, and $S\models\alpha$ for $\forall s\in S, S\models \alpha[s]$, $S\models \{\alpha\}p\{\beta\}$ for
all $s\in S$, $\nu\subseteq \mathbb{E}\cup G$, $S\models \alpha[s]$, $\langle p,s\rangle\xrightarrow{\nu}\langle p',s'\rangle$, $S\models\beta[s']$ with $s'\in S$. It is obvious that
$S\models \{\alpha\}p\{\beta\}\Leftrightarrow \alpha p\approx_{prbp}(\approx_{prbs},\approx_{prbhp},\approx_{prbhhp})\alpha p\beta$.

We design a proof system $H$ to deriving partial correct formulas over terms of $APPTC_G$ as Table \ref{H} shows. Let $\Gamma$ be a set of conditions and partial correct formulas, we
denote $\Gamma\vdash\{\alpha\}t\{\beta\}$ iff we can derive $\{\alpha\}t\{\beta\}$ in $H$, note that $t$ does not need to be closed terms. And we write $\alpha\rightarrow \beta$ for
$S\models \alpha\Rightarrow S\models\beta$.

\begin{center}
    \begin{table}
        $(H1)\quad\{wp(e,\alpha)\}e\{\alpha\}\textrm{ if }e\in\mathbb{E}$

        $(H2)\quad\{\alpha\}\phi\{\alpha\cdot\phi\}\textrm{ if }\phi\in G$

        $(H3)\quad\frac{\{\alpha\}t\{\beta\}\quad\{\alpha\}t'\{\beta\}}{\{\alpha\}t+t'\{\beta\}}$

        $(PH1)\quad\frac{\{\alpha\}t\{\beta\}\quad\{\alpha\}t'\{\beta\}}{\{\alpha\}t\boxplus_{\pi}t'\{\beta\}}$

        $(H4)\quad\frac{\{\alpha\}t\{\alpha'\}\quad\{\alpha'\}t'\{\beta\}}{\{\alpha\}t\cdot t'\{\beta\}}$

        $(H5)\quad\frac{\{\alpha\}t\{\alpha'\}\quad\{\beta\}t'\{\beta'\}}{\{\alpha\leftmerge\beta\}t\between t'\{\alpha'\leftmerge\beta'\}}$

        $(H6)\quad\frac{\{\alpha\}t\{\beta\}}{\{\alpha\}\Theta(t)\{\beta\}}$

        $(H7)\quad\frac{\{\alpha\}t\{\beta\}}{\{\alpha\}\partial_H(t)\{\beta\}}$

        $(H8)\quad\frac{\{\alpha\}t\{\beta\}}{\{\alpha\}\tau_I(t)\{\beta\}}$

        $(H9)\quad\frac{\alpha\rightarrow\alpha'\quad\{\alpha'\}t\{\beta'\}\quad\beta'\rightarrow\beta}{\{\alpha\}t\{\beta\}}$

        $(H10)\quad$ For $E=\{x=t_x|x\in V_E\}$ a guarded linear recursive specification $\forall y\in V_E$ and $z\in V_E$:
        \[\frac{\{\alpha_x\}t_x\{\beta_x\}\quad\cdots\quad\{\alpha_y\}t_y\{\beta_y\}}{\{\alpha_z\}\langle z|E\rangle\{\beta_z\}}\]

        $(H10')\quad$ For $E=\{x=t_x|x\in V_E\}$ a guarded linear recursive specification $\forall y\in V_E$ and $z\in V_E$:
        \[\frac{\{\alpha_{x_1}\leftmerge\cdots\leftmerge\alpha_{x_{nx}}\}t_x\{\beta_{x_1}\leftmerge\cdots\leftmerge\beta_{x_{nx}}\}\quad\cdots\quad\{\alpha_{y_1}\leftmerge\cdots \leftmerge\alpha_{y_{ny}}\}t_y\{\beta_{y_1}\leftmerge\cdots\leftmerge\beta_{y_{ny}}\}}{\{\alpha_{z_1}\leftmerge\cdots\leftmerge\alpha_{z_{nz}}\}\langle z|E\rangle\{\beta_{z_1}\leftmerge\cdots\leftmerge\beta_{z_{nz}}\}}\]
        \caption{The proof system $H$}
        \label{H}
    \end{table}
\end{center}

\begin{theorem}[Soundness of $H$]
Let $Tr_S$ be the set of conditions that hold in $S$. Let $p$ be a closed term of $APPTC_{G_{\tau}}$ with guarded linear recursion and $VR_1,VR_2,VR_3$, and $\alpha,\beta\in G$ be guards. Then

\begin{eqnarray}
Tr_S\vdash\{\alpha\}p\{\beta\}&\Rightarrow&APPTC_{G_{\tau}}\textrm{ with guarded linear recursion and } VR_1,VR_2,VR_3\vdash \alpha p=\alpha p\beta\nonumber\\
&\Leftrightarrow&\alpha p\approx_{prbs}(\approx_{prbp},\approx_{prbhp},\approx_{prbhhp})\alpha p\beta\nonumber\\
&\Leftrightarrow&S\models\{\alpha\}p\{\beta\}\nonumber
\end{eqnarray}
\end{theorem}

\begin{proof}
We only need to prove

$$Tr_S\vdash\{\alpha\}p\{\beta\}\Rightarrow APPTC_{G_{\tau}}\textrm{ with guarded linear recursion and } VR_1,VR_2,VR_3\vdash \alpha p=\alpha p\beta$$

For $H1$-$H10$, by induction on the length of derivation, the soundness of $H1$-$H10$ are straightforward. We only prove the soundness of $H10'$.

%For the proof of $H10$ and $H10'$, it is similar to the soundness proof in \cite{HLPA}, the only difference is that the guards maybe in parallel, and we do not repeat any more, please refer to \cite{HLPA} for details.

Let $E=\{x_i=t_i(x_1,\cdots,x_n)|i=1,\cdots,n\}$ be a guarded linear recursive specification. Assume that

$$Tr_S,\{\{\alpha_1\leftmerge\cdots\leftmerge\alpha_{n_i}\}x_i\{\beta_1\leftmerge\cdots \leftmerge\beta_{n_i}\}|i=1,\cdots,n\}\vdash\{\alpha_1\leftmerge\cdots\leftmerge\alpha_{n_j}\}t_j(x_1,\cdots,x_n)\{\beta_1\leftmerge\cdots \leftmerge\beta_{n_j}\}$$

for $j=1,\cdots,n$. We would show that $APPTC_{G_{\tau}}\textrm{ with guarded linear recursion and } VR_1,VR_2,VR_3\vdash (\alpha_1\leftmerge\cdots\leftmerge\alpha_{n_j})X_j=(\alpha_1\leftmerge\cdots\leftmerge\alpha_{n_j}) X_j(\beta_1\leftmerge\cdots\leftmerge\beta_{n_j})$.

We write recursive specifications $E'$ and $E''$ for

$$E'=\{y_i=(\alpha_1\leftmerge\cdots\leftmerge\alpha_{n_i})t_i(y_1,\cdots,y_n)|i=1,\cdots,n\}$$

$$E''=\{z_i=(\alpha_1\leftmerge\cdots\leftmerge\alpha_{n_i})t_i(z_1(\beta_1\leftmerge\cdots\leftmerge\beta_{n_1}),\cdots,z_n(\beta_1\leftmerge\cdots\leftmerge\beta_{n_n}))|i=1,\cdots,n\}$$

and would show that for $j=1,\cdots,n$,

(1) $(\alpha_1\leftmerge\cdots\leftmerge\alpha_{n_j})X_j=Y_j$;

(2) $Z_j(\beta_1\leftmerge\cdots\leftmerge\beta_{n_j})=Z_j$;

(3) $Z_j=Y_j$.

For (1), we have

\begin{eqnarray}
(\alpha_1\leftmerge\cdots\leftmerge\alpha_{n_j})X_j&=&(\alpha_1\leftmerge\cdots\leftmerge\alpha_{n_j})t_j(X_1,\cdots,X_n)\nonumber\\
&=&(\alpha_1\leftmerge\cdots\leftmerge\alpha_{n_j})t_j((\alpha_1\leftmerge\cdots\leftmerge\alpha_{n_1})X_1,\cdots,(\alpha_1\leftmerge\cdots\leftmerge\alpha_{n_n})X_n)\nonumber
\end{eqnarray}

by RDP, we have $(\alpha_1\leftmerge\cdots\leftmerge\alpha_{n_j})X_j=Y_j$.

For (2), we have

\begin{eqnarray}
Z_j(\beta_1\leftmerge\cdots\leftmerge\beta_{n_j})&=&(\alpha_1\leftmerge\cdots\leftmerge\alpha_{n_j}) t_j(Z_1(\beta_1\leftmerge\cdots\leftmerge\beta_{n_1}),\cdots,Z_n(\beta_1\leftmerge\cdots \leftmerge\beta_{n_n}))(\beta_1\leftmerge\cdots\leftmerge\beta_{n_j})\nonumber\\
&=&(\alpha_1\leftmerge\cdots\leftmerge\alpha_{n_j}) t_j(Z_1(\beta_1\leftmerge\cdots\leftmerge\beta_{n_1}),\cdots,Z_n(\beta_1\leftmerge\cdots \leftmerge\beta_{n_n}))\nonumber\\
&=&(\alpha_1\leftmerge\cdots\leftmerge\alpha_{n_j}) t_j((Z_1(\beta_1\leftmerge\cdots\leftmerge\beta_{n_1}))(\beta_1\leftmerge\cdots\leftmerge\beta_{n_1}),\cdots,(Z_n(\beta_1\leftmerge\cdots \leftmerge\beta_{n_n}))\nonumber\\
&&(\beta_1\leftmerge\cdots\leftmerge\beta_{n_n}))\nonumber
\end{eqnarray}

by RDP, we have $Z_j(\beta_1\leftmerge\cdots\leftmerge\beta_{n_j})=Z_j$.

For (3), we have

\begin{eqnarray}
Z_j&=&(\alpha_1\leftmerge\cdots\leftmerge\alpha_{n_j})t_j(Z_1(\beta_1\leftmerge\cdots\leftmerge\beta_{n_1}),\cdots,Z_n(\beta_1\leftmerge\cdots\leftmerge\beta_{n_n}))\nonumber\\
&=&(\alpha_1\leftmerge\cdots\leftmerge\alpha_{n_j})t_j(Z_1,\cdots,Z_n)\nonumber
\end{eqnarray}

by RDP, we have $Z_j=Y_j$.
\end{proof}

%\input{section6/section6.7.tex}

%\newpage\input{section7.tex}
%\newpage\input{section8.tex}


\newpage\begin{thebibliography}{Lam94}
\bibitem{CC}R. Milner. (1989). Communication and concurrency. Printice Hall.

\bibitem{CCS}R. Milner. (1980). A calculus of communicating systems. LNCS 92, Springer.

\bibitem{ACP}W. Fokkink. (2007). Introduction to process algebra 2nd ed. Springer-Verlag.

%\bibitem{CFAR}F.W. Vaandrager. (1986). Verification of two communication protocols by means of process algebra. Report CS-R8608, CWI, Amsterdam.

%\bibitem{ABP}K.A. Bartlett, R.A. Scantlebury, and P.T. Wilkinson. (1969). A note on reliable full-duplex transmission over half-duplex links. Communications of the ACM, 12(5):260-261.

\bibitem{PI1}R. Milner, J. Parrow, and D. Walker. (1992). A Calculus of Mobile Processes, Part I. Information and Computation, 100(1):1-40.

\bibitem{PI2}R. Milner, J. Parrow, and D. Walker. (1992). A calculus of mobile processes, Part II. Information and Computation, 100(1):41-77.

\bibitem{CTC}Y. Wang. (2017). A calculus for true concurrency. Manuscript, arxiv: 1703.00159.

\bibitem{ATC}Y. Wang. (2016). Algebraic laws for true concurrency. Manuscript, arXiv: 1611.09035.

\bibitem{PITC}Y. Wang. (2017). A calculus of truly concurrent mobile processes. Manuscript, arXiv:1704.07774.

\bibitem{PPA}S. Andova. (2002). Probabilistic process algebra. Annals of Operations Research 128(2002):204-219.

\bibitem{PPA2}S. Andova, J. Baeten, T. Willemse. (2006). A Complete Axiomatisation of Branching Bisimulation for Probabilistic Systems with an Application in Protocol
Verification. International Conference on Concurrency Theory. Springer Berlin Heidelberg.

\bibitem{PPA3}S. Andova, S. Georgievska. (2009). On Compositionality, Efficiency, and Applicability of Abstraction in Probabilistic Systems. Conference on Current Trends in Theory and
Practice of Computer Science. Springer-Verlag.

\bibitem{HL}C.A.R. Hoare. (1969). An axiomatic basis for computer programming. Communications of the ACM, 12(10).

\bibitem{HLPA}J. F. Groote, A. Ponse. (1994). Process algebra with guards: combining hoare logic with process algebra. Formal Aspects of Computing, 6(2): 115-164.

\bibitem{ILM}F. Moller. (1990). The importance of the left merge operator in process algebras. In M.S. Paterson, ed., Proceedings 17th Colloquium on Automata, Languages and Programming (ICALP'90), Warwick, LNCS 443, 752-764. Springer.
\end{thebibliography}
\end{document}